%% file: main.tex
\documentclass[11pt]{article}
\usepackage{lmodern} 
\usepackage[T1]{fontenc}%

\usepackage{color}
\usepackage[colorlinks=true,linkcolor=blue,citecolor=blue]{hyperref}
\usepackage{amsmath, amssymb, amsthm}
\usepackage{mathtools}
\usepackage[margin=1in]{geometry}
\usepackage{graphics}
\usepackage{pifont}
\usepackage{tikz}
\usepackage{bbm, bm}
\DeclareMathOperator*{\argmax}{arg\,max}
\DeclareMathOperator*{\argmin}{arg\,min}
\usetikzlibrary{arrows.meta}
\usepackage{environ}
\usepackage{framed}
\usepackage{url}
\usepackage[linesnumbered,ruled,vlined]{algorithm2e}
\usepackage[noend]{algpseudocode}
\usepackage[labelfont=bf]{caption}
\usepackage{framed}
\usepackage[framemethod=tikz]{mdframed}
\usepackage{appendix}
\usepackage{graphicx}
\usepackage[textsize=tiny]{todonotes}
\usepackage{tcolorbox}
\allowdisplaybreaks[1]
\usepackage{nicefrac}
\usepackage{thm-restate}
\usepackage[noabbrev,capitalize,nameinlink]{cleveref}
\crefname{equation}{}{}

\usepackage{subcaption}

\DontPrintSemicolon
\SetKw{KwAnd}{and}
\SetProcNameSty{textsc}
\SetFuncSty{textsc}

\newcommand\remove[1]{}

\newtheorem{lemma}{Lemma}[section]
\newtheorem*{lemma*}{Lemma}
\newtheorem{theorem}[lemma]{Theorem}
\newtheorem{corollary}[lemma]{Corollary}
\newtheorem*{corollary*}{Corollary}
\newtheorem{claim}[lemma]{Claim}

\newtheorem{itheorem}[lemma]{Informal Theorem}
\newtheorem*{theorem*}{Theorem}
\newtheorem*{inducthyp*}{Inductive Hypothesis}
\newtheorem*{definition*}{Definition}
\newtheorem{definition}[lemma]{Definition}
\newtheorem{rem}[lemma]{Remark}
\newtheorem{assump}[lemma]{Assumption}

\newtheorem*{rem*}{Remark}

\newcommand\R{\mathbb{R}}

\newcommand\E{\mathbb{E}}

\newcommand{\eps}{\varepsilon}
\renewcommand{\O}{\widetilde{O}}
\newcommand{\pe}{\preceq}

\renewcommand{\l}{\langle}
\renewcommand{\r}{\rangle}

\newcommand{\assign}{\leftarrow}

\renewcommand{\forall}{\mathrm{\text{ for all }}}

\newcommand{\g}{\nabla}

\newcommand{\vone}{\mathbf{1}}

\newcommand{\cost}{\mathsf{cost}}
\newcommand{\vol}{\mathrm{vol}}
\newcommand{\econg}{\mathrm{econg}}
\newcommand{\vcong}{\mathrm{vcong}}
\newcommand{\length}{\mathrm{length}}
\newcommand{\proj}{\mathrm{proj}}

\newcommand{\pconcat}{\oplus}

\newcommand{\bc}{\bm{c}}
\newcommand{\bd}{\boldsymbol{d}}
\renewcommand{\bf}{\bm{f}}
\newcommand{\bg}{\boldsymbol{g}}
\newcommand{\bp}{\boldsymbol{p}}

\newcommand{\bw}{\boldsymbol{w}}
\newcommand{\bss}{\boldsymbol{s}}
\newcommand{\bu}{\boldsymbol{u}}
\newcommand{\bv}{\boldsymbol{v}}
\newcommand{\bz}{\boldsymbol{z}}
\newcommand{\bx}{\boldsymbol{x}}
\newcommand{\by}{\boldsymbol{y}}
\newcommand{\bell}{\boldsymbol{\ell}}

\newcommand{\bDelta}{\boldsymbol{\Delta}}
\newcommand{\blambda}{\boldsymbol{\lambda}}
\newcommand{\diag}{\mathrm{diag}}
\renewcommand{\root}{\mathsf{root}}

\newcommand{\Enc}{\textsc{Enc}}

\newcommand{\DynamicBranchingChain}{\textsc{DynamicBranchingChain}}
\newcommand{\DynamicSparseCore}{\textsc{DynamicSparseCore}}
\newcommand{\DynamicCore}{\textsc{DynamicCore}}
\newcommand{\Rebuild}{\textsc{Rebuild}}
\newcommand{\Initialize}{\textsc{Initialize}}

\renewcommand{\hat}{\widehat}
\renewcommand{\tilde}{\widetilde}

\DeclareFontFamily{U}{mathb}{\hyphenchar\font45}
\DeclareFontShape{U}{mathb}{m}{n}{<5> <6> <7> <8> <9> <10> gen * mathb
<10.95> mathb10 <12> <14.4> <17.28> <20.74> <24.88> mathb12}{}
\DeclareSymbolFont{mathb}{U}{mathb}{m}{n}
\DeclareMathSymbol{\rcirclearrow}{\mathbin}{mathb}{'367}

\newcommand{\wt}{\widetilde}
\newcommand{\wh}{\widehat}

\newcommand{\barDelta}{\overline{\bDelta}}
\newcommand{\barbd}{\overline{\bd}}
\renewcommand{\bar}{\overline}

\newcommand{\init}{\mathrm{(init)}}

\foreach \x in {A,...,Z}{%
	\expandafter\xdef\csname m\x\endcsname{\noexpand\mathbf{\x}}
}

\foreach \x in {A,...,Z}{%
	\expandafter\xdef\csname c\x\endcsname{\noexpand\mathcal{\x}}
}

\newif\ifrandom
\randomtrue

\newcommand{\defeq}{\stackrel{\mathrm{\scriptscriptstyle def}}{=}}

\newcommand{\poly}{{\mathrm{poly}}}
\newcommand{\Cong}{{\mathsf{cong}}}

\newcommand{\str}{{\mathsf{str}}}

\newcommand{\wstr}{{\wt{\str}}}

\newcommand{\len}{{\mathrm{len}}}
\newcommand{\wlen}{{\wt{\len}}}

\newcommand{\rev}{\mathsf{rev}}

\newcommand{\lvl}{\mathsf{level}}
\newcommand{\prev}{\mathsf{prev}}

\newcommand{\bb}{\boldsymbol{b}}
\newcommand{\ba}{\boldsymbol{a}}
\newcommand{\bPi}{\boldsymbol{\Pi}}
\newcommand{\last}{\mathsf{last}}
\newcommand{\bzeta}{\bar{\zeta}}

\renewcommand{\l}{\langle}
\renewcommand{\r}{\rangle}

\newcommand{\norm}[1]{\left\lVert#1\right\rVert}
\newcommand{\Abs}[1]{\left|#1\right|}
\newcommand{\Set}[2]{\left\{#1 ~\middle\vert~ #2\right\}}
\newcommand{\setof}[1]{\left\{#1\right\}}

\newcommand{\todolater}[1]{}

\newcommand{\urgent}[1]{\textbf{\color{red}[URGENT: FIX ME: #1]}}
\newcommand{\sushant}[1]{\textbf{\color{orange}[Sushant: #1]}}
\newcommand{\rasmus}[1]{\textbf{\color{orange}[Rasmus: #1]}}
\newcommand{\rklow}[1]{\textbf{\color{orange}[Rasmus: #1]}}
\newcommand{\rp}[1]{\textbf{\color{blue}[Richard: #1]}}
\newcommand{\mprobst}[1]{\textbf{\color{olive}[Max: #1]}}
\newcommand{\yang}[1]{\textbf{\color{red}[Yang: #1]}}
\newcommand{\li}[1]{\textbf{\color{orange}[Li: #1]}}
\renewcommand{\todo}[1]{\textbf{\color{violet}[TODO: #1]}}

\renewcommand{\urgent}[1]{}
\renewcommand{\sushant}[1]{}
\renewcommand{\rasmus}[1]{}
\renewcommand{\rklow}[1]{}
\renewcommand{\rp}[1]{}
\renewcommand{\mprobst}[1]{}
\renewcommand{\yang}[1]{}
\renewcommand{\li}[1]{}
\renewcommand{\todo}[1]{}

\renewcommand{\SS}{\mathcal{S}}
\newcommand{\abs}[1]{\left|#1\right|}

\DeclareUnicodeCharacter{2113}{$\ell$}

\usepackage[backend=biber, isbn=false, style=alphabetic, backref=true, doi=false, url=false, maxcitenames=10, mincitenames=5, maxalphanames=10, maxbibnames=10, minbibnames=5, minalphanames=5, defernumbers=true, sortlocale=en_US]{biblatex}
\begin{document}
\pagenumbering{gobble}

\title{Maximum Flow and Minimum-Cost Flow
in Almost-Linear Time}

\author{
Li Chen\thanks{Li Chen was supported by NSF Grant CCF-2106444.}\\ Georgia Tech\\ lichen@gatech.edu
\and
Rasmus Kyng\thanks{The research leading to these results has received funding from the grant ``Algorithms and complexity for high-accuracy flows and convex optimization'' (no. 200021 204787) of the Swiss National Science Foundation.}\\ ETH Zurich \\ kyng@inf.ethz.ch 
\and
Yang P. Liu\thanks{Yang P. Liu was supported by NSF CAREER Award CCF-1844855
and NSF Grant CCF-1955039.} \\ Stanford University \\ yangpliu@stanford.edu
\and
Richard Peng \\ University of Waterloo
\footnote{Part of this work was done while at the Georgia Institute of Technology.}\\
y5peng@uwaterloo.ca
\and
Maximilian Probst Gutenberg\footnotemark[2]\\ ETH Zurich\\ maxprobst@ethz.ch
\and
Sushant Sachdeva\thanks{Sushant Sachdeva's research is supported by an NSERC (Natural Sciences and Engineering Research Council of Canada) Discovery Grant.
} \\ University of Toronto \\ sachdeva@cs.toronto.edu
}

\maketitle

\begin{abstract}
We give an algorithm that computes exact maximum flows and minimum-cost flows on directed graphs with $m$ edges and polynomially bounded integral demands, costs, and capacities in $m^{1+o(1)}$ time. Our algorithm builds the flow through a sequence of $m^{1+o(1)}$ approximate undirected minimum-ratio cycles, each of which is computed and processed in amortized $m^{o(1)}$ time using a new dynamic graph data structure.

Our framework extends to algorithms running in $m^{1+o(1)}$ time for computing flows that minimize general edge-separable convex functions to high accuracy. This gives almost-linear time algorithms for several problems including entropy-regularized optimal transport, matrix scaling, $p$-norm flows, and $p$-norm isotonic regression on arbitrary directed acyclic graphs.
\end{abstract}

\newpage
\setcounter{tocdepth}{2}
\tableofcontents
\listoffigures
\listoftables

\normalsize
\pagebreak
\pagenumbering{arabic}

\input{intro}

\input{overview}

\input{prelim}
\input{ipm}
\input{spanner}

\input{jtree}

\input{combine}
\input{generalconvex}

\printbibliography[heading=bibintoc]

\appendix

\input{previous}

\input{proofs}

\input{scaling}

\input{applications}

\end{document}

%% file: intro.tex
\section{Introduction}
\label{sec:intro}

The maximum flow problem and its generalization, the minimum-cost flow problem, are classic combinatorial graph problems that find numerous applications in engineering and scientific computing.
These problems have been studied extensively over the last seven decades, starting from the work of Dantzig and Ford-Fulkerson, and several important algorithmic problems can be reduced to min-cost flows (e.g. max-weight bipartite matching, min-cut, Gomory-Hu trees).
The origin of numerous significant algorithmic developments such as the simplex method, graph sparsification, and link-cut trees, can be traced back to seeking faster algorithms for max-flow and min-cost flow.

Formally, we are given a directed graph $G=(V, E)$ with $|V|=n$ vertices and $|E|=m$ edges,  upper/lower edge capacities $\bu^+, \bu^- \in \R^E,$ edge costs $\bc \in \R^E,$ and vertex demands $\bd \in \R^V$ with $\sum_{v \in V} \bd_v = 0$. Our goal is to find a flow $\bf \in \R^E$ of minimum cost $\bc^{\top} \bf$ that respects edge capacities $\bu^-_e \le \bf_e \le \bu_e^+$ and satisfies vertex demands $\bd.$
The vertex demand constraints are succinctly captured as 
$\mB^\top \bf = \bd,$ where $\mB \in \R^{E \times V}$ is the edge-vertex incidence matrix defined as $\mB_{((a,b),v)}$ is 1 if $v=a,$ $-1$ if $v=b,$ and $0$ otherwise. 
To compare running times, we assume that all $\bu^+_e, \bu^-_e, \bc_e$ and $\bd_v$ are integral, and $\abs{\bu^+_e}, \abs{\bu^-_e} \le U$ and    $\abs{\bc_e} \le C.$

There has been extensive  work on max-flow and min-cost flow.
While we defer a longer discussion of the related works to~\cref{sec:maxflowprevious}, a brief discussion will help place our work in context.
Starting from the first pseudo-polynomial time algorithm by Dantzig~\cite{D51}
that ran in $O(mn^2U)$ time, the approach
to designing faster flow algorithms was primarily combinatorial,
working with various adaptations of augmenting paths, cycle cancelling, blocking flows, and capacity/cost scaling.
A long line of work led to a running time of $\O(m\min\{m^{\nicefrac{1}{2}}, n^{\nicefrac{2}{3}}\}\log U)$~\cite{HK73, K73, ET75, GR98} for max-flow, and $\O(mn\log U)$~\cite{GT87} for min-cost flow. These bounds stood for decades.

In their breakthrough work on solving Laplacian systems and computing electrical flows, Spielman and Teng~\cite{ST04} introduced the idea of combining continuous optimization primitives with graph-theoretic constructions for designing flow algorithms. This is often referred to as the \emph{Laplacian Paradigm}.
Daitch and Spielman~\cite{DS08} demonstrated the power of this paradigm by combining Interior Point methods (IPMs) with fast Laplacian systems solvers to achieve an $\O(m^{1.5}\log^2 U)$ time algorithm for min-cost flow, the first progress in two decades.
A key advantage of IPMs is that they reduce flow problems on directed graphs to problems on undirected graphs, which are
easier to work with.
The Laplacian paradigm achieved several successes, including $\O(m\epsilon^{-1})$ time $(1+\epsilon)$-approximate undirected max-flow and multicommodity flow~\cite{CKMST11, KLOS14, S13, P16, S17}, 
and an $m^{\nicefrac{4}{3}+o(1)}U^{\nicefrac{1}{3}}$ time algorithm for
bipartite matching and
unit capacity max-flow~\cite{M13, M16, LS20, KLS20, AMV20},
and $pm^{1+o(1)}$ time unweighted $p$-norm minimizing flow for large $p$~\cite{KPSW19, AS20}.

Efficient graph data-structures have played a key role in the development of faster algorithms for flow problems, e.g. dynamic trees~\cite{ST83}.
Recently, the development of special-purpose data-structures for efficient implementation of IPM-based algorithms has led to progress on min-cost flow for some cases -- including an $\O(m\log U+n^{1.5}\log^2 U)$ time algorithm~\cite{BLSS20,BLNPSSW20,BLLSSSW21}, an $\O(n \log U)$ time algorithm for planar graphs~\cite{DLY21, DGGLPSY22}, and small improvements for general graphs, resulting in an $\O(m^{3/2 - 1/58}\log^{O(1)}U)$ time algorithm for min-cost flow~\cite{BGS21, GLP21:arxiv, AMV21:arxiv, BGJLLPS21:arxiv}.
Yet, despite this progress, the best running time bounds in general graphs are far from linear.
We give the first almost-linear time algorithm for min-cost flow, achieving the optimal running time up to subpolynomial factors.
\begin{theorem}
\label{thm:main}
There is an algorithm that, on a graph $G = (V, E)$ with $m$ edges, vertex demands, upper/lower edge capacities, and edge costs, all integral with capacities bounded by $U$ and costs bounded by $C$,
computes an exact min-cost flow in $m^{1+o(1)} \log U \log C$ time with high probability.
\end{theorem}
Our algorithm implements a new IPM that solves min-cost flow via a sequence of slowly-changing undirected min-ratio cycle subproblems.
We exploit randomized tree-embeddings to design new data-structures to efficiently maintain approximate solutions to these subproblems.
 
A direct reduction from max-flow to min-cost flow gives us an algorithm for max-flow with only a $\log U$ dependence on the capacity range $U$.\footnote{
$s,t$ max-flow can be reduced to min cost circulation by adding a new edge $t\to s$ with lower capacity 0 and upper capacity $mU.$ Set all demands to be 0. The cost of the $t \to s$ edge is $-1.$ All other edges have zero cost.}
\begin{corollary}
\label{cor:main}
There is an algorithm that on a graph G with $m$ edges with integral  capacities in $[1, U]$ computes a maximum flow between two vertices in time $m^{1+o(1)}\log U$ with high probability.
\end{corollary}

By classic capacity scaling techniques \cite{G85, GT88, AGOT92}, it suffices to work with graphs with $U, C = \poly(m)$ to show \cref{thm:main,cor:main}. For completeness, we include our version of the reductions in \cref{appendix:scaling}, as we could not find a readily citable version.

\subsection{Applications}

Our result in \Cref{thm:main} has a wide range of applications. By standard reductions, it gives the first $m^{1+o(1)}$ time algorithm for the bipartite matching problem and $m^{1+o(1)}\log U\log C$ time algorithms for its generalizations, including the worker assignment and optimal transport problems.

In directed graphs with possibly negative edge weights, assuming integral weights bounded by $W$ in absolute value, we obtain the first almost-linear time algorithm to compute single-source shortest paths and to detect a negative cycle,  running in $m^{1 + o(1)}\log{W}$ time (see \cref{sec:applications} for a reduction).
In an independent work, Bernstein, Nanongkai, and Wulff-Nilsen~\cite{BNW22} claim the first  $m\cdot \poly(\log m) \log{W}$ time algorithm for this problem.

Using recent reductions from various connectivity problems to max-flow, we also obtain the first $m^{1+o(1)}$ time algorithms for various such problems, most prominently to compute vertex connectivity and Gomory-Hu trees in undirected, unweighted graphs, and $(1+\eps)$-approximate Gomory-Hu trees in undirected weighted graphs. We also obtain the fastest current algorithm to find the global min-cut in a directed graph. Finally, we obtain the first almost linear time algorithms to compute approximate sparsest cuts in directed graphs.  
We defer the discussion of these results and precise statements to \Cref{sec:applications}.

Additionally, we extend our algorithm to compute flows that minimize general edge-separable convex objectives. This allows us to solve regularized versions of optimal transport (equivalently, matrix scaling), as well as $p$-norm flow problems and $p$-norm isotonic regression for all $p \in [1,\infty]$. We state an informal version of our main result \cref{thm:general} on general convex flows.
\begin{itheorem}
Consider a graph $G$ with demands $\bd$, and an edge-separable convex cost function $\cost(f) = \sum_e \cost_e(f_e)$ for ``computationally efficient'' edge costs $\cost_e$. Then in $m^{1+o(1)}$ time, we can compute a (fractional) flow $\bf$ that routes demands $\bd$ and $\cost(\bf) \le \cost(\bf^*) + \exp(-\log^C m)$ for any constant $C>0$, where $\bf^*$ minimizes $\cost(\bf^*)$ over flows with demands $\bd$.
\end{itheorem}
We remark that the optimal solution $\bf^*$ to the above convex flow problem can be non-integral, whereas in the case of max-flow and min-cost flow with integral demands/capacities, there exists an integral optimal flow. 

\subsection{Key Technical Contributions}
\label{subsec:keyalgo}

Towards proving our results, we make several algorithmic contributions. We informally describe the key pieces here, and present a more detailed overview in \cref{sec:overview}.

Our first contribution is a new potential reduction IPM for min-cost flow, inspired by~\cite{K84}, that reduces min-cost flow to a sequence of $m^{1+o(1)}$ slowly-changing instances of undirected minimum-ratio cycle. 
Each instance of undirected min-ratio cycle is specified by an undirected graph where every edge $e$ is assigned a positive length $\bell_e$ and a signed gradient $\bg_e,$ and the goal is to find a circulation $\bc \in \R^E,$ i.e.  $\bc$ satisfies $\mB^\top\bc = 0,$ with the smallest ratio $\bg^{\top} \bc / \|\mL\bc\|_1$, where $\mL = \mathrm{diag}(\bell)$ is the diagonal length matrix.
 Note that the graph is undirected in the sense that each edge can be traversed in either direction, and has the same length in either direction, however, the contribution of the edge gradient changes sign depending on the direction that the edge is traversed in.

Below is an informal statement summarizing the IPM guarantees proven in \cref{sec:ipm}.
\begin{itheorem}[$\ell_1$ IPM Algorithm]
\label{thm:ipm:informal}
We give an IPM algorithm that reduces solving min-cost flow exactly to sequentially solving $m^{1+o(1)}$ instances of undirected min-ratio cycle, each up to an $m^{o(1)}$ approximation. Further, the resulting problem instances are ``stable'', i.e. they satisfy,
1) the direction from the current flow to the (unknown) optimal flow  is a good enough solution for each of the instances, and,
2) the length and gradient input parameters to the instances change only for an amortized $m^{o(1)}$ edges every iteration.
\end{itheorem}
The standard IPM approach reduces min-cost flow to solving $\O(\sqrt{m})$ instances of electrical flow, which is an $\ell_2$  minimization problem, to constant accuracy.
At the cost of solving a larger number of resulting subproblems, our algorithm offers several advantages --  undirected min-ratio cycle is an $\ell_1$ minimization problem which is hopefully simpler (e.g. note that the optimal solution must be a simple cycle) and we can afford a large $m^{o(1)}$ approximation factor in the subproblems.
Most analogous to our approach is an early interior point method by \cite{WZ92}\footnote{We thank an attentive reader for making us aware of this connection.} which solved minimum cost flow using (exact) $\ell_1$ min-ratio cycle subproblems.
Their subproblems, however, do not satisfy the stability guarantees that are essential for our approach to quickly solving the subproblems approximately.
Our IPM is robust to updates with much worse approximation factors than those required in the the recent works on \emph{robust interior point methods} (\cite{CLS19} and many later works)
and establishes a different notion of stability w.r.t. gradients, lengths, and solution witnesses. This perspective may be of independent interest.

In contrast to most IPMs that work with the log barrier, our IPM uses a power barrier which aggressively penalizes constraints that are very close to being violated, more so than the usual log barrier. This ensures polylogarithmic bit-complexity throughout our algorithm.

Since a large approximation suffices, one can use a probabilistic low stretch spanning tree $T$ \cite{alon1995graph,AN19:journal} computed with respect to the lengths $\bell$ and use a fundamental tree cycle to find an $\O(1)$ approximate solution in time $\O(m)$ (see \cref{overview:mmcOblivious}). However, the changes to gradient and lengths by the IPM due to the flow updates during the IPM iterations forces us to compute a new probabilistic low stretch spanning tree $T'$ with respect to the new edge lengths. But computing a new tree in time $\Omega(m)$ per iteration results in much too large a runtime.

Our approach instead rebuilds only parts of the probabilistic low-stretch spanning tree after each IPM iteration to adapt to the changes in lengths.
To implement this, we design a data structure which maintains a recursive sequence of instances of the min-ratio cycle problem on graphs with fewer vertices and fewer edges. These smaller instances give worse approximate solutions, but are cheaper to maintain.
We use a $j$-tree style approach \cite{M10} where we interleave vertex reduction by partial embeddings into trees with edge reductions via spanners, and exploit the stability of the IPM. However, using a $j$-tree as in \cite{M10} na\"{i}vely still requires $m^{1+o(1)}$ time per instance.
Our second contribution is to push this approach much further, to give a randomized data structure 
that can return $m^{o(1)}$ approximate solutions to \emph{all} $m^{1+o(1)}$ undirected min-ratio cycle instances generated by the IPM in  $m^{1+o(1)}$ total time.
Our approach leads to a strong form of a dynamic vertex sparsifier (in the spirit of \cite{CGHPS20}).
The stability of the instances generated by our IPM algorithm is essential to achieve low amortized time per instance. 
\begin{itheorem}[Hidden Stable-Flow Chasing. \cref{thm:MMCHiddenStableFlow}]
We design a randomized data structure for approximately solving a sequence of ``stable'' (as defined in \cref{thm:ipm:informal}) undirected min-ratio cycle instances. 
The data structure maintains a collection of $m^{o(1)}$ spanning trees and supports the following operations with high probability in amortized $m^{o(1)}$ time: 1) Return an $m^{o(1)}$-approximate 
min-ratio cycle (implicitly represented as the union of $m^{o(1)}$ off-tree edges and tree paths on one of the maintained trees), 2) route a circulation along such a cycle 3) insert/delete edge $e,$ or update $\bg_e$ and $\bell_e,$ and 4) identify edges that have accumulated significant flow.
\end{itheorem}

To achieve efficient edge reduction over the entire sequence of subproblems, we give an algorithm that can efficiently maintain a spanner of a given graph (a sparse subgraph that can embed the original graph using short paths) with explicit embeddings under edge deletions/insertions and vertex splits.
Removing edges can \emph{completely} destroy the min-ratio cycles in the graph.
However, in that case, we can find a good approximate min-ratio cycle using the removed edges along with their explicit spanner embeddings.
This spanner is our third key contribution.

\begin{itheorem}[Dynamic Spanner w/ Embeddings. \cref{thm:spanner}]
We give a randomized data-structure that for an unweighted, undirected graph $G$ undergoing edge updates (insertions/deletions/vertex splits), maintains a subgraph $H$ with $\O(n)$ edges, along with an explicit path embedding of every $e \in G$ into $H$ of length $m^{o(1)}.$ 
The amortized number of edge changes in $H$ is $m^{o(1)}$ for every edge update.
Moreover, the set of edges that are embed into a fixed edge $e \in H$ is decremental for all edges $e$, except for an amortized set of $m^{o(1)}$ edges per update.

This algorithm can be implemented efficiently.
\end{itheorem}
By designing a spanner which changes very little under input graph modifications including edge insertions/deletions and vertex splits, we make it possible to dynamically combine edge and vertex sparsification very efficiently, even in a recursive construction.

Finally, note that our data-structures for hidden stable-flow chasing and spanner maintenance are utilized to efficiently implementing the $\ell_1$ IPM algorithm. 
Thus, the subsequent undirected min-ratio cycle instances can change depending on the approximately optimal cycles returned by our algorithm. In the terminology of dynamic graph algorithms, the sequence of undirected min-ratio cycle problems we need to solve is not oblivious (to the answers returned by the algorithm).
This adaptivity creates significant additional challenges for the data-structures that need addressing.

\subsection{Paper Organization}
\label{sec:organization}

The remainder of the paper is organized as follows. In \cref{sec:overview} we elaborate on each major piece of our algorithm: the $\ell_1$-IPM based on undirected minimum-ratio cycles, the construction of the data structure for maintaining undirected minimum-ratio cycles for ``stable'' update sequences, and a spanner with explicit path embeddings in dynamic graphs.
In \cref{sec:prelim} we give the preliminaries.

The algorithm to obtain our main result (\cref{thm:main}), the min-cost flow algorithm, is given on pages \pageref{sec:ipm}-\pageref{lastPageOfMaxFlow} in Sections \ref{sec:ipm}-\ref{sec:combine}, with some omitted proofs in \cref{sec:proofs}. The rest of the paper addresses generalization to convex costs, connections to the broader flow literature, and applications.

In \cref{sec:ipm} we give an iterative method which shows that a minimum cost flow can be computed to high accuracy in $m^{1+o(1)}$ iterations, each of which augments by a $m^{o(1)}$-approximate undirected minimum-ratio cycle. In \cref{sec:spanner} we construct our dynamic spanner with path embeddings. The goal of \cref{sec:jtree,sec:routing,sec:rebuilding} is to show our main data structure (\cref{thm:MMCHiddenStableFlow}) for maintaining undirected minimum-ratio cycles. \cref{sec:jtree} sets up the framework for describing ``stable'' update sequences, and describes the main data structure components. \cref{sec:routing} formally constructs the data structure modulo a technical issue, which we resolve by introducing and solving the \emph{rebuilding game} in \cref{sec:rebuilding}. In \cref{sec:combine} we combine all the pieces we have developed to give a min-cost flow algorithm running in time $m^{1+o(1)}$. 

In the last part of the paper, \cref{sec:general}, we extend the IPM analysis to handle general edge-separable convex, nonlinear objectives, such as normed flows, isotonic regression, entropy-regularized optimal transport, and matrix scaling.

The appendix contains an overview of previous max-flow and min-cost flow approaches in \cref{sec:maxflowprevious}, omitted proofs in \cref{sec:proofs}, a proof of capacity scaling for min-cost flows in \cref{appendix:scaling}, and an extensive description of applications of our algorithms in \cref{sec:applications}.

%% file: overview.tex
\section{Overview}
\label{sec:overview}
In this section, we give a technical overview of the key pieces developed in this paper.
\cref{overview:ipm} describes an optimization method based on interior point methods that reduces min-cost flow to a sequence of $m^{1+o(1)}$ undirected minimum-ratio cycle computations.
In particular, we reduce the problem to computing approximate min-ratio cycles on a slowly changing graph.
This can be naturally formulated as a data structure problem of maintaining min-ratio cycles approximately on a dynamic graph.

We build a data structure for solving this dynamic min-ratio cycle problem and solve it with $m^{o(1)}$ amortized time per cycle update for our IPM, giving an overall running time of $m^{1+o(1)}$.
\cref{overview:mmcOblivious} gives an overview of our data structure for this dynamic min-ratio cycle problem, with pointers to the rest of the overview which provides a more in-depth picture of the construction.
The data structure creates a recursive hierarchy of graphs with fewer and fewer vertices and edges.
In \cref{overview:core} we describe how to reduce the number of vertices, before describing the overall recursive data structure in \cref{overview:dssetup}.
Na\"{i}vely, the resulting data structure works only against oblivious adversaries where updates and queries to the data structure are fixed beforehand.
We cannot utilize it directly because the optimization routine updates the dynamic graph based on past outputs from the data structure.
Therefore, the cycles output by the data structure may not be good enough to make progress.
\cref{overview:routing} discusses the interaction between the optimization routine and the data structure when we directly apply it.
It turns out one can leverage properties of the interaction and adapt the data structure for the optimization routine.
\cref{overview:game} presents an online algorithm that manipulates the data structure so that it always outputs cycles that are good enough to make progress in the optimization routine.
Finally, the overview ends with \cref{overview:spanner} which gives an outline of our dynamic spanner data structure.
We use this spanner to reduce the number of edges at each level of our recursive hierarchy, one of the main algorithmic elements of our data structure.

\subsection{Computing Min-Cost Flows via Undirected Min-Ratio Cycles}
\label{overview:ipm}
The goal of this section is to describe an optimization method which computes a min-cost flow on a graph $G = (V, E)$ in $m^{1+o(1)}$ computations of $m^{o(1)}$-approximate min-ratio cycles:
\begin{align}
    \min_{\mB^\top\bDelta = 0} \frac{\bg^\top\bDelta}{\|\mL\bDelta\|_1}
    \label{eq:mmc}
\end{align}
for gradient $\bg \in \R^E$ and lengths $\mL = \diag(\bell)$ for $\bell \in \R_{>0}^E$. Note that the value of this objective is negative, as $-\bDelta$ is a circulation if $\bDelta$ is.

Towards this, we work with the linear-algebraic setup of the min-cost flow problem:
\begin{align}
    \bf^* \in \argmin_{\substack{\mB^\top\bf = \bd \\ \bu^-_e \le \bf_e \le \bu^+_e \forall e \in E}} \bc^\top \bf \label{eq:mincostflow}
\end{align}
for demands $\bd \in \R^E$, lower and upper capacities $\bu^-, \bu^+ \in \R^E$, and cost vector $\bc \in \R^E$. Our goal is to compute an optimal flow $\bf^*$. Let $F^* = \bc^\top \bf^*$ be the optimal cost.

Our algorithm is based on a potential reduction interior point method \cite{K84}, where each iteration we reduce the value of the potential function
\begin{align}
    \Phi(\bf) \defeq 20m \log(\bc^\top \bf - F^*) + \sum_{e \in E} \left((\bu^+_e - \bf_e)^{-\alpha} + (\bf_e - \bu^-_e)^{-\alpha} \right)
\end{align}
for $\alpha = 1/(1000 \log mU)$. The reader can think of the barrier $x^{-\alpha}$ as the more standard $-\log x$ for simplicity instead.
We use $x^{-\alpha}$ to ensure that all lengths/gradients encountered during the algorithm can be represented using $\O(1)$ bits, and explain why this holds later in the section.
When $\Phi(\bf) \le -200m\log mU$, we can terminate because then $\bc^\top \bf - F^* \le (mU)^{-10}$, at which point standard techniques let us round to an exact optimal flow \cite{DS08}. Thus if we can reduce the potential by $m^{-o(1)}$ per iteration, the method terminates in $m^{1+o(1)}$ iterations.

Previous analyses of IPMs used $\ell_2$ subproblems, i.e. replacing the $\ell_1$ norm in \eqref{eq:mmc} with an $\ell_2$ norm, which can be solved using a linear system. \cite{K84} shows that using $\ell_2$ subproblems such a method converges in $\O(m)$ iterations. Later analyses of path-following IPMs \cite{R88} showed that a sequence of $\O(\sqrt{m})$ $\ell_2$ subproblems suffice to give a high-accuracy solution. Surprisingly, we are able to argue that a solving sequence of $\O(m)$ $\ell_1$ minimizing subproblems of the form in \eqref{eq:mmc} suffice to give a high accuracy solution to \eqref{eq:mincostflow}. In other words, changing the $\ell_2$ norm to an $\ell_1$ norm does not increase the number of iterations in a potential reduction IPM.
The use of an $\ell_1$-norm-based subproblem gives us a crucial advantage: Problems of this form must have optimal solutions in the form of cycles---and our new algorithm finds approximately optimal cycles vastly more efficiently than any known approaches for $\ell_2$ subproblems.

There are several reasons we choose to use a potential reduction IPM with this specific potential. The most important reason is the flexibility of a potential reduction IPM allows our data structure for maintaining solutions to \eqref{eq:mmc} to have large $m^{o(1)}$ approximation factors. This contrasts with recent works towards solving min-cost flow and linear programs using a \emph{robust IPM}~(see \cite{CLS19} or the tutorial \cite{LV21}), which require $(1+o(1))$--approximate solutions for the iterates.

Finally, we use the barrier $x^{-\alpha}$ as opposed to the more standard logarithmic barrier in order to guarantee that all lengths/gradients encountered during the method are bounded by $\exp(\log^{O(1)} m)$ throughout the method. This follows because if $(\bu^+_e - \bf_e)^{-\alpha} \le \O(m)$, then
\[ \bu^+_e - \bf_e \ge \O(m)^{-1/\alpha} \ge \exp(-O(\log^2 Um)). \] Such a guarantee does not hold for the logarithmic barrier.\footnote{The reason that path-following IPMs for max-flow \cite{DS08} do not encounter this issue is because one can show that primal-dual optimality actually guarantees that the lengths/resistances are polynomially bounded. We do not maintain any dual variables, so such a guarantee does not hold for our algorithm.}

To conclude, we discuss a few specifics of the method, such as how to pick the lengths and gradients, and how to prove that the method makes progress. Given a current flow $\bf$ we define the gradient and lengths we use in \eqref{eq:mmc} as $\bg(\bf) \defeq \g \Phi(\bf)$ and $\bell(\bf)_e \defeq \left(\bu^+_e - \bf_e\right)^{-1-\alpha} + \left(\bf_e - \bu^-_e\right)^{-1-\alpha}$. Now, let $\bDelta$ be a circulation with $\bg(\bf)^\top \bDelta / \left\|\mL\bDelta\right\|_1 \le -\kappa$ for some $\kappa < 1/100$, scaled so that $\left\|\mL\bDelta\right\|_1 = \kappa/50$. A direct Taylor expansion shows that $\Phi(\bf + \bDelta) \le \Phi(\bf) - \kappa^2/500$ (\cref{lemma:phidecrease}).

Hence it suffices to show that such a $\bDelta$ exists with $\kappa = \widetilde{\Omega}(1)$, because then a data structure which returns an $m^{o(1)}$-approximate solution still has $\kappa = m^{-o(1)}$, which suffices. Fortunately, the \emph{witness circulation} $\bDelta(\bf)^* = \bf^* - \bf$ satisfies
$\bg(\bf)^\top \bDelta / \left\|\mL\bDelta\right\|_1 \le -\widetilde{\Omega}(1)$ (\cref{lemma:optCirculationQuality}).

We emphasize that the fact that $\bf^* - \bf$ is a good enough witness circulation for the flow $\bf$ is essential for proving that our randomized data structures suffice, even though the updates seem adaptive. At a high level, this guarantee helps because even though we do not know the witness circulation $\bf^* - \bf$, we know how it changes between iterations, because we can track changes in $\bf$. We are able to leverage such guarantees to make our data structures succeed for the updates coming from the IPM.
To achieve this, we end up carefully designing our adversary model with enough power to capture our IPM, but with enough restrictions that our min-ratio cycle data structure to win against the adversary.
 We elaborate on this point in \cref{overview:mmcOblivious,overview:routing}.

\subsection{High Level Overview of the Data Structure for Dynamic Min-Ratio Cycle}
\label{overview:mmcOblivious}
As discussed in the previous section, our algorithm computes a min-cost flow by solving a sequence of $m^{1+o(1)}$ min-ratio cycle problems $\min_{\mB^\top\bDelta = 0} \bg^\top\bDelta/\|\mL\bDelta\|_1$ to $m^{o(1)}$ multiplicative accuracy.
Because our IPM ensures stability for lengths and gradients (see \cref{lemma:stablebell,lemma:stablebg}), and is even robust to approximations of lengths and gradients,
we can show that over the course of the algorithm we only need to update the entries of the gradients $\bg$ and lengths $\bell$ at most $m^{1+o(1)}$ total times. Efficiency gains based on leveraging stability has appeared in the earliest works on efficiently maintaining IPM iterates \cite{K84,V90} as well as most recent progress on speeding up linear programs.

\paragraph{Warm-Up: A Simple, Static Algorithm.} A simple approach to finding an $\O(1)$-approximate min-ratio cycle is the following: given our graph $G$, we find a probabilistic low stretch spanning tree $T$, i.e., a tree such that for each edge $e = (u,v) \in G$, the stretch of $e$, defined as $\str^{T,\bell}_e \defeq \frac{\sum_{f \in T[u,v]} \bell(f)}{ \bell(e)}$ where $T[u,v]$ is the unique path from $u$ to $v$ along the tree $T$, is $\O(1)$ in expectation. Such a tree can be found in $\O(m)$ time \cite{alon1995graph, AN19:journal}. 

Let $\bDelta^*$ be the witness circulation that minimizes \eqref{eq:mmc}, and assume wlog that $\bDelta^*$ is a cycle that routes one unit of flow along the cycle. We assume for convenience, that edges on $\bDelta^*$ are oriented along the flow direction of $\bDelta^*$, i.e. that $\bDelta^* \in \mathbb{R}^{E}_{\geq 0}$. Then, for each edge $e = (u,v)$ on the cycle $\bDelta^*$, the fundamental tree cycle of $e$ in $T$ denoted $e \oplus T[v,u]$, representing the cycle formed by edge $e$ concatenated with the path in $T$ from its endpoint $v$ to $u$. To work again with vector notation, we denote by $\bp(e \oplus T[v,u]) \in \mathbb{R}^{E}$ the vector that sends one unit of flow along the cycle $e \oplus T[v,u]$ in the direction that aligns with the orientation of $e$. A classic fact from graph theory now states that $\bDelta^* = \sum_{e : \bDelta^*_e > 0} \bDelta^*_e \cdot \bp(e \oplus T[v,u])$ (note that the tree-paths used by adjacent off-tree edges cancel out, see \cref{fig:TreeCycleCancel}). In particular, this implies that $\bg^{\top} \bDelta^* =  \sum_{e : \bDelta^*_e > 0} \bDelta^*_e \cdot \bg^{\top} \bp(e \oplus T[v,u])$.
\begin{figure}
\centering
\begin{subfigure}{.5\textwidth}
  \centering
  \includegraphics[width=0.8\linewidth]{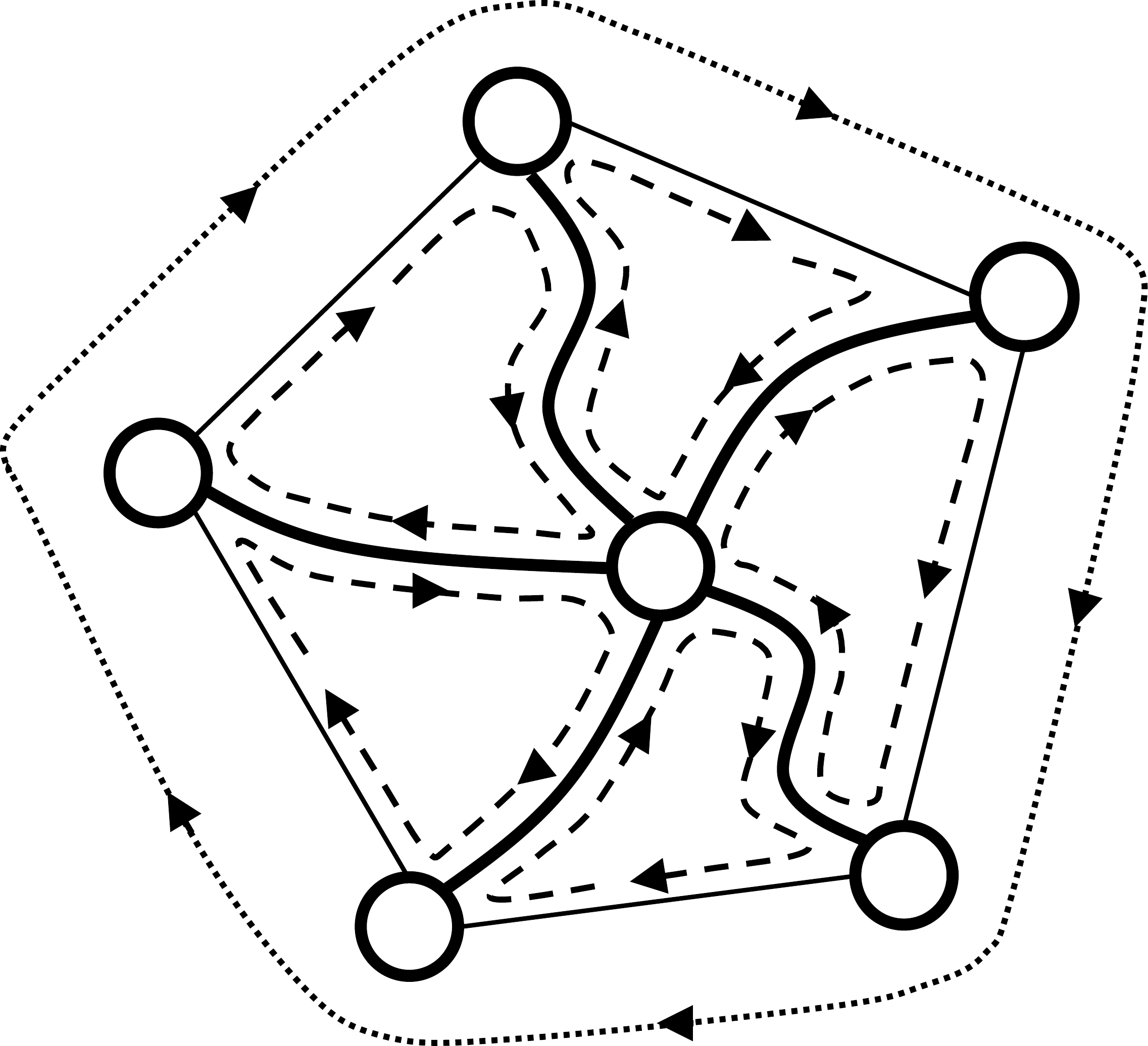}
\end{subfigure}%
\begin{subfigure}{.5\textwidth}
  \centering
  \includegraphics[width=0.8\linewidth]{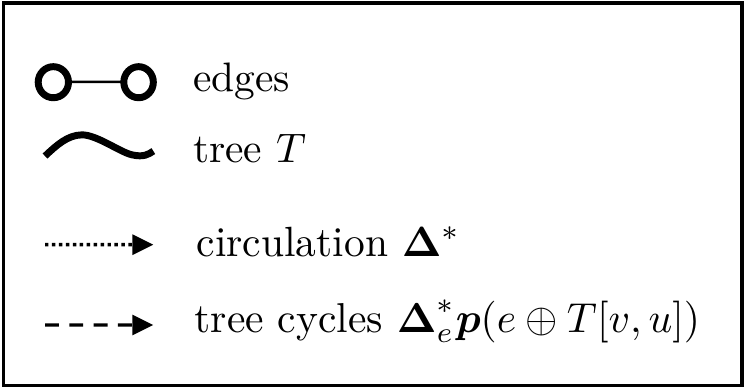}
\end{subfigure}
\caption[Finding good min-ratio cycles using a low-stretch tree]{Illustrating the decomposition $\bDelta^* = \sum_{e : \bDelta^*_e > 0} \bDelta^*_e \cdot \bp(e \oplus T[v,u])$ of a circulation into tree cycles given by off-trees and the corresponding tree paths. }
\label{fig:TreeCycleCancel}
\end{figure}

This fact will allow us to argue that with probability at least $\frac{1}{2}$, one of the tree cycles is an $\O(1)$-approximate solution to \eqref{eq:mmc}. Therefore, since the stretch $\str^{T,\bell}_e$ of edges $e \in E$ is small in expectation, we can, by Markov's inequality, argue that with probability at least $\frac{1}{2}$, the circulation $\bDelta^*$ is not stretched by too much. Formally, we have that $\sum_{e : \bDelta^*_e > 0} \bDelta^*_e \cdot \|\mL \; \bp(e \oplus T[v,u])\|_1 \leq \gamma \|\mL \bDelta^*\|_1$ for $\gamma = \O(1)$. Combining our insights, we can thus derive that
\[
\frac{\bg^{\top} \bDelta^* }{\|\mL \bDelta^*\|_1} \geq \frac{1}{\gamma} \cdot \frac{\sum_{e : \bDelta^*_e > 0} \bDelta^*_e \cdot \bg^{\top} \bp(e \oplus T[v,u])}{\sum_{e : \bDelta^*_e > 0} \bDelta^*_e \cdot \|\mL \; \bp(e \oplus T[v,u])\|_1} \geq \frac{1}{\gamma} \min_{e : \bDelta^*_e > 0}  \frac{ \bg^{\top} \bp(e \oplus T[v,u])}{ \|\mL \; \bp(e \oplus T[v,u])\|_1}
\]
where the last inequality follows from the fact that $ \min_{i\in[n]} \frac{\bx_i}{\by_i} \le \frac{\sum_{i\in[n]} \bx_i}{\sum_{i\in[n]} \by_i}$ (recall also that $\bg^{\top} \bDelta^*$ is negative). But this exactly says that for the edge $e$ minimizing the expression on the right, the tree cycle $e \oplus T[v,u]$ is a $\gamma$-approximate solution to \eqref{eq:mmc}, as desired.

Since the low stretch spanning tree $T$ stretches circulation $\bDelta^*$ reasonably with probability at least $\frac{1}{2}$, we could boost the probability by sampling $\O(1)$ trees $T_1, T_2, \ldots, T_s$ independently at random and conclude that w.h.p. one of the fundamental tree cycles gives an approximate solution to \eqref{eq:mmc}.

Unfortunately, after updating the flow $\bf$ to $\bf'$ along such a fundamental tree cycle, we cannot reuse the set of trees $T_1, T_2, \ldots, T_s$ because the next solution to \eqref{eq:mmc} has to be found with respect to gradients $\bg(\bf')$ and lengths $\bell(\bf')$ depending on $\bf'$ (instead of $\bg = \bg(\bf)$ and $\bell = \bell(\bf)$). But $\bg(\bf')$ and $\bell(\bf')$ depend on the randomness used in trees $T_1, T_2, \ldots, T_s$. Thus, naively, we have to recompute all trees, spending again $\Omega(m)$ time. But this leads to run-time $\Omega(m^2)$ for our overall algorithm which is far from our goal.

\paragraph{A Dynamic Approach.}

Thus we consider the data structure problem of maintaining an $m^{o(1)}$ approximate solution to \eqref{eq:mmc} over a sequence of at most $m^{1+o(1)}$ changes to entries of $\bg, \bell$. To achieve an almost linear time algorithm overall, we want our data structure to have an amortized $m^{o(1)}$ update time. Motivated by the simple construction above, our data structure will ultimately maintain a set of $s = m^{o(1)}$ spanning trees $T_1, \dots, T_s$ of the graph $G$. Each cycle $\bDelta$ that is returned is represented by $m^{o(1)}$ off-tree edges and paths connecting them on some $T_i$. 

To obtain an efficient algorithm to maintain these trees $T_i$, we turn to a recursive approach. 
In each level of our recursion, we first reduce the number of vertices, and then the number of  edges in the  graphs we recurse on.
To reduce the number of vertices, we produce a \emph{core graph} on a subset of the original vertex set, and we then compute a \emph{spanner} of the core graph which reduces the number of edges.
Both of these objects need to be maintained dynamically, and we ensure they are very stable under changes in the graphs at shallower levels in the recursion. 
In both cases, our notion of stability relies on some subtle properties of the interaction between the data structure and the hidden witness circulation.

We maintain a recursive hierarchy of graphs. At the top level of our hierarchy, for the input graph $G$, we produce $B = O(\log n)$ core graphs. To obtain each such core graph, for each $i \in [B]$, we sample a (random) forest $F_i$ with $\O(m/k)$ connected components for some size reduction parameter $k$. The associated core graph is the graph $G / F_i$ which denotes $G$ after contracting the vertices in the same components of $F_i$. 
We can define a map that lifts circulations $\wh{\bDelta}$ in the core graph $G / F_i$, to circulations  $\bDelta$ in the graph $G$ by routing flow along the contracted paths in $F_i$. 
The lengths in the core graph $\wh{\bell}$ (again let $\wh{\mL} = \diag(\wh{\bell})$) and are chosen to upper bound the length of circulations when mapped back into $G$ such that $\|\wh{\mL}\wh{\bDelta}\|_1 \geq \|\mL\bDelta\|_1 $. 
Crucially, we must ensure these new lengths $\wh{\bell}$ do not stretch the witness circulation $\bDelta^*$ when mapped into $G/F_i$ by too much, so we can recover it from $G/F_i$. To achieve this goal, we choose $F_i$ to be a low stretch forest, i.e. a forest with properties similar to those of a low stretch tree.
In \Cref{overview:core}, we summarize the central aspects of our core graph construction.

While each core graph $G/F_i$ now has only $\O(m/k)$ vertices, it still has $m$ edges which is too large for our recursion.  To overcome this issue we build a spanner $\SS(G,F_i)$ on $G/F_i$ to reduce the number of edges to $\O(m/k)$, which guarantees that for every edge $e=(u,v)$ that we remove from $G/F_i$ to obtain $\SS(G,F_i)$, there is a  $u$-to-$v$ path in $\SS(G,F_i)$ of length $m^{o(1)}$.
Ideally, we would now recurse on each spanner $\SS(G,F_i)$, again approximating it with a collection of smaller core graphs and spanners.
However, we face an obstacle: removing edges could destroy the witness circulation, so that possibly no good circulation exists in any $\SS(G,F_i)$.
To solve this problem, we compute an explicit embedding $\Pi_{G /F_i \to \SS(G,F_i)}$ that maps each edge $e = (u,v) \in G/F_i$ to a short $u$-to-$v$ path in $\SS(G,F_i)$. 
We can then show the following dichotomy: Let $\wh{\bDelta}(\bf)^*$ denote the witness circulation when mapped into the core graph $G/F_i$. Then, \emph{either} one of the edges $e \in E_{G /F_i} \setminus E_{\SS(G,F_i)}$ has a spanner cycle consisting of $e$ combined with $ \Pi_{G /F_i \to \SS(G,F_i)}(e)$ which is almost as good as $\wh{\bDelta}(\bf)^*$, \emph{or} re-routing $\wh{\bDelta}(\bf)^*$ into  $\SS(G,F_i)$ roughly preserves its quality. \cref{fig:SpannerCyclesOrProjGood} illustrates this dichotomy.
Thus, either we find a good cycle using the spanner, or we can recursively find a solution on $\SS(G,F_i)$ that almost matches $\wh{\bDelta}(\bf)^*$ in quality.
To construct our dynamic spanner with its strong stability guarantees under changes in the input graph, we use a new approach that diverges from other recent works on dynamic spanners; we give an outline of the key ideas in \Cref{overview:spanner}.
\begin{figure}
  \centering
  \includegraphics[width=\linewidth]{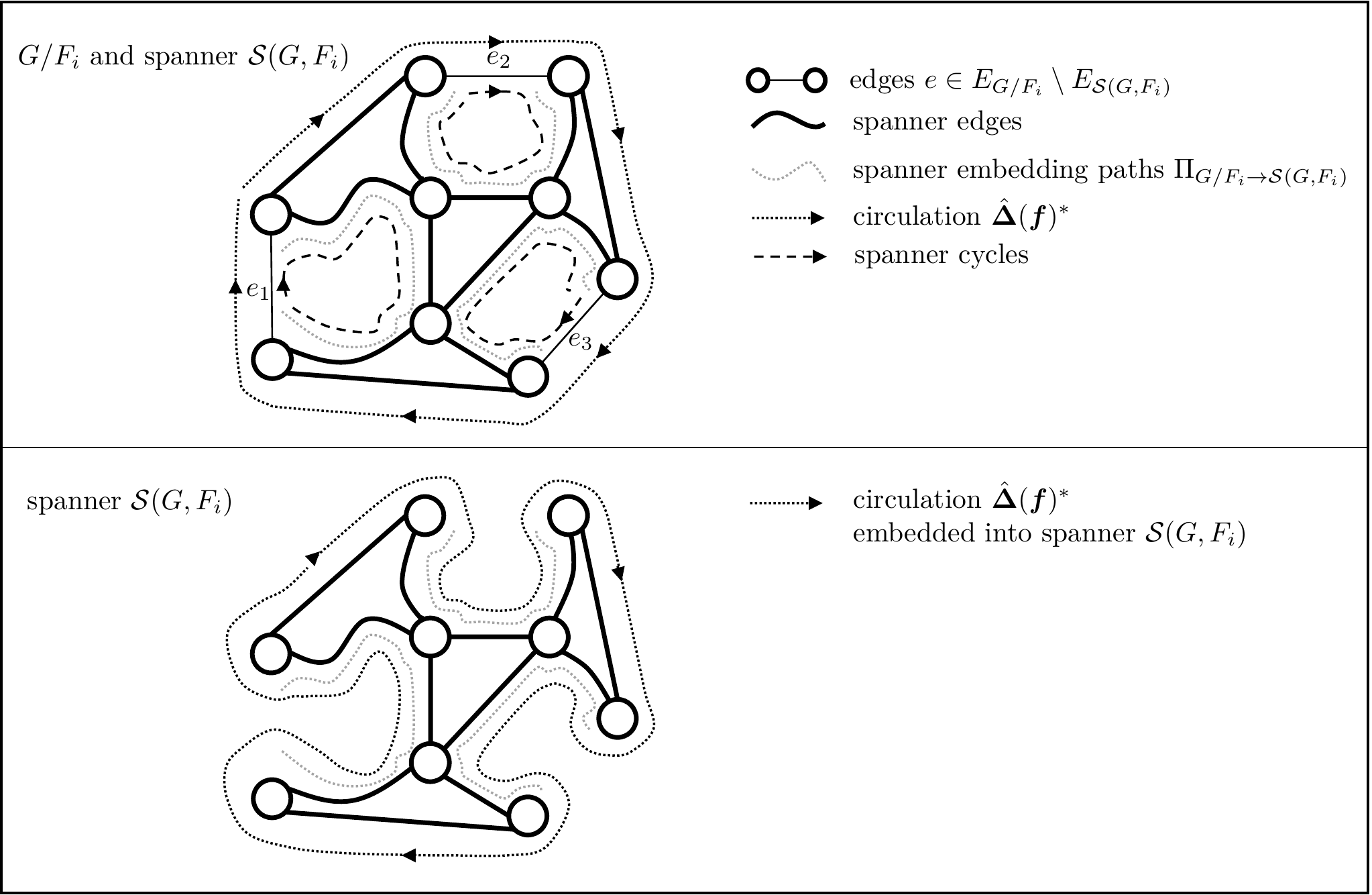}
\caption[A dichotomy: Either the spanner preserves the circulation, or a fundamental spanner cycle is good]{Illustration of a dichotomy: either one of the edges $e \in E_{G /F_i} \setminus E_{\SS(G,F_i)}$ has a spanner cycle consisting of $e$ combined with $ \Pi_{G /F_i \to \SS(G,F_i)}(e)$ which is almost as good as $\wh{\bDelta}(\bf)^*$, \emph{or} re-routing $\wh{\bDelta}(\bf)^*$ into  $\SS(G,F_i)$ roughly preserves its quality.}
\label{fig:SpannerCyclesOrProjGood}
\end{figure}

Our recursion uses $d$ levels, where we choose the size reduction factor $k$ such that $k^d \approx m$ and the bottom level graphs have $m^{o(1)}$ edges.
Note that since we build $B$ trees on $G$ and recurse on the spanners of $G/F_1, G/F_2, \ldots, G/F_B$, our recursive hierarchy has a branching factor of $B = O(\log n)$ at each level of recursion. 
Thus, choosing $d \le \sqrt{\log n}$, we get $B^d = m^{o(1)}$ leaf nodes in our recursive hierarchy.
Now, consider the forests $F_{i_1}, F_{i_2},\ldots, F_{i_d}$ on the path from the top of our recursive hierarchy to a leaf node. We can patch these forests together to form a tree associated with the leaf node.
Each of these trees, we maintain as a link-cut tree data structure. Using this data structure, whenever we find a good cycle, we can route flow along it and detect edges where the flow has changed significantly.
The cycles are either given by an off-tree edge or a collection of $m^{o(1)}$ off-tree edges coming from a spanner cycle.
We call the entire construction a \emph{branching tree chain}, and in \Cref{overview:dssetup}, we elaborate on the overall composition of the data structure.

What have we achieved using this hierarchical construction compared to our simple, static algorithm? 
First, consider the setting of an oblivious adversary, where the gradient and length update sequences and the optimal circulation after each update is fixed in advance.
In this setting, we can show that our spanner-of-core graph construction can survive through $m^{1-o(1)}/k^i$ updates at level $i$.
Meanwhile, we can rebuild these constructions in time $m^{1+o(1)}/k^{i-1}$, leading to an amortized cost per update of $k m^{o(1)} \leq m^{o(1)}$ at each level.
This gives the first dynamic data structure for our undirected min-ratio problem with $m^{o(1)}$ query time against an oblivious adversary.

However, our real problem is harder: the witness circulation in each round is $\bDelta(\bf)^* = \bf^* - \bf$ and depends on the updates we make to $\bf$, making our problem adaptive.
Instead of modelling our IPM as giving rise to a fully-dynamic problem against an adaptive adversary, 
the promise that the witness circulation can always be written as $\bf^* - \bf$ lets us express the IPM with an adversary that is much more restricted.
Our data structure needs to ensure that the flow $\bf^* - \bf$ is stretched by $m^{o(1)}$ on average w.r.t. the lengths $\bell$.
At a high level, we achieve this by forcing the forests at every level to have stretch $1$ on edges where $\bf_e$ changes significantly and could affect the total stretch of our data structure on $\bf^* - \bf$.
\Cref{overview:routing} describes the guarantees we achieve using this strategy.
However, the data structure at this point is not yet guaranteed to succeed. 
Instead, we very carefully characterize the failure condition.
In particular, to induce a failure, the adversary must create a situation where the current value of $\|\mL\bDelta(\bf)^*\|_1$ is significantly less than the value when the levels of our data structure were last rebuilt.
This means we can counteract from this failure by rebuilding the data structure levels.
Due to the high cost of rebuilding the shallowest levels of the data structure, na\"{i}vely rebuilding the entire data structure is much too expensive, and we need a more sophisticated strategy.
We describe this strategy in \Cref{overview:game}, where we design a game that expresses the conflict between our data structure and the adversary, and we show how to win this game without paying too much runtime for rebuilds.

\subsection{Building Core Graphs}
\label{overview:core}
In this section, we describe our core graph construction (\cref{def:coregraph}), which maps our dynamic undirected min-ratio cycle problem on a graph $G$ with at most $m$ edges and vertices into a problem of the same type on a graph with only $\O(m/k)$ vertices and $m$ edges, and handles $\O(m/k)$ updates to the edges before we need to rebuild it.
Our construction is based on constructing low-stretch decompositions using forests and portal routing (\cref{lemma:globalstretch}).
We first describe how our portal routing uses a given forest $F$ to construct a core graph $G/F$.
We then discuss how to use a collection of (random) forests $F_1, \ldots, F_{B}$ to produce a low-stretch decomposition of $G$, which will ensure that one of the core graphs $G/F_i$ preserves the witness circulation well.
Portal routings played a key role in the ultrasparsifiers of \cite{ST04} and has been further developed in many works since.

\paragraph{Forest Routings and Stretches.}
To understand how to define the stretch of an edge $e$ with respect to a forest $F$, it is useful to define how to \emph{route} an edge $e$ in $F$.
Given a spanning forest $F$, every path and cycle in $G$ can be mapped to $G/F$ naturally (where we allow $G/F$ to contain self-loops).
On the other hand if every connected component in $F$ is rooted, where $\root^F_{u}$ denotes the root corresponding to a vertex $u \in V$, we can map every path and cycle in $G/F$ back to $G$ as follows.
Let $P = e_1, \dots, e_k$ be any (not necessarily simple) path in $G/F$ where the preimage of every edge $e_i$ is $e^G_i = (u^G_i, v^G_i) \in G.$
The preimage of $P$, denoted $P^G$, is defined as the following concatenation of paths: 
\begin{align*}
    P^G \defeq \bigoplus_{i=1}^k F[\root^F_{u^G_i}, u^G_i] \oplus e^G_i \oplus F[v^G_i, \root^F_{v^G_i}],
\end{align*}
where we use $A \oplus B$ to denote the concatenation of paths $A$ and $B$, and $F[a, b]$ to denote the unique $ab$-path in the forest $F.$
When $P$ is a circuit (i.e. a not necessarily simple cycle), $P^G$ is a circuit in $G$ as well.
One can extend these maps linearly to all flow vectors and denote the resulting operators as $\bPi_F : \R^{E(G)} \to \R^{E(G / F)}$ and $\bPi_F^{-1}: \R^{E(G / F)} \to \R^{E(G)}$.
Since we let $G/F$ have self-loops, there is a bijection between edges of $G$ and $G/F$ and thus $\bPi_F$ acts like the identity function. 

To make our core graph construction dynamic, the key operation we need to support is the dynamic addition of more root nodes, which results in forest edges being deleted to maintain the invariant each connected component has a root node.
Whenever an edge is changing in $G$, we ensure that $G/F$ approximates the changed edge well by forcing both its endpoints to become root notes, which in turn makes the portal routing of the new edge trivial and this guarantees its stretch is $1$.
An example of this is shown in \cref{fig:CoreGraphMaintenance}.

\begin{figure}
\begin{subfigure}{\textwidth}
  \centering
  \includegraphics[width=\linewidth]{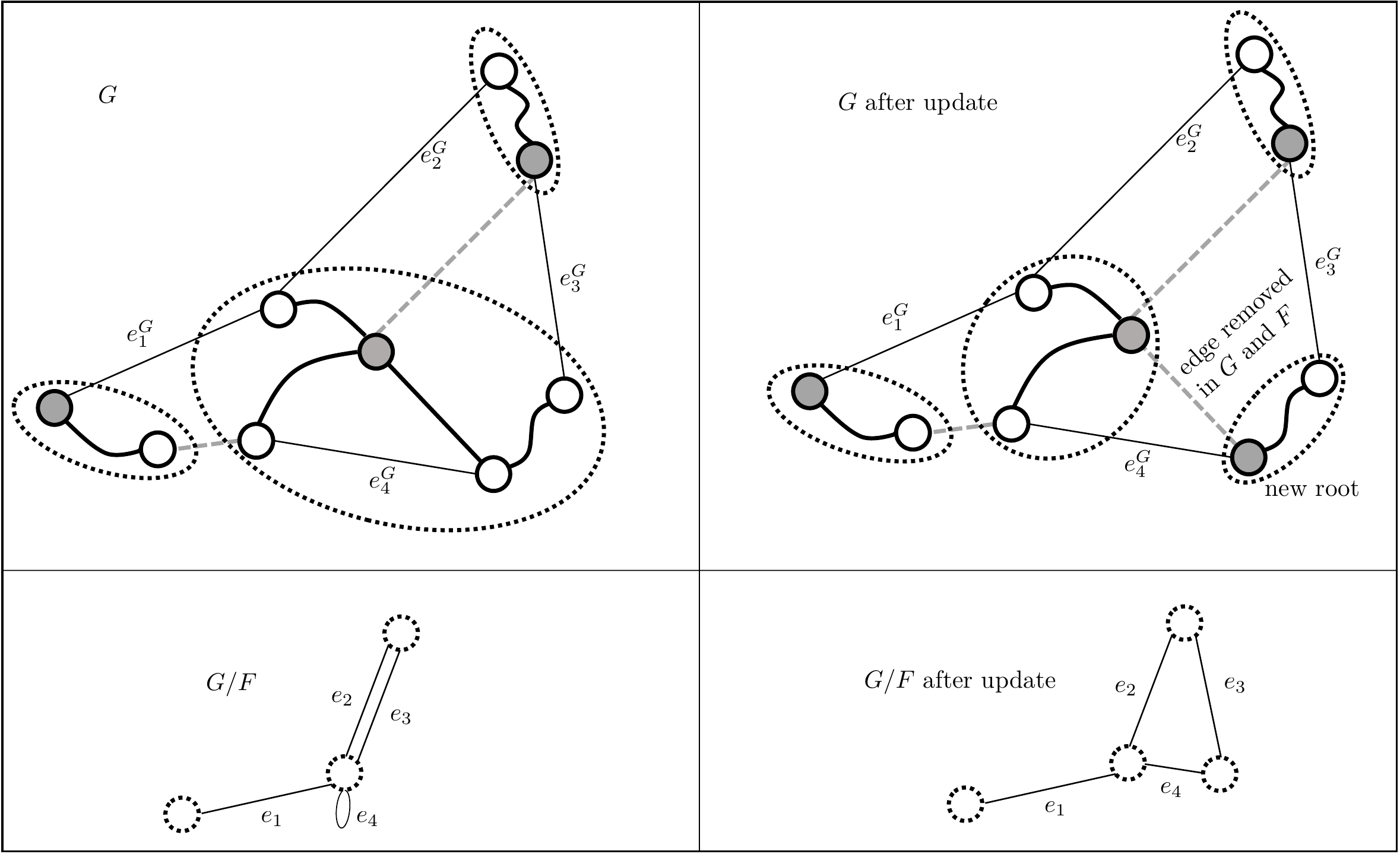}
  \vspace{1em}
\end{subfigure}
\raggedleft
\begin{subfigure}{.4\textwidth}
  \includegraphics[width=\linewidth]{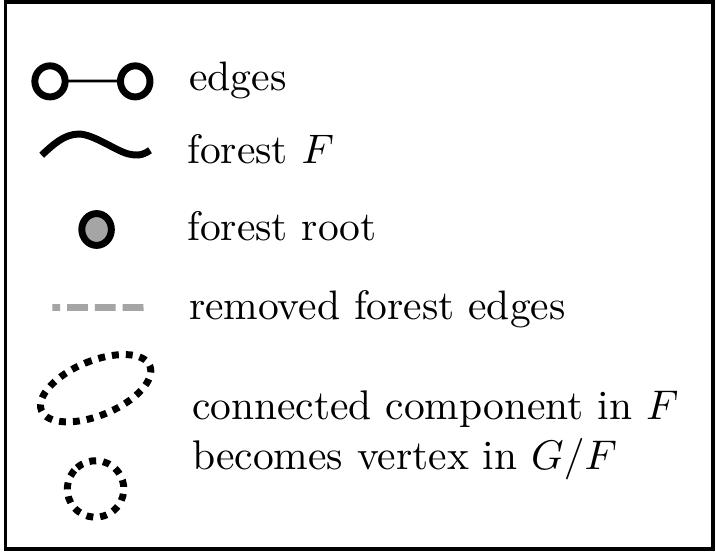}
\end{subfigure}
\caption[Changes in the core graph under edge deletions]{Illustration of the core graph $G/F$ changing as an edge is deleted in $G$ (and in $F$).}
\label{fig:CoreGraphMaintenance}
\end{figure}

For any edge $e^G = (u^G, v^G)$ in $ G$ with image $e$ in $G / F$, we set $\wh{\bell}^F_e$, the edge length of $e$ in $G/F$, to be \emph{an upper bound} on the length of the \emph{forest routing} of $e$, i.e. the path $F[\root^F_{u^G}, u^G] \oplus e^G \oplus F[v^G, \root^F_{v^G}]$.
Meanwhile, we define $\wstr_e \defeq \wh{\bell}^F_e / \bell_e$, as an overestimate on the stretch of $e$ w.r.t. the forest routing.
A priori, it is unclear how to provide a single upper bound on the stretch of every edge, as the root nodes of the endpoints are changing over time. 
Providing such a bound for every edge is important for us as the lengths in $G/F$ could otherwise be changing too often when the forest changes.
We guarantee these bounds by scheme that makes auxiliary edge deletions in the forest in response to external updates, with these additional roots chosen carefully to ensure the length upper bounds.

Now, for any flow $\bf$ in $G/F$, its length in $G/F$ is at least the length of its pre-image in $G$, i.e. $\norm{\mL \bPi_F^{-1} \bf}_1 \le \norm{\wh{\mL}^F \bf}_1$.
Let $\bDelta^*$ be the optimal solution to \eqref{eq:mmc}.
We will show later how to build $F$ such that $\norm{\wh{\mL}^F \bDelta^*}_1 \le \gamma \norm{\mL \bDelta^*}_1$ holds for some $\gamma = m^{o(1)}$, solving \eqref{eq:mmc} on $G/F$ with edge length $\wh{\bell}$ and properly defined gradient $\wh{\bg}$ on $G/F$ yields an $\frac{1}{\gamma}$-approximate solution for $G.$
The gradient $\wh{\bg}$ is defined so that the total gradient of any circulation $\bDelta$ on $G/F$ and its preimage $\bPi_{F}^{-1}\bDelta$ in $G$ is the same, i.e. $\wh{\bg}^\top\bDelta = \bg^\top \bPi_{F}^{-1}\bDelta$.
The idea of incorporating gradients into portal routing was introduced in~\cite{KPSW19}; our version of this construction is somewhat different to allow us to make it dynamic efficiently.

\paragraph{Collections of \emph{Low Stretch Decompositions (LSD)}.}
The first component of the data structure is constructing and maintaining forests of $F$ that form a \emph{Low Stretch Decomposition (LSD)} of $G$.
Variations of which (such as $j$-trees) have been used to construct several recursive graph preconditioners~\cite{M10, S13, KLOS14, CPW21:arxiv} and dynamic algorithms~\cite{CGHPS20}.
Informally, a $k$-LSD is a rooted forest $F \subseteq G$ that decomposes $G$ into $O(m/k)$ vertex disjoint components.
Given some positive edge weights $\bv \in \R^E_{>0}$ and reduction factor $k > 0$, we compute a $k$-LSD $F$ and length upper bounds $\wh{\bell}^F$ of $G/F$ that satisfy two properties:
\begin{enumerate}
    \item $\wstr^F_e = \wh{\bell}^F_{e} / \bell_{e^G} = \O(k)$ for any edge $e^G \in G$ with image $e$ in $G / F$, and \label{item:lsdWidth}
    \item The weighted average of $\wstr^F_e$ w.r.t. $\bv$ is only $\O(1)$, i.e. $\sum_{e^G \in G} \bv_{e^G} \cdot \wstr^F_e \le \O(1) \cdot \norm{\bv}_1.$ \label{item:lsdAvg}
\end{enumerate}
\cref{item:lsdWidth} guarantees that the solution to \eqref{eq:mmc} for $G/F$ yields a $\O(k)$-approximate one for $G.$
However, this guarantee is not sufficient for our data structure, as our $B$-branching tree chain has $d \approx \log_k m$ levels of recursion and the quality of the solution from the deepest level would only be $\O(k)^d \approx m^{1+o(1)}$-approximate.

Instead, like \cite{M10,S13,KLOS14} we compute $k$ different edge weights $\bv_1, \ldots, \bv_k$ via multiplicative weight updates (\cref{lemma:strMWU}) so that the corresponding LSDs $F_1, \ldots, F_k$ have $\O(1)$ average stretch on every edge in $G$:
  $ \sum_{j = 1}^k \wstr^{F_j}_e = \O(k),~\forall e^G \in G$ with image $e$ in $G/F.$

By Markov's inequality, for any fixed flow $\bf$ in $G$, $\norm{\wh{\mL}^{F_j} \bf}_1 \le \O(1) \norm{\mL \bf}_1$ holds for at least half the LSDs corresponding to $F_1, \dots, F_k.$
Taking $\O(1)$ samples uniformly from $F_1, \dots, F_k$, say $F_1, \dots, F_B$ for $B = \O(1)$ we get that
with high probability
\begin{align}
\label{eq:polySizedWorks}
    \min_{j \in [B]} \norm{\wstr^{F_j} \circ \mL\bDelta^*}_1 \le \O(1) \norm{\mL\bDelta^*}_1.
\end{align}
That is, it suffices to solve \eqref{eq:mmc} on $G/F_1, \dots, G/F_B$ to find an $\O(1)$-approximate solution for $G$.

We provide all details including definitions and construction of the core graph in \Cref{sec:jtree}.

\subsection{Maintaining a Branching Tree Chain}
\label{overview:dssetup}
The goal of this section is to elaborate on how we combine core graphs and spanners to produce our overall data structure for our undirected min-ratio cycle problem, the $B$-branching tree chain.
We also describe how the data structure is maintained under dynamic updates, which is more formally shown in \cref{sec:routing}.
A central reason our hierarchical data structure works is that the components, both core graphs and spanners, are designed to remain very stable under dynamic changes to the input graphs they approximate. In the literature on dynamic graph algorithms, this is referred to as having \emph{low recourse}.
\begin{enumerate}
    \item Sample and maintain $B = O(\log n)$ $k$-LSDs $F_1, F_2, \dots, F_B$, and their associated core graphs $G/F_i$. Over the course of $O(m/k)$ updates at the top level, the forests $F_i$ are \emph{decremental}, i.e. only undergo edge deletions (from root insertions), and will have $\O(m/k)$ connected components.
    \item Maintain spanners $\SS(G, F_i)$ of the core graphs $G/F_i$, and embeddings $\Pi_{E(G/F_i) \to \SS(G, F_i)}$, say with length increase $\gamma_{\ell} = m^{o(1)}$.
    \item Recursively process the graphs $\SS(G, F_i)$, i.e. maintains LSDs and core graphs on those, and spanners on the contracted graphs, etc. Go for $d$ total levels, for $k^d = m$.
    \item Whenever a level $i$ accumulates $m/k^i$ total updates, hence doubling the number of edges in the graphs at that level, we rebuild levels $i, i+1, \dots, d$.
\end{enumerate}
Recall that on average, the LSDs stretch lengths by $\O(1)$, and the spanners $\SS(G, F_i)$ stretch lengths by $\gamma_{\ell}$. Hence the overall data structure stretches lengths by $\O(\gamma_{\ell})^d = m^{o(1)}$ (for appropriately chosen $d$).

We now discuss details on how to update the forests $G/F_i$ and spanners $\SS(G, F_i)$. Intuitively, every time an edge $e = (u, v)$ is changed in $G$, we will delete $\O(1)$ additional edges from $F_i$.
This ensures that no edge's total stretch/routing-length increases significantly due to the deletion of $e$ (\cref{lemma:globalstretch}).
As the forest $F_i$ undergoes edge deletions, the graph $G/F_i$ undergoes \emph{vertex splits}, where a vertex has a subset of its edges moved to a newly inserted vertex. Thus, a key component of our data structure is to maintain spanners and embeddings of graphs undergoing vertex splits (as well as edge insertions/deletions). It is important that the amortized recourse (number of changes) to the spanner $\SS(G, F_i)$ is $m^{o(1)}$ independent of $k$, even though the average degree of $G/F_i$ is $\Omega(k)$, and hence on average $\Omega(k)$ edges will move per vertex split in $G/F_i$. We discuss the more precise guarantees in \cref{overview:spanner}.

Overall, let every level have recourse $\gamma_r = m^{o(1)}$ (independent of $k$) per tree. Then each update at the top level induces $O(B\gamma_r)^d$ (as each tree branches into $B$ trees) updates in the data structure overall. Intuitively, for the proper choice of $d = \omega(1)$, both the total recourse $O(B\gamma_r)^d$ and approximation factor $\O(\gamma_{\ell})^d$ are $m^{o(1)}$ as desired.

\subsection{Going Beyond Oblivious Adversaries by using IPM Guarantees}
\label{overview:routing}

The precise data structure in the previous section only works for \emph{oblivious adversaries}, because we used that if we sampled $B = O(\log n)$ LSDs, then whp. there is a tree whose average stretch is $\O(1)$ with respect to a \emph{fixed flow} $\bf$. However, since we are updating the flow along the circulations returned by our data structure, we influence future updates, so the optimal circulations our data structure needs to preserve are not independent of the randomness used to generate the LSDs.
To overcome this issue we leverage the key fact that the flow $\bf^* - \bf$ is a good witness for the min-ratio cycle problem at each iteration.

\cref{lemma:optCirculationQuality} states that for any flow $\bf$, $\bg(\bf)^\top \bDelta(\bf) / (100m + \norm{\mL(\bf)\bDelta(\bf)}_1) \le -\wt{\Omega}(1)$ holds where $\bDelta(\bf) = \bf^* - \bf$.
Then, the best solution to \eqref{eq:mmc} among the LSDs $G/F_1, \dots, G/F_B$ maintains an $\O(1)$-approximation of the quality of the witness $\bDelta(\bf) = \bf^* - \bf$ as long as
\begin{align}
\label{eq:overviewLSDGoal1}
    \min_{j \in [B]} \norm{\wh{\mL}^{F_j} \bDelta(\bf)}_1 \le \O(1) \norm{\mL(\bf) \bDelta(\bf)}_1 + \O(m).
\end{align}
In this case, let $\wh{\bDelta}$ be the best solution obtained from $G/F_1, \dots, G/F_B$.
We have
\begin{align*}
    \frac{\bg(\bf)^\top \wh{\bDelta}}{\norm{\mL(\bf) \wh{\bDelta}}_1} \le \frac{\bg(\bf)^\top \bDelta(\bf)}{\O(1) \norm{\mL(\bf) \bDelta(\bf)}_1 + \O(m)} = -\wt{\Omega}(1).
\end{align*}
The additive $\O(m)$ term is there for a technical reason discussed later.

To formalize this intuition, we define the \emph{width} $\bw(\bf)$ of $\bDelta(\bf)$ as $\bw(\bf) = 100 \cdot \mathbf{1} + \Abs{\mL(\bf) \bDelta(\bf)}.$
The name comes from the fact that $\bw(\bf)_e$ is always at least $\Abs{\bell(\bf)_e(\bf^*_e - \bf_e)}$ for any edge $e.$
We show that the width is also slowly changing (\cref{lemma:setuphidden}) across IPM iterations, in that if the width changed by a lot, then the residual capacity of $e$ must have changed significantly. This gives our data structure a way to predict which edges' contribution to the length of the witness flow $\bf^* - \bf$ could have significantly increased.

Observe that for any forest $F_j$ in the LSD of $G$, we have $\norm{\wh{\mL}^{F_j} \bDelta(\bf)}_1 \le \norm{\wstr^{F_j} \circ \bw(\bf)}_1.$
Thus, we can strengthen \eqref{eq:overviewLSDGoal1} and show that the IPM potential can be decreased by $m^{-o(1)}$ if
\begin{align}
\label{eq:overviewLSDGoal2}
    \min_{j \in [B]} \norm{\wstr^{F_j} \circ \bw(\bf)}_1 \le \O(1) \norm{\bw(\bf)}_1.
\end{align}
\eqref{eq:overviewLSDGoal2} also holds with w.h.p if the collection of LSDs are built after knowing $\bf.$
However, this does not necessarily hold after augmenting with $\bDelta$, an approximate solution to \eqref{eq:mmc}.

Due to stability of $\bw(\bf)$, we have $\bw(\bf + \bDelta)_e \approx \bw(\bf)_e$ for every edge $e$ whose length does not change a lot.
For other edges, we update their edge length and force the stretch to be $1$, i.e. $\wstr^{F_j}_e = 1$ via the dynamic LSD maintenance, by shortcutting the routing of the edge $e$ at its endpoints.
This gives that for any $j \in [B]$, the following holds:
\begin{align*}
    \norm{\wstr^{F_j} \circ \bw(\bf + \bDelta)}_1 \lesssim \norm{\wstr^{F_j} \circ \bw(\bf)}_1 + \norm{\bw(\bf + \bDelta)}_1.
\end{align*}
Using the fact that $\min_{j \in [B]} \norm{\wstr^{F_j} \circ \bw(\bf)}_1 \le \O(1) \norm{\bw(\bf)}_1$, we have the following:
\begin{align*}
    \min_{j \in [B]} \norm{\wstr^{F_j} \circ \bw(\bf + \bDelta)}_1 \lesssim \O(1)\norm{\bw(\bf)}_1 + \norm{\bw(\bf + \bDelta)}_1.
\end{align*}

Thus, solving \eqref{eq:mmc} on the updated $G / F_1, \dots, G / F_B$ yields a good enough solution for reducing IPM potential as long as the width of $\bw(\bf+\bDelta)$ has not increased significantly, i.e. $\norm{\bw(\bf + \bDelta)}_1 \le \O(1)\norm{\bw(\bf)}_1.$

If the solution on the updated graphs $G / F_1, \dots, G / F_B$ does not have a good enough quality, we know by the above discussion that $\norm{\bw(\bf + \bDelta)}_1 \ge 100 \norm{\bw(\bf)}_1$ must hold.
Then, we re-compute the collection of LSDs of $G$ and solve \eqref{eq:mmc} on the new collection of $G / F_1, \dots, G / F_B$ again. Because each recomputation reduces the $\ell_1$ norm of the width by a constant factor, and all the widths are bounded by $\exp(\log^{O(1)}m)$ (as discussed in \cref{overview:ipm}), there can be at most $\O(1)$ such recomputations. At the top level, this only increases our runtime by $\O(1)$ factors.

The real situation is much more complicated since we recursively maintain the solutions on the spanners of each $G/F_1, \dots, G/F_B.$ Hence, it is possible that lower levels in the data structure are the ``reason'' that the quality of the solution is poor. More formally, let $T$ be the total number of IPM iterations.
We use $t \in [T]$ to index each iteration and use superscript $x^{(t)}$ to denote the state of any variable $x$ after $t$-th iteration.
For example, $\bf^{(t)}$ is the flow computed so far after $t$ IPM iterations and we define $\bw^{(t)} \defeq \bw(\bf^{(t)})$ to be the width w.r.t. $\bf^{(t)}$.
Recall that every graph maintained in the dynamic $B$-Branching Tree Chain re-computes its collection of LSDs after certain amount of updates.
When some graph at level $i$ re-computes, we enforce every graph at the same level to re-compute as well.
Since there's only $m^{o(1)}$ such graphs at each level, this scheme results in a $m^{o(1)}$ overhead on the update time which is tolerable.
For every level $i = 0, \dots, d$, we define $\prev_i^{(t)}$ to be the most recent iteration at or before $t$ that a re-computation of LSDs occurs at level $i.$
For graphs at level $d$ which contain only $m^{o(1)}$ vertices, we enforce a rebuild everytime and always have $\prev_d^{(t)} = t$.
We show in \cref{lemma:dynamicchain} that the cycle output by the data structure in the $t$-th IPM iteration has length at most
\begin{align*}
    m^{o(1)} \sum_{i=0}^d \|\bw^{(\prev_i^{(t)})}\|_1.
\end{align*}
This inequality is a natural generalization of the $\O(1)\left(\norm{\bw(\bf)}_1 + \norm{\bw(\bf + \bDelta)}_1\right)$-bound when taking recursive structure into account. 

At this point, we want to emphasize that the fact that we can prove this guarantee depends on certain ``monotonicity'' properties of both our core and spanner graph constructions.
In the core graph construction, it is essential that we can provide a fixed length upper bound for most edges. In the spanner construction, we crucially use that the set of edges routing into any fixed edge in the spanner is \emph{decremental} for most spanner edges. This allows us to produce an initial upper bound on the width for edges in the spanner and continue using this bound as long as the spanner edge routes a decremental set.

The cycle output by the data structure yields enough decrease in the IPM potential if its 1-norm is small enough.
Otherwise, the 1-norm of the output cycle is large and we know that $\sum_{i=0}^d \|\bw^{(\prev_i^{(t)})}\|_1$ is much more than $m^{o(1)}\|\bw^{(t)}\|_1.$ In this way, the data structure can fail because some lower level $i$ has $\|\bw^{(\prev_i^{(t)})}\|_1 \gg \|\bw^{(t)}\|_1$.
A possible fix is to rebuild the entire data structure which sets $\prev_i^{(t)} = t$ at any level $i.$
However, this costs linear time per rebuild, and this may need to happen almost every iteration because there are multiple levels.
In the next section we show how to leverage that lower levels have cheaper rebuilding times (levels $i, i+1, \dots, d$ can be rebuilt in time approximately $m^{1+o(1)}/k^i$) to design a more efficient rebuilding schedule.

\subsection{The Rebuilding Game}
\label{overview:game}
Our goal in \Cref{sec:rebuilding}
is to develop a strategy that finds approximate min-ratio cycles without spending too much time rebuilding our data structure when it fails to do so. 
In the previous overview section, we carefully characterized the conditions under which our data structure can fail against adversarial updates, given the promise that $\bf^* - \bf$ remains a good witness circulation.
In this section, we set up a game which abstracts the properties of the data structure and the adversary.
The player in this game wants to ensure our data structure works correctly by rebuilding levels of it when it fails.
We show that the player can win without spending too much time on rebuilding.

Recall $\bw^{(t)} \defeq \bw(\bf^{(t)})$ is a hidden vector that we use to upper bound the $\ell_1$ cost of the hidden witness circulation $\bDelta(\bf)$.
We will refer to $\|\bw^{(t)} \|_1$ as the total width at time $t$.
We argued in the previous \cref{overview:routing} that our branching-tree data structure can find a good cycle whenever
the total width $\|\bw^{(t)} \|_1$ is not too small compared to the total widths at the times when the levels $0,1,\ldots,d$ of the data structure were last initialized or rebuilt.
We let $\prev_i^{(t)}$ denote the stage when level $i$ was last rebuilt, and refer to $\|\bw^{(\prev_i^{(t)})} \|_1$ as the total width at level $i$.
As we saw in the previous section, the only way our cycle-finding data structure can fail to produce a good enough cycle is if  $\sum_{i=0}^d \|\bw^{(\prev_i^{(t)})}\|_1 \gg m^{o(1)} \| \bw^{(t)}\|_1 $.
We can estimate the quality of the cycles we find, and if we fail to find a good cycle we can conclude this undesired condition holds.
However, even if the condition holds, we might still find a good cycle ``by accident'', so finding a cycle does not prove that the data structure currently estimates the total width well.
Because the total widths $\|\bw^{(t)} \|_1$  are hidden from us, we do not know which level(s) cause the problem when we fail to find a cycle.

We turn this into a game that abstracts the data structure and IPM and supposes that  
total width $\|\bw^{(t)} \|_1$ is an arbitrary positive number chosen by an adversary, while a player (our protagonist) manages the data structure by rebuilding levels of the data structure to set $\prev_i^{(t)} = t$ when necessary.
Now, because of well-behaved numerical properties of our IPM, we are guaranteed that $\log(\| \bw^{(t)}\|_1) \in [-\poly\log(m),\poly\log(m)]$, and we impose this condition on the total width in our game as well.
By developing a strategy that works against any adversary choosing such total widths, we ensure our data structure will work with our IPM as a special case.
In \Cref{def:game} we formally define our rebuilding game.

In our branching tree data structure, level $i$ can be rebuilt at a cost of $m^{1+o(1)}/k^i$ and it can last through roughly $m^{1-o(1)}/k^i$ cycle updates before we have to rebuild it because the core graph has grown too large (we call this a ``winning rebuild'').
But, if we are unable to find a good cycle, we are forced to rebuild sooner (we call this a ``losing rebuild'').
Which level should we rebuild if we are unable to find a good cycle?
The answer is not immediately clear, because any level could have too large total width.
However, by tuning our parameters such that the $m^{o(1)}$ factor in our condition $\sum_{i=0}^d \|\bw^{(\prev_i^{(t)})}\|_1 \gg m^{o(1)} \|\bw^{(t)}\|_1 $ is larger than $2(d+1),$ we can deduce that if a failure occurs, then $\max_{i=0}^d \|\bw^{(\prev_i^{(t)})}\|_1 > 2 \|\bw^{(t)}\|_1.$ 
Thus, \emph{if} the total width at level $i$ is too large, then a losing rebuild at level $i$  (and hence updating $\bw^{(\prev_i^{(t+1)})}$ to $\bw^{(t)}$)  will reduce its total width by at least a factor 2.

This means that for any level $i,$ if we do a losing rebuild of level $i$ $\poly\log(m)$ times before a winning rebuild of level $i,$ we can conclude that the too-large total width is not at level $i.$
This leads to the following strategy: Starting at the lowest level, do a losing rebuild of each level $i$ up to $\poly\log(m)$ times after each winning rebuild,
and then move to rebuilding level $i-1$ in case of more failures.
We state this strategy more formally in \Cref{alg:rebuildStrategy}.
This leads to a cost of $O(m^{o(1)}(m+T))$ to process $T$ cycle updates in the rebuilding game, as we prove in \Cref{lemma:rebuildinggame}.

Finally, at the end of \Cref{sec:rebuilding}, we combine the data structure designed in the previous sections with our strategy for the rebuilding game to create a data structure that handles successfully finds update cycles in our hidden stable-flow chasing setting in amortized $m^{o(1)}$ cost per cycle update, which is encapsulated in \Cref{thm:MMCHiddenStableFlow}.

\subsection{Dynamic Embeddings into Spanners of Decremental Graphs}
\label{overview:spanner}

It remains to describe the algorithm to maintain a spanner $\SS(G, F_i)$ on the graphs $G / F_i$. Let us recall the requirements on the spanner given in \Cref{overview:dssetup}:
\begin{enumerate}
    \item Sparsity: at all times the spanner should be sparse, i.e. consist of at most $\O(|V(\SS(G, F_i))|)$ edges. This is crucial for reducing the problem size and as we ensure that $F_i$ has only $\O(m/k)$ connected components, we have that $\SS(G,F_i)$ consists of $\tilde{O}(m/k)$ edges, reducing the problem size by a factor of almost $k$. 
    
    \item Low Recourse: we further require that for each update to $G / F_i$, there are at most $\gamma_r = m^{o(1)}$ changes to $\SS(G, F_i)$ on average. This is crucial as otherwise the updates to $\SS(G, F_i)$ could trigger even more updates in the $B$-Branching Tree Chain (see \cref{overview:dssetup}). 
    
    \item {Short Paths with Embedding:} we maintain the spannner such that for every edge $e$ in $G$, its endpoints in $\SS(G, F_i)$ are at distance at most $\gamma_l \cdot \ell(e)$ and even maintain witness paths $\Pi_{G \to \SS(G, F_i)}(e)$ between the endpoints consisting of $\gamma_l$ edges. This is crucial as we need an explicit way to check whether $e \oplus \Pi_{G/F_i \to \SS(G, F_i)}(e)$ is a good solution to the min-ratio cycle problem.
    
    \item Small Set of New Edges That We Embed Into: we ensure that after each update, we return a set $D$ consisting of $m^{o(1)}$ edges such that each edge $e$ in $G/F_i$ is embedded into a path $\Pi_{G/F_i \to \SS(G, F_i)}(e)$ consisting of the edges on the path of the old embedding path $\Pi_{G/F_i \to \SS(G, F_i)}(e)$ of $e$ and edges in $D$.
    
    \item {Efficient Update Time:} we show how to maintain $\SS(G, F_i)$ with amortized update time $k m^{o(1)}$. 
\end{enumerate}

We note that additionally, we need our spanner to work against adaptive adversaries since the update sequence is influenced by the output spanner. Although spanners have been studied extensively in the dynamic setting, there is currently only a single result that works against adaptive adversaries. While this spanner given in \cite{bernstein2020fully} appears promising, it does not ensure our desired low recourse property for vertex splits and this seems inherent to the algorithm (additionally, it also  does not maintain an embedding $\Pi_{G/F_i \to \SS(G, F_i)}$). 

While we use similar elements as in \cite{bernstein2020fully} to obtain spanners statically, we arrive at a drastically different algorithm that can deal well with vertex splits. We focus first on obtaining an algorithm with low recourse and discuss afterwards how to implement it efficiently.

\paragraph{A Static Algorithm.} We first consider the static version of the problem on a graph $G/F_i$, i.e. to give a static algorithm that computes a spanner with short path embeddings. By using a simple bucketing scheme over edge lengths, we can assume wlog that all lengths have unit-weight. We partition the graph into edge-disjoint expander graphs $H_1, H_2, \ldots, H_k$ where each $H_i$ has roughly uniform degree, i.e. its average degree is at most a polylogarithmic factor larger than its minimum degree $\Delta_{min}(H_i)$, and each vertex $v$ in $G$ is in at most $\O(1)$ graphs $H_i$. Here, we define an expander to be a graph $H_i$ that has no cut $(X, \overline{X})$ where $\overline{X} = V(H_i) \setminus X$ with $|E_{H_i}(X, \overline{X})| < \Omega\left(\frac{1}{\log^3(m)}\right) \min\{ \vol_{H_i}(X), \vol_{H_i}(\overline{X})\}$ where $E_{H_i}(X, \overline{X})$ is the set of edges in $H_i$ with endpoints in $X$ and $\overline{S}$ and $\vol_{H_i}(Y)$ is the sum of degrees over the vertices $y \in Y$.

Next, consider any such expander $H_i$. It is well-known that sampling edges in expanders with probability $p_i \sim \frac{\log^4(m)}{\Delta_{min}(H_i)}$ gives a cut-sparsifier $\SS_i$ of $H_i$, i.e. a graph such that for each cut $(X, \overline{X})$, we have $|E_{H_i}(X, \overline{X})| \approx |E_{\SS_i}(X, \overline{X})|/p_i$ (see \cite{ST04, bernstein2020fully}). This ensures that also $\SS_i$ is an expander. It is well-known that any two vertices in the same expander are at small distance, i.e. there is a path of length at most $\O(1)$ between them. We use a dynamic shortest paths data structure \cite{CS21} for expander graphs on $\SS_i$ to find such short paths between the endpoints of each edge $e$ in $G / F_i$ and take them to be the embedding paths (here we lose an $m^{o(1)}$ factor in the length of the paths due to the data structure).

It remains to observe that each spanner $\SS_i$ has a nearly linear number of edges because each graph $H_i$ has average degree close to its minimum degree, and edges are sampled independently with probability $p_i$. Thus, letting $\SS(G, F_i)$ be the union of all graphs $\SS_i$ and using that each vertex is in at most $\O(1)$ graphs $H_i$, we conclude the desired sparsity bound on $\SS(G, F_i)$. We take $\Pi_{G/F_i \to \SS(G, F_i)}$ to be the union of the embeddings constructed above and observe that the length of embedding paths is at most $m^{o(1)}$ as desired.

\paragraph{The Dynamic Algorithm.} To make the above algorithm dynamic, let us assume that there is a spanner $\SS(G, F_i)$ with corresponding embedding $\Pi_{G/F_i \to \SS(G, F_i)}$ and after its computation, a batch of updates $U$ is applied to $G/F_i$ (consisting of edge insertions/deletions and vertex splits). Clearly, after forwarding the updates $U$ to the current spanner $\SS(G, F_i)$, by deleting edges that were deleted from $G/F_i$ and splitting vertices, we have that for some edges $e \in G /F_i$, the updated embedding $\Pi_{G/F_i \to \SS(G, F_i)}(e)$ might no longer be a proper path.

We therefore need to add new edges to $\SS(G, F_i)$ and fix the embedding. We start by defining $S$ to be the vertices that are touched by an update in $U$, meaning for the deletion/insertion of edge $(u,v)$ we add $u$ and $v$ to $S$ and for a vertex split of $v$ into $v$ and $v'$, we add $v$ and $v'$ to $S$. Note that $|S| \leq 2|U|$ and that all $\Pi_{G/F_i \to \SS(G, F_i)}(e)$ that are no longer proper paths intersect with $S$. 

We now fix the embedding by constructing a new static spanner on a special graph $J$ over the vertices of $S$. More precisely, for each $e = (a,b)$ in $G / F_i$ where $\Pi_{G/F_i \to \SS(G, F_i)}(e)$ intersects with $S$, we find the vertices $\hat{a},\hat{b}$ in $S$ that are closest to $a$ and $b$ on $\Pi_{G/F_i \to \SS(G, F_i)}(e)$, and then insert an edge $\hat{e} = (\hat{a}, \hat{b})$ into the graph $J$. We say that $e$ is the pre-image of $\hat{e}$ (and $\hat{e}$ the image of $e$ in $J$). 

Finally, we run the static algorithm from the last paragraph to find a sparsifier $\tilde{J}$ of $J$ and let $\Pi_{J \to \tilde{J}}$ be the corresponding embedding. Then, for each edge $\hat{e}$ that was sampled into $\tilde{J}$, we add its pre-image $e$ to the current sparsifier $\SS(G, F_i)$. 

To fix the embedding, for each $\hat{e} = (\hat{a}, \hat{b}) \in \tilde{J}$, we observe that since $e = (a,b)$ was added to $\SS(G, F_i)$, we can simply embed the edge into itself. We define for each such edge $\hat{e}$ the path 
\[
P_{\hat{e}} = \Pi_{G/F_i \to \SS(G, F_i)}(e)[\hat{a}, a] \oplus (a,b) \oplus \Pi_{G/F_i \to \SS(G, F_i)}(e)[b, \hat{b}]
\]
which is a path between the endpoints of $\hat{e}$. This path is in the current graph $\SS(G, F_i)$ since we added $(a,b)$ to the spanner and by definition of $\hat{a}$, we have that $\Pi_{G/F_i \to \SS(G, F_i)}(e)[\hat{a}, a]$ is still a proper path, the same goes for $\hat{b}$.

But this means we can embed each edge $f = (c,d)$ even if its image $\hat{f} = (\hat{c}, \hat{d}) \not\in \tilde{J}$, since we can simply set it to the path
\[
    \Pi_{G/F_i \to \SS(G, F_i)}(f)[c, \hat{c}] \oplus \left( \bigoplus_{\hat{e} \in \Pi_{J \to \tilde{J}}(\;\hat{f}\;)} P_{\hat{e}} \right)  \oplus \Pi_{G/F_i \to \SS(G, F_i)}(f)[\hat{d}, d]. 
\]
By the guarantees from the previous paragraph, we have that the sparsifier $\tilde{J}$ has average degree $\O(1)$, and we only added the pre-images of edges in $\tilde{J}$ to $\SS(G, F_i)$. Since $J$ (and $\tilde{J}$) are taken over the vertex set $S$, we can conclude that we only cause $\O(|S|) = \O(|U|)$ recourse to the spanner. Further, since each new path $\Pi_{G \to \SS(G, F_i)}(e)$ for each $e$ now consists of $\O(1)$ path segments from the old embedding $\Pi_{G \to \SS(G, F_i)}$ (plus $\O(1)$ edges), the maximum length of the the embedding paths has only increased by a factor of $\O(1)$ overall. %
Finally, we take $D$ to be the set of edges on $P_{\hat{e}}$ for all $\hat{e} \in \tilde{J}$. Clearly, each edge $f$ embeds into a subpath of its previous embedding path (to reach the first and last vertex in $S$) and into some paths $P_{\hat{e}}$ all of which now have edges in $D$. To bound the size of $D$, we observe that also each path $P_{\hat{e}}$ is of short length since it is obtained from combining two old embedding paths (which were short) and a single edge. Thus, we have $|D| = |\bigcup_{\hat{e} \in \tilde{J}} P_{\hat{e}}\;| = \O(|\tilde{J}|) = \O(|U|)$ which again is only $\O(1)$ when amortizing over the number of updates. \cref{fig:SpannerMaintenance} gives an example of this spanner maintenance procedure in action.
\begin{figure}
  \centering
  \includegraphics[width=\linewidth]{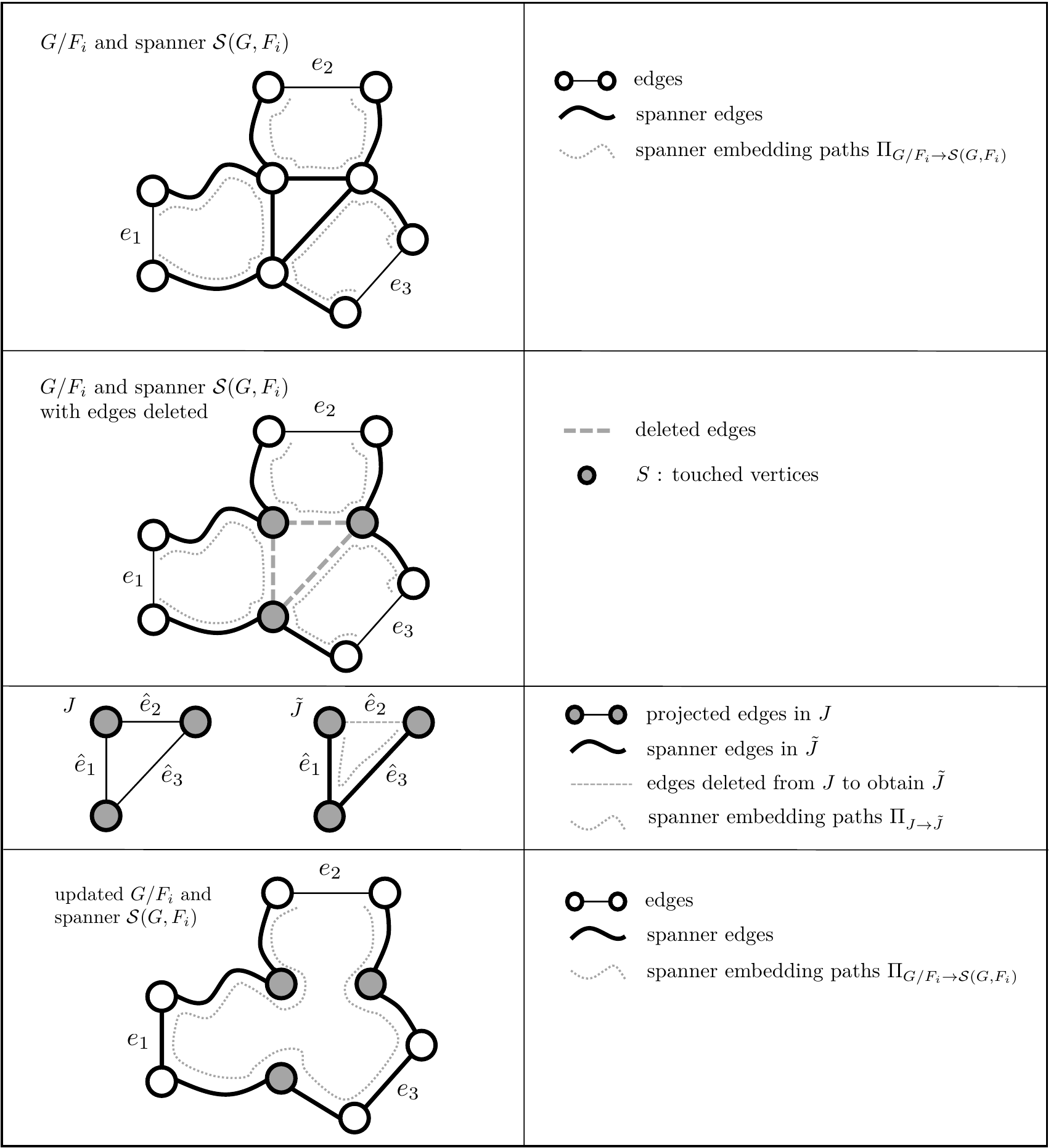}
\caption[Maintaining the Sparsifier under edge deletions]{Illustration of the procedure for maintaining  $\SS(G, F_i)$ under edge deletions.}
\label{fig:SpannerMaintenance}
\end{figure}

By using standard batching techniques, we can also deal with sequences of update batches $U^{(1)}, U^{(2)}, \ldots$ to the spanner and ensure that we cause only $m^{o(1)}$ amortized recourse per update/ size of $D$ to the spanner.

\paragraph{An Efficient Implementation.} While the algorithm above achieves low recourse, so far, we have not reasoned about the run-time. To do so, we enforce low \emph{vertex-congestion} of $\Pi_{G/F_i \to \SS(G, F_i)}$ defined to be the maximum number of paths $\Pi_{G/F_i \to \SS(G, F_i)}(e)$ that any vertex $v \in V(G/F_i)$ occurs on. More precisely, we implement the algorithm above such that the vertex congestion of $\Pi_{G/F_i \to \SS(G, F_i)}(e)$ remains of order $\gamma_c \Delta_{max}(G/F_i)$ for some $\gamma_c = m^{o(1)}$ over the entire course of the algorithm. We note that by a standard transformation, we can assume wlog that $\Delta_{max}(G/F_i) = \O(k)$.

Crucially, using our bound on the vertex congestion, we can argue that the graph $J$ has maximum degree $\gamma_c \Delta_{max}(G/F_i)$. Since we can implement the static spanner algorithm in time near-linear in the number of edges, this implies that the entire algorithm to compute a sparsifier $\tilde{J}$ only takes time $\sim |U| \gamma_c \Delta_{max}(G /F_i) \approx |U| m^{o(1)}k$, and thus in amortized time $k m^{o(1)}$ per update.

It remains to obtain this vertex congestion bound. Let us first discuss the static algorithm. Previously, we exploited that each sparsifier $\SS_i$ is expander since it is a cut-sparsifier of $H_i$ in a rather crude way. But it is not hard to see via the multi-commodity max-flow min-cut theorem \cite{leighton1999multicommodity} that this property can be used to argue the existence of an embedding $\Pi_{H_i \to \SS_i}$ that uses each edge in $\SS_i$ on at most $\O(1/p_i)$ embedding paths and therefore each path has average length $\O(1)$. In fact, using the shortest paths data structures on expanders \cite{CS21}, we can find such an embedding and turn the average length guarantee into a worst-case guarantee.

This ensures that each edge has congestion at most $\O(1/p_i) = \O(\Delta_{max}(G/F_i))$ and because $\SS(G, F_i)$ has average degree $\O(1)$, this also bounds the vertex congestion. We need to refine this argument carefully for the dynamic version but can then argue that due to the batching we only increase the vertex congestion slightly. We refer the reader to \Cref{sec:spanner} for the full implementation and analysis.

%% file: prelim.tex
\section{Preliminaries}
\label{sec:prelim}

\paragraph{Model of Computation.} In this article, for problem instances encoded with $z$ bits, all algorithms work in fixed-point arithmetic where words have $O(\log^{O(1)} z)$ bits, i.e. we prove that all numbers stored are in $[\exp(-\log^{O(1)}z), \exp(\log^{O(1)}z)]$.

\paragraph{General notions.} We denote vectors by boldface lowercase letters. We use uppercase boldface to denote matrices. Often, we use uppercase matrices to denote the diagonal matrices corresponding to lowercase vectors, such as $\mL = \diag(\bell)$. For vectors $\bx, \by$ we define the vector $\bx \circ \by$ as the entrywise product, i.e. $(\bx \circ \by)_i = \bx_i\by_i$. We also define the entrywise absolute value of a vector $|\bx|$ as $|\bx|_i = |\bx_i|$. We use $\l \cdot, \cdot \r$ as the vector inner product: $\l \bx, \by \r = \bx^\top\by = \sum_i \bx_i\by_i$. We elect to use this notation when $\bx, \by$ have superscripts (such as time indexes) to avoid cluttering. For positive real numbers $a, b$ we write $a \approx_{\alpha} b$ for some $\alpha > 1$ if $\alpha^{-1}b \le a \le \alpha b$. For positive vectors $\bx, \by \in \R^{[n]}_{>0}$, we say $\bx \approx_{\alpha} \by$ if $\bx_i \approx_{\alpha} \by_i$ for all $i \in [n]$. This notion extends naturally to positive diagonal matrices. We will need the standard Chernoff bound.

\begin{theorem}[Chernoff Bound]\label{thm:chernoffBound}
Suppose $X_1, X_2, \ldots, X_k \in [0,W]$ are independent random variables, $X = \sum_i X_i$ and $\mu = \mathbb{E} X$. For any $\delta \ge 1/2$, we have $\mathbb{P}[X \in [(1-\delta)\mu, (1+\delta)\mu]] \ge 1 - 2e^{-\frac{\delta \mu}{6W}}$.
\end{theorem}

\paragraph{Graphs.} In this article, we consider multi-graphs $G$, with edge set $E(G)$ and vertex set $V(G)$. When the graph is clear from context, we use the short-hands $E$ for $E(G)$, $V$ for $V(G)$, $m = |E|, n = |V|$. We assume that each edge $e \in E$ has an implicit direction, used to define its edge-vertex incidence matrix $\mB$. Abusing notation slightly, we often write $e = (u,v) \in E$ where $e$ is an edge in $E$ and $u$ and $v$ are the tail and head of $e$ respectively (note that technically multi-graphs do not allow for edges to be specified by their endpoints). We let $\rev(e)$ be the edge $e$ reversed: if $e = (u, v)$ points from $u$ to $v$, then $\rev(e)$ points from $v$ to $u$.  

We say a flow $\bf \in \R^E$ routes a demand $\bd \in \R^V$ if $\mB^\top\bf=\bd$. For an edge $e = (u, v) \in G$ we let $\bb_e \in \R^V$ denote the demand vector of routing one unit from $u$ to $v$.

We denote by $\deg_G(v)$ the degree of $v$ in $G$, i.e. the number of incident edges. We let $\Delta_{\max}(G)$ and $\Delta_{\min}(G)$ denote the maximum and minimum degree of graph $H$. We define the volume of a set $S \subseteq V$ as $\vol_G(S) \defeq \sum_{v \in S} \deg_G(v)$. 

\paragraph{Dynamic Graphs.} \label{para:dynG} 
We say $G$ is a \emph{dynamic} graph, if it undergoes \emph{batches} $U^{(1)}, U^{(2)}, \ldots$ of updates consisting of edge insertions/ deletions and/or vertex splits that are applied to $G$. We stress that results on dynamic graphs in this article often only consider a subset of the update types and we therefore explicitly state for each dynamic graph which updates are allowed. We say that the graph $G$, after applying the first $t$ update batches $U^{(1)}, U^{(2)}, \ldots, U^{(t)}$, is at \emph{stage} $t$ and denote the graph at this stage by $G^{(t)}$. Additionally, when $G$ is clear, we often denote the value of a variable $x$ at the end of stage $t$ of $G$ by $x^{(t)}$, or a vector $\bx$ at the end of stage $t$ of $G$ by $\bx^{(t)}$. 

For each update batch $U^{(t)}$, we encode edge insertions by a tuple of tail and head of the new edge and deletions by a pointer to the edge that is about to be deleted. We further also encode vertex splits by a sequence of edge insertions and deletions as follows: if a vertex $v$ is about to be split and the vertex that is split off is denoted $v^{\text{NEW}}$, we can delete all edges that are incident to $v$ but should be incident to $v^{\text{NEW}}$ from $v$ and then re-insert each such edge via an insertion (we allow insertions to new vertices, that do not yet exist in the graph). 

For technical reasons, we assume that in an update batch $U^{(t)}$, the updates to implement the vertex splits are last, and that we always encode a vertex split of $v$ into $v$ and $v^{\text{NEW}}$ such that $\deg_{G^{(t+1)}}(v^{\text{NEW}}) \leq \deg_{G^{(t+1)}}(v)$. We let the vertex set of graph $G^{(t)}$ consist of the union of all endpoints of edges in the graph (in particular if a vertex is split, the new vertex $v^{\text{NEW}}$ is added due to having edge insertions incident to this new vertex $v^{\text{NEW}}$ in $U^{(t)}$).

$\Enc(u)$ of an update $u \in U^{(t)}$ be the size of the encoding of the update and note that for edge insertions/ deletions, we have $\Enc(u) = \tilde{O}(1)$ and for a vertex split of $v$ into $v$ and $v^{\text{NEW}}$ as described above we have $\Enc(u) = \tilde{O}(\deg_{G^{(t+1)}}(v^{\text{NEW}}))$. For a batch of updates $U$, we let $\Enc(U) = \sum_{u \in U} \Enc(u)$. In this article, we only consider dynamic graphs where the total size of the encodings of all update batches is polynomially bounded in the size of the initial graph $G^{(0)}$.

We point out in particular that the number of updates $|U|$ in an update batch $U$ can be completely different from the actual encoding size $\Enc(U)$ of the update batch $U$.

\paragraph{Paths, Flows, and Trees.} Given a path $P$ in $G$ with vertices $u,v$ both on $P$, then we let $P[u,v]$ denote the path segment on $P$ from $u$ to $v$. We note that if $v$ precedes $u$ on $P$, then the segment $P[u,v]$ is in the reverse direction of $P$. For forests $F$, we similarly define $F[u,v]$ as the path from $u$ to $v$ along edges in the forest $F$. We ensure that $u, v$ are in the same connected component of $F$ whenever this notation is used.

We let $\bp(F[u,v]) \in \R^{E(G)}$ denote the flow vector routing one unit from $u$ to $v$ along the path in $F$. In this way, $|\bp(F[u,v])|$ is the indicator vector for the path from $u$ to $v$ on $F$. Note that $\bp(F[u,v])+\bp(F[v,w])=\bp(F[u,w])$ for any vertices $u, v, w \in V$.

The \emph{stretch} of $e = (u, v)$ with respect to a tree $T$ is defined as
\[ \str^{T,\bell}_e \defeq 1 + \frac{\langle \bell, |\bp(T[u, v])|\rangle}{\bell_e} = 1 + \frac{\sum_{e' \in T[u,v]} \bell_{e'}}{\bell_e}. \]
This differs slightly from the more common definition of stretch because of the $1+$ term -- we do this to ensure that $\str^{T,\bell}_e \ge 1$ for all $e$. It is known how to efficiently construct trees with polylogarithmic average stretch with respect to underlying weights. These are called low-stretch spanning trees (LSSTs).
\begin{theorem}[Static LSST \cite{AN19:journal}]
\label{thm:an}
Given a graph $G = (V, E)$ with lengths $\bell \in \R^E_{>0}$ and weights $\bv \in \R^E_{>0}$ there is an algorithm that runs in time $\O(m)$ and computes a tree $T$ such that $\sum_{e \in E} \bv_e \str^{T, \bell}_e \le O(\|\bv\|_1 \log n \log \log n).$
\end{theorem}
We let $\gamma_{LSST} \defeq O(\log n \log \log n)$.

\paragraph{Graph Embeddings.} %
Given graphs $G$ and $H$ with $V(G) \subseteq V(H)$, we say that $\Pi_{G \xrightarrow{} H}$ is an \emph{graph-embedding} from $G$ into $H$ if it maps each edge $e^G = (u,v) \in E(G)$ to a $u$-$v$ path $\Pi_{G \xrightarrow{} H}(e^G)$ in $H$. We define congestion of an edge $e^H$ by $\econg(\Pi_{G \xrightarrow{} H}, e^H) \defeq |\{ e^G \in E(G) \;|\; e^H \in \Pi_{G \xrightarrow{} H}(e^G) \}|$ and of the embedding by $\econg(\Pi_{G \xrightarrow{} H}) \defeq \max_{e^H \in E(H)} \econg(\Pi_{G \xrightarrow{} H}, e^H)$. Analogously, the congestion of a vertex $v^H \in V(H)$ is defined by $\vcong(\Pi_{G \xrightarrow{} H}, v^H) \defeq |\{ e^G \in E(G) \;|\; v^H \in \Pi_{G \xrightarrow{} H}(e^G) \}|$ and the 
vertex-congestion of the embedding by $\vcong(\Pi_{G \xrightarrow{} H}) \defeq \max_{v^H \in V(H)} \vcong(\Pi_{G \xrightarrow{} H}, v^H)$. We define the length by $\length(\Pi_{G \xrightarrow{} H}) \defeq \max_{e^G \in E(G)} |\Pi_{G \xrightarrow{} H}(e^G)|$.
We let (boldface) $\bPi_{G\to H}(e) \in \R^E$ for $e = (u, v)$ denote a vector representing the flow from $u \to v$. Thus $\mB^\top\bPi_{G\to H}(e) = \bb_e$.

For a path $p_1$ with endpoints $u \to v$ and $p_2$ with endpoints $v \to w$, we define $p_1 \oplus p_2$ as the concatenation, which is a path from $u \to w$.

Sometimes we consider the edges that route into an edge $e \in E(H)$.
Given graphs $G, H$ and embedding $\Pi_{G\to H}$, for edge $e \in E(H)$ we define $\Pi_{G\to H}^{-1}(e) \defeq \left\{e' \in G : e \in \Pi(e') \right\}.$ The notation is natural if we think of $\Pi$ as a function from an edge $e$ to the set of edges in its path, and hence $\Pi^{-1}$ is the inverse/preimage of a one-to-many function.

\paragraph{Dynamic trees.} Our algorithms make heavy use of dynamic tree data structures, so we state a lemma describing the variety of operations that can be supported on a dynamic tree.
This includes path updates either of the form adding a directed flow along a tree path, or adding a positive value to each edge on a tree path. Additionally, the data structure can support changing edges in the tree, and querying flow values on edge. Each of these operations can be performed in amortized $\O(1)$ time.
\begin{lemma}[Dynamic trees, see \cite{ST83}]
\label{algo:LCT}
There is a deterministic data structure $\mathcal{D}^{(T)}$ that maintains a dynamic tree $T \subseteq G = (V, E)$ under insertion/deletion of edges with gradients $\bg$ and lengths $\bell$, and supports the following operations:
\begin{enumerate}
    \item Insert/delete edges $e$ to $T$, under the condition that $T$ is always a tree, or update the gradient $\bg_e$ or lengths $\bell_e$. The amortized time is $\O(1)$ per change.
    \item For a path vector $\bDelta = \bp(T[u,v])$ for some $u, v \in V$, return $\l \bg, \bDelta \r$ or $\l \bell, |\bDelta|\r$ in time $\O(1)$.
    \item \label{item:nonpositiveflow} Maintain a flow $\bf \in \R^E$ under operations $\bf \assign \bf + \eta\bDelta$ for $\eta \in \R$ and path vector $\bDelta = \bp(T[u,v])$, or query the value $\bf_e$ in amortized time $\O(1)$.
    \item \label{item:positiveflow} Maintain a positive flow $\bf \in \R^E_{>0}$ under operations $\bf \assign \bf + \eta|\bDelta|$ for $\eta \in \R_{\ge0}$ and path vector $\bDelta = \bp(T[u,v])$, or or query the value $\bf_e$ in amortized time $\O(1)$.
    \item $\textsc{Detect}()$. For a fixed parameter $\eps$, and under positive flow updates (item \ref{item:positiveflow}), where $\bDelta^{(t)}$ is the update vector at time $t$, returns
    \begin{align}
        \label{eq:detect} S^{(t)} \defeq \left\{ e \in E : \bell_e \sum_{t' \in [\last^{(t)}_e+1,t]} |\bDelta_e^{(t')}| \ge \eps \right\}
    \end{align}
    where $\last^{(t)}_e$ is the last time before $t$ that $e$ was returned by $\textsc{Detect}()$. Runs in time $\O(|S^{(t)}|)$.
\end{enumerate}
\end{lemma}
\begin{proof}
Every operation described is standard except for \textsc{Detect}, which we now give an algorithm for.
Note that \eqref{eq:detect} is equivalent to the following:
\begin{align*}
    \sum_{t' \in [\last^{(t)}_e+1, t]} |\bDelta_e^{(t')}| -\frac{\eps}{\bell_e} \ge 0.
\end{align*}
This value can be maintained using positive flow updates (item \ref{item:positiveflow}), i.e. $\Abs{\bDelta}$ to a tree path.
We reset the value of an edge $e$ to $-\eps / \bell_e$ once it is detected.
Locating and collecting edges satisfying \eqref{eq:detect} is reduced to finding edges with nonnegative values, which can be done in $\O(|S^{(t)}|)$ time by repeatedly querying the largest value on the tree, and checking whether it is nonnegative.
\end{proof}
The \textsc{Detect} operation allows our algorithm to decide when we need to change the gradients and lengths of an edge $e$ in our IPM.

%% file: ipm.tex
\section{Potential Reduction Interior Point Method}
\label{sec:ipm}

The goal of this section is to present a primal-only potential reduction IPM \cite{K84} that solves the min-cost flow problem on a graph $G = (V, E)$ with demands $\bd \in \mathbb{Z}^V$, lower and upper capacities $\bu^-, \bu^+ \in \mathbb{Z}^E$, and costs $\bc \in \mathbb{Z}^E$ such that all integers are bounded by $U$:
\begin{align} \bf^* \defeq \argmin_{\substack{\mB^\top\bf=\bd \\ \bu^-_e \le \bf_e \le \bu^+_e \forall e \in E}} \bc^\top \bf. \label{eq:mincostopt} \end{align}

Instead of using the standard logarithmic barrier, we elect to use the barrier $x^{-\alpha}$ for small $\alpha$. This is because we do not know how to prove that the lengths encountered during the algorithms are quasipolynomially bounded for the logarithmic barrier. Precisely, we consider the following potential function, where $F^* \defeq \bc^\top \bf^*$ is the optimal value for \eqref{eq:mincostopt}, and $\alpha \defeq 1/(1000 \log mU)$. We assume that we know $F^*$, as running our algorithm allows us to binary search for $F^*$.
\begin{align}
    \Phi(\bf) \defeq 20m \log(\bc^\top \bf - F^*) + \sum_{e \in E} \left((\bu^+_e - \bf_e)^{-\alpha} + (\bf_e - \bu^-_e)^{-\alpha} \right) \label{eq:karmarkar}
\end{align}
We show in \cref{subsec:initialfinal} that we can initialize a flow $\bf$ on a larger graph (still with $O(m)$ edges) such that the potential $\Phi(\bf)$ is initially $O(m \log mU)$ (\cref{lemma:initialpoint}). Additionally, given a nearly optimal solution, we can recover an exactly optimal solution to the original min-cost flow problem in linear time (\cref{lemma:finalpoint}).
A simple observation is that if the potential is sufficiently small, then the cost of the flow is nearly optimal.
\begin{lemma}
\label{lemma:obvious}
We have $\bc^\top \bf - F^* \le \exp\left(\Phi(\bf)/(20m)\right).$ In particular, if $\Phi(\bf) \le -200m \log mU$ then $\bc^\top \bf - F^* \le (mU)^{-10}$.
\end{lemma}
\begin{proof}
From \eqref{eq:karmarkar} and the fact that $\bu^-_e \le \bf_e \le \bu^+_e$, we get
\[ \Phi(\bf) \ge 20m \log(\bc^\top \bf - F^*). \]
Rearranging this gives the desired result.
\end{proof}
Given a flow $\bf \in \R^E$ we define lengths $\bell \in \R_{>0}^E$ and gradients $\bg \in \R^E$ to capture the next $\ell_1$ problem we solve to decrease the potential.
\begin{definition}[Lengths and gradients]
\label{def:lg}
Given a flow $\bf \in \R^E$ we define lengths $\bell \in \R^E$ as
\begin{align}
    \bell(\bf)_e \defeq \left(\bu^+_e - \bf_e\right)^{-1-\alpha} + \left(\bf_e - \bu^-_e\right)^{-1-\alpha} \label{eq:bell}
\end{align}
and gradients $\bg \in \R^E$ as $\bg(\bf) \defeq \g\Phi(\bf)$. More explicitly,
\begin{align}
    \bg(\bf)_e \defeq \left[\g\Phi(\bf)\right]_e = 20m(\bc^\top\bf-F^*)^{-1}\bc_e + \alpha \left(\bu_e^+ - \bf_e\right)^{-1-\alpha} - \alpha \left(\bf_e - \bu^-_e\right)^{-1-\alpha} \label{eq:bg}
\end{align}
\end{definition}
The remainder of the section is split into three parts. In \cref{subsec:onestep} we show that approximately solving the cycle problem induced by gradients and lengths approximating those in \cref{def:lg} allows us to decrease the potential additively by an almost constant quantity in a single iteration. Then in \cref{subsec:stability} we bound how such iterations affect the lengths and gradients in order to show that approximate versions of them only need to be modified $m^{1 + o(1)}$ times across the entire algorithm, and in \cref{subsec:initialfinal} we discuss how to get an initial flow and extract an exact min-cost flow from a nearly optimal flow.

The following theorem summarizes the results of this section.
\begin{theorem}
\label{thm:IPM}
Suppose we are given a min-cost flow instance given by Equation~\eqref{eq:mincostopt}. Let $\bf^*$ denote an optimal solution to the instance.

For all $\kappa \in (0,1),$ there is a potential reduction interior point method for this problem, that, given an initial flow $\bf^{(0)} \in \R^E$ such that $\Phi(\bf^{(0)}) \le 200m\log mU,$ the algorithm proceeds as follows:

The algorithm runs for $\O(m\kappa^2)$ iterations. At each iteration,  let $\bg(\bf^{(t)}) \in \R^E$ denote that gradient and $\bell(\bf^{(t)}) \in \R^E_{>0}$ denote the lengths given by \cref{def:lg}.
Let $\wt{\bg} \in \R^E$ and $\wt{\bell} \in \R^E_{>0}$ be any vectors such that $\left\|\mL(\bf^{(t)})^{-1}\left(\wt{\bg} - \bg(\bf^{(t)})\right) \right\|_\infty \le \kappa/8$ and  $\wt{\bell} \approx_2 \bell(\bf)$.
\begin{enumerate}
\item At each iteration, the hidden circulation $\bf^* - \bf^{(t)}$ satisfies    
\[ \frac{\wt{\bg}^\top(\bf^* - \bf^{(t)})}{100m + \left\|\wt{\mL}(\bf^* - \bf^{(t)})\right\|_1} \le -\alpha/4. \]
    \item At each iteration, given any $\bDelta$ satisfying $\mB^\top\bDelta = 0$ and $\nicefrac{\wt{\bg}^\top \bDelta}{\left\|\wt{\mL}\bDelta\right\|_1} \le -\kappa,$ it 
updates $\bf^{(t+1)} \assign \bf^{(t)} + \eta \bDelta$ for $\eta \leftarrow \kappa^2/(50\cdot\abs{\wt{\bg}\Delta}).$
    \item At the end of $\O(m\kappa^2)$ iterations, we have $\bc^\top \bf^{(t)} \le \bc^{\top}\bf^* + (mU)^{-10}.$
\end{enumerate}
\end{theorem}
Intuitively, the algorithm will compute a sequence of flows $\bf^{(0)}, \bf^{(1)}$, and maintain approximations $\wt{\bg}, \wt{\bell}$ of $\bg(\bf^{(t)}), \bell(\bf^{(t)})$ respectively. Each iteration, the algorithm will call an oracle for approximating the minimum-ratio cycle, i.e. $\min_{\mB^\top\bDelta = 0} \nicefrac{\wt{\bg}^\top \bDelta}{\left\|\wt{\mL}\bDelta\right\|_1}$. The first item shows that the optimal ratio is at most $-\alpha/4$. Thus if the oracle returns an $m^{o(1)}$ approximation, the returned circulation has $\kappa \ge m^{-o(1)}$. Scaling $\bDelta$ appropriately and adding it to $\bf^{(t)}$ decreases the potential by $\Omega(\kappa^2)$, hence the potential drops to $-O(m \log m)$ within $\O(m\kappa^{-2})$ iterations.

In \cref{sec:combine} we will give a formal description of the interaction of the algorithm of \cref{thm:IPM} and our data structures to implement each step in amortized $m^{o(1)}$ time. As part of this, we argue that we can change $\wt{\bg}$ and $\wt{\bell}$ only $\O(m\kappa^{-2})$ total times. This is encapsulated in \cref{lemma:sizedt}.

\subsection{One Step Analysis}
\label{subsec:onestep}
Consider a current flow $\bf$ and lengths/gradients $\bell(\bf), \bg(\bf)$ defined in \cref{def:lg}, with $\mL = \diag(\bell)$. The problem we will solve approximately in each iteration will be
\begin{align}
    \min_{\mB^\top\bDelta = 0} \frac{\bg(\bf)^\top  \bDelta}{\left\|\mL(\bf)\bDelta\right\|_1}. \label{eq:approxl1}
\end{align}
Alternatively, this can be viewed as constraining $\mB^\top\bDelta = 0$ and $\bg(\bf)^\top \bDelta = -1$, and then minimizing $\left\|\mL(\bf)\bDelta\right\|_1$. Our first goal is to show that an approximate solution to \eqref{eq:approxl1} for approximations of the gradient and lengths allows us to decrease the potential.
\begin{lemma}
\label{lemma:phidecrease}
Let $\wt{\bg} \in \R^E$ satisfy $\left\|\mL(\bf)^{-1}\left(\wt{\bg} - \bg(\bf)\right) \right\|_\infty \le \kappa/8$ for some $\kappa \in (0,1)$, and $\wt{\bell} \in \R^E_{>0}$ satisfy $\wt{\bell} \approx_2 \bell(\bf)$.
Let $\bDelta$ satisfy $\mB^\top\bDelta = 0$ and $\nicefrac{\wt{\bg}^\top \bDelta}{\left\|\wt{\mL}\bDelta\right\|_1} \le -\kappa$.
Let $\eta$ satisfy $\eta \wt{\bg}^\top\bDelta  = -\kappa^2/50.$ Then
\[ \Phi(\bf + \eta\bDelta) \le \Phi(\bf) - \frac{\kappa^2}{500}. \]
\end{lemma}
Before showing this, we need simple bounds on the Taylor expansion of the logarithmic barrier and $x^{-\alpha}$ in the region where the second derivative is stable.
\begin{lemma}[Taylor expansion for $x^{-\alpha}$]
\label{lemma:taylorxa}
If $|\bDelta_e| \le \frac{1}{10}\min\left(\bu^+_e - \bf_e, \bf_e - \bu^-_e \right)$ for $e \in E$ then
\begin{align}
    &\left((\bu^+_e - \bf_e - \bDelta_e)^{-\alpha} + (\bf_e + \bDelta_e - \bu^-_e)^{-\alpha} \right) \le \left((\bu^+_e - \bf_e)^{-\alpha} + (\bf_e - \bu^-_e)^{-\alpha} \right) \nonumber \\
    &~+ \alpha\left(\left(\bu_e^+ - \bf_e\right)^{-1-\alpha} - \left(\bf_e - \bu^-_e\right)^{-1-\alpha} \right)\bDelta_e + \alpha\left(\left(\bu_e^+ - \bf_e\right)^{-2-\alpha} + \left(\bf_e - \bu^-_e\right)^{-2-\alpha} \right)\bDelta_e^2
    \label{eq:taylorxa}
\end{align}
Also we have that
\begin{align}
    \left| (\bu^+_e - \bf_e - \bDelta_e)^{-1-\alpha} - (\bu^+_e - \bf_e)^{-1-\alpha} \right| \le 2|\bDelta_e|(\bu^+_e - \bf_e)^{-2-\alpha} \label{eq:uppertaylor}
\end{align}
and
\begin{align}
    \left| (\bf_e + \bDelta_e - \bu^-_e)^{-1-\alpha} - (\bf_e - \bu^-_e)^{-1-\alpha} \right| \le 2|\bDelta_e|(\bf_e - \bu^-_e)^{-2-\alpha} \label{eq:lowertaylor}.
\end{align}
\end{lemma}
\cref{eq:uppertaylor} and \cref{eq:lowertaylor} are useful for analyzing how a step improves the value of the potential function $\Phi(\bf)$, as well as showing that the gradients $\bg(\bf)$ and lengths $\bell(\bf)$ are stable, i.e. change only $m^{1+o(1)}$ times over $m^{1+o(1)}$ iterations.
\begin{proof}
Define $\phi(x) \defeq (\bu^+_e - x)^{-\alpha} + (x - \bu^-_e)^{-\alpha}$. $\phi$ is a convex function with derivative
\[ \phi'(x) = \alpha\left((\bu^+_e - x)^{-1-\alpha} - (x - \bu^-_e)^{-1-\alpha} \right) \]
and second derivative
\[ \phi''(x) = \alpha(1+\alpha)\left((\bu^+_e - x)^{-2-\alpha} + (x - \bu^-_e)^{-2-\alpha} \right). \]
In particular note that $\phi''(\bf_e + \delta) \approx_{1.3} \phi''(\bf_e)$ for any $|\delta| \le \frac{1}{10}\min\left(\bu^+_e - \bf_e, \bf_e - \bu^-_e \right)$, because $1.1^{2+\alpha} \le 1.3$ by the choice of $\alpha$. Thus by Taylor's theorem we get that
\begin{align*}
    \phi(\bf_e + \bDelta_e) &\le \phi(\bf_e) + \phi'(\bf_e)\bDelta_e + \frac{1}{2}\max_{y \in [\bf_e, \bf_e + \bDelta_e]} \phi''(y) \bDelta_e^2 \le
    \phi(\bf_e) + \phi'(\bf_e)\bDelta_e + 1.3 \phi''(\bf_e)\bDelta_e^2 \\
    &\le \phi(\bf_e) + \phi'(\bf_e)\bDelta_e + \alpha\left(\left(\bu_e^+ - \bf_e\right)^{-2-\alpha} + \left(\bf_e - \bu^-_e\right)^{-2-\alpha} \right)\bDelta_e^2,
\end{align*}
which when expanded yields the desired bound. \eqref{eq:uppertaylor}, \eqref{eq:lowertaylor} follow from a similar application of Taylor's theorem on a first order expansion.
\end{proof}
\begin{lemma}[Taylor expansion for $\log x$]
\label{lemma:taylorlog}
If $|y| \le x/10$ for $x > 0$ then
\begin{align}
    \log(x+y) \le \log(x) + y/x + y^2/x^2. \label{eq:taylorlog}
\end{align}
\end{lemma}
\begin{proof}
This is equivalent to $\log(1+y/x) \le y/x + y^2/x^2$ for $|y/x| \le 1/10$, which follows from the Taylor expansion $\log(1+z) = \sum_{k \ge 0} z^k/k$ for $|z| < 1$.
\end{proof}
\begin{proof}[Proof of \cref{lemma:phidecrease}]
We first bound $\bg(\bf)^\top\bDelta$ by
\begin{align*}
    \left|\bg(\bf)^\top\bDelta - \wt{\bg}^\top\bDelta\right| &\overset{(i)}{\le} \left\|\mL(\bf)^{-1}(\wt{\bg}-\bg(\bf))\right\|_\infty \left\|\mL(\bf) \bDelta\right\|_1 \overset{(ii)}{\le} 2\eps/\kappa \cdot |\wt{\bg}^\top\bDelta| \le \wt{\bg}^\top\bDelta/2,
\end{align*} where
$(i)$ follows from H\"{o}lder's inequality with the $\ell_1,\ell_\infty$ norms,
$(ii)$ follows from the lemma hypotheses and $\|\mL(\bf)\bDelta\|_1 \le 2\|\wt{\mL}\bDelta\|_1$, and the final inequality follows from $\eps \le \kappa/8$. Hence
\begin{align} 
2\wt{\bg}^\top\bDelta \le \bg(\bf)^\top\bDelta \le \wt{\bg}^\top\bDelta/2. \label{eq:wtbg}
\end{align}

We can also bound $\bc^\top\bDelta$ by
\begin{align}
    20m(\bc^\top\bf-F^*)^{-1}\left|\bc^\top\bDelta\right| &= \left|\bg(\bf)^\top\bDelta - \alpha \sum_{e \in E} \left(\left(\bu_e^+ - \bf_e\right)^{-1-\alpha} - \left(\bf_e - \bu^-_e\right)^{-1-\alpha}\right)\bDelta_e \right| \label{eq:stablebc1} \\
    &\overset{(i)}{\le} \left|\bg(\bf)^\top\bDelta\right| + \alpha\left\|\mL(\bf) \bDelta\right\|_1 \overset{(ii)}{\le} (2 + 2\alpha/\kappa)|\wt{\bg}^\top \bDelta|. \label{eq:stablebc2}
\end{align}
where $(i)$ follows from the triangle inequality, and $(ii)$ from \eqref{eq:wtbg} plus the problem hypotheses. In particular, we deduce that
\begin{align}
    \left|\bc^\top\bDelta\right| \le \frac{\bc^\top\bf - F^*}{20m} \cdot (2+2\alpha/\kappa)|\wt{\bg}^\top \bDelta|. \label{eq:cdelta}
\end{align}
Let $\barDelta \defeq \eta\bDelta,$ the circulation that we add. From \eqref{eq:cdelta} we get that
\begin{align*} |\bc^\top \barDelta| \le \eta \cdot \frac{\bc^\top\bf - F^*}{20m} \cdot (2+2\alpha/\kappa)|\wt{\bg}^\top \bDelta| \le \frac{\bc^\top\bf - F^*}{20m} \cdot 4\eta/\kappa|\wt{\bg}^\top \bDelta| \le \frac{\kappa}{500m}\left(\bc^\top\bf - F^*\right)
\end{align*}
by the choice of $\eta$ in the problem hypothesis. Additionally, we have
\begin{align} \left\|\mL(\bf)\barDelta\right\|_1 \le 2\left\|\wt{\mL}\barDelta\right\|_1 \le 2/\kappa \cdot \left\|\wt{\bg}\barDelta\right\|_1 \le 2\eta/\kappa \cdot \left\|\wt{\bg}\bDelta\right\|_1 \le \kappa/25 \label{eq:lbardeltabound} \end{align}
by the choice of $\eta$. This implies that
\[ |\barDelta_e| \le \kappa/25 \cdot (\bu^+_e - \bf_e)^{1+\alpha} \le \kappa/25 \cdot (2U)^{\alpha} \cdot (\bu^+_e - \bf_e) \le \kappa/10 \cdot (\bu^+_e - \bf_e), \]
where the last inequality follows from the choice $\alpha = 1/(1000 \log mU)$.
$|\barDelta_e| \le \kappa/10 \cdot (\bf_e - \bu^-_e)$ follows similarly. 

This bound allows us to apply:
\begin{itemize}
    \item \cref{lemma:taylorxa} on the current $\bu^{+}$, $\bu^{-}$,
    $\bf$, and $\bDelta$, and
    \item \cref{lemma:taylorlog} for $x = \bc^\top \bf - F^*$ and $y = \bc^\top \barDelta$ 
\end{itemize}
to get
\begin{align*}
    \Phi(\bf + \barDelta) - \Phi(\bf) &\le \bg(\bf)^\top\barDelta + 20m\left(\frac{\kappa}{500m}\right)^2 + \sum_{e \in E} \alpha\left(\left(\bu_e^+ - \bf_e\right)^{-2-\alpha} + \left(\bf_e - \bu^-_e\right)^{-2-\alpha} \right)\barDelta_e^2 \\
    &\overset{(i)}{\le} \wt{\bg}^\top\barDelta/2 + \frac{\kappa^2}{10000} + \alpha\kappa/10 \cdot \left\|\mL(\bf)\barDelta\right\|_1 \overset{(ii)}{\le} -\frac{\kappa^2}{100} + \frac{\kappa^2}{10000} + \frac{\alpha\kappa}{250} \le -\frac{\kappa^2}{500}.
\end{align*}
Here, $(i)$ follows from \eqref{eq:wtbg}, and the bound $|\barDelta_e| \le \kappa/10 \cdot \min\left(\bu^+_e - \bf_e, \bf_e - \bu^-_e \right)$, and $(ii)$ from $\wt{\bg}^\top \barDelta = \eta \wt{\bg}^\top \bDelta = -\kappa^2/50$ and \eqref{eq:lbardeltabound}.
\end{proof}

Our next goal is to show that a straight line to $\bf^*$, i.e. $\bDelta = \bf^* - \bf$ satisfies the guarantees of \cref{lemma:phidecrease} for some $\kappa \ge \widetilde{\Omega}(1)$. This has two purposes. First, it shows that an $m^{o(1)}$-optimal solution to \cref{eq:approxl1} allows us to decrease the potential by $m^{-o(1)}$ per step, so that the algorithm terminates in $m^{1+o(1)}$ steps. Second, it shows that the problems \eqref{eq:approxl1} encountered during the method are not fully adaptive, and we are able to use this guarantee on a good solution to inform our data structures.

\begin{lemma}[Quality of $\bf^* - \bf$]
\label{lemma:optCirculationQuality}
Let $\wt{\bg} \in \R^E$ satisfy $\left\|\mL(\bg)^{-1}\left(\wt{\bg} - \bg(\bf)\right) \right\|_\infty \le \eps$ for some $\eps < \alpha/2$, and $\wt{\bell} \in \R^E_{>0}$ satisfy $\wt{\bell} \approx_2 \bell(\bf)$. If $\Phi(\bf) \le 200m\log mU$ and $\log(\bc^\top\bf - F^*) \ge -10 \log mU$,
then
\[ \frac{\wt{\bg}^\top(\bf^* - \bf)}{100m + \left\|\wt{\mL}(\bf^* - \bf)\right\|_1} \le -\alpha/4. \]
\end{lemma}
The additional $100m$ in the denominator is for a technical reason, and intuitively says that the bound is still fine even if we force every edge to pay at least $100$ towards the $\ell_1$ length of the circulation.
\begin{proof}
We can bound $\bg(\bf)^\top(\bf^* - \bf)$ using
\begin{align*} \bg(\bf)^\top(\bf^* - \bf) &= 20m \frac{\bc^\top(\bf^* - \bf)}{\bc^\top \bf - F^*} + \sum_{e \in E} \left(\alpha \left(\bu_e^+ - \bf_e\right)^{-1-\alpha} - \alpha \left(\bf_e - \bu^-_e\right)^{-1-\alpha}\right)(\bf^*_e - \bf_e) \\
&\overset{(i)}{\le} -20m - \alpha \sum_{e \in E} \left(\left(\bu_e^+ - \bf_e\right)^{-1-\alpha} + \left(\bf_e - \bu^-_e\right)^{-1-\alpha}\right)|\bf^*_e - \bf_e| \\ &~+ 2\alpha \sum_{e \in E} \left(\left(\bu_e^+ - \bf_e\right)^{-\alpha} + \left(\bf_e - \bu^-_e\right)^{-\alpha}\right) \\
&= -20m - \alpha\left\|\mL(\bf)(\bf^* - \bf)\right\|_1 + 2\alpha\left(\Phi(\bf) - 20m \log(\bc^\top\bf - F^*) \right) \\
&\overset{(ii)}{\le} -20m - \alpha\left\|\mL(\bf)(\bf^* - \bf)\right\|_1 + 2\alpha \cdot 400 m \log mU \\ &\overset{(iii)}{\le} -19m - \alpha\left\|\mL(\bf)(\bf^* - \bf)\right\|_1.
\end{align*}
where $(i)$ follows from the bound $\bu^-_e - \bf_e \le \bf^*_e - \bf_e \le \bu^+_e - \bf_e \forall e \in E$, so 
\begin{align*} (\bu^+_e - \bf_e)^{-1-\alpha}(\bf^*_e - \bf_e) & = (\bu^+_e - \bf_e)^{-1-\alpha}(\bf^*_e - \bu^+_e + \bu^+_e - \bf_e) \\
&= (\bu^+_e - \bf_e)^{-\alpha} - (\bu^+_e - \bf_e)^{-1-\alpha}|\bf^*_e - \bu^+_e|
\\ &\le 2(\bu^+_e - \bf_e)^{-\alpha} - (\bu^+_e - \bf_e)^{-1-\alpha}\left|\bf^*_e - \bf_e\right|, 
\end{align*}

and similar for the $-(\bu^+_e - \bf_e)^{-1-\alpha}(\bf^*_e - \bf_e)$ term, $(ii)$ follows from the lemma hypotheses, and $(iii)$ from the choice $\alpha = 1/(1000 \log mU)$. We now bound
\begin{align*}
    \wt{\bg}^\top(\bf^* - \bf) &= \bg(\bf)^\top(\bf^* - \bg) + \left(\wt{\bg} - \bg(\bf) \right)^\top(\bf^* - \bf) \\
    &\overset{(i)}{\le} -19m - \alpha\left\|\mL(\bf)(\bf^* - \bf)\right\|_1 + \left\|\mL(\bf)^{-1}\left(\wt{\bg} - \bg(\bf)\right)\right\|_\infty \left\|\mL(\bf)(\bf^* - \bf) \right\|_1 \\
    &\overset{(ii)}{\le} -19m - \alpha\left\|\mL(\bf)(\bf^* - \bf)\right\|_1 + \eps \left\|\mL(\bf)(\bf^* - \bf) \right\|_1 \\
    &\le -19m - \alpha/2 \cdot \left\|\mL(\bf)(\bf^* - \bf)\right\|_1,
\end{align*}
where $(i)$ follows from the above bound on $\bg(\bf)^\top(\bf^* - \bf)$ and H\"{o}lder for the $\ell_1/\ell_\infty$ norms, and $(ii)$ follows from the the conditions on $\wt{\bg}$. Thus we get that
\[ \frac{\wt{\bg}^\top(\bf^* - \bf)}{100m + \left\|\wt{\mL}(\bf^* - \bf) \right\|_1} \le \frac{-19m - \alpha/2 \cdot \left\|\mL(\bf)(\bf^* - \bf)\right\|_1}{100m + 2\left\|\mL(\bf)(\bf^* - \bf)\right\|_1} \le -\alpha/4, \]
where we have used the above bound on $\wt{\bg}^\top(\bf^* - \bf)$ and $\bell \approx_2 \wt{\bell}$.
\end{proof}

\subsection{Stability Bounds}
\label{subsec:stability}
Our algorithm ultimately approximately solves \eqref{eq:approxl1} by using approximations $\wt{\bg}$ of $\bg(\bf)$ and $\wt{\bell}$ of $\bell(\bf)$ satisfying the conditions of \cref{lemma:phidecrease}. The goal of this section is to show that $\bg(\bf)$ and $\bell(\bf)$ are slowly changing relative to the lengths, so that our dynamic data structure can only update their values on $m^{o(1)}$ edges per iteration.

We start by showing that the residual cost is very slowly changing, by about $1/m$ per iteration.
\begin{lemma}[Residual stability]
\label{lemma:stablebc}
Let $\wt{\bg} \in \R^E$ satisfy $\left\|\mL(\bf)^{-1}\left(\wt{\bg} - \bg(\bf)\right) \right\|_\infty \le \eps$ for some $\eps \in (0,1/2]$, and $\wt{\bell} \in \R^E_{>0}$ satisfy $\wt{\bell} \approx_2 \bell(\bf)$.
Let $\bDelta$ satisfy $\mB^\top\bDelta = 0$ and $\nicefrac{\wt{\bg}^\top \bDelta}{\left\|\wt{\mL}\bDelta\right\|_1} \le -\kappa$ for $\kappa \in (0, 1)$. Then
\[ \frac{|\bc^\top \bDelta|}{\bc^\top \bf - F^*} \le |\wt{\bg}^\top\bDelta|/(\kappa m). \]
\end{lemma}
\begin{proof}
We can write
\begin{align*} \frac{|\bc^\top \bDelta|}{\bc^\top \bf - F^*} &\overset{(i)}{\le} \frac{1}{20m}(|\bg^\top\bDelta| + \alpha\|\mL(\bf)\bDelta\|_1) \\
&\overset{(ii)}{\le} \frac{1}{20m}(|\wt{\bg}^\top\bDelta| + \|\mL(\bf)^{-1}(\wt{\bg}-\bg(\bf))\|_\infty\|\mL(\bf)\bDelta\|_1 + \alpha\|\mL(\bf)\bDelta\|_1) \\
&\overset{(iii)}{\le} (|\wt{\bg}^\top\bDelta|+2\|\wt{\mL}\bDelta\|_1)/(20m) \le |\wt{\bg}^\top\bDelta|/(\kappa m),
\end{align*}
where $(i)$ uses the triangle inequality, $(ii)$ uses the triangle inequality and $|\bx^\top\by| \le \|\bx\|_\infty\|\by\|_1$, and $(iii)$ uses the hypotheses and $\wt{\bell} \approx_2 \bell(\bf)$.
\end{proof}
Hence if $|\bg(f)^\top\bDelta| \le O(\kappa^2)$ and $\|\mL(\bf)\bDelta\|_1 = O(\kappa)$ such as in the hypotheses of  \cref{lemma:phidecrease}, the residual cost changes by at most a $1/m$ factor per iteration.

We show that if the residual capacity of an edge does not change much, then its length is stable.
\begin{lemma}[Length stability]
\label{lemma:stablebell}
If $\|\mL(\bf)(\bf-\bar{\bf})\|_\infty \le \eps$ for some $\eps \le 1/100$ then $\bell(\bf) \approx_{1+3\eps} \bell(\bar{\bf})$.
\end{lemma}
\begin{proof}
Because $\left\|\mL(\bf)(\bf-\bar{\bf})\right\|_\infty \le 1/100$, we have for all $e \in E$
\[ |\bf_e - \bar{\bf}_e| \le \eps(\bu^+_e - \bf_e)^{1+\alpha} = \eps(\bu^+_e - \bf_e)(2U)^\alpha \le 2\eps(\bu^+_e - \bf_e).\]
Similarly, $|\bf_e - \bar{\bf}_e| \le 2\eps(\bf_e - \bu^-_e)$. Hence $\bu^+_e - \bf_e \approx_{\exp(2\eps)} \bu^+_e-\bar{\bf}_e$ and $\bf_e - \bu^-_e \approx_{\exp(2\eps)} \bar{\bf}_e-\bu^-_e$, so we get
\begin{align*}
\bell(\bar{\bf})_e &= (\bu^+_e-\bar{\bf}_e)^{-1-\alpha} + (\bar{\bf}_e-\bu^-_e)^{-1-\alpha} \\
&\approx_{\exp(2\eps(1+\alpha))} (\bu^+_e-\bf_e)^{-1-\alpha} + (\bf_e-\bu^-_e)^{-1-\alpha} = \bell(\bf)_e.
\end{align*}
This completes the proof, as $2\eps(1+\alpha) \le 3\eps$.
\end{proof}
Next we show a similar stability claim for gradients. Here, we scale by the residual cost $\bc^\top\bf-F^*$ to ensure that the leading term is $20m\bc$. Thus, the gradient is stable if the residual capacity of an edge does not change much, and if the residual cost is stable. We know that the residual cost is stable over $\widetilde{O}(m)$ iterations by \cref{lemma:stablebc}.
\begin{lemma}[Gradient stability]
\label{lemma:stablebg}
If $\|\mL(\bf)(\bf-\bar{\bf})\|_\infty \le \eps$ and $r \approx_{1+\eps} \bc^\top\bar{\bf} - F^*$ then $\wt{\bg}$ defined as
\[ \wt{\bg}_e = 20m\bc_e/r + \alpha(\bu^+_e - \bf_e)^{-1-\alpha} - \alpha(\bf_e - \bu^-_e)^{-1-\alpha} \forall e \in E \] satisfies
\[ \left\|\mL(\bf)^{-1}\left(\wt{\bg} - (\bc^\top\bar{\bf}-F^*)/r \cdot \bg(\bar{\bf})\right) \right\| \le 6\alpha\eps.\]
\end{lemma}
\begin{proof}
We can first compute that
\begin{align} &\left[\wt{\bg} - (\bc^\top\bar{\bf}-F^*)/r \cdot \bg(\bar{\bf})\right]_e \nonumber \\ =~&\left(1 - \frac{\bc^\top\bar{\bf}}{r}\right)\left(\alpha(\bu^+_e-\bar{\bf}_e)^{-1-\alpha} - \alpha(\bar{\bf}_e-\bu^-_e)^{-1-\alpha}\right) \label{eq:term1} \\
+~&\alpha((\bu^+_e-\bf_e)^{-1-\alpha} - (\bu^+_e-\bar{\bf}_e)^{-1-\alpha}) - \alpha((\bf_e-\bu^-_e)^{-1-\alpha}-(\bar{\bf}_e-\bu^-_e)^{-1-\alpha}). \label{eq:term2}
\end{align}
We bound the terms in \eqref{eq:term1} and \eqref{eq:term2} separately. To bound \eqref{eq:term1} we write
\begin{align*}
&\left(1 - \frac{\bc^\top\bar{\bf}}{r}\right)\alpha\left\|\mL(\bf)^{-1}\left((\bu^+ - \bar{\bf})^{-1-\alpha} - (\bar{\bf} - \bu^-)^{-1-\alpha} \right) \right\|_\infty \\
=~& \left(1 - \frac{\bc^\top\bar{\bf}}{r}\right)\alpha\left\|\mL(\bf)^{-1}\bell(\bar{\bf})\right\|_\infty \overset{(i)}{\le} 1.1\eps\alpha \cdot (1+3\eps) \le 2\eps\alpha,
\end{align*}
where $(i)$ follows from the hypothesis and \cref{lemma:stablebell}. To bound \eqref{eq:term2} we write
\begin{align*}
    &\alpha\left\|\mL(\bf)^{-1}\left(((\bu^+-\bf)^{-1-\alpha}-(\bu^+-\bar{\bf})^{-1-\alpha}) - ((\bf-\bu^-)^{-1-\alpha}-(\bar{\bf}-\bu^-)^{-1-\alpha}) \right) \right\|_\infty \\
    \overset{(i)}{\le}~& 2\alpha\left\|\mL(\bf)^{-1}((\bu^+-\bf)^{-2-\alpha}+(\bf-\bu^-)^{-2-\alpha})(\bf - \bar{\bf}) \right\|_\infty \\
    \le~&2\alpha \max_{e \in E}\{(\bu^+_e-\bf_e)^\alpha, (\bf_e-\bu^-_e)^\alpha\} \|\mL(\bf)(\bf - \bar{\bf})\|_\infty
    \le 2\alpha(2U)^\alpha\eps \le 4\alpha\eps
\end{align*}
where $(i)$ follows from \cref{lemma:taylorxa}, specifically \eqref{eq:uppertaylor}, \eqref{eq:lowertaylor}. Summing these gives the desired bound.
\end{proof}
We now show the main result of this section, \cref{thm:IPM}.
\begin{proof}[Proof of \cref{thm:IPM}]
The first item follows from \cref{lemma:optCirculationQuality}. To show the third item, note that the update in the second item exactly corresponds to \cref{lemma:phidecrease}, so $\Phi(\bf^{(t)}) \le \Phi(\bf^{(t-1)}) - \Omega(\kappa^2)$. Once the potential has reduced to $-O(m \log m)$, then $\bc^\top \bf^{(t)} - \bc^\top\bf \le (mU)^{-10}$ (\cref{lemma:obvious}), so the algorithm takes $\O(m\kappa^{-2})$ total iterations.
\end{proof}

\subsection{Initial and Final Point}
\label{subsec:initialfinal}
In this section we discuss how to initialize our method and how to get an exact optimal solution from a nearly optimal solution. For the latter piece, we can directly cite previous work which gives a rounding method using the Isolation Lemma.

\begin{lemma}[{\cite[Lemma 8.10]{BLNPSSW20}}]
\label{lemma:finalpoint}
Consider a min-cost flow instance $\mathcal{I} = (G, \bd, \bc)$ on a graph $G = (V, E)$ with demands $\bd \in \{-U, \dots, U\}^V$ and cost $\bc \in \{-U, \dots, U\}^E.$ Assume that all optimal flows have congestion at most $U$ on every edge.

Consider a perturbed instance $\wt{\mathcal{I}} = (G, \bd, \wt{\bc})$ on the same graph $G = (V, E)$ and demand $\bd$, but with modified cost vector $\wt{\bc} \in \R^E$ defined as $\wt{\bc}_e = \bc_e + \bz_e$ for independent, random $\bz_e \in \left\{\frac{1}{4m^2U^2}, \frac{2}{4m^2U^2}, \dots, \frac{2mU}{4m^2U^2} \right\}$ for all $e \in E$. Let $\wt{\bf}$ be a solution for $\wt{\mathcal{I}}$ whose cost is at most $\frac{1}{12m^3U^3}$ from optimal. Let $\bf$ be obtained by rounding $\wt{\bf}$ to the nearest integer on every edge. Then $\bf$ is an optimal flow for the instance $\mathcal{I}$ with probability at least $1/2$.
\end{lemma}
It is worth noting that scaling up the cost vector $\wt{\bc}$ of the  perturbed instance $\wt{\mathcal{I}}$ in \cref{lemma:finalpoint} by $4m^2U^2$ results in a min-cost flow instance with integral demands and costs again.

Now we describe how to augment our original graph with additional edges without affecting the optimal solution, but allows us to initialize a solution with bounded potential. The proof is deferred to \cref{app:initialpoint}.
\begin{lemma}[Initial Point]
\label{lemma:initialpoint}
There is an algorithm that given a graph $G = (V, E)$ and min-cost flow instance $\mathcal{I} = (G, \bd, \bc, \bu^+, \bu^-)$ with demands $\bd \in \{-U, \dots, U\}^V$, and costs and lower/upper capacities $\bc, \bu^-, \bu^+ \in \{-U, \dots, U\}^V$, constructs a min-cost flow instance $\wt{\mathcal{I}} = (\wt{G}, \wt{\bd}, \wt{\bc}, \wt{\bu}^+, \wt{\bu}^-)$ with $O(m)$ edges and $\wt{\bd} \in \{-2mU, \dots, 2mU\}^{V(\wt{G})}$, $\wt{\bc}, \wt{\bu}^+, \wt{\bu}^- \in \{-4mU^2, \dots, 4mU^2\}^{E(\wt{G})}$, and a flow $\bf^{\init}$ on $\wt{G}$ routing $\wt{\bd}$ such that $\Phi(\bf^{\init}) \le 200m \log mU$.

Also, given an optimal flow $\wt{\bf}$ for $\wt{\mathcal{I}}$, the algorithm can either compute an optimal flow $\bf$ for $\mathcal{I}$ or conclude that $\mathcal{I}$ admits no feasible flow. The algorithm runs in time $O(m)$.
\end{lemma}

%% file: spanner.tex
\section{Decremental Spanner and Embedding}
\label{sec:spanner}

The main result of this section is summarized in the following theorem. Intuitively, the theorem states that given a low-degree graph $G$, one can maintain a sparsifier $H$ of $G$ and embed $G$ with short paths and low congestion into $H$.

\begin{theorem}
\label{thm:spanner}
Given an $m$-edge $n$-vertex unweighted, undirected, dynamic graph $G$ undergoing update batches $U^{(1)}, U^{(2)}, \ldots$ consisting only of edge deletions and $\tilde{O}(n)$ vertex splits. There is a randomized algorithm with parameter $1 \le L \le o\left(\frac{\sqrt{\log(m)}}{\log\log m}\right)$, that maintains a spanner $H$ and an embedding $\Pi_{G \to H}$ such that
\begin{enumerate}
    \item \label{prop:sparsitySpanner} \underline{Sparsity and Low Recourse:} initially $H^{(0)}$ has sparsity $\tilde{O}(n)$. At any stage $t \geq 1$, the algorithm outputs a batch of updates $U_H^{(t-1)}$ that when applied to $H^{(t-1)}$ produce $H^{(t)}$ such that $H^{(t)} \subseteq G^{(t)}$, $H^{(t)}$ consists of at most $\tilde{O}(n)$ edges and $\sum_{t' \leq t} \Enc(U_H^{(t')}) = \tilde{O}\left(n + \sum_{t' \leq t} |U^{(t')}| \cdot n^{1/L}\right)$, and 
    
    \item \underline{Low Congestion, Short Paths Embedding:} $\length(\Pi_{G \to H}) \le (\gamma_l)^{O(L)}$ and $\vcong(\Pi_{G \to H}) \le (\gamma_c)^{O(L)} \Delta_{\max}(G)$, for $\gamma_l, \gamma_c = \exp(O(\sqrt{\log m} \cdot \log\log m))$, and
    \item
    \underline{Low Recourse Re-Embedding:} \label{prop:lowRecourseSpanner}
    the algorithm further reports after each update batch $U^{(t)}$ at stage $t$ is processed, a (small) set $D^{(t)} \subseteq E(H^{(t)})$ of edges, such that for all other edges $e \in E(H^{(t)}) \setminus D^{(t)}$, there exists \textbf{no} edge $e' \in E(G^{(t)})$ whose embedding path $\Pi^{(t)}_{G \to H}(e')$ contains $e$ at the current stage but did not before the stage. The algorithm ensures that at any stage $t$, we have $\sum_{t' \leq t} |D^{(t')}| = \tilde{O}\left( \sum_{t' \leq t} |U^{(t)}| \cdot n^{1/L}(\gamma_c\gamma_l)^{O(L)}\right)$, i.e. that the sets $D^{(t)}$ are roughly upper bounded by the size of $U^{(t)}$ on average. \label{item:lowReEmbed} 
\end{enumerate}
The algorithm takes initialization time $\tilde{O}(m \gamma_l)$ and processing the $t$-th update batch $U^{(t)}$ takes amortized update time $\tilde{O}(\textsc{Enc}(U^{(t)}) \cdot n^{1/L}(\gamma_c\gamma_l)^{O(L)}\Delta_{\max}(G))$, and succeeds with probability at least $1-n^{-C}$ for any constant $C > 0$, specified before the procedure is invoked.
\end{theorem}

Taking $L = (\log m)^{1/4}$ in \cref{thm:spanner} gives a parameter $\gamma_s = \exp(O(\log^{3/4}m\log\log m))$ such that the amortized runtime, lengths of the embeddings, and amortized size of $D$ are all $O(\gamma_s)$. We emphasize that the guarantees \ref{prop:sparsitySpanner} and \ref{item:lowReEmbed} are with respect to the number of updates in each batch $U^{(t)}$ and \emph{not} with respect to the (possibly much larger) encoding size of $U^{(t)}$. This is of utmost importance for our application.

In this section, we will prove \Cref{thm:spanner} under the assumption that the update sequence is bounded by $n-1$ and that each update batch consists only of a single update. This is without loss of generality as one can restart the algorithm every $n$ updates without affecting any of the bounds.

\subsection{The Algorithm}

\paragraph{Data Structures.} Our algorithm to implement \Cref{thm:spanner} works with a rather standard batching approach with $L+1$ levels. The algorithm therefore maintains graphs $H_0, H_1, \dots, H_{L}$ which form the sparsifier $H \defeq \bigcup H_j$ implicitly. The algorithm recomputes each graph $H_j$ every so often, with shallower levels being recomputed less often than deeper levels. However, each $H_j$ undergoes frequent changes since after each stage we apply the updates to $G$ to each graph $H_j$ (if applicable) such that the graphs remain subgraphs of $G$ with the same vertex set at any stage.

It further maintains embeddings $\Pi_0, \Pi_1, \dots, \Pi_L$ where each embedding $\Pi_j$ maps a subset of $E(G)$ into the graph $H_{\leq j} \defeq \bigcup_{i \leq j} H_i$. In the algorithm, the pre-image of the embeddings $\Pi_0, \Pi_1, \ldots, \Pi_L$ is not disjoint, and in fact, we let $\Pi_0$ always have the full set $E(G)$ in its pre-image. We define the embedding $\Pi_{\leq j}$ which to be the embedding that maps each edge $e \in E(G)$ via the embedding $\Pi_i$ with the largest $i \leq j$ that has $e$ in its pre-image. 

Whenever we recompute a graph $H_j$, we also recompute $\Pi_j$ such that after recomputation $\Pi_{\leq j}$ embeds the current graph $G$ into the current graph $H_{\leq j}$. As for the graphs $H_{j}$, we apply updates to $G$ to the embedding paths in $\Pi_j$ which means that eventually an edge $e \in E(G)$ with endpoints $a,b$ might not be mapped by $\Pi_j(e)$ to an actual $a$-$b$ path; either because edges on $\Pi_j(e)$ are deleted, or vertices are split or both. However, the algorithm ensures that most edges are correctly mapped via $\Pi_{\leq j}$ at all times and the small fraction of edges that is not properly dealt with are then dealt with by embeddings $\Pi_{j+1}, \Pi_{j+2}, \dots, \Pi_{L}$ on deeper levels. The embedding $\Pi_{G \to H}$ is again maintained implicitly and defined $\Pi_{G \to H} \defeq \Pi_{\leq L}$. 

Finally, we maintain sets $S_0, S_1, \ldots, S_L$. Each set $S_j$ consists of the vertices that are touched by the updates to $G$ since the last time that $H_j$ was recomputed. We give a formal definition for touched vertices below.

\begin{definition}\label{def:touchedVertex}
We say that the $t$-th update to $G$ \emph{touches} a vertex $u$ if the update is an edge deletion and $u$ is one of its endpoints, or if the update is a vertex split and $u$ is one of the resulting vertices from the split.
\end{definition}

\paragraph{Initialization.} We start our algorithm by running the sparsification procedure below to initialize $H_0$ and $\Pi_0$. 
\begin{theorem}\label{cor:sparsifier}
Given an unweighted, undirected graph $G$. There is a procedure $\textsc{Sparsify}(G)$ that produces a sparse subgraph %
$H_0 \subseteq G$ and an embedding $\Pi_{0}$ of vertex-congestion at most $2\gamma_c \Delta_{\max}(G)$ and length at most $\gamma_l$, with high probability. The algorithm takes time $\tilde{O}(m\gamma_l)$.
\end{theorem}

For $j > 0$, we initialize $H_j$ to the empty graph and $\Pi_j$ to be an empty map. We set all sets $S_j$ (including $S_0$) to the empty vertex set.

\paragraph{Updates.} At each stage $t = 1, 2, \dots$, we invoke the procedure $\textsc{Update}(t)$ that is given in \Cref{alg:updateSparsifier}. As mentioned earlier, the algorithm works by batching updates made to the graph $G$. In \Cref{lne:phaseNumSparsifier}, the algorithm determines the current batch level $j$ which in turn determines the batch size that has to be handled at the current stage $t$. Once $j$ is determined, we also find the last stage $t_{j-1}$ that a level-$(j-1)$ update occurred.

The algorithm then sets all graphs and embeddings of level $j, j+1, \ldots, L$ to the empty graphs/ embeddings in \Cref{lne:touchSj}.

It then forms the graph $J$ which is the key object in this section. This graph is constructed by finding all edges $e$ whose embedding path in $\Pi_{<j}(e)$ was affected by updates since the last recomputations of $\Pi_0, \Pi_1, \ldots, \Pi_{j-1}$ and adds some projection $\hat{e}$ of $e$ onto the vertex set $S_{j-1}$ to $J$. 

The idea behind this projection is that as $S_{j-1}$ is the set of vertices touched since this stage, we have that an affected edge $e$ has some vertex of $S_{j-1}$ on its path. Assuming for the moment that the edge $e$ itself is not incident to a vertex that was split since the recomputation stage, we essentially project the endpoints $a,b$ of $e$ to the nearest vertices $\hat{a}, \hat{b}$ in $S_{j-1}$ on the old embedding path $\Pi_{<j}(e)$. Note in particular that by definition the path $\Pi_{<j}(e)[a, \hat{a}]$ and $\Pi_{<j}(e)[ \hat{b}, b]$ are then still in $G$. 

We give the following more formal definition that also defines a projection for the slightly more involved case where the edge $e$ is incident to a vertex that splits over time.

\begin{definition}[Edge-Embedding Projection]\label{def:edgeProj}
For any $0 < j \leq L$, embedding $\Pi_{< j}$, set $S_{j-1}$ being the set of all vertices touched by updates to $G$ since the last stage that $\Pi_{< j}$ was modified, and edge $e \in E(G)$ such that $\Pi_{< j}(e) \cap S_{j-1} \neq \emptyset$. Then, we let $\proj_{j-1}(e)$ be a new edge $\hat{e}$ that is associated with $e$ and has endpoints $\hat{a}$ and $\hat{b}$ being the closest vertices in $S_{j-1}$ to the endpoints of $e$ in the current graph $\Pi_{<j}(e)$, respectively. Here $\Pi_{<j}(e)$ refers to the graph over the entire vertex set $V(G)$ obtained from adding the edges that are on $\Pi_{< j}(e)$ and then applying the relevant updates that took place on $G$ since $\Pi_{< j}$ was last modified.
\end{definition}

As previously mentioned, we project these edges $e$ whose embedding path $\Pi_{<j}(e)$ was affected by updates onto $S_{j-1}$ to obtain edge $\hat{e}$ which is added to the graph $J$. Note that as projected edges are associated with edges $e$ in $G$, the graph $J$ can be a multigraph. 

\begin{algorithm}
Update all sparsifiers $H_0, H_1, \ldots, H_L$ with the $t$-th update if it applies.\;
Add the vertices touched by the $t$-th update to each of the sets $S_0, S_1, \ldots, S_L$.\label{lne:touchSl}\;
$j \gets \min \{ j' \in \mathbb{Z}_{\geq 0} \;|\; t \text{ is divisible by } n^{1-j'/L}\}$.\label{lne:phaseNumSparsifier}\tcp*{Determine $j$ level of stage $t$.}
$t_{j-1} \gets \lfloor t / n^{1-(j-1)/L} \rfloor \cdot n^{1-(j-1)/L}$. \tcp*{$t_{j-1}$ is the most recent level-$(j-1)$ stage.}
\For(\tcp*[f]{Re-set level-$\geq j$ sparsifiers.}){$i = j, j+1, , \ldots, L$}{$H_i \gets (V, \emptyset)$; Set $\Pi_i$ to be an empty map; $S_i \gets \emptyset$.\label{lne:touchSj}}
\tcp{Auxiliary Graph and embedding by projecting embedding paths onto $S_{j-1}$.}
$J \gets (S_{j-1}, \emptyset)$; $\Pi_{J \to H_{< j} \cup E_{\mathit{affected}}} \gets \emptyset$.\;
$E_{\mathit{affected}} \gets \{ e \in E(G) \text{ with } \Pi_{<j}(e) \cap S_{j-1} \neq \emptyset\}$.\;
\ForEach(\label{lne:updateFirstForeach}){edge $e \in E_{\mathit{affected}}$}{
  $\hat{e} \gets \proj_{j-1}(e)$.\tcp*{Find Projected Edge of $e$.}
      Let $a$ and $b$ be the endpoints of $e$, and $\hat{a}$ and $\hat{b}$ be the endpoints of $\hat{e}$.\;
    Add $\hat{e}$ to $J$.\label{lne:addProjectedEdgeToJ}\;
    $\Pi_{J \to H_{< j} \cup E_{\mathit{affected}}}(\hat{e}) \gets \Pi_{<j}(e)[\hat{a}, a] \pconcat e \pconcat \Pi_{<j}(e)[b, \hat{b}]$.\label{lne:preImageWillSufficeInFuture}
}
\tcp{Sparsify Auxiliary Graph and translate back to re-build $H$.}
$(\tilde{J}, \Pi_{J \to \tilde{J}}) \gets \textsc{Sparsify}(H_{<j} \cup E_{\mathit{affected}}, J, \Pi_{J \to H_{< j} \cup E_{\mathit{affected}}})$.\label{lne:computeSparsifierJ}\;

\lForEach{edge $e \in E_{\mathit{affected}}$ and $\hat{e} = \proj_{j-1}(e)\in \tilde{J}$}{Add $e$ to $H_j$.\label{lne:addEdgeToSparsifier}}
\ForEach(\label{lne:forEachRembed}){edge $e \in E_{\mathit{affected}}$}{%
    $\hat{e} = \proj_{j-1}(e)$\;
    Let $a$ and $b$ be the endpoints of $e$, and $\hat{a}$ and $\hat{b}$ be the endpoints of $\hat{e}$.\;
    $\Pi_{j}(e) \gets \Pi_{<j}(e)[a, \hat{a}] \pconcat  [\Pi_{J \to H_{< j} \cup E_{\mathit{affected}}} \circ \Pi_{J \to\tilde{J}}](\hat{e}) \pconcat \Pi_{<j}(e)[\hat{b}, b]$.\label{lne:setPijCorrectly}
}
\Return $D = \Pi_{J \to H_{< j} \cup E_{\mathit{affected}}}(\tilde{J})$.
\caption{$\textsc{Update}(t)$}
\label{alg:updateSparsifier}
\end{algorithm}

Along with $J$, there is also a natural embedding of the projected edge $\hat{e}$ into the sparsifier $H_{< j}$: we can simply take the path $\Pi_{<j}(e)[\hat{a}, a] \pconcat e \pconcat \Pi_{<j}(e)[b, \hat{b}]$ which we already argued to exist in the current graph $H_{< j}$. This embedding is constructed in \Cref{lne:preImageWillSufficeInFuture}.

Finally, the procedure $\textsc{Sparsify}$ is invoked on the graph $J$. While we have seen this procedure before in the initialization stage, here, we use a generalized version that incorporates the embedding from $J$ into $H$ constructed above. The guarantees of our generalized procedure $\textsc{Sparsify}$ are summarized below. Note that letting $J = G$; and letting $\Pi_{J \to G}$ be the identity function, we recover \Cref{cor:sparsifier} as a corollary.

\begin{restatable}{theorem}{staticEmbedding}\label{lma:staticEmbed}
Given unweighted, undirected graphs $H'$ and $J$ with $V(J) \subseteq V(H')$ and an embedding $\Pi_{J \to H'}$ from $J$ into $H'$. Then, there
is a randomized algorithm $\textsc{Sparsify}(H', J, \Pi_{J \to H'})$ that returns a sparsifier $\tilde{J} \subseteq J$ with $|E(\tilde{J})| = \tilde{O}(|V(J)|)$ and an embedding $\Pi_{J \to\tilde{J}}$ from $J$ to $\tilde{J}$ such that
\begin{enumerate}
    \item $\vcong(\Pi_{J \to H'} \circ \Pi_{J \to\tilde{J}}) \leq \gamma_c \cdot \left(\vcong(\Pi_{J \to H'}) + \Delta_{\max}(J)\right)$, and
    \item $\length(\Pi_{J \to H'} \circ \Pi_{J \to\tilde{J}}) \leq \gamma_l \cdot \length(\Pi_{J \to H'})$. 
\end{enumerate}
The algorithm runs in time $\tilde{O}(|E(J)| \cdot \gamma_l)$ and succeeds with probability at least $1-n^{-C}$ for any constant $C$, specified before the procedure is invoked.
\end{restatable}

We defer the proof of the theorem to 
\Cref{subsec:sparsify} and finish the description of \Cref{alg:updateSparsifier}. Given the sparsified graph $\tilde{J}$ (along with the embedding map), we find the pre-images of the edges in $\tilde{J}$ and add them to $H$. We then re-embed all edges $(a,b)$ in $G$ that were no longer properly embedded into $H$ by using the embedding $\Pi_{J \to H} \circ \Pi_{J \to\tilde{J}}$ to get from $\hat{a}$ to $\hat{b}$ and then prepend (append) the path $\Pi_{J \to \tilde{J}}(e)$ from $a$ to $\hat{a}$ (from $\hat{b}$ to $b$). To gain better intuition for the resulting embedding, we recommend the reader to follow the analysis in \Cref{lma:dynamicEmbeddingIsRealEmbedding}.

\paragraph{Proof of \Cref{thm:spanner}.} As a first part of the analysis, we establish that the algorithm indeed maintains an actual embedding.

\begin{lemma}\label{lma:dynamicEmbeddingIsRealEmbedding}
For any $0 \leq i \leq L$, and stage $t$ divisible by $n^{1-i/L}$, $\Pi^{(t)}_{\leq i}$ embeds $G^{(t)}$ into $H^{(t)}_{\leq i}$. In particular, at any stage $t$, $\Pi_{G \to H}^{(t)} = \Pi_{\leq L}^{(t)}$ embeds $G^{(t)}$ into $H^{(t)}$.
\end{lemma}
\begin{proof}
We prove by induction on the stage $t$. For the base case ($t=0$), we observe that in the initialization phase $\Pi_0$ embeds $G^{(0)}$ into $H_0$ via the algorithm in \Cref{cor:sparsifier}. For all other $j > 0$, $\Pi_j$ is an empty embedding and therefore, the claim follows.

Let us next consider the inductive step $t-1 \mapsto t$: We let $j = j^{(t)}$ and $t_{j-1} = t_{j-1}^{(t)}$. Consider any edge $e \in E(G^{(t)})$. We first note that the path $\Pi^{(t_{j-1})}_{< j}(e)$ in $H^{(t_{j-1})}_{< j}$ exists since we can invoke the induction hypothesis by the fact that $t_{j-1} < t$ which follows from the minimality of $j$ (see \Cref{lne:phaseNumSparsifier}). By definition of $S_{j-1}$ being the vertices touched by all updates to $G$ since $t_{j-1}$, we have that if $\Pi^{(t)}_{< j}(e) \cap S^{(t)}_{j-1} = \emptyset$, that $\Pi^{(t)}_{<j}(e)$ still embeds $e$ properly into $H^{(t)}_{< j}$. Since the foreach-loop in \Cref{lne:forEachRembed} is not entered for such edges $e$, we further have that $\Pi^{(t)}_{\leq j}(e) = \Pi^{(t)}_{<j}(e)$.
    
It remains to analyze the case $\Pi^{(t)}_{< j}(e) \cap S^{(t)}_{j-1} \neq \emptyset$ where we note that $\hat{e}=\proj_{j-1}(e)$ is well-defined. We start with the observation that $\Pi_{J \to H_{< j} \cup E_{\mathit{affected}}}(\hat{e})$ restricted to the edges in $\tilde{J}$ only maps to the edges in $H^{(t)}_{<j}$ and the pre-images of edges $\hat{e} \in \tilde{J}$ under the $\proj_{j-1}(\cdot)$ map as can be seen easily from its construction in \Cref{lne:preImageWillSufficeInFuture}. But $H^{(t)}_j$ consists exactly of the pre-images of $\tilde{J}$ as can be seen from  \Cref{lne:addEdgeToSparsifier}, so each such embedding path is in $H^{(t)}_{\leq j}$. Since $\Pi_{J \to \tilde{J}}$ maps all edges $\hat{e}$ in $J$ to paths in $\tilde{J}$, we thus have $[\Pi_{J \to H_{< j} \cup E_{\mathit{affected}}} \circ \Pi_{J \to\tilde{J}}](\hat{e})$ in $H^{(t)}_{\leq j}$.
    
By the way the algorithm sets the new embedding path $\Pi^{(t)}_j(e)$ in \Cref{lne:setPijCorrectly}, we thus only have to argue about the path segments  $\Pi^{(t)}_{<j}(e)[a, \hat{a}]$ and $\Pi^{(t)}_{<j}(e)[\hat{b}, b]$ where $a,b$ are the endpoints of $e$ and $\hat{a},\hat{b}$ are the endpoints of $\hat{e}$. But again, by definition of $\hat{e}$ (see \Cref{def:edgeProj}), we have that both of these path segments are contained in $H^{(t)}_{< j}$.

Finally, for all $i > j$, $\Pi^{(t)}_i$ is set to be empty and therefore $\Pi^{(t)}_{\leq i} = \Pi_{\leq j}^{(t)}$.
\end{proof}

It turns out that the proof of \Cref{lma:dynamicEmbeddingIsRealEmbedding} is already the most complicated part of the analysis. We next bound congestion and length of the embedding.

\begin{claim}\label{clm:vertexCongSparsifier}
For any $0 \leq i \leq L$, we have $\vcong(\Pi^{(t)}_{\leq i}) \le 4^i \gamma_c^{i+1} \Delta_{\max}(G)$.
\end{claim}
\begin{proof}
Again, we prove by induction on stage $t$. We have for $i = 0$, that $\Pi^{(0)}_0$ as computed in the initialization stage has vertex-congestion at most  $\gamma_c\Delta_{\max}(G)$ by \Cref{cor:sparsifier}. For $i > 0$, we have that $\Pi^{(0)}_i$ is empty; therefore its congestion is $0$.

For $t-1 \mapsto t$, we define $j = j^{(t)}$ and $t_{j-1} = t_{j-1}^{(t)}$. Observe that for each edge $e$ considered in the first foreach-loop starting in \Cref{lne:updateFirstForeach}, $\Pi_{J \to H_{< j} \cup E_{\mathit{affected}}}(e)$ consists only of the edges in $\Pi^{(t)}_{< j}(e)$ and the edge $e$ itself, it follows that every embedding path that contributes to vertex congestion of a vertex $v$ in $\vcong(\Pi_{J \to H_{< j} \cup E_{\mathit{affected}}})$ also contributes to the vertex congestion of $v$ in $\vcong(\Pi^{(t)}_{<j})$, and hence
$\vcong(\Pi_{J \to H_{< j} \cup E_{\mathit{affected}}}) \leq \vcong(\Pi^{(t)}_{<j})$. Further, we can see from the construction of the graph $J$ that $\Delta(J) \leq \vcong(\Pi^{(t)}_{<j})$. 

Let us next analyze $\vcong(\Pi^{(t)}_{<j})$. By minimality of $j$ (see \Cref{lne:phaseNumSparsifier}), we have $t_{j-1} < t$ and we can use the induction hypothesis to get $\vcong(\Pi^{(t_{j-1})}_{<j}) \le 4^{j-1}\gamma_c^{j} \Delta_{\max}(G)$. It is further immediate to see that since the embedding $\Pi_{<j}$ was not affected by any recomputations since stage $t_{j-1}$ that the vertex congestion can only have dropped ever since.

Thus, when the graph $J$ is sparsified in \Cref{lne:computeSparsifierJ}, by \Cref{lma:staticEmbed}, we can conclude
\[
\vcong\left(
\Pi_{J \to H_{< j} \cup E_{\mathit{affected}}}
\circ \Pi_{J \to\tilde{J}}\right)
\leq
2 \cdot 4^{j-1}\gamma_c^{j+1} \Delta_{\max}\left(G\right).
\]

Finally, when we construct the embedding $\Pi_j$ in \Cref{lne:setPijCorrectly}, the path segments $[\Pi_{J \to H_{< j} \cup E_{\mathit{affected}}} \circ \Pi_{J \to\tilde{J}}](\hat{e})$ incur vertex congestion at most $\vcong(
\Pi_{J \to H_{< j} \cup E_{\mathit{affected}}}\circ \Pi_{J \to\tilde{J}})$, and the path segments
$\Pi_{<j}(e)[a, \hat{a}]$ and $\Pi_{<j}(e)[\hat{b}, b]$ incur total vertex congestion at most $\vcong(\Pi^{(t_{j})}_{<j})$. 

As congestion is additive, we can upper bound the total congestion of $\Pi_{j}^{(t)}$ by $( 4^{j-1}\gamma_c^{j} + 2 \cdot 4^{j-1}\gamma_c^{j+1})\Delta_{\max}(G)$ and can finally use $\vcong(\Pi_{\leq j}^{(t)}) \leq \vcong(\Pi_{j}^{(t)}) + \vcong(\Pi_{< j}^{(t)}) \leq 4^j \gamma_c^{j+1} \Delta_{\max}(G)$. For all $i > j$, we note that $\Pi^{(t)}_i$ is empty and therefore, $\vcong(\Pi^{(t)}_{\leq i}) \leq \vcong(\Pi^{(t)}_{\leq j})$.
\end{proof}

\begin{claim}\label{clm:lengthsEmbedding}
For any $0 \leq i \leq L$ and stage $t$ divisible by $n^{1-i/L}$, we have $\length(\Pi^{(t)}_{\leq i}) \le 2^i \gamma_l^{i+1}$.
\end{claim}
\begin{proof}
We again take induction over time $t$. For $t=0$, we note that $\Pi_0$ has length $\gamma_l$ by \Cref{cor:sparsifier}. For $t-1 \to t$, for $j = j^{(t)}$ and $t_{j-1} = t_{j-1}^{(t)}$, we have by induction hypothesis that $\Pi^{(t_{j-1})}_{< j}(e) \leq 2^{j-1} \gamma_l^{j}$. But note that when we set the path $\Pi^{(t)}_{j}(e)$ in \Cref{lne:setPijCorrectly}, then the segments $\Pi^{(t)}_{<j}(e)[a, \hat{a}]$ and $ \Pi^{(t)}_{<j}(e)[\hat{b}, b]$ (combined) are of length at most $2^{j-1} \gamma_l^{j}$ because they survived from $\Pi^{(t_{j-1})}_{< j}(e)$ by definition of $S_{j-1}$. Further, the segment $[\Pi_{J \to H_{< j} \cup E_{\mathit{affected}}} \circ \Pi_{J \to\tilde{J}}](\hat{e})$ has length at most $\gamma_l \cdot \length(\Pi_{J \to H_{< j} \cup E_{\mathit{affected}}})$ by \Cref{lma:staticEmbed}. But by construction in \Cref{lne:preImageWillSufficeInFuture}, the embedding $\Pi_{J \to H_{< j} \cup E_{\mathit{affected}}}$ has length at most $\length(\Pi^{(t_{j-1})}_{< j}(e)) + 1$. 

Combining these insights, we have $\length(\Pi_{\leq j}^{(t)}) \leq 2^{j-1} \gamma_l^{j} + \gamma_l (\length(\Pi^{(t_{j-1})}_{< j}(e)) + 1) \leq 2^{j}\gamma_l^{j+1}$. For $i > j$, we have $\length(\Pi_{\leq j}) = \length(\Pi_{\leq i})$.
\end{proof}

Now that we established all properties of the embedding, it remains to analyze the sparsifier $H$.

\begin{lemma}\label{clm:runningTimeSparsifier}
At any stage, $H$ consists of at most $\tilde{O}(n)$ edges and the amortized number of changes to the edge set of $H$ per update is $\tilde{O}(n^{1/L})$. $D^{(t)}$ is of amortized size $\tilde{O}(n^{1/L}(\gamma_c\gamma_l)^{O(L)})$. Initialization time of the algorithm is $\tilde{O}(m \gamma_l)$ and it has  amortized update time $\tilde{O}(n^{1/L}(\gamma_c\gamma_l)^{O(L)}\Delta_{\max}(G))$.
\end{lemma}
\begin{proof}
The graph $H_0$ is computed during initialization and remains fixed and therefore consists of $\tilde{O}(n)$ edges by \Cref{cor:sparsifier} and contributes no recourse. For each $j > 0$, $H_j$ is initially empty and only has edges added in stages $t$ divisible by $n^{1-j/L}$ (but not by $n^{1-(j-1)/L}$) in \Cref{lne:addEdgeToSparsifier}. In each such stage $t$, the graph $J$ is formed over the vertices $S_{j-1}$. It is straight-forward to see by \Cref{def:touchedVertex} and  \Cref{lne:touchSj} 
that $S_{j-1}$ is of size at most $n^{1-(j-1)/L}$ at any stage. Thus, when the graph $\tilde{J}$ is computed, it consists of at most $\tilde{O}(n^{1-(j-1)/L})$ edges by \Cref{lma:staticEmbed}. The bounds on overall sparsity of $H$ follow. 

For the claim on the recourse, we note that in stages $t$ divisible by $n^{1-j/L}$ (but not by $n^{1-(j-1)/L}$), we recompute a spanner on the vertices $S_{j-1}$ which are a subset of the vertices in $G^{(t)}$ and add $\tilde{O}(n^{1-(j-1)/L})$ edges. For the graphs $H^{(t)}_{j'}$ for $j' > j$, the graphs are empty after the algorithm finishes. Using an inductive argument, we can argue that the number of edge deletions at stage $t$ can also be upper bound by $\tilde{O}(n^{1-(j-1)/L})$. Thus, it is not hard to see that at most $\tilde{O}(n^{1/L})$ amortized changes to the edge set of $H$ are made. It remains to argue about a rather subtle detail: if the update is a vertex split applied to $G^{(t-1)}$ to obtain $G^{(t)}$, then we also need to account for the recourse caused by the vertex split to the graphs $H_{j'}^{(t-1)}$ for $j' < j$. But note that we only pay in recourse cost for edges that are moved from a vertex $v$ to a vertex $v'$ if the degree of $v'$ after the vertex split is at most half the degree of $v$'s degree. Thus, we can charge each edge that is moved this way. Further, if $v'$'s degree is then again increased by a factor of $3/2$, we can further re-pay that cost of moving by charging the newly inserted edges. Following this charging scheme, we can argue that each edge can be charged to pay $O(\log(n))$ on insertion and an additional $O(\log(n))$ in recourse for the halving of degrees (after being recompensated if the degree goes up again). Since there are $\tilde{O}(n)$ edges initially in $H$ and at most $\tilde{O}(n^{1/L})$ new edges after each update appear, our recourse bound follows.

To obtain the bound on $D^{(t)}$,
we first observe that for each path $\Pi_j(e')$ constructed in \Cref{lne:setPijCorrectly}, by induction over time, we can straight-forwardly establish that $\Pi_{<j}(e')[a, \hat{a}]$ and $\Pi_{<j}(e')[\hat{b},b]$ are subpaths of $\Pi_{G \to H}^{(t-1)}(e')$ by the properties of set $S_{j-1}$. Thus, the only edges $e$ on any such $\Pi_j(e')$ not already on the path $\Pi_{G \to H}^{(t-1)}(e')$ are the edges in the subpath $[\Pi_{J \to H_{< j} \cup E_{\mathit{affected}}} \circ \Pi_{J \to\tilde{J}}](\hat{e'})$. But clearly, $[\Pi_{J \to H_{< j} \cup E_{\mathit{affected}}} \circ \Pi_{J \to\tilde{J}}](\hat{e'}) \subseteq \Pi_{J \to H_{< j} \cup E_{\mathit{affected}}}(\tilde{J})$. It remains to use our bound on the number of edges in $\tilde{J}$ and the fact that the map $\Pi_{J \to H_{< j} \cup E_{\mathit{affected}}}$ maps edges to paths
of length $\gamma_c^{O(L)}$ in $H$ by \Cref{clm:lengthsEmbedding}. 

For the running time, we use \Cref{cor:sparsifier} for the initialization, and observe that each vertex in $J$ as analyzed above has degree at most $O(\gamma_c^{O(L)} \Delta_{\max}(G))$ as discussed in the proof of \Cref{clm:vertexCongSparsifier}. Thus, using \Cref{lma:staticEmbed} computing each sparsifier $\tilde{J}$ of $J$ only takes time $\tilde{O}(|V(J)| \gamma_c^{O(L)} \Delta_{\max}(G))$. By standard amortization arguments and the fact that the time to  compute the sparsifier dominates the update time of $\textsc{Update}(t)$ asymptotically, the lemma follows. 
\end{proof}

To complete the proof of \Cref{thm:spanner}, we only have to analyze the success probability, which is straight-forward as the only random event at each stage is the invocation of the procedure $\textsc{Sparsify}$. Thus, taking a simple union bound over these events at all stages gives the desired result.

\subsection{Implementing the Sparsification Procedure}
\label{subsec:sparsify}

It remains to prove the procedure that statically sparsifies graphs.

\staticEmbedding*

\paragraph{Additional Tools.} At a high level, the proof of \Cref{lma:staticEmbed} follows by performing an expander decomposition, uniformly subsampling each expander to produce a sparsifier, and then embedding each expander into its sparsifier by using a data structure for outputting short paths between vertices in decremental expanders. To formalize this, we start by surveying some tools on expander graphs. Recall the definiton of expanders.

\begin{definition}[Expander]\label{def:expander}
Let $G$ be an unweighted, undirected graph and $\phi \in (0,1]$, then we say that $G$ is a $\phi$-expander if for all $\emptyset \neq S \subsetneq V$, $|E_G(S, V \setminus S)| \geq \phi \cdot \min\{\vol_G(S), \vol_G(V \setminus S)\}$.
\end{definition}

We can further get a collection of expander decomposition with near uniform degrees in the expanders. The proof of this statement follows almost immediately from \cite{SW19} and is therefore deferred to \Cref{app:expanderStatement}.

\begin{restatable}{theorem}{decompose}\label{thm:expanderStatement}
Given an unweighted, undirected graph $G$, there is an algorithm $\textsc{Decompose}(G)$ that computes an edge-disjoint partition of $G$ into graphs $G_0, G_1, \dots, G_{\ell}$ for $\ell = O(\log n)$ such that for each $0 \leq i \leq \ell$, $|E(G_{i})| \leq 2^{i} n$ and for each nontrivial connected component $X$ of $G_i$, $G_i[X]$ is a $\psi$-expander for $\psi = \Omega(1/\log^3(m))$, and each $x \in X$ has $\deg_{G_i}(x) \geq \psi 2^{i}$. The algorithm runs correctly in time $O(m \log^7(m))$, and succeeds with probability at least $1-n^{-C}$ for any constant $C$, specified before the procedure is invoked.
\end{restatable}

Further, we use the following result from \cite{CS21}. Given a $\phi$-expander undergoing edge deletions the data structure below implicitly maintains a subset of the expander that still has large conductance using standard expander pruning techniques (see for example \cite{NSW17,SW19}). Further on the subset of the graph that still has good conductance, it can output a path of length $m^{o(1)}$ between any pair of queried vertices.

\begin{theorem}[see Theorem 3.9 in arXiv v1 in \cite{CS21}]\label{thm:apspDataStructure}
Given an unweighted, undirected graph $G$ that is $\phi$-expander for some $\phi > 0$.
There is a deterministic data structure $\mathcal{DS}_{\mathit{ExpPath}}$
that explicitly grows a monotonically increasing 
``forbidden'' vertex subset $\hat{V} \subseteq V(G)$
while handling the following operations:
\begin{itemize}
    \item $\textsc{Delete}(e)$: Deletes edge $e$ from $E(G)$ and then explicitly outputs a set of vertices that were added to $\hat{V}$ due to the edge deletion. 
    \item $\textsc{GetPath}(u,v)$: for any $u,v \in V(G) \setminus \hat{V}$ returns a path consisting of at most $\gamma_{\mathit{ExpPath}}$ edges
    between $u$ and $v$ in the graph $G[V(G) \setminus \hat{V}]$. Each path query can be implemented in time $\gamma_{\mathit{ExpPath}}$, where $\gamma_{\mathit{ExpPath}} = (\log(m)/\phi)^{O(\sqrt{\log(m)})}$. The operation does not change the set $\hat{V}$.
\end{itemize}
The data structure ensures that after $t$ edge deletions $\vol_{G^{(0)}}(\hat{V}) \leq \gamma_{\mathit{del}} t/\phi$ for some constant $\gamma_{\mathit{del}} = O(1)$. The total update time taken by the data structure for initialization and over all deletions is $O(|E(G^{(0)})| \gamma_{\mathit{ExpPath}})$.
\end{theorem}

\paragraph{The Algorithm.} We can now use these tools to give \Cref{alg:sparsify} that implements the procedure $\textsc{Sparsify}(H', J, \Pi_{J \to H'})$.

\begin{algorithm}[!ht]
$J_0, J_1, \ldots, J_{\ell} \gets \textsc{Decompose}(J)$.\;
$\tilde{J} \gets (V, \emptyset)$.\;
\lForEach{$e \in E(J)$}{$\Pi_{J \rightarrow \tilde{J}}(e) \gets \emptyset$.}
\ForEach{$i \in [0,\ell]$ and connected component $X$ in $J_i$}{
    \tcc{Sample the edges that are added to the sparsifier $\tilde{J}$.}
    $p_{X,i} \defeq \min\left\{ \frac{96C \log(m)}{\psi \Delta_{min}(J_i[X])}, 1\right\}$.
     \label{lne:pXiDefinition}\;
    Construct graph $\tilde{J}_{X,i}$ by sampling each edge $e \in E(J_i[X])$ independently with probability $p_{X,i}$.\label{lne:edgesAreSampledIid}\;
    
    Add all edges in $\tilde{J}_{X,i}$ to $\tilde{J}$.\;

    \tcc{Embed all edges in $J_i[X]$ into the sampled local graph $\tilde{J}_{X,i}$.}
    \lForEach{$e \in \tilde{J}_{X,i}$}{
        $\Pi_{J \rightarrow \tilde{J}}(e) \gets e$.\label{lne:embedByItselfIfSampled}
    }
    \While(\label{lne:outerWhileStaticEmbed}){there exists an edge $e \in E(J[X_i])$ with $\Pi_{J \rightarrow \tilde{J}}(e) = \emptyset$}{
    $\tilde{J}_{APSP} \gets \text{a copy of } \tilde{J}_{X,i}$.\; 
    Initialize $\mathcal{DS}_{\mathit{ExpPath}}$ on graph $\tilde{J}_{APSP}$ with parameter $\phi \defeq \psi/4$ maintaining set $\hat{V}$.\label{lne:initAPSP}\;
        \lForEach{$e \in E(\tilde{J}_{X,i})$}{$cong(e) \gets 0$.}
        \While(\label{lne:innerWhileStaticEmbed}){there exists an edge $e \in E(J[V \setminus \hat{V}])$ with $\Pi_{J \rightarrow \tilde{J}}(e) = \emptyset$}{
            Let $u$ and $v$ be the endpoints of edge $e$.\;
            $\Pi_{J \rightarrow \tilde{J}}(e) \gets \mathcal{DS}_{\mathit{ExpPath}}.\textsc{GetPath}(u,v)$.\label{lne:computeMCFlowGetPath1}\;
            \ForEach{$e \in \Pi_{J \rightarrow \tilde{J}}(e)$}{
                $cong(e) \gets cong(e) + 1$\;
                    \If{$cong(e) \geq \tau \defeq \frac{\gamma_{\mathit{ExpPath}}\gamma_{\mathit{del}}}{\psi p_{X,i}}$
                    }{ 
                    Remove edge $e$ from $\tilde{J}_{APSP}$ via $\mathcal{DS}_{\mathit{ExpPath}}.\textsc{Delete}(e)$.\label{lne:removeCOngestedEdge1} }
            }
        }
    }
}
\Return $(\tilde{J}, \Pi_{J \rightarrow \tilde{J}})$
\caption{$\textsc{Sparsify}(H', J, \Pi_{J \to H'})$}
\label{alg:sparsify}
\end{algorithm}

The algorithm has two key steps:
\begin{enumerate}
    \item \underline{Sampling:} The graph $J$ is first decomposed via \Cref{thm:expanderStatement}. Then, the algorithm iterates over $\psi$-expanders $J_i[X]$ with near-uniform degrees. It is well-known that to obtain a sparsifier $\tilde{J}_{X,i}$ of such graphs, one can simply sample each edge with probability roughly $\frac{\log(m)}{\psi \Delta_{min}(J_i[X])}$. To obtain the final sparsifier $\tilde{J}$, we only have to take the union over all samples $\tilde{J}_{X,i}$.
    \item \underline{Embedding:} We then proceed to find an embedding for edges in $J_i[X]$ into $\tilde{J}_{X,i}$. The sampled edges can be handled trivially by embedding them into themselves. To embed the remaining edges $e \in E(J)$ into $\tilde{J}_{X,i}$, we exploit that $\tilde{J}_{X,i}$ is an expander graph which allows us to employ the data structure from \Cref{thm:apspDataStructure} on $\tilde{J}_{X,i}$ to query for a path between the endpoints of $e$ in $\tilde{J}_{X,i}$ efficiently. We further keep track of the congestion of each edge in $\tilde{J}_{X,i}$ by our embedding and remove edges that are too congested (at least until we cannot embed anymore in any other way).
\end{enumerate}
We point out that  \Cref{alg:sparsify} in no way uses the embedding $\Pi_{J \rightarrow H'}$. Still, we show that due to the structure given, we can tightly upper bound the congestion and length of the embedding given by the composition $\Pi_{J \rightarrow H'} \circ \Pi_{J \rightarrow \tilde{J}}$.

\paragraph{Proof of \Cref{lma:staticEmbed}.} We start by proving the following structural claim. For the rest of the section, we condition on the event that it holds for each relevant $i$ and $X$.

\begin{claim}\label{clm:HisCutSparsifier}
For each $i$, and connected component $X$ in $J_i$, the corresponding sample $\tilde{J}_{X,i}$ satisfies for each $S \subseteq X$ that $\frac{1}{2}|E_{\tilde{J}_{X,i}}(S, X \setminus S)|/p_{X,i} \leq |E_{J_i}(S, X \setminus S)| \leq 2 |E_{\tilde{J}_{X,i}}(S, X \setminus S)|/p_{X,i}$ with probability at least $1- n^{-2C}$.
\end{claim}
\begin{proof}
Since for $i = 0$, $J_i[X] = \tilde{J}_{X,i}$ and $p_{X,i} = 1$, the claim is vacuously true. For $i > 0$, consider any cut $(S, X \setminus S)$ in $J_i[X]$ and assume wlog $k = |S| \leq |X \setminus S|$. Since $J_i[X]$ is a $\psi$-expander, we have that $|E_{J_i}(S, X \setminus S)| \geq \psi \Delta_{min}(J_i[X]) |S|$ by \Cref{def:expander}. The algorithm samples each such edge $e$ into the sample $\tilde{J}_{X,i}$ independently with probability $p_{X,i}$. Thus, $\mathbb{E}|E_{\tilde{J}_{X,i}}(S, X \setminus S)| = |E_G(S, X \setminus S)| \cdot p_{X,i} \geq 48C \log(m)|S|$.

Using a Chernoff bound as in \Cref{thm:chernoffBound} on the random variable $|E_{\tilde{J}_{X,i}}(S, X \setminus S)|$, we can thus conclude that our claim is correct on the cut $(S, X \setminus S)$ with probability at least $1- 2m^{-4Ck}$. Since there are at most $\binom{|X|}{k} \leq |X|^{2k}$ cuts where the smaller side has exactly $k$ vertices, we can finally use a union bound over all cuts to complete the proof. 
\end{proof}

By the claim above, and the fact that each graph $J_i[X]$ (for $i > 0$) is a $\psi$-expander by \Cref{thm:expanderStatement}, we can conclude that each $\tilde{J}_{X,i}$ is a $\psi/4$-expander.

\begin{corollary}\label{cor:sparsifierIsExpander}
For any $i > 0$ and $X$ as used in \Cref{alg:sparsify}, $\tilde{J}_{X,i}$ is a $\psi/4$-expander.
\end{corollary}

This implies that our initializations of the data structure $\mathcal{DS}_{\mathit{ExpPath}}$ in \Cref{lne:initAPSP} are legal according to \Cref{thm:apspDataStructure}. Next, let us give an upper bound on the congestion of the embedding $\Pi_{J \rightarrow \tilde{J}}$.

\begin{claim}\label{clm:edgeCongestion}
For any $i \in [0, \ell]$ and $X$ as used in \Cref{alg:sparsify}, we have $\econg(\Pi_{J \rightarrow \tilde{J}}|_{ E(J_i[X])}) \leq \gamma_{X,i} = O\left(\frac{\gamma_{\mathit{ExpPath}}\gamma_{\mathit{del}}\log(m)}{\psi p_{X,i}}\right)$ where $\Pi_{J \rightarrow \tilde{J}}|_{ E(J_i[X])}$ denotes the embedding $\Pi_{J \rightarrow \tilde{J}}$ restricted to edges in $J_i[X]$.
\end{claim}
\begin{proof}
We first observe that only in the foreach loop iteration on $i$ and $X$, can any congestion be added to edges in $\tilde{J}_{X,i}$ by the disjointness of the graphs $J_j$ (see \Cref{thm:expanderStatement}). Further, in the particular iteration on $i$ and $X$, up to the while-loop starting in \Cref{lne:outerWhileStaticEmbed}, the congestion of $\Pi_{J \rightarrow \tilde{J}} | E(J_i[X])$ is at most $1$ since only edges sampled into $\tilde{J}_{X,i}$ are embedded into themselves. 

It is straight-forward to observe that the congestion of the partial embedding $\Pi_{J \rightarrow \tilde{J}}$ (restricted to $E(J_i[X])$) throughout each iteration of the outer-while loop starting in \Cref{lne:outerWhileStaticEmbed} is increased by at most $\tau$ as the algorithm track congestion of the current iteration explicitly and removes edges that are too congested in \Cref{lne:removeCOngestedEdge1}.
It thus remains to bound the number of iterations of the outer-while loop starting in \Cref{lne:outerWhileStaticEmbed} by by $O(\log(m))$. We can then conclude that the total congestion on the edges is at most $O(\tau \log(m))$.

To bound this number of iterations, let us analyze a single outer while-loop iteration (starting at \Cref{lne:outerWhileStaticEmbed}), and fix the end of such an iteration. Let $t$ be the number of deletions processed by the data structure $\mathcal{DS}_{\mathit{ExpPath}}$ throughout the iteration and $\hat{V}^{\mathit{FINAL}}$, the set $\hat{V}$ at the end of the while-loop iteration. 
Using that the while-loop terminates, we can further conclude that the only edges not embedded after the current iteration are  
those outside $E(J_i[X \setminus \hat{V}^{\mathit{FINAL}}])$.
By \Cref{thm:apspDataStructure}, $\vol_{\tilde{J}_{X,i}}(\hat{V}^{\mathit{FINAL}}) \leq 4\gamma_{\mathit{del}} t/\psi$ and therefore by \Cref{clm:HisCutSparsifier}, we have $|E(J_i[X]) \setminus E(J_i[X \setminus \hat{V}^{\mathit{FINAL}}])| \leq \vol_{J_i[X]}(\hat{V}^{\mathit{FINAL}}) \leq \frac{8\gamma_{\mathit{del}} t}{p_{X,i}\psi}$. 
But at the same time, we know that at least $t \cdot \tau/ \gamma_{\mathit{ExpPath}}$ edges have been embedded in the current while-loop iteration, since each edge embedding adds at most $\gamma_{\mathit{ExpPath}}$ units to the total congestion. 
We conclude that each iteration, we embed at least a $\frac{1}{16}$-fraction of the edges in $J_i[X]$ that where not embedded before the current while-loop iteration. It follows that there are at most  $O(\log(m))$ iterations, which establishes our claim.
\end{proof}

\begin{claim}\label{clm:staticEmbeddingCongestion}
$\vcong(\Pi_{J \to H'} \circ \Pi_{J \to \tilde{J}}) \leq \gamma_c \cdot \left(\vcong(\Pi_{J \to H'}) + \Delta_{\max}(J)\right)$ with probability at least $1 - n^{-2c}$.
\end{claim}
\begin{proof}
Let us fix any vertex $v \in V(H')$. We define $E_{v} = \{ e \in E(J) \;|\; v \in \Pi_{J \to H'}(e) \}$ to be the edges in $J$ whose embedding path contains $v$. By definition of vertex congestion for embeddings, $|E_v| \leq \vcong(\Pi_{J \to H'})$.

Next, for each edge $e \in E_v$, let $e$ be in $J_i$ after the decomposition in \Cref{lne:decompFirstLine} in the component $X$, we define the random variable
\[
    Y_{e} = \begin{cases} 
        \gamma_{X,i} & \text{if } e \in E(\tilde{J})\\
        0 & \text{otherwise}
    \end{cases}
\]
Note that the random variables $Y_e$ are independent as edges are sampled independently at random into $\tilde{J}$ in \Cref{lne:edgesAreSampledIid}. Further, by \Cref{clm:edgeCongestion}, we have for each edge $e$, $\econg(\Pi_{J \rightarrow \tilde{J}}|_{ E(J_i[X])}, e) \leq Y_e$ and thus $\vcong(\Pi_{J \to H'} \circ \Pi_{J \to \tilde{J}}, v) \leq \sum_{e \in E_v} Y_{e}$.

We will bound this sum using a Chernoff bound. We first observe that every variable $Y_e \in [0,W]$ for $W = \gamma_{\mathit{var}}\gamma_{\mathit{ExpPath}} \Delta_{\max}(J)$ for some scalar $\gamma_{\mathit{var}}$ which follows from the definitions of $\gamma_{X,i}$ in \Cref{clm:edgeCongestion} and $p_{X,i}$ in \Cref{lne:pXiDefinition}.
Across all the edge congestion variables $Y_e$, we have a uniform bound $\mu_{\text{edge}} $ on the expectation given by $\mathbb{E}[Y_{e}] = p_{X,i} \cdot \gamma_{X,i} = \mu_{\text{edge}}$.
Therefore $\mathbb{E}[\sum_{e \in E_v} Y_e] \leq \mu_{\text{edge}} \cdot \vcong(\Pi_{J \to H'})$. We can conclude by \Cref{thm:chernoffBound} that
\[
\mathbb{P}\left[\sum_{e \in E_v} Y_e \leq 24C \log(n) \cdot W + 2\mu_{\text{edge}} \vcong(\Pi_{J \to H'}) \right] > 1 - 2n^{-4C}. 
\]
We can now set $\gamma_c = 24C \log(n) \cdot \gamma_{\mathit{ExpPath}} \cdot \gamma_{\mathit{var}} + 2\mu_{\text{edge}} = (\gamma_{\mathit{ExpPath}})^{O(1)}$ which is consistent with our requirements on $\gamma_c$. Finally, it remains to take a simple union bound over all vertices $v \in V(H')$. 
\end{proof}

\begin{claim}\label{clm:lengthsEmbeddingJ}
$\length(\Pi_{J \to H'} \circ \Pi_{J \to\tilde{J}}) \leq \gamma_l \cdot \length(\Pi_{J \to H'})$. 
\end{claim}
\begin{proof}
Consider any edge $e \in E(J)$. If $e$ is sampled into $\tilde{J}$, then $\Pi_{J \to\tilde{J}}(e) = e$, as can be seen in \Cref{lne:embedByItselfIfSampled}. Otherwise, $\Pi_{J \to\tilde{J}}(e)$ is of length at most $\gamma_{\mathit{ExpPath}}$ as can be seen from  \Cref{lne:computeMCFlowGetPath1} and \Cref{thm:apspDataStructure}. Setting $\gamma_l = \gamma_{\mathit{ExpPath}}$ thus ensures our claim.
\end{proof}

Combining \Cref{clm:staticEmbeddingCongestion} and \Cref{clm:lengthsEmbeddingJ}, we have established the properties claimed in \Cref{lma:staticEmbed}. The success probability follows by taking a straight-forward union bound over the events used in the analysis above and the success of \Cref{thm:expanderStatement}. The run-time of the algorithm can be seen from inspecting \Cref{alg:sparsify}, \Cref{thm:expanderStatement}, the fact that the while-loop in \Cref{lne:outerWhileStaticEmbed} runs at most $O(\log(n))$ times for each iteration of the outer foreach-loop (established in the proof of \Cref{clm:edgeCongestion}) and finally the run-time guarantees on the data structure in \Cref{thm:apspDataStructure}.

%% file: jtree.tex
\section{Data Structure Chain}
\label{sec:jtree}

The goal of \cref{sec:jtree,sec:routing,sec:rebuilding} is to build a data structure to dynamically maintain $m^{o(1)}$-approximate undirected minimum-ratio cycles under changing costs and lengths, i.e. for gradients $\bg \in \R^E$ and lengths $\bell \in \R_{>0}^E$ return a (compactly represented) cycle $\bDelta$ satisfying $\mB^\top\bDelta=0$ and
\begin{align} \frac{\l \bg, \bDelta \r}{\|\mL\bDelta\|_1} \le m^{-o(1)} \min_{\mB^\top\bf=0} \frac{\l \bg, \bf \r}{\|\mL\bf\|_1}. \label{eq:mmc2} \end{align}
Our data structure does not work against fully adaptive adversaries. However, it works for updates coming from the IPM. We capture this notion with the following definition.

\begin{restatable}[Hidden Stable-Flow Chasing Updates]{definition}{defHiddenStableFlowChasing}
  \label{def:hiddenStableFlowChasing}
  Consider a dynamic graph $G^{(t)}$ undergoing batches of updates $U^{(1)}, \dots, U^{(t)}, \dots$ consisting of edge insertions/deletions and vertex splits.
  We say the sequences $\bg^{(t)}, \bell^{(t)},$ and $U^{(t)}$ satisfy the \emph{hidden stable-flow chasing} property if there are hidden dynamic circulations $\bc^{(t)}$ and hidden dynamic upper bounds $\bw^{(t)}$ such that the following holds at all stages $t$:
  \begin{enumerate}
  \item $\bc^{(t)}$ is a circulation: $\mB_{G^{(t)}}^\top \bc^{(t)} = 0.$ \label{item:circulation}
  \item
  $\bw^{(t)}$ upper bounds the length of $\bc^{(t)}$: $|\bell^{(t)}_e\bc^{(t)}_e| \le \bw^{(t)}_e$ for all $e \in E(G^{(t)})$.
  \label{item:width}
  
  \item
  For any edge $e$ in the current graph $G^{(t)}$, and any stage $t' \leq t$, if the edge $e$ was already present in $G^{(t')}$, i.e. $e \in G^{(t)} \setminus \bigcup_{s=t' + 1}^{t} U^{(s)}$, then $\bw^{(t)}_e \le 2\bw^{(t')}_e$.
  
  \label{item:widthstable}
  \item \label{item:quasipoly} Each entry of $\bw^{(t)}$ and $\bell^{(t)}$ is quasipolynomially lower and upper-bounded:
    \[ \log \bw^{(t)}_e \in [-\log^{O(1)} m, \log^{O(1)} m] \text{   and   } \log \bell^{(t)}_e \in [-\log^{O(1)} m, \log^{O(1)} m] \forall e \in E(G^{(t)}). \]
  \end{enumerate}
\end{restatable}

Intuitively \cref{def:hiddenStableFlowChasing} says that even while $\bg^{(t)}$ and $\bell^{(t)}$ change, there is a witness circulation $\bc^{(t)}$ that is fairly stable.
More precisely, there is some upper bound $\bw^{(t)}$ on the coordinate-wise lengths of $\bc^{(t)}$ that increases by at most a factor of $2$, except on edges that are explicitly updated.
Interestingly, even though both $\bc^{(t)}$ and $\bw^{(t)}$ are hidden from the data structure, their existence is sufficient.

The IPM guarantees in \cref{sec:ipm} can be connected to \cref{def:hiddenStableFlowChasing} by setting $\bc^{(t)} = \bf^* - \bf^{(t)}$ and $\bw^{(t)} = 10 + |\bell(\bf^{(t)}) \circ \bc^{(t)}|$, where $\bf^{(t)}$ is the current flow maintained by our algorithm. The guarantees of \cref{def:hiddenStableFlowChasing} then hold by a combination of \cref{lemma:optCirculationQuality,lemma:stablebell,lemma:stablebg}. This is formalized in \cref{lemma:setuphidden} in \cref{sec:combine}.

Our main data structure dynamically maintains min-ratio cycles under hidden stable-flow chasing updates.
\begin{theorem}[Dynamic Min-Ratio Cycle with Hidden Stable-Flow Chasing Updates]
  \label{thm:MMCHiddenStableFlow}
  There is a data structure that on a dynamic graph $G^{(t)}$ maintains a collection of $s = O(\log n)^d$ spanning trees $T_1, T_2, \dots, T_s \subseteq G^{(t)}$ for $d = O(\log^{1/8}m)$, and supports the following operations:
  \begin{itemize}
  \item $\textsc{Update}(U^{(t)}, \bg^{(t)}, \bell^{(t)}):$ Update the gradients and lengths to $\bg^{(t)}$ and $\bell^{(t)}$.
  For the update to be supported, we require that $U^{(t)}$ contains only edge insertions/deletions and $\bg^{(t)}, \bell^{(t)}$ and $U^{(t)}$ satisfy the hidden stable-flow chasing property (\cref{def:hiddenStableFlowChasing}) with hidden circulation $\bc^{(t)}$ and upper bounds $\bw^{(t)}$, and for a parameter $\alpha$,
    \[ \frac{\langle \bg^{(t)}, \bc^{(t)}\rangle}{\|\bw^{(t)}\|_1} \le -\alpha. \]
  \item $\textsc{Query}()$: Return a tree $T_i$ for $i \in [s]$ and a cycle $\bDelta$ represented as $m^{o(1)}$ paths on $T_i$ (specified by their endpoints and the tree index) and $m^{o(1)}$ explicitly given off-tree edges such that for $\kappa = \exp(-O(\log^{7/8}m \cdot \log\log m))$,
    \[ \frac{\langle \bg^{(t)}, \bDelta\rangle}{\|\mL^{(t)}\bDelta\|_1} \le -\kappa\alpha. \]
  \end{itemize}
  Over $\tau$ stages the algorithm succeeds whp. with total runtime $m^{o(1)}(m+Q)$ for $Q = \sum_{t \in [\tau]} |U^{(t)}|$.
\end{theorem}

To interpret \cref{thm:MMCHiddenStableFlow}, note that $\bDelta = \bc^{(t)}$ would be a valid output by the guarantees in \cref{def:hiddenStableFlowChasing}, i.e. $\norm{\mL^{(t)} \bDelta}_1 \le \norm{\bw^{(t)}}_1$ from \cref{item:width}. Thus the data structure guarantee can be interpreted as efficiently representing and returning a cycle whose quality is within a $m^{o(1)}$ factor of $\bc^{(t)}$. Eventually, we will add $\bDelta$ to our flow efficiently by using link-cut trees.

\cref{sec:jtree} focuses on introducing the general layout of the data structure, and is definition-heavy. \cref{sec:routing} explains how to plug in the circulations $\bc$ and upper bounds $\bw$ in our data structure, and shows a weaker version of \cref{thm:MMCHiddenStableFlow} in \cref{thm:norebuild}. We use the weaker \cref{thm:norebuild} to show the full cycle-finding result \cref{thm:MMCHiddenStableFlow} by defining a rebuilding game in \cref{sec:rebuilding}.

\subsection{Dynamic Low-Stretch Decompositions (LSD)}
\label{sec:lsd}
\begin{table}[!ht]
  \centering
  \begin{tabular}{|c|c|}
    \hline
    \textbf{Variable} & \textbf{Definition} \\
    \hline
    $\bell^{(t)}, \bg^{(t)}$ & Lengths and gradients on a dynamic graph $G^{(t)}$ \\
    \hline
    $\bc^{(t)}, \bw^{(t)}$ & Hidden circulation \& upper bounds with $|\bell^{(t)} \circ \bc^{(t)}| \le \bw^{(t)}$ (\cref{def:hiddenStableFlowChasing})  \\
    \hline
    $F$ & Rooted spanning forest of $G$ (\cref{def:spanningforest}). \\
    \hline
    $\bp(F[u, v])$ & Path vector from $u \to v$ in a forest $F$ \\
    \hline
    $\str^{F,\bell}_e$ & Stretch of edge $e$ with respect to spanning forest $F$ and lengths $\bell$ (\cref{def:stretchf}) \\
    \hline
    $\wstr_e$ & Stretch overestimates stable under edge deletions (\cref{lemma:globalstretch}) \\
    \hline
    $\mathcal{C}(G, F)$ & Core graph from a spanning forest $F$ (\cref{def:coregraph}) \\
    \hline
    $\hat{e}$ & Image of edge $e \in E(G)$ into the core graph $\mathcal{C}(G, F)$ \\
    \hline
    $\SS(G, F)$ & Sparsified core graph $\SS(G, F) \subseteq \mathcal{C}(G, F)$ (\cref{def:sparsecore}) \\
    \hline
    $\mathcal{G}_0,\dots,\mathcal{G}_d$ & $B$-branching tree chain (\cref{def:chain}) \\
    \hline
    $G_0,\dots,G_d$ & Tree chain (\cref{def:chain}) \\
    \hline
    $T^{G_0,\dots,G_d}$ & Tree in $G$ corresponding to tree chain $G_0,\dots,G_d$ (\cref{def:treesfromchain}) \\
    \hline
    $\mathcal{T}^G$ & Collection of $B^d$ trees on $G$ from $B$-branching tree chain (\cref{def:treesfromchain}) \\
    \hline
    $\prev^{(t)}_i$ & Previous rebuild times of branching tree chain (\cref{def:rebuildtimes}) \\
    \hline
  \end{tabular}
  \caption[Important definitions and notation to describe the data structure.]{Important definitions and notation to describe the data structure. In general a $(t)$ superscript is the corresponding object at time $t$ of a sequence of updates.}
  \label{tab:glossaryDataStructureChain}
\end{table}
In the following subsections we describe the components of the data structure we maintain to show \cref{thm:MMCHiddenStableFlow}.
At a high level, our data structure consists of $d$ levels, each of which has approximately a factor of $k = m^{1/d}$ fewer edges than the previous level.
The edge reduction is achieved in two parts.
First, we reduce the number of vertices to $\O(m/k)$ by maintaining a spanning forest $F$ of $G$ with $\O(m/k)$ connected components, and then recurse on $G/F$, the graph where each connected component of $F$ in $G$ is contracted to a single vertex.
While $G/F$ now has $\O(m/k)$ vertices, it still potentially has up to $m$ edges, so we need to employ the dynamic sparsification procedure in \cref{thm:spanner} to reduce the number of edges to $\widehat{O}(m/k)$.

We start by defining a rooted spanning forest and its induced stretch.
\begin{definition}[Rooted Spanning Forest]
  \label{def:spanningforest}
  A \emph{rooted spanning forest} of a graph $G = (V, E)$ is a forest $F$ on $V$ such that each connected component of $F$ has a unique distinguished vertex known as the \emph{root}. We denote the root of the connected component of a vertex $v \in V$ as $\root^F_v$.
\end{definition}
\begin{definition}[Stretches of $F$]
  \label{def:stretchf}
  Given a rooted spanning forest $F$ of a graph $G = (V, E)$ with lengths $\bell \in \R_{>0}^E$, the stretch of an edge $e = (u, v) \in E$ is given by
  \begin{align*}
    \str^{F,\bell}_e \defeq
    \begin{cases}
      1 + \left\langle \bell, |\bp(F[u,v])|\right\rangle/\bell_e &~\text{ if } \root^F_u = \root^F_v \\
      1 + \left\langle \bell, |\bp(F[u, \root^F_u])| + |\bp(F[v, \root^F_v])| \right\rangle/\bell_e &~\text{ if } \root^F_u \neq \root^F_v,
    \end{cases}
  \end{align*}
  where $\bp(F[\cdot,\cdot]),$ as defined in \cref{sec:prelim}, maps a path to its signed indicator vector.
\end{definition}
When $F$ is a spanning tree \cref{def:stretchf} coincides with the definition of stretch for a LSST.

The goal of the remainder of this section is to give an algorithm to maintain a \emph{Low Stretch Decomposition (LSD)} of a dynamic graph $G.$
As a spanning forest decomposes a graph into vertex disjoint connected subgraphs, a LSD consists of a spanning forest $F$ of low stretch.
The algorithm produces stretch upper bounds that hold throughout all operations, and the number of connected components of $F$ grows by amortized $\O(1)$ per update.
At a high level, for any edge insertion or deletion, the algorithm will force both endpoints to become roots of some component of $F$.
This way, any inserted edge will actually have stretch $1$ because both endpoints are roots.

\begin{lemma}[Dynamic Low Stretch Decomposition]
  \label{lemma:globalstretch}
  There is a deterministic algorithm with total runtime $\O(m)$ that on a graph $G = (V, E)$ with lengths $\bell \in \R^E_{>0}$, weights $\bv \in \R^E_{>0}$, and parameter $k,$ initializes a tree $T$ spanning $V$, and a rooted spanning forest $F \subseteq T$, a edge-disjoint partition $\cW$ of $F$ into $O(m/k)$ sub trees and stretch overestimates $\wstr_e$.
  The algorithm maintains $F$ decrementally against $\tau$ batches of updates to $G$, say $U^{(1)}, U^{(2)}, \dots, U^{(\tau)}$, such that $\wstr_e \defeq 1$ for any new edge $e$ added by either edge insertions or vertex splits, and:
  \begin{enumerate}
  \item $F$ has initially $O(m/k)$ connected components and $O(q \log^2 n)$ more after $t$ update batches of total encoding size $q \defeq \sum_{t=1}^{\tau} \Enc(U^{(i)})$ satisfying $q \le \O(m)$.  \label{item:cccount}
  \item $\str^{F,\bell}_e \le \wstr_e \le O(k \gamma_{LSST}\log^4 n)$ for all $e \in E$ at all times, including inserted edges $e$. \label{item:stretchbound}
  \item $\sum_{e \in E^{(0)}} \bv_e \wstr_e \le O(\|\bv\|_1 \gamma_{LSST}\log^2 n)$, where $E^{(0)}$ is the initial edge set of $G$. \label{item:avgstretchbound}
  \item
  Initially, $\cW$ contains $O(m/k)$ subtrees.
  For any piece $W \in \cW, W \subseteq V$, $\Abs{\partial W} \le 1$ and $\vol_G(W \setminus R) \le O(k\log^2n)$ at all times, where $R \supseteq \partial \cW$ is the set of roots in $F$.
  Here, $\partial W$ denotes the set of \emph{boundary vertices} that are in multiple partition pieces.
    \label{item:degbound}
  \end{enumerate}
\end{lemma}

Intuitively, the first property says that $F$ has $O(m/k)$ roots initially and each update $x$ adds $\O(\Enc(x))$ roots to it.
For example, each edge update adds $\O(1)$ roots to $F.$
This allows us to satisfy the second property, which is that the stretch of $e$ with respect to $F$ (\cref{def:stretchf}) is upper bounded by some global upper bound $\wstr_e$.
Note that $\wstr_e$ stays the same for any edge $e$ across the execution of the algorithm.
The third property says that these global upper bounds are still good on average with respect to the weights $\bv$ up to $\O(1)$ factors.
The final property is useful for interacting with our sparsifier in \cref{thm:spanner} whose runtime and congestion depend on the maximum degree of the input graphs.

We defer the proof of \cref{lemma:globalstretch} to \cref{app:globalstretch}.

\subsection{Worst-Case Average Stretch via Multiplicative Weights}
\label{sec:mwu}

By doing a multiplicative weights update procedure (MWU) on top of \cref{lemma:globalstretch}, we can build a distribution over partial spanning tree routings whose average stretch on every edge is $\O(1)$. This is very similar to MWUs done in works of \cite{R08,KLOS14} for building $\ell_\infty$ oblivious routings, and cut approximators~\cite{M10, S13}.

\begin{lemma}[MWU]
  \label{lemma:strMWU}
  There is a deterministic algorithm that on a graph $G = (V, E)$ with lengths $\bell$ and a positive integer $k$ computes $t$ spanning trees, rooted spanning forests, and stretch overestimates $\{(T_i, F_i \subseteq T_i, \wstr^i_e)\}_{i=1}^t$ (\cref{lemma:globalstretch}) for some $t = \O(k)$ such that
  \begin{align}
    \label{eq:strMWU}
    \sum_{i=1}^t \blambda_i \wstr^{i}_e \le O(\gamma_{LSST}\log^2 n) \forall e \in E,
  \end{align}
  where $\blambda \in \R_{>0}^{[t]}$ is the uniform distribution over the set $[t]$, i.e. $\blambda = \vec{1} / t.$
  
  The algorithm runs in $\O(mk)$-time.
\end{lemma}
The proof is standard and deferred to \cref{app:mwu}.

If we sample a single tree/index from the distribution $\blambda$, then any fixed flow will be stretched by $O(\gamma_{LSST}\log^2 n)$ on average.
Hence any fixed flow will be stretched by $O(\gamma_{LSST}\log^2 n)$ by at least one of $O(\log n)$ trees sampled from $\blambda$ with high probability.
We will leverage this fact to analyze how the witness circulation $\bc^{(t)}$ in \cref{def:hiddenStableFlowChasing} and \cref{thm:MMCHiddenStableFlow} is stretched by a random forest.

\subsection{Sparsified Core Graphs and Path Embeddings}
\label{sec:core}

Given a rooted spanning forest $F$, we will recursively process the graph $G/F$ where each connected component of $F$ is contracted to a single vertex represented by the root. We call this the \emph{core graph}, and define the lengths and gradients on it as follows. Below, we should think of $G$ as the result of edge insertions/deletions to an earlier graph $G^{(0)}$, so $\wstr_e = 1$ for edge inserted to get from $G^{(0)}$ to $G$, as enforced in \cref{lemma:globalstretch}.

\begin{definition}[Core graph]
  \label{def:coregraph}
  Consider a tree $T$ and a rooted spanning forest $E(F) \subseteq E(T)$ on a graph $G$ equipped with stretch overestimates $\wstr_e$ satisfying the guarantees of \cref{lemma:globalstretch}. We define the \emph{core graph} $\mathcal{C}(G, F)$ as a graph with the same edge and vertex set as $G/F$. For $e = (u, v) \in E(G)$ with image $\hat{e} \in E(G/F)$ we define its length as $\bell^{\mathcal{C}(G,F)}_{\hat{e}} \defeq \wstr_e\bell_e$ and gradient as $\bg^{\mathcal{C}(G,F)}_{\hat{e}} \defeq \bg_e + \langle \bg, \bp(T[v, u]) \rangle$.
\end{definition}
\begin{rem}
  \label{rem:coregraph}
  In our usage, we maintain $\cC(G, F)$ where $G$ is a dynamic graph and $F$ is a decremental rooted spanning forest.
  In particular, $T$, $F$, and $\wstr$ are initialized and maintained via \cref{lemma:globalstretch}.
  As $G$ undergoes dynamic updates such as edge deletions or vertex splits which adds new vertices to $G$, $T$ won't be a spanning tree of $G$ anymore.
  \cref{def:coregraph} responds to such situation by allowing $T$ not being a spanning tree nor a subgraph of $G.$
  
  Thus, for $e = (u, v) \in E(G)$, $u$ and $v$ may not be connected in $T$.
  In this case, the value of $\bg^{\mathcal{C}(G,F)}_{\hat{e}}$ is simply $\bg_e$.
  Also, the support of the gradient vector $\bg$ is $E(G) \cup E(T).$
  This corresponds to the case when some edge in $T$ is removed from $G$, we keep the gradient on that edge as it is.
\end{rem}

Note that the length and gradient of the image of an edge $e \in E(G)$ in \cref{def:coregraph} do not change under edge deletions to $F$, because they are defined with respect to the tree $T$. This important property will be useful later in efficiently maintaining a sparsifier of the core graph, which we require to reduce the number of edges in the sparsified core graph to $\widehat{O}(m/k)$.

\begin{definition}[Sparsified core graph]
  \label{def:sparsecore}
  Given a graph $G$, forest $F$, and parameter $k$, define a $(\gamma_s, \gamma_c, \gamma_{\ell})$-\emph{sparsified core graph with embedding} as a subgraph $\SS(G, F) \subseteq \mathcal{C}(G, F)$ and embedding $\Pi_{\mathcal{C}(G, F) \to \SS(G, F)}$ satisfying
  \begin{enumerate}
  \item For any $\hat{e} \in E(\mathcal{C}(G, F))$, all edges $\hat{e}' \in \Pi_{\mathcal{C}(G, F) \to \SS(G, F)}(\hat{e})$ satisfy $\bell_{\hat{e}}^{\mathcal{C}(G,F)} \approx_2 \bell_{\hat{e}'}^{\mathcal{C}(G,F)}$. \label{item:samelen}
  \item $\length(\Pi_{\mathcal{C}(G, F) \to \SS(G, F)}) \le \gamma_l$ and $\econg(\Pi_{\mathcal{C}(G, F) \to \SS(G, F)}) \le k\gamma_c$.
  \item $\SS(G, F)$ has at most $m\gamma_s/k$ edges.
  \item The lengths and gradients of edges in $\SS(G, F)$ are the same as in $\mathcal{C}(G, F)$ (\cref{def:coregraph}).
  \end{enumerate}
\end{definition}
In \cref{sec:routing} we give a dynamic algorithm for maintaining a sparsified core graph of a graph $G$ undergoing edge insertions and deletions. We defer the formal statement to \cref{lemma:hintedsparsecore} where we not only maintain a sparsified core graph but also show that the witness circulation $\bc^{(t)}$ and upper bounds $\bw^{(t)}$ from \cref{def:hiddenStableFlowChasing} are preserved approximately.

\subsection{Full Data Structure Chain}
\label{sec:chain}

Our data structure has $d$ levels. The graphs at the $i$-th level have about $m/k^i$ edges, and each such graph branches into $O(\log n)$ graphs sampled from the distribution $\blambda$ from \cref{lemma:strMWU}.
\begin{definition}[Branching Tree-Chain]
  \label{def:chain}
  For a graph $G$, parameter $k$, and branching factor $B$, a $B$-\emph{branching tree-chain} consists of collections of graphs $\{\cG_i\}_{0 \le i \le d}$,  such that $\cG_0 \defeq \{G\}$, and we define $\cG_i$ inductively as follows,
  \begin{enumerate}
  \item For each $G_i \in \cG_i$,  $i < d,$  we have a collection of $B$ trees $\cT^{G_i} = \left\{T_1, T_2, \dots, T_B\right\}$ and a collection of $B$ forests $\cF^{G_i} = \left\{F_1, F_2, \dots, F_B \right\}$ such that $E(F_j) \subseteq E(T_j)$ satisfy the conditions of \cref{lemma:globalstretch}.
  \item For each $G_i \in \cG_i$, and $F \in \mathcal{F}^{G_i},$ 
  we maintain $(\gamma_s, \gamma_c, \gamma_l)$-sparsified core graphs and embeddings $\SS(G_i, F)$ and $\Pi_{\mathcal{C}(G_i, F) \to \SS(G_i, F)}$.
  \item We let $\cG_{i+1} \defeq \{\SS(G_i, F) : G_i \in \cG_i, F \in F^{G_i}\}$.
  \end{enumerate}
  Finally, for each $G_d \in \cG_d,$ we maintain a low-stretch tree $F$.

  We let a \emph{tree-chain} be a single sequence of graphs $G_0, G_1, \dots, G_d$ such that $G_{i+1}$ is the $(\gamma_s, \gamma_c, \gamma_l)$-sparsified core graph $\SS(G_i, F_i)$ with embedding $\Pi_{\mathcal{C}(G_i, F_i) \to \SS(G_i, F_i)}$ for some $F_i \in \mathcal{F}^{G_i}$ for $0 \le i < d$, and a low-stretch tree $F_d$ on $G_d$.
\end{definition}
In general, we will have $B = O(\log n)$ throughout, and we will omit $B$ when discussing branching tree-chains. Note that level $i$ of a branching tree-chain, i.e. the collection of graphs in $\mathcal{G}_i$, has at most $B^i = O(\log n)^i$ graphs for $B = O(\log n)$. A branching tree-chain can alternatively be viewed as a set of $O(\log n)^d$ tree-chains, each of which naturally corresponds to a spanning tree of the top level graph $G$.
\begin{definition}[Trees from Tree-Chains]
  \label{def:treesfromchain}
  Given a graph $G$ and tree-chain $G_0, G_1, \dots, G_d$ where $G_0 = G$, define the corresponding spanning tree $T^{G_0, G_1, \dots, G_d} \defeq \bigcup_{i=0}^d F_i$ of $G$ as the union of preimages of edges of $F_i$ in $G = G_0$.

  Define the set of trees corresponding to a branching tree-chain of graph $G$ as the union of $T^{G_0, G_1, \dots, G_d}$ over all tree-chains $G_0, G_1, \dots, G_d$ where $G_0 = G$: \[ \mathcal{T}^G \defeq \{ T^{G_0, G_1, \dots, G_d} : G_0, G_1, \dots, G_d \text{ s.t. } G_{i+1} = \SS(G_i, F_i) \forall 0 \le i < d \} \]
\end{definition}

We can dynamically maintain a branching tree-chain such that we rebuild $\mathcal{G}_{i+1}$ from $\mathcal{G}_i$ every approximately $m/k^i$ updates. Between rebuilds, the trees $\mathcal{T}^{G'}$ of graphs $G' \in \mathcal{G}_i$ stay the same, while the forests in $\mathcal{F}^{G'}$ are decremental as guaranteed in \cref{lemma:globalstretch}.
\begin{definition}[Previous Rebuild Times]
  \label{def:rebuildtimes}
  Given a dynamic graph $G^{(t)}$ with updates indexed by times $t = 0, 1, \dots$ and corresponding dynamic branching tree-chain (\cref{def:chain}), we say that nonnegative integers $\prev^{(t)}_0 \le \prev^{(t)}_1 \le \dots \le \prev^{(t)}_d = t$ are \emph{previous rebuild times} if $\prev^{(t)}_i$ was the most recent time at or before $t$ that $\mathcal{G}_i$ was rebuilt, i.e. for $G \in \mathcal{G}_i$ the set of trees $\mathcal{T}^G$ was reinitialized and sampled.
\end{definition}
We will assume that our algorithm rebuilds all $G_i \in \mathcal{G}_i$ at the same time: if we recompute a set of trees $\mathcal{T}^G$ for some $G_i \in \mathcal{G}_i$, then we also recompute the trees $\mathcal{T}^{G_i'}$ for all other $G_i' \in \mathcal{G}_i$. In the following \cref{sec:routing} we show \cref{thm:norebuild}, we gives a data structure whose guarantee is weaker than \cref{thm:MMCHiddenStableFlow}. Precisely, the quality of the cycle returned depends on the previous rebuild times. We later boost this to an algorithm for \cref{thm:MMCHiddenStableFlow} by solving a \emph{rebuilding game} in \cref{sec:rebuilding}.

\section{Routings and Cycle Quality Bounds}
\label{sec:routing}
The goal of this section is to explain how to route the witness circulations $\bc^{(t)}$ and length upper bounds $\bw^{(t)}$ through the branching tree-chain, and eventually recover an approximately optimal flow $\bDelta$. The main theorem we show in this section is the following.
\begin{theorem}
  \label{thm:norebuild}
  Let $G = (V, E)$ be a dynamic graph undergoing $\tau$ batches of updates $U^{(1)}, \dots, U^{(\tau)}$ containing \emph{only} edge insertions/deletions with edge gradient $\bg^{(t)}$ and length $\bell^{(t)}$ such that the update sequence satisfies the hidden stable-flow chasing property (\cref{def:hiddenStableFlowChasing}) with hidden dynamic circulation $\bc^{(t)}$ and width $\bw^{(t)}.$
  There is an algorithm on $G$ that maintains a $O(\log n)$-branching tree chain corresponding to $s = O(\log n)^d$ trees $T_1, T_2, \dots, T_s$ (\cref{def:treesfromchain}), and at stage $t$ outputs a circulation $\bDelta$ represented by $\exp(O(\log^{7/8}m\log\log m))$ off-tree edges and paths on some $T_i, i \in [s]$.
  
  The output circulation $\bDelta$ satisfies $\mB^\top \bDelta = 0$ and for some $\kappa = \exp(-O(\log^{7/8}m\log\log m))$
  \begin{align*}
      \frac{\l\bg^{(t)}, \bDelta\r}{\norm{\bell^{(t)} \circ \bDelta}_1} \le \kappa \frac{\l\bg^{(t)}, \bc^{(t)}\r}{\sum_{i=0}^d \|\bw^{(\prev_i^{(t)})}\|_1},
  \end{align*}
  where $\prev_i^{(t)}, i \in [d]$ are the previous rebuild times (\cref{def:rebuildtimes}) for the branching tree chain.
  
  The algorithm succeeds w.h.p. with total runtime $(m + Q)m^{o(1)}$ for $Q \defeq \sum_{t=1}^{\tau} \Abs{U^{(t)}} \le \poly(n)$.
  Also, levels $i, i+1, \dots, d$ of the branching tree chain can be rebuilt at any point in $m^{1+o(1)} / k^i$ time.
  
\end{theorem}
The final sentence about rebuilding levels $i, i+1, \dots, d$ allows us to force $\prev^{(t)}_i = \prev^{(t)}_{i+1} = \dots = \prev^{(t)}_d = t$. This is necessary because it is possible that $\|\bw^{(\prev^{(t)}_i)}\|_1$ is much larger than $\|\bw^{(t)}\|_1$ for some $0 \le i \le d$.
This could result in $\frac{\l \bg^{(t)}, \bc^{(t)} \r}{\sum_{i=0}^d \|\bw^{(\prev^{(t)}_i)}\|_1}$ being much more than $\frac{\l \bg^{(t)}, \bc^{(t)} \r}{\|\bw^{(t)}\|_1}$.
This is not sufficient to show \cref{def:hiddenStableFlowChasing}, which only guarantees that the latter quality is at most $-\alpha$, but does not assume a bound on the former.
We will resolve this issue in \cref{sec:rebuilding} by carefully using our ability to rebuild levels $i, i+1, \dots, d$ periodically whenever the cycle $\bDelta$ returned by \cref{thm:norebuild} is not good enough, and we deduce that $\|\bw^{(\prev^{(t)}_i)}\|_1$ is much larger than  $\|\bw^{(t)}\|_1$ for some $0 \le i \le d$.

\subsection{Passing Circulations and Length Upper Bounds Through a Tree-Chain}
\label{sec:passtreechain}
Towards proving \cref{thm:norebuild} we define how to pass the witness circulation $\bc$ and length upper bounds $\bw$ downwards in a tree-chain. It is convenient to define a \emph{valid pair} of $\bc,\bw$ with respect to a graph $G$ with lengths $\bell$. Essentially, this means that $\bc$ is indeed a circulation and $\bw$ are valid length upper bounds, i.e. items \ref{item:circulation} and \ref{item:width} of the hidden stable-flow chasing property \cref{def:hiddenStableFlowChasing}.
\begin{definition}[Valid pair]
  \label{def:validpair}
  For a graph $G = (V, E)$ with lengths $\bell \in \R^E_{>0}$, we say that $\bc, \bw \in \R^E$ are a \emph{valid pair} if $\bc$ is a circulation and $|\bell_e\bc_e| \le \bw_e$ for all $e \in E$.
\end{definition}

\subsubsection{Passing Circulations and Length Upper Bounds to the Core Graph}
We first describe how to pass $\bc, \bw$ from $G$ to a core graph $\cC(G, F)$ (\cref{def:coregraph}).
\begin{definition}[Passing $\bc, \bw$ to core graph]
  \label{def:passcore}
  Given a graph $G = (V, E)$ with a tree $T$, arooted spanning forest $E(F) \subseteq E(T)$, and a stretch overestimates $\wstr_e$ as in \cref{lemma:globalstretch}, circulation $\bc \in \R^E$ and length upper bounds $\bw \in \R^E_{>0}$, we define vectors $\bc^{\cC(G, F)} \in \R^{E(\cC(G, F))}$ and $\bw^{\cC(G, F)} \in \R^{E(\cC(G, F))}_{>0}$ as follows. For $\hat{e} \in E(\cC(G, F))$ with preimage $e \in E$, define $\bc^{\cC(G, F)}_{\hat{e}} \defeq \bc_e$ and $\bw^{\cC(G, F)}_{\hat{e}} \defeq \wstr_e\bw_e$.
\end{definition}
We verify that $\bc^{\cC(G, F)}$ is a circulation on $\cC(G, F)$ and that $\bw^{\cC(G, F)}$ are length upper bounds.
\begin{lemma}[Validity of \cref{def:passcore}]
  \label{lemma:passcore}
  Let $\bc,\bw$ be a valid pair (\cref{def:validpair}) on a graph $G$ with lengths $\bell$.
  As defined in \cref{def:passcore},
  $\bc^{\cC(G, F)}, \bw^{\cC(G, F)}$ are a valid pair on $\cC(G, F)$ with lengths $\bell^{\cC(G,F)}$ (\cref{def:coregraph}),
  and
  \[ \left\|\bw^{\cC(G, F)}\right\|_1 \le \sum_{e \in E(G)} \wstr_e \bw_e. \]
\end{lemma}
\begin{proof}
  The proof is primarily checking the definitions.
  Recall that the edge set of $\cC(G,F)$ is $G/F$. Contracting vertices preserves circulations, hence $\bc^{\cC(G,F)}$ is a circulation (as $\bc$ is).

  In \cref{def:coregraph} we define $\bell^{\cC(G, F)}_{\hat{e}} = \wstr_e\bell_e$. So
  \begin{align*}
  |\bell^{\cC(G, F)}_{\hat{e}} \bc^{\cC(G, F)}_{\hat{e}}| &= \wstr_e|\bell_e\bc_{\hat{e}}| = \wstr_e|\bell_e\bc_e| \le \wstr_e\bw_e = \bw^{\cC(G, F)}_{\hat{e}},
  \end{align*}
  where the inequality holds because $\bc,\bw$ are a valid pair.

  The bound on $\|\bw^{\cC(G,F)}\|_1$ follows trivially by definition. The reason for the inequality (instead of equality) is that some edges may be contracted and disappear.
\end{proof}
Finally we state an algorithm which takes hidden stable-flow chasing updates on a dynamic graph $G^{(t)}$ and produces a dynamic core graph.
Below, we let $\bc^{(t),\cC(G,F)},\bw^{(t),\cC(G,F)}$ denote the result of using \cref{def:passcore} for $\bc = \bc^{(t)}$ and $\bw = \bw^{(t)}$, and similar definitions for $\bg^{(t),\cC(G,F)},\bell^{(t),\cC(G,F)}$ used later in the section.
\begin{lemma}[Dynamic Core Graphs]
  \label{lemma:hintedcore}
  \cref{algo:dynacore} takes as input a parameter $k$, a dynamic graph $G^{(t)}$ undergoes $\tau$ batches of updates $U^{(1)}, \dots U^{(\tau)}$ with gradients $\bg^{(t)}$, and lengths $\bell^{(t)}$ at stage $t = 0, \ldots, \tau$ that satisfies $\sum_{t=1}^{\tau} \Enc(U^{(t)}) \le m/(k\log^2 n)$ and the hidden stable-flow chasing property with the hidden circulation $\bc^{(t)}$ and width $\bw^{(t)}$.
  
  For each $j \in [B]$ with $B = O(\log n)$, the algorithm maintains a static tree $T_j$, a decremental rooted forest $F^{(t)}_j$ with $O(m/k)$ components satisfying the conditions of \cref{lemma:globalstretch}, and a core graph $\cC(G^{(t)}, F^{(t)}_j)$:
  \begin{enumerate}
      \item \label{item:CGRec} \underline{Core Graphs have Bounded Recourse:}
      the algorithm outputs update batches $U^{(t)}_j$ that produce $\cC(G^{(t)}, F^{(t)}_j)$ from $\cC(G^{(t-1)}, F^{(t-1)}_j)$ such that $\sum_{t' \le t} |U^{(t')}_j| = O\left(\sum_{t' \le t} \Enc(U^{(t')}) \cdot \log^2 n\right).$
      
      \item \label{item:CGWidth} \underline{The Widths on the Core Graphs are Small:}
      whp. there is an $j^* \in [B]$ only depending on $\bw^{(0)}$ such that for $\bw^{(t),\cC(G^{(t)}, F^{(t)}_j)}$ as defined in \cref{def:passcore} for $E(F_j) \subseteq E(T_j)$, and all stages $t \in \{0, \dots, \tau\}$,
      \begin{align}
        \|\bw^{(t),\cC(G^{(t)}, F^{(t)}_{j^*})}\|_1
        \le O(\gamma_{LSST}\log^2 n)\|\bw^{(0)}\|_1 + \|\bw^{(t)}\|_1.
        \label{eq:corewbound}
      \end{align}
  \end{enumerate}
  
  The algorithm runs $\O(mk)$-time.
\end{lemma}
\begin{algorithm}[!ht]
  \caption{Dynamically maintains a core graph (\cref{def:coregraph}). Procedure $\Initialize$ initializes all variables, and $\DynamicCore$ takes updates to $G^{(t)}$. \label{algo:dynacore}}
  \SetKwProg{Globals}{global variables}{}{}
  \SetKwProg{Proc}{procedure}{}{}
  \Globals{}{
    $B \assign O(\log n)$: number of instantiations of \cref{lemma:globalstretch} \\
    $\cA^{(LSD)}_j$ for $j \in [B]$: algorithms implementing \cref{lemma:globalstretch} \\
    $T_j$ for $j \in [B]$: trees initialized by $\cA^{(LSD)}_j$ \\
    $\wstr^{j}$ for $j \in [B]$: stretch overestimates initialized by $\cA^{(LSD)}_j$.
  }
  \Proc{$\Initialize(G = G^{(0)}, \bell, k)$}{
    Let $\blambda$ and $\{T_1', \dots, T_t'\}$ be returned by \cref{lemma:strMWU} on $G$ with lengths $\bell$ and $t = \O(k)$ \\
    For $j \in [B]$ sample $i_j \in [t]$ proportional to $\blambda$, and $T_j \assign T_{i_j}'$ \\
    Initialize $\cA^{(LSD)}_j$ on $T_j$ for $j \in [B]$
  }
  \Proc{$\DynamicCore(G^{(t)}, U^{(t)}, \bg^{(t)}, \bell^{(t)})$}{
    \For{$j \in [B]$}{
      Pass $U^{(t)}$ to $\cA^{(LSD)}_j$ which updates $F^{(t-1)}_j$ to $F^{(t)}_j$. \\
      \tcp{All edges $e \in G^{(t)} \cap \left(\bigcup_{i\le t} U^{(i)}\right)$ have $\wstr^{j}_e = 1$ by \cref{lemma:globalstretch}}
      Let $U^{(t)}_j$ be the batch of vertex splits that updates $\cC(G^{(t-1)}, F^{(t-1)}_j)$ to $\cC(G^{(t-1)}, F^{(t)}_j)$. \\
      Append $U^{(t)}_j$ with $U^{(t)}$ which updates $\cC(G^{(t-1)}, F^{(t)}_j)$ to $\cC(G^{(t)}, F^{(t)}_j)$.
    }
  }
\end{algorithm}
\cref{algo:dynacore} initializes $B$ trees $T_1, \dots, T_B$ from the MWU distribution output by \cref{lemma:strMWU}. For each of these $B$ trees, we maintain a forest $E(F^{(t)}_j) \subseteq E(T_j)$ satisfying the conditions of \cref{lemma:globalstretch} with the goal of forcing the stretch of every newly appeared edge $e$ in $G^{(t)}$ to be $1$, i.e. $\wstr^{j}_e = 1 \forall j \in [B]$.

Given an update batch $U^{(t)}$, \cref{algo:dynacore} first updates the forest $F^{(t - 1)}_j$ to $F^{(t)}_j$ for any $j \in [B]$ using the algorithm of \cref{lemma:globalstretch}.
For any update $x \in U^{(t)}$, if $x$ updates some edge $e$, both endpoints are roots in the forest $F^{(t)}_j$ and they appear in the core graph.
If $x$ splits some vertex $u$, $u$ is made a root in the forest and it appears in the core graph as well.
In both cases, the update $x$ can be performed in the core graph $\cC(G^{(t-1)}, F^{(t)}_j)$ (notice that it is not $\cC(G^{(t)}, F^{(t)}_j)$).
Thus, we can apply the entire batch $U^{(t)}$ to produce $\cC(G^{(t)}, F^{(t)}_j)$ from $\cC(G^{(t-1)}, F^{(t)}_j)$.

When a vertex $u \in V(G^{(t-1)})$ is split, the algorithm of \cref{lemma:globalstretch} treats it as a sequence of one isolated vertex insertion and $O(\deg_{G^{(t)}}(u^{NEW}))=O(\Enc(x))$ edge insertions/deletions.
The newly added isolated vertex stays isolated in the forests $F_j, j \in [B]$ as they are maintained decrementally edge-wise.

We adapt the reduction when applying $U^{(t)}$ to produce $\cC(G^{(t)}, F^{(t)}_j)$ from $\cC(G^{(t-1)}, F^{(t)}_j)$.
Thus, the number of updates in the core graph is at least the \emph{total encoding size} of updates in the original graph.
As we will show, it is upper-bounded by the total encoding size as well.

\begin{proof}[Proof of \cref{lemma:hintedcore}]
Note that $\sum_{t \in [\tau]}\Enc(U^{(t)}) = m/(k\log^2 n)$, we can take $q = O(m/(k\log^2 n))$ in \cref{lemma:globalstretch}.
Thus by item \ref{item:cccount} of \cref{lemma:globalstretch}, $F_j^{(t)}$ has $O(m/k)$ connected components.

Next, we prove \cref{item:CGRec} that bounds the number of updates to the core graph.
After $t$ batches of updates $U^{(1)}, \dots, U^{(t)}$, $\cA^{(LSD)}_j$ increases the number of components in $F_j$ by $O\left(\sum_{t' \le t} \Enc(U^{(t')}) \cdot \log^2 n\right)$ according to \cref{item:cccount} of  \cref{lemma:globalstretch}.
Every new components appeared in $F_j$ splits a vertex in the core graph $\cC(G, F_j).$
Thus, there will be $O\left(\sum_{t' \le t} \Enc(U^{(t')}) \cdot \log^2 n\right)$ vertex splits happened in the core graph.
After updating $F_j$, every update batch $U^{(t')}$ updates $\cC(G, F_j)$ as it updates $G.$
Thus, we can bound the number of updates to $\cC(G, F_j)$ up to first $t$ stages by $O\left(\sum_{t' \le t} \Enc(U^{(t')}) \cdot \log^2 n\right)$.

To show \cref{item:CGWidth}, by \cref{lemma:passcore} we first get that
  \begin{align}
    \left\|\bw^{(t),\cC\left(G^{(t)}, F^{(t)}_j\right)}\right\|_1
    &\le \sum_{e \in E(G^{(t)})} \wstr^{j}_e\bw^{(t)}_e
    \overset{(i)}{\le} \sum_{e \in G^{(t)} \setminus \left(\bigcup_{s=1}^t U^{(s)}\right)} \wstr^{j}_e\bw^{(t)}_e + \sum_{e \in G^{(t)} \cap \left(\bigcup_{s=1}^t U^{(s)}\right)} \bw^{(t)}_e \nonumber \\
    & \overset{(ii)}{\le} 2\sum_{e \in G^{(0)}}\wstr^{j}_e\bw^{(0)}_e + \left\|\bw^{(t)}\right\|_1, \label{eq:firstboundcore}
  \end{align}
  where $(i)$ follows because every edge $e$ appeared in $G^{(t)}$ due to some update in some $U^{(i)}$ has $\wstr^{j}_e = 1$ by a condition of \cref{lemma:globalstretch}.
  $(ii)$ follows because the hidden stable-flow chasing property (\cref{def:hiddenStableFlowChasing} item \ref{item:widthstable}) gives that any edge $e \in G^{(t)} \setminus \bigcup_{s=1}^t U^{(s)} \subseteq G^{(0)}$ has $\bw^{(t)}_e \le 2\bw^{(0)}_e$.

Now recall that $T_j$ is sampled from the collection $\{T'_1, \ldots, T'_t\}$ of trees given by \cref{lemma:strMWU}, with probabilities proportional to $\blambda$.
Hence
\[ \E_{T_j}\Big[\sum_{e \in G^{(0)}}\wstr^{j}_e\bw^{(0)}_e\Big] = \sum_{e \in G^{(0)}} \bw^{(0)}_e \sum_{i=j}^t \blambda_j\wstr^{j}_e \le O(\gamma_{LSST}\log^2 n)\|\bw^{(0)}\|_1, \]
by the guarantees of \cref{lemma:strMWU}, so by Markov's inequality
\[ \Pr_{T_j}\Big[\sum_{e \in G^{(0)}}\wstr^{j}_e\bw^{(0)}_e \le O(\gamma_{LSST}\log^2 n)\|\bw^{(0)}\|_1 \Big] \ge 1/2. \]
Since we sample $B$ independent trees $T,$ for $B = \Theta(\log n),$ we get that there exists an $i^*$ satisfying \eqref{eq:corewbound} with probability at least $1 - 2^{-B} \ge 1 - n^{-\Theta(1)}.$

  Finally, the algorithm runs in total time $\O(mk)$ by \cref{lemma:strMWU} and \cref{lemma:globalstretch}.
\end{proof}

\subsubsection{Passing Circulations and Length Upper Bounds to the Sparsified Core Graph}
We describe how to pass $\bc^{\mathcal{C}(G, F)}, \bw^{\mathcal{C}(G, F)}$ on a core graph to a sparsified core graph $\SS(G, F)$.
\begin{definition}[Passing $\bc,\bw$ to sparsified core graph]
  \label{def:passsparsecore}
  Consider a graph $G$ with spanning forest $F$, and circulation $\bc^{\mathcal{C}(G, F)} \in \R^{E(\mathcal{C}(G, F))}$ and upper bound $\bw^{\mathcal{C}(G, F)} \in \R^{E(\mathcal{C}(G, F))}_{>0}$, and embedding $\Pi_{\mathcal{C}(G, F) \to \SS(G, F)}$ for a $(\gamma_s,\gamma_c,\gamma_l)$-sparsified core graph $\SS(G, F) \subseteq \mathcal{C}(G, F)$. Define
  \begin{align} 
    \bc^{\SS(G, F)} &= \sum_{\hat{e} \in E(\mathcal{C}(G, F))} \bc^{\mathcal{C}(G, F)}_{\hat{e}}\bPi_{\mathcal{C}(G, F) \to \SS(G, F)}(\hat{e}) \label{eq:bcsparsecore} \\
    \bw^{\SS(G, F)} &= 2\sum_{\hat{e} \in E(\mathcal{C}(G, F))} \bw^{\mathcal{C}(G, F)}_{\hat{e}}\left|\bPi_{\mathcal{C}(G, F) \to \SS(G, F)}(\hat{e})\right| \label{eq:bwsparsecore}.
  \end{align}
\end{definition}
We check that $\bc^{\SS(G, F)}$ is a circulation on $\SS(G, F)$ and $\bw^{\SS(G, F)}$ are length upper bounds.
\begin{lemma}[Validity of \cref{def:passsparsecore}]
  \label{lemma:passsparsecore}
  Let $\bc^{\mathcal{C}(G,F)}, \bw^{\mathcal{C}(G,F)}$ be a valid pair on graph $\mathcal{C}(G,F)$ with lengths $\bell^{\mathcal{C}(G,F)}$. As defined in \cref{def:passsparsecore}, $\bc^{\SS(G, F)}, \bw^{\SS(G,F)}$ is a valid pair on $\mathcal{S}(G,F)$ with lengths $\bell^{\mathcal{S}(G,F)}$ (\cref{def:sparsecore}). Also,
  \[ \|\bw^{\mathcal{C}(G,F)}\|_1 \le \|\bw^{\SS(G, F)}\|_1 \le O(\gamma_l)\|\bw^{\mathcal{C}(G,F)}\|_1. \]
\end{lemma}
\begin{proof}
  Let $\mB_{\SS}, \mB_{\mathcal{C}}$ be the incidence matrices of $\SS(G,F), \mathcal{C}(G,F)$ respectively. To see that $\bc^{\SS(G,F)}$ is a circulation, we write
  \begin{align*}
    \mB_{\SS}^{\top}\bc^{\SS(G,F)} &= \sum_{\hat{e} \in E(\mathcal{C}(G,F))} \bc^{\SS(G,F)}_{\hat{e}} \mB_{\SS}^\top\bPi_{\mathcal{C}(G, F) \to \SS(G, F)}(\hat{e}) = \sum_{\hat{e} \in E(\mathcal{C}(G,F))} \bc^{\SS(G,F)}_{\hat{e}} \bb_{\hat{e}} = \mB_{\mathcal{C}}^\top \bc^{\mathcal{C}(G,F)} = 0.
  \end{align*}
  To see that $\bw^{\SS(G,F)}$ are valid upper bounds, for all $\hat{e}' \in E(\SS(G,F))$
  \begin{align*}
    |\bell^{\SS(G,F)}_{\hat{e}'}\bc^{\SS(G,F)}_{\hat{e}'}| &= \bell^{\SS(G,F)}_{\hat{e}'}\Big|\sum_{\hat{e} : \hat{e}' \in \Pi_{\mathcal{C}(G,F)\to\SS(G,F)}(\hat{e})} \bc^{\mathcal{C}(G,F)}_{\hat{e}}\Big| \overset{(i)}{\le} 2\sum_{\hat{e} : \hat{e}' \in \Pi_{\mathcal{C}(G,F)\to\SS(G,F)}(\hat{e})} |\bell^{\SS(G,F)}_{\hat{e}}\bc^{\mathcal{C}(G,F)}_{\hat{e}}| \\
                                                           &\le 2\sum_{\hat{e} : \hat{e}' \in \Pi_{\mathcal{C}(G,F)\to\SS(G,F)}(\hat{e})} \bw^{\mathcal{C}(G,F)}_{\hat{e}} = 2\bw^{\SS(G,F)}_{\hat{e}'}.
  \end{align*}
  Throughout, we used several properties guaranteed in \cref{def:sparsecore}, and $(i)$ specifically follows by item \ref{item:samelen}. The final equality follows by the definition of $\bw^{\SS(G,F)}$ in \eqref{eq:bwsparsecore}.

  Finally we upper-bound $\|\bw^{\SS(G,F)}\|_1$ by
  \begin{align*}
    \|\bw^{\SS(G,F)}\|_1 &\le 2\sum_{\hat{e} \in E(\mathcal{C}(G,F))} \bw^{\mathcal{C}(G,F)}_{\hat{e}} \|\bPi_{\mathcal{C}(G,F)\to\SS(G,F)}(\hat{e}) \|_1 \\ &\le 2\|\bw^{\mathcal{C}(G,F)}\|_1 \length(\Pi_{\mathcal{C}(G,F)\to\SS(G,F)}) \le O(\gamma_l)\|\bw^{\mathcal{C}(G,F)}\|_1,
  \end{align*}
  because $\SS(G,F)$ is a $(\gamma_s,\gamma_c,\gamma_l)$ sparsified core graph.
  $\|\bw^{\cC(G,F)}\|_1 \le \|\bw^{\SS(G,F)}\|_1$ follows directly from the definition.
\end{proof}
We can now give an algorithm that takes a dynamic graph $G^{(t)}$ undergoing hidden stable-flow chasing updates, and maintain a sparsified core graph also undergoing hidden stable-flow chasing updates, such that the total size of updates increases by a factor of at most $m^{o(1)}$. This shows how to pass from level $i$ to $i+1$ in a tree-chain (\cref{def:chain}).

\begin{lemma}[Dynamic Sparsified Core Graphs]
  \label{lemma:hintedsparsecore}
  
  \cref{algo:dynasparsecore} takes as input a parameter $k$, a dynamic graph $G^{(t)}$ undergoes $\tau$ batches of updates $U^{(1)}, \dots, U^{(\tau)}$ with gradients $\bg^{(t)}$, lengths $\bell^{(t)}$ at stage $t = 0, \dots, \tau$ that satisfies $\sum_{t=1}^{\tau} \Enc(U^{(t)}) \le m/(k\log^2 n)$ and the hidden stable-flow chasing property with the hidden circulation $\bc^{(t)}$ and width $\bw^{(t)}$.

  The algorithm maintains for each $j \in [B]$ (for $B = O(\log n)$), a decremental forest $F^{(t)}_j$, a static tree $T_j$ satisfying the conditions of \cref{lemma:globalstretch}, and a $(\gamma_s,\gamma_l,\gamma_c)$-sparsified core graph $\cS(G^{(t)}, F^{(t)}_j)$ for parameters $\gamma_s = \gamma_c = \gamma_l = \exp(O(\log^{3/4}m\log\log m))$ with embedding $\Pi_{\mathcal{C}(G^{(t)}, F^{(t)}_j)\to\SS(G^{(t)}, F^{(t)}_j)}$:
  \begin{enumerate}
      \item \label{item:SCGLowRec} \underline{Sparsified Core Graphs have Low Recourse:}
      the algorithm outputs update batches $U_{\SS, j}^{(t)}$ that produce $\SS(G^{(t)}, F^{(t)}_j)$ from $\SS(G^{(t-1)}, F^{(t-1)}_j)$ such that $\sum_{t' \le t} \Enc(U_{\SS, j}^{(t')}) = \gamma_r \cdot \sum_{t' \le t} \Enc(U^{(t')})$ for some $\gamma_r = \exp(O(\log^{3/4}m\log\log m))$,
      
      \item \label{item:SCGHSFC} \underline{Sparsified Core Graphs undergo Hidden Stable-Flow Chasing Updates:}
      for each $j \in [B]$, the update batches $U_{\SS, j}^{(t)}$ to the sparsified core graph along with the associated gradients $\bg^{(t),\SS(G^{(t)} F^{(t)}_j)}$, and lengths $\bell^{(t),\SS(G^{(t)}, F^{(t)}_j)}$ as defined in \Cref{def:sparsecore} satisfy the hidden stable-flow chasing property (see \cref{def:hiddenStableFlowChasing}) with the hidden circulation $\bc^{(t),\SS(G^{(t)}, F^{(t)}_j)}$, and width $\bw^{(t),\SS(G^{(t)}, F^{(t)}_j)}$ as defined in \Cref{def:passsparsecore}, and
      
      \item \label{item:SCGWidth} \underline{The Widths on the Sparsified Core Graphs are Small:}
      for each $j \in [B]$, the width on the sparsified core graph $\cS(G^{(t)} F^{(t)}_j)$ is bounded as follows:
      \begin{align*}
        \norm{\bw^{(t),\cS(G^{(t)}, F^{(t)}_j)}}_1
        \le \O(\gamma_l)\left(\|\bw^{(0),\cC(G^{(0)}, F^{(0)}_j)}\|_1 + \norm{\bw^{(t)}}_1\right).
      \end{align*}
      Also, whp. there is an $j^* \in [B]$ only depending on $\bw^{(0)}$ such that
      \begin{align}
        \left\|\bw^{(t),\SS(G^{(t)}, F^{(t)}_{j^*})}\right\|_1
        \le \O(\gamma_l)\left(\left\|\bw^{(0)}\right\|_1 + \left\|\bw^{(t)}\right\|_1\right).
        \label{eq:hintedsparsecore} 
      \end{align}
  \end{enumerate}
The algorithm runs in total time $\O(mk \cdot \gamma_r)$.
\end{lemma}

\begin{algorithm}[!ht]
  \caption{Dynamically maintains a sparsified core graph (\cref{def:coregraph}). Procedure $\Initialize$ initializes all variables, and $\DynamicSparseCore$ takes updates to $G^{(t)}$.\label{algo:dynasparsecore}
  }
  \SetKwProg{Globals}{global variables}{}{}
  \SetKwProg{Proc}{procedure}{}{}
  \Globals{}{
    $B \assign O(\log n)$: number of instantiations of \cref{lemma:globalstretch} in \cref{lemma:hintedcore} \\
    $\cA^{(Core)}$: algorithm implementing \cref{lemma:hintedcore} \\
    $\cA^{(Spanner)}_j$ for $j \in [B]$: algorithms implementing \cref{thm:spanner}
  }
  \Proc{$\Initialize(G = G^{(0)}, \bell, k)$}{
    $\cA^{(Core)}.\Initialize(G, \bell, k)$ \\
    \For{$j \in [B]$}{
        Let $\cW$ be the partition in \cref{lemma:globalstretch} item \ref{item:degbound}, and $R \supseteq \partial \cW$ an initial set of roots obtained from running the algorithm in \cref{lemma:globalstretch}. \\
        Create graph $\wt{\cC}_j$ by splitting vertices of $\cC(G,F_j)$ into vertices $u_r$ for each $r \in R$, and a vertex $u_W$ for the set of vertices $W \setminus R$ for each $W \in \cW$.
        \label{line:splitBeforeSpanner} \tcp{The vertices $u_r$ will not be split further, and $u_W$ all have degree at most $\O(k)$. Also, a deletion to $F_j$ will only split a single vertex.} 
        Let $\Lambda_{j} \defeq \Lambda_{\wt{\cC}_j \to \cC(G,F_j)}$ be the bijection between $E\left(\wt{\cC}_j\right)$ and $E\left(\cC(G,F_j)\right)$.\\
        Initialize $\cA^{(Spanner)}_j$ on $\wt{\cC}_j$, the split version of $\cC(G, F_j).$ \\
        Let $\wt{\cS}_j$ be the spanner maintained by $\cA^{(Spanner)}_j$ and $\cS(G, F_j) \defeq \Lambda_j(\wt{\cS}_j).$
    }
  }
  \Proc{$\DynamicSparseCore(G^{(t)}, U^{(t)}, \bg^{(t)}, \bell^{(t)})$}{
    $\cA^{(Core)}.\DynamicCore(G^{(t)}, U^{(t)}, \bg^{(t)}, \bell^{(t)})$ \\
    \For{$j \in [B]$}{
        Let $U^{(t)}_j$ be the update batch that produce $\cC(G^{(t)}, F^{(t)}_j)$ from $\cC(G^{(t-1)}, F^{(t-1)}_j)$. \\
        Let $U^{(t)+}_j \subseteq U^{(t)}_j$ contain all edge insertions. \\
        Let $U^{(t)-}_j \subseteq U^{(t)}_j$ contain the rest. \\
        Update $\wt{\cS}^{(t-1)}_j$ to $\wt{\cS}^{(t-0.5)}_j$ with $\Lambda_{j}^{-1}(U^{(t)-}_j)$ using $\cA^{(Spanner)}_j$. \\
        Update $\wt{\cS}^{(t-0.5)}_j$ to $\wt{\cS}^{(t)}_j$ via inserting edges of $\Lambda_{j}^{-1}(U^{(t)+}_j)$ directly. \\
        Let $R^{(t)}_j \subseteq E(\wt{\cS}^{(t)}_j)$ be the re-embedded set output by $\cA^{(Spanner)}_j$. \\
        Let $U^{(t)}_{\cS, j}$ be the corresponding update batch that produce $\cS(G^{(t)}, F^{(t)}_j)$ from $\cS(G^{(t-1)}, F^{(t-1)}_j)$. \\
        Append $U^{(t)}_{\cS, j}$ with $\Lambda_j(R_j)$ and output $U^{(t)}_{\cS, j}$. \tcp{Despite edges in $\Lambda_j(R_j)$ remain unchanged in $\cS(G^{(t)}, F^{(t)}_j)$, we force re-insertions on them in the output batch of updates.}
    }
  }
\end{algorithm}

\cref{algo:dynasparsecore} essentially maintains the sparsified core graphs $\SS(G^{(t)}, F^{(t)}_j)$ by passing the core graphs $\mathcal{C}(G^{(t)}, F^{(t)}_j)$ into the dynamic spanner \cref{thm:spanner}. Intuitively, because $F^{(t)}_j$ is decremental, the graph $\mathcal{C}(G^{(t)}, F^{(t)}_j)$ changes by undergoing vertex splits, plus additional edge insertions and deletions induced by the update batch $U^{(t)}$.

Similar to \cref{lemma:globalstretch}, \cref{algo:dynasparsecore} treats each update batch $U^{(t)}$ to $G$ as $O(\Enc(U^{(t)}))$ edge insertions/deletions and isolated vertex insertions.
In particular, for any update $x \in U^{(t)}$ that splits a vertex $u \in G^{(t-1)}$, it is treated as an update sequence of inserting one isolated vertex $u^{NEW}$ and then deleting/inserting $\deg_{G^{(t)}}(u^{NEW})$ edges.

However, each edge insertion/deletion causes $\O(1)$ vertex splits in the core graph $\cC(G^{(t)}, F^{(t)}_j)$.
As vertices in $\cC(G^{(t)}, F^{(t)}_j)$ could have degree $\Omega(k)$, we cannot afford treating vertex splits in the core graph as a sequence of edge insertions/deletions.
This would represent $\cS(G^{(t)}, F^{(t)}_j)$ using updates of total encoding size $O(k \cdot \sum_t \Enc(U^{(t)})) = O(m)$ instead of $O(m^{1+o(1)} / k).$
Using the dynamic spanner of \cref{thm:spanner} resolves the issue as it handles vertex splits with low recourse.
In particular, $\cS(G^{(t)}, F^{(t)}_j)$ can be represented using a sequence of updates with total encoding size $O(m^{o(1)} \cdot \sum_t \Enc(U^{(t)})).$

Formalizing this approach requires discussion of several technical points.
First, we cannot simply maintain the spanner of $\cC(G^{(t)}, F^{(t)}_j)$ using \cref{thm:spanner} which does not support edge insertions.
Instead of modifying the dynamic spanner algorithm, we deal with edge insertions na\"ively by inserting each of them to $\cS(G^{(t)}, F^{(t)}_j)$.
As the total number of edge insertion is at most $\sum_{t \in [\tau]}\Enc(U^{(t)}) = o(m / k)$, $\cS(G^{(t)}, F^{(t)}_j)$ is still sparse enough.

Second, vertices in core graphs $\cC(G^{(t)}, F^{(t)}_j), j \in [B]$ might have degree $\Omega(k).$
To ensure a maximum degree bound of $\O(k)$ which is required by \cref{thm:spanner}, we artificially split vertices in $\cC(G^{(t)}, F^{(t)}_j)$ to create a graph $\wt{\cC}^{(t)}_j$ on which we maintain the spanner.
Precisely, we create a new vertex $u_W$ in $\wt{\cC}^{(0)}_j$ for each piece $W$ in the partition $\cW$ of the forest $F_j^{(0)}$, and a vertex $u_r$ for each root in the initial forest $F_j^{(0)}$.
Throughout the execution, we ensure that every vertex of $\wt{\cC}^{(t)}_j$ is either $u_r$ for some $r$ being a root in the current forest, or $u_X$ for some connected component $X \subseteq W$ of an initial piece $W \in \cW.$
In the former case, $u_r$ corresponds to a single vertex in the original graph $G^{(t)}$ and thus it is never split due to edge removals from the forest $F^{(t)}_j.$
In the later case, $u_X$ corresponds to the set of vertices $X \setminus R$ and thus its degree is bounded by $\O(k)$ due to \ref{item:degbound} of \cref{lemma:globalstretch}.

\begin{proof}[Proof of \Cref{lemma:hintedsparsecore}]

We first argue that the graph $\wt{\cC}_j$ for all $j \in [B]$ has maximum degree $\O(k)$ and $O(m/k)$ vertices, and undergoes a total of $O(m/k)$ vertex splits, and edge insertions/deletions.
This shows that the application of the dynamic spanner algorithm in \cref{thm:spanner} is efficient.

For any $j \in [B]$, each vertex of $\wt{\cC}_j$ is either $u_r$ for some root $r \in R$ or $u_{X}$ for some connected component $X$ of an initial piece $W \in \cW.$
In the case of $u_r$, it will not be split further.
In the case of $u_{X}$, it corresponds to the set of vertices $X \setminus R$.
Since $X$ is a connected component in $F$ of an initial piece $W \in \cW$, the degree of $u_{X}$ is at most $\deg_G(W \setminus R)$ which is $\O(k)$ due to \ref{item:degbound} of \cref{lemma:globalstretch}.

The data structure implementing \cref{lemma:globalstretch} inside $\cA^{(Core)}$ ensures that $F_j$ is decremental.
Edge deletions in $F_j$ does not affect either gradients nor lengths of edges in $\cC(G, F_j)$ (\cref{def:coregraph}).
Thus, one edge deletion in $F_j$ corresponds to only a single vertex split in $\cC(G, F_j).$
The total number of vertex splits happened to $\cC(G, F_j)$ can be bounded by the number of edge removals in $F_j.$
The number is $O(m/k + q \log^2 n) = O(m/ k)$ for $q \defeq \sum_{t=1}^{\tau} \Enc(U^{(t)}) \le m / (k\log^2 n)$ by item \ref{item:cccount} of \cref{lemma:globalstretch}.
Similarly, each edge deletion to $F_j$ causes one vertex split in $\wt{\cC}_j$.
To see this, first note that no root vertices $u_r \in \wt{\cC}_j$ are ever split.
For the deletion of an edge $e$ to $F_j^{(t)}$, let $W \in \cW$ be the partition piece containing $e$.
The vertex $u_W$ may have been split further already, so let $e$ be currently inside the connected component $X \subseteq W$.
Now, because $\partial W \subseteq R$ at all times, we get that only $u_{X}$ was split in $\wt{\cC}_j$, as desired.

After updating $F_j^{(t)}$ and the enlarged vertex set of $\cC(G^{(t)}, F_j^{(t)})$, we process every update of $\bigcup_t U^{(t)}$ naively as $O(q)$ edge updates.
As each edge update to $\cC(G^{(t)}, F_j)$ corresponds to one edge update to $\wt{\cC}_j$, the number of edge updates happened to $\wt{\cC}_j$ is also $O(q) = O(m/(k \log^2 n)).$
It remains to bound the initial number of vertices in $\wt{\cC}_j$ by $O(m/k).$
As noted in \cref{algo:dynasparsecore}, there is one vertex per root of the initial forest $F_j$ and one vertex per cluster of the partition $\cW.$
The number of roots initially is $O(m/k)$ (\cref{lemma:globalstretch}).
The number of clusters in $\cW$ is also $O(m/k)$ (\cref{lemma:globalstretch}).
Thus, $\wt{\cC}_j$ has $O(m/k)$ vertices initially.

As noted in \cref{algo:dynasparsecore}, it is at most twice the initial number of roots in $F_j$ which is $O(m/k).$

\paragraph{Bounding the total size of $U^{(t)}_{\cS, j}$ (\cref{item:SCGLowRec}):}
  Fix some $j \in [B].$
  As discussed above, processing all updates in the data structure of \Cref{lemma:globalstretch} causes at most $O(m / k +  \sum_{t} \Enc(U^{(t)})\log^2 m) = O(m / k)$ vertex splits to $\wt{\cC}_j$.
  So, the graph $\wt{\mathcal{C}}_j$ undergoes at most $O(m/k)$ vertex splits and edge insertions/deletions.
  By item \ref{item:lowReEmbed} of \cref{thm:spanner}, the data structure $\cA^{(Spanner)}_i$ outputs the re-embedded set $R_j^{(t)}$ of amortized size at most $\gamma_r$, by taking $L = (\log m)^{1/4}$ in \cref{thm:spanner}.
  Thus, the total size of re-embedded edges $\sum_t |R_j^{(t)}|$ is bounded by $O(m\gamma_r/k)$. Similarly, \cref{thm:spanner} also shows that $\SS(G^{(t)},F_j^{(t)})$ are $(\gamma_s,\gamma_c,\gamma_l)$ sparsified core graphs with the embeddings $\Pi^{(t)}_j$.

Now, we move towards checking the remaining conditions: showing the hidden stable-flow chasing property of the outputs on $\SS(G^{(t)},F^{(t)}_j)$ for all $j \in [B]$, and \eqref{eq:hintedsparsecore}.
For simplicity, we use $\Pi^{(t)}_j$ to denote the embedding $\Pi_{\cC(G^{(t)}, F^{(t)}_j) \to \cS(G^{(t)}, F^{(t)}_j)}$ throughout the remainder of this proof.
\paragraph{Showing hidden stable-flow chasing property (\cref{item:SCGHSFC}):}

$\bc^{(t),\SS(G^{(t)}, F^{(t)}_j)}$ and $\bw^{(t),\cS(G^{(t)}, F^{(t)}_j)}$ form a valid pair by \cref{lemma:passsparsecore}. Therefore, items \ref{item:circulation} and \ref{item:width} of \cref{def:hiddenStableFlowChasing} are satisfied.

Next, we prove item \ref{item:widthstable} of \cref{def:hiddenStableFlowChasing}.
At any stage $t \in [\tau]$ and any edge $e \in \cS(G^{(t)}, F_j^{(t)})$ for some $j \in [B]$, suppose $e$ also appears in an earlier stage $t'$, i.e. $e \in \cS(G^{(t')}, F_j^{(t')})$ for some $t' < t.$
$e$ is not included in any of $U^{(s)}_{\cS, j}, s \in (t', t].$
Thus, we have $(\Pi^{(t)}_j)^{-1}(e) \subseteq (\Pi^{(t')}_j)^{-1}(e)$ otherwise $e$ is included in some $U^{(s)}_{\cS, j}, s \in (t', t]$ due to the definition of re-embedded set (\cref{prop:lowRecourseSpanner} of \cref{thm:spanner}).

For any edge $e' \in (\Pi^{(t)}_j)^{-1}(e)$, it exists in the core graph at both stage $t$ and $t'$, i.e. $e' \in \cC(G^{(t)}, F_j^{(t)})$ and $\cC(G^{(t')}, F_j^{(t')}).$
Let ${e'}^G$ be its pre-image in $G.$
${e'}^G$ also exists in $G$ at both stage $t$ and $t'$.
Since $G$ is undergoing hidden stable-flow chasing updates, by item \ref{item:widthstable} of \cref{def:hiddenStableFlowChasing} we have
\begin{align*}
    \bw^{(t), G^{(t)}}_{{e'}^G} \le 2 \cdot \bw^{(t'), G^{(t')}}_{{e'}^G}.
\end{align*}
\cref{def:passcore} and the immutable nature of $\wstr$ from \cref{lemma:globalstretch} yields
\begin{align}
\label{eq:coreWidthStable}
    \bw^{(t), \cC(G^{(t)}, F^{(t)}_j)}_{e'} = \wstr^{T_j, \bell}_{{e'}^G}  \bw^{(t), G^{(t)}}_{{e'}^G} \le 2 \cdot \wstr^{T_j, \bell}_{{e'}^G}  \bw^{(t'), G^{(t')}}_{{e'}^G} = 2 \cdot \bw^{(t'), \cC(G^{(t')}, F^{(t')}_j)}_{e'}.
\end{align}

Combining with the fact that $(\Pi^{(t)}_j)^{-1}(e) \subseteq (\Pi^{(t')}_j)^{-1}(e)$ and \cref{def:passsparsecore} yields the following and proves item \ref{item:widthstable} of \cref{def:hiddenStableFlowChasing}:
\begin{align*}
    \bw^{(t),\cS(G^{(t)}, F^{(t)}_j)}_{e} 
    &= 2 \cdot \sum_{e' \in \left(\Pi^{(t)}_j\right)^{-1}(e)} \bw^{(t), \cC(G^{(t)}, F^{(t)}_j)}_{e'} \\
    &\le 2 \cdot 2 \cdot \sum_{e' \in \left(\Pi^{(t)}_j\right)^{-1}(e)} \bw^{(t'), \cC(G^{(t')}, F^{(t')}_i)}_{e'} \\
    &\le 2 \cdot 2 \cdot \sum_{e' \in \left(\Pi^{(t')}_j\right)^{-1}(e)} \bw^{(t'), \cC(G^{(t')}, F^{(t')}_j)}_{e} = 2 \cdot \bw^{(t'),\cS(G^{(t')}, F^{(t')}_i)}_{e}.
\end{align*}

  Item~\ref{item:quasipoly} follows directly from the definition of $\bell^{(t),\cS(G^{(t)}, F^{(t)}_j)}$ and $\bw^{(t),\cS(G^{(t)}, F^{(t)}_j)}.$

\paragraph{Upper-bounding $\|\bw^{(t),\cS(G^{(t)}, F^{(t)}_{i^*})}\|_1$ (\cref{item:SCGWidth}):}
  For any $i$, \Cref{lemma:passsparsecore} yields
  \begin{align*}
    \norm{\bw^{(t),\SS(G^{(t)}, F^{(t)}_i)}}_1 \le O(\gamma_l) \norm{\bw^{(t),\cC(G^{(t)}, F^{(t)}_i)}}_1.
  \end{align*}
  \Cref{lemma:hintedcore} gives that there is an $i^* \in [B]$ such that for all $t$,
  \begin{align*}
    \norm{\bw^{(t),\cC(G^{(t)}, F^{(t)}_{i^*})}}_1 \le \O(\gamma_{LSST}\log^2 n)\|\bw^{(0)}\|_1 + \|\bw^{(t)}\|_1.
  \end{align*}
  Combining these gives the desired bound.

\paragraph{Runtime:}
  Time spent on the data structure implementing \Cref{lemma:globalstretch} is $\O(mk).$
  Before using the dynamic spanner of \Cref{thm:spanner}, we split each $\cC(G, F_j), j \in [B]$ in  Line~\ref{line:splitBeforeSpanner}. 
  This makes the max degree of the input graph to each of $\cA^{(Spanner)}_j, j \in [B]$ being $\Theta(k).$ By \cref{lemma:hintedcore} none of these vertices is split in an update to $\cC(G, F_j)$, so we may still apply \Cref{thm:spanner}.
  Thus, the time spent on every dynamic spanner is $O(mk\gamma_r)$.
\end{proof}

\subsubsection{Maintaining a Branching Tree-Chain}
Note that definitions \cref{def:passcore,def:passsparsecore} give a way to pass $\bc^{(t)},\bw^{(t)}$ from the top level graph $G$ downwards through a tree-chain (\cref{def:chain}). We formalize this by proving that we can dynamically maintain a branching tree-chain (\cref{def:chain}).
\begin{lemma}[Dynamic Branching Tree-Chain]
  \label{lemma:dynamicchain}
  \cref{algo:dynamicchain} takes as input a parameter $d$, a dynamic graph $G^{(t)}$ undergoes $\tau$ batches of updates $U^{(1)}, \dots, U^{(\tau)}$ with gradients $\bg^{(t)}$, length $\bell^{(t)}$ at stage $t=0, \dots, \tau$ that satisfies the hidden stable-flow chasing property (\cref{def:hiddenStableFlowChasing}) with hidden circulation $\bc^{(t)}$, and width $\bw^{(t)}.$
  The algorithm explicitly maintains a $B$-branching tree-chain (\cref{def:chain}) with previous rebuild times $\prev^{(t)}_0,\dots,\prev^{(t)}_d$ (\cref{def:rebuildtimes}). If $\bc^{(t),G},\bw^{(t),G}$ for $G \in \mathcal{G}^{(t)}_i$ for all $0 \le i \le d$ are recursively defined via \cref{def:passcore}, \ref{def:passsparsecore} then there is a tree-chain $G_0,\dots, G_d$ with
  \begin{align}
    \label{eq:dynamicchain}
    \|\bw^{(t),G_i}\|_1 \le \O(\gamma_l)^i\left(\sum_{j=0}^i\|\bw^{(\prev^{(t)}_j)}\|_1 + \|\bw^{(t)}\|_1\right) \forall i \in \{0,1,\dots,d\}.
  \end{align}
  The algorithm succeeds with high probability and runs in total time  $m^{1/d}\O(\gamma_s\gamma_r)^{O(d)}(m+Q)$ for $Q \defeq \sum_t \Enc(U^{(t)}) \le \poly(n)$.
\end{lemma}
\begin{rem}
  \label{rem:dynamicchain}
  \cref{thm:norebuild} maintains the data structure implementing \cref{lemma:dynamicchain} on dynamic graphs undergoing only edge insertions/deletions.
  However, it can be modified to also support vertex splits since it is built using \cref{lemma:hintedsparsecore,lemma:globalstretch,thm:spanner}, all which support vertex splits.
\end{rem}

\begin{algorithm}[t]
  \caption{Dynamically maintains a $B$-branching tree chain (\cref{def:chain}). Procedure $\Initialize$ initializes all variables, $\Rebuild$ rebuilds the data structure of level at least $d_s$ at stage $t_0$, and $\DynamicBranchingChain$ takes updates to $G^{(t)}$. \label{algo:dynamicchain}}
  \SetKwProg{Globals}{global variables}{}{}
  \SetKwProg{Proc}{procedure}{}{}
  \Globals{}{
    $d \gets \log^{1/8}n$: number of levels in the maintained branching tree chain. \\
    $k \assign m^{1 / d}$: reduction factor used in \cref{lemma:globalstretch}. \\
    $B \assign O(\log n)$: number of sparsified core graphs maintained in \cref{lemma:hintedsparsecore}. \\
  }
  \Proc{$\Initialize(G^{(0)}, \bell)$}{
    Initialize $\cG_0 = \{G^{(0)}\}.$ \\
    $\Rebuild(0, 0)$
  }
  \Proc{$\Rebuild(i_0, t_0)$}{
    \For{$i = i_0, \dots, d - 1$}{
      $\prev_{i+1} \assign t_0.$ \\
      $\cG_{i+1} \assign \{\}$ \\
      \For{$G \in \cG_i$}{
        $\cA^{(SparseCore)}_G.\Initialize(G, \bell_G)$ \\
        For $j \in [B]$, add $\cS(G, F_j)$ to $\cG_{i + 1}$.
      }
    }
  }
  \Proc{$\DynamicBranchingChain(G^{(t)}, U^{(t)}, \bg^{(t)}, \bell^{(t)})$}{
    $U^{(t)}_{G^{(t)}} \assign U^{(t)}$ \\
    \For{$i = 0, \dots, d - 1$}{
      \If{The accumulated encoding size of updates of any $G \in \cG_i$ exceeds $m (\gamma_s / k)^{i+1} / \log^2 n$}{
        $\Rebuild(i, t)$
      }
      \For{$G \in \cG_i$}{
        $\Set{U^{(t)}_{\cS(G, F_j)}}{j \in [B]} \assign \cA^{(SparseCore)}_G.\DynamicSparseCore(G, U^{(t)}_G)$
      }
    }
  }
\end{algorithm}

\cref{algo:dynamicchain} initializes a $B$-branching tree chain as in \cref{def:chain}.
For every graph $G \in \cG_i$ for some level $i$, it maintains a collection of forests, trees, and sparsified core graph using the dynamic data structure from \cref{lemma:hintedsparsecore}.

However, the data structure of \cref{lemma:hintedsparsecore} can only take up to $m /(k\log^2 n)$ updates if the input graph has at most $m$ edges at all time.
This forces us to rebuild the data structure every once in a while.
In particular, we rebuild everything at every level $i \ge i_0$ if any of the data structures of \cref{lemma:hintedsparsecore} on some level $i_0$ graph $G \in \cG_{i_0}$ has accumulated too many updates (approximately $m/k^{i_0}$).
We will show that the cost for rebuilding amortizes well across dynamic updates.

\begin{proof}[Proof of \cref{lemma:dynamicchain}]
  At any level $i = 0, \dots, d$, there are at most $O(\log n)^i$ graphs maintained at $i$-th level at any given stage $t$ due to \cref{lemma:hintedsparsecore}.
  That is, we have $|\cG^{(t)}_i| \le O(\log n)^i$ for any $t$ and $i.$
  At any stage $t$ and level $i > 0$, every graph $G \in \cG^{(t)}_i$ has at most $m \gamma_s^{i-1} / k^{i}$ vertices and $m (\gamma_s / k)^{i}$ edges.
  This is again due to \cref{lemma:hintedsparsecore}.

  To analyze the runtime, note that every $m \gamma_s^i / \left(k^{i+1} \log^2 n\right)$ updates to some graph $G \in \cG_i$ create $O(m\gamma_s^i\gamma_r/k^{i+1})$ updates to every $\cS(G, F_j)$ for $j \in [B]$ by \cref{lemma:hintedsparsecore}.
  Therefore, over the course of $Q$ updates to the top level graph $G$, the total number of updates to the $B = O(\log n)$ graphs at level $i$ is
  \begin{align*}
    Q \cdot O\left(\gamma_r \log^3 n\right)^i.
  \end{align*}

  Next we analyze the runtime cost of rebuilding level $i_0$.
  By \cref{lemma:hintedsparsecore}, it takes $O(m(\gamma_s/k)^ik\gamma_r)$ time to initialize $\cA^{(\ref{lemma:hintedsparsecore})}$ for any graph $G \in \cG_i$ at any level $i.$
  Therefore, the cost of rebuilding the graphs of every level $i \ge i_0$ is
  \begin{align*}
    \sum_{i = i_0}^d O(\log n)^i \cdot m (\gamma_s / k)^{i}k\gamma_r = \frac{m}{k^{i_0}} \cdot \O(\gamma_s)^dk\gamma_r.
  \end{align*}
  However, the rebuild happens at most every $m (\gamma_s / k)^{i_0} / \log^2 n$ total updates to graphs at level $i_0.$
  Thus, over the course of at most $Q \cdot O\left(\gamma_r \log^3 n\right)^{i_0}$ updates to every graph at level $i_0$, the total runtime cost spent on rebuilding level $i_0$ is at most
  \begin{align*}
    \left(m + Q\right)k\gamma_r\O(\gamma_s\gamma_r)^{O(d)}.
  \end{align*}
  The overall runtime bound follows because $k = m^{1/d}$.

  We now show \eqref{eq:dynamicchain} by induction on $t$ and the the level $i$, and prove the result for level $i+1$ at a time $t$ given a partial tree chain $G_0, \dots, G_i$ satisfying \eqref{eq:dynamicchain}. Let $G^{(\prev^{(t)}_i)}_i$ be the version of graph $G_i$ when it was rebuilt at time $\prev^{(t)}_i$.

  If $\prev^{(t)}_{i+1} < t$, we can use the same chain $G_0, \dots, G_{i+1}$ as the change from stage $\prev^{(t)}_{i+1}$ because \cref{lemma:hintedsparsecore} guarantees that there is an index $j^* \in [B]$ which satisfies \eqref{eq:hintedsparsecore} for all stages in $[\prev^{(t)}_{i+1}, t]$, where $G_{i+1} = \SS(G_i, F_{j^*})$. By induction and \cref{lemma:hintedsparsecore}, we deduce that
    \begin{align*}
\norm{\bw^{\left(t\right),G_{i+1}}}_1
&
\overset{\left(i\right)}{\le}
\O\left(\gamma_l\right)
\left(\norm{
  \bw^{\left(\prev^{\left(t\right)}_i\right),
  G^{\left(\prev^{\left(t\right)}_i\right)}_i}
}_1
+
\norm{\bw^{\left(t\right),G_i}}_1\right) \\
                            &\overset{(ii)}{\le} \O(\gamma_l)\left[\O(\gamma_l)^i\sum_{j=0}^i \norm{\bw^{\left(\prev^{(t)}_j\right)}}_1 + \O(\gamma_l)^i\left(\sum_{j=0}^i \norm{\bw^{\left(\prev^{(t)}_j\right)}}_1 + \norm{\bw^{(t)}}_1\right) \right] \\
                            &\le \O(\gamma_l)^{i+1}\left(\sum_{j=0}^i \norm{\bw^{\left(\prev^{(t)}_j\right)}}_1 + \norm{\bw^{(t)}}_1\right).
  \end{align*}
  $(i)$ is because the vector $\bw^{(0)}$ in \cref{lemma:hintedsparsecore} corresponds to $\bw^{(\prev^{(t)}_i),G^{(\prev^{(t)}_i)}_i}$, as $\prev^{(t)}_i$ is the initialization time of level $i$. $(ii)$ is by induction on the stage $t$ and level $i$. Thus we have shown \eqref{eq:dynamicchain} by induction, which completes the proof.
\end{proof}

\subsection{Finding Approximate Min-Ratio Cycles in a Tree-Chain}
\label{sec:findcycles}
In this section we explain how to extract a cycle $\bDelta$ from a branching tree-chain (such as the one maintained in \cref{lemma:dynamicchain}) with large quality $|\bg^\top\bDelta|/\|\mL\Delta\|_1$, satisfying the guarantees of \cref{thm:norebuild}. As a branching tree-chain consists of $O(\log n)^d$ tree-chains, we focus on getting a cycle $\bDelta$ out of a single tree-chain. More formally, our setting for much of this section will be a tree-chain $G_0, G_1, \dots, G_d$ (\cref{def:chain}), with a corresponding tree $T \defeq T^{G_0,\dots,G_d}$ as defined in \cref{def:treesfromchain}.
For $\bg,\bell$ and a valid pair $\bc,\bw$ (\cref{def:validpair}), we can define $\bc^{G_0} \defeq \bc$ and $\bw^{G_0} \defeq \bw$, and $\bc^{G_i}$ and $\bw^{G_i}$ recursively for $1 \le i \le d$ via \cref{def:passcore,def:passsparsecore}. Let $\bell^{G_i},\bg^{G_i}$ be the lengths and gradients on the graphs $G_i$, and $\bell^{\mathcal{C}(G_i,F_i)},\bg^{\mathcal{C}(G_i,F_i)}$ be the lengths and gradients on the core graphs.

Note that every edge $e^G \in E(G) \setminus E(T)$ has a ``lowest'' level that the image of it (which we call $e$) exists in a tree chain, after which it is not in the next sparsified core graph. In this case, the edge plus its path embedding induce a cycle, which we call the \emph{sparsifier cycle} associated to $e$. In the below definition we assume that the path embedding of a self-loop $e$ in $\mathcal{C}(G_i,F_i)$ is empty.
\begin{definition}
\label{def:sparsifiercycle}
Consider a tree-chain $G_0 = G,\dots,G_d$ (\cref{def:chain}) with corresponding tree $T \defeq T^{G_0,\dots,G_d}$ where for every $0 \le i \le d$, we have a core graph $\mathcal{C}(G_i, F_i)$ and sparsified core graph $\SS(G_i, F_i) \subseteq \mathcal{C}(G_i, F_i)$, with embedding $\Pi_{\mathcal{C}(G_i, F_i)\to\SS(G_i, F_i)}$.

We say an edge $e^G \in E(G)$ is at level $\lvl_{e^G} = i$ if its image $e$ is in $E(\mathcal{C}(G_i, F_i)) \backslash E(\SS(G_i, F_i))$. Define the \emph{sparsifier cycle} $a(e)$ of such an edge $e = e_0 \in \mathcal{C}(G_i, F_i)$ to be the cycle $a(e) = e_0  \oplus \rev(\Pi_{\mathcal{C}(G_i, F_i)\to\SS(G_i, F_i)}(e_0)) = e_0  \oplus e_1 \oplus \cdots \oplus e_L$. We define the preimage of this sparsifier cycle in $G$ to be the \emph{fundamental chain cycle}
     \[ 
     a^G(e^G) = e_0^G \oplus T[v^G_0, u^G_1] \oplus e_1^G \oplus T[v^G_1, u^G_2] \oplus \dots \oplus e_L^G \oplus T[v^G_L, u^G_{L+1}],
     \]
    where $e^G_i = (u_i^G, v_i^G)$ is the preimage of edge $e_i$ in $G$ for each $i \in [L]$ and where we define $u^G_{L+1} = u^G_0$.
\end{definition}
We let $\ba(e)$ and $\ba^G(e^G)$ be the associated flow vectors for the sparsifier cycle $a(e)$ and fundamental chain cycle $a^G(e^G)$.

At a high level, our algorithm will maintain the total gradient of every fundamental chain cycle explicitly. Note that this implies that the gradient of at most $m^{o(1)}$ fundamental chain cycles change per iteration on average. Also, the algorithm maintains length \emph{overestimates} of each fundamental chain cycle, as maintaining the true length dynamically is potentially expensive. The algorithm will return the overall best quality fundamental chain cycle.
\begin{definition}
  \label{def:fundChainCycleOverest}
  Consider a tree-chain $G=G_0,\dots,G_d$ with corresponding spanning tree $T \defeq T^{G_0, \dots, G_d}$.
  For any edge $e^G \in E(G) \setminus E(T)$ at level $i$ with image $e$ in $\cC(G_i,F_i) \setminus \SS(G_i,F_i)$ we define $\wlen_{e^G}^G$, an overestimate on the length of $e^G$'s fundamental chain cycle, as $\wlen_{e^G} \defeq \langle \bell^{\cC(G_i,F_i)}, |\ba(e)| \rangle.$
\end{definition}
Because the lengths and gradients on all edges in all the $G_i$, and embeddings $\Pi_{\cC(G_i,F_i)\to\SS(G_i,F_i)}$ are maintained explicitly in \cref{lemma:dynamicchain}, we can store length overestimates $\wlen_{e^G}$ for all fundamental chain cycles, and their total gradients with a constant overhead.

There are two more important pieces to check. First, we need to check that the gradients defined on the core graphs \cref{def:coregraph} indeed given the correct total gradient for each cycle, and that the values $\wlen_{e^G}$ are indeed overestimates for the lengths of all the fundamental chain cycles $\ba^G(e^G)$. Then we will show that using the length overestimates $\wlen_{e^G}$ still allows us to return a sufficiently good fundamental chain cycle.
\begin{lemma}[Gradient correctness]
\label{lemma:gradcorrect}
Let $e^G \notin E(T)$ be an edge with $\lvl_{e^G} = i$ and let $e$ be its image in $G_i$. Then the total gradient of the cycle $a(e)$ and its preimage $a^G(e^G)$ are the same, i.e. $\langle \bg^{\cC(G_i,F_i)}, \ba(e) \rangle = \langle \bg, \ba^G(e^G) \rangle$.
\end{lemma}
\begin{lemma}[Length overestimates]
\label{lemma:lencorrect}
Let $e^G \notin E(T)$ be an edge with $\lvl_{e^G} = i$ and let $e$ be its image in $G_i$. Then the values $\wlen_{e^G}$ overestimate the length of the preimage cycle $a^G(e^G)$, i.e. $\wlen_{e^G} \ge \langle \bell, |\ba^G(e^G)| \rangle$.
\end{lemma}
It is useful to define the concept of \emph{lifting} a cycle back from $\cC(G_i,F_i)$ to $G_i$ in order to show \cref{lemma:gradcorrect,lemma:lencorrect}.
\begin{definition}[Lifted cycle]
\label{def:lifted}
Consider a cycle $\hat{C}$ in $\mathcal{C}(G_i,F_i)$ with edges $\hat{e_1} \oplus \hat{e_2} \oplus \dots \oplus \hat{e_L}$, such that $e_i = (u_i, v_i)$ is the preimage of $\hat{e_i}$ in $G_i$. We define the lift of $\hat{C}$ into $G_i$ as the cycle
\[ e_1 \oplus F_i[v_1,u_2] \oplus e_2 \oplus \dots \oplus e_L \oplus F_i[v_L,u_1]. \]
\end{definition}
Now we can show \cref{lemma:gradcorrect,lemma:lencorrect} by repeatedly lifting cycles until we are back in the top level graph $G$.
\begin{proof}[Proof of \cref{lemma:gradcorrect}]
It suffices to show that any cycle $\hat{C}$ in $\mathcal{C}(G_i,F_i)$ and its lift $C$ in $G_i$ have the same gradient. Precisely, if we let $\hat{\bc}$ and $\bc$ denote the flow vectors of $\hat{C}$ and $C$ respectively, we wish to show $\langle \bg^{\cC(G_i,F_i)}, \hat{\bc} \rangle = \langle \bg^{G_i}, \bc \rangle$. To see this, recall that by the definition of $\bg^{\mathcal{C}(G_i,F_i)}$ in \cref{def:coregraph}, for $\hat{C} = \hat{e_1} \oplus \dots \oplus \hat{e_L}$ for $e_j = (u_j,v_j)$ and $E(F_i) \subseteq E(T_i)$ for some tree $T_i$ (namely the tree used to initialize the forest $F_i$ of $\cC(G_i,F_i)$),
\begin{align*}
    \langle \bg^{\mathcal{C}(G_i,F_i)}, \hat{\bc} \rangle
    &= \sum_{j=1}^L \bg^{\mathcal{C}(G_i,F_i)}_{\hat{e_j}} = \sum_{j=1}^L \bg^{G_i}_{e_j} + \langle \bg^{G_i}, \bp(T_i[v_j, u_j]) \rangle \\
    &\overset{(i)}{=} \sum_{j=1}^L \bg^{G_i}_{e_j} + \langle \bg^{G_i}, \bp(T_i[v_j, u_{j+1}]) \rangle \\ &\overset{(ii)}{=} \sum_{j=1}^L \bg^{G_i}_{e_j} + \langle \bg^{G_i}, \bp(F_i[v_j, u_{j+1}]) = \langle \bg^{G_i}, \bc \rangle,
\end{align*}
where $(i)$ follows because $\sum_{j=1}^L \bp(T_i[v_j, u_j]) = \sum_{j=1}^L \bp(T_i[v_j, u_{j+1}])$ (for $u_{L+1} \defeq u_1$) because both sides route the same demand on a tree $T_i$. $(ii)$ follows because $F_i \subseteq T$, and $v_j, u_{j+1}$ are in the same connected component of $F_i$. The last equality follows by the definition of $C$.
\end{proof}
\begin{proof}[Proof of \cref{lemma:lencorrect}]
Similar to the above proof of \cref{lemma:gradcorrect}, by repeatedly lifting until we get to $G$, it suffices to show that the length of a cycle $\hat{C}$ in $\mathcal{C}(G_i,F_i)$ is larger than that of its lift $C$.
Formally, if $\hat{\bc}$ and $\bc$ denote the flow vectors of $\hat{C}$ and $C$ respectively, we wish to show $\langle \bell^{\mathcal{C}(G_i,F_i)}, |\hat{\bc}| \rangle \ge \langle \bell^{G_i}, |\bc| \rangle$. For $\hat{C} = \hat{e_1} \oplus \dots \oplus \hat{e_L}$ for $e_j = (u_j, v_j)$ and forest $F_i$, we have
\begin{align*}
    \langle \bell^{\mathcal{C}(G_i,F_i)}, |\hat{\bc}| \rangle
    &= \sum_{j=1}^L \bell^{\mathcal{C}(G_i,F_i)}_{\hat{e_j}}
    \overset{(i)}{=} \sum_{j=1}^L \wstr^{i}_{e_j} \bell^{G_i}_{e_j} 
    \overset{(ii)}{\ge} \sum_{j=1}^L \str^{F_i,\bell^{G_i}}_{e_j} \bell^{G_i}_{e_j}
    \\ &= \sum_{j=1}^L \bell^{G_i}_{e_j} + \langle \bell^{G_i}, |\bp(F_i[u_j, \root^{F_i}_{u_j}])| \rangle + \langle \bell^{G_i}, |\bp(F_i[v_j, \root^{F_i}_{v_j}])| \rangle \\
    &= \sum_{j=1}^L \bell^{G_i}_{e_j} + \langle \bell^{G_i}, |\bp(F_i[v_j, \root^{F_i}_{v_j}])| \rangle + \langle \bell^{G_i}, |\bp(F_i[u_{j+1}, \root^{F_i}_{u_{j+1}}])| \rangle \\
    &\overset{(iii)}{\ge} \sum_{j=1}^L \bell^{G_i}_{e_j} + \langle \bell^{G_i}, |\bp(F_i[v_j, u_{j+1}])| \rangle = \langle \bell^{G_i}, |\bc| \rangle,
\end{align*}
where $(i)$ follows from the definition of $\bell^{\mathcal{C}(G_i,F_i)}$ in \cref{def:coregraph}, $(ii)$ follows from \cref{lemma:globalstretch} item \ref{item:stretchbound}, and $(iii)$ follows from $\root^{F_i}_{v_j} = \root^{F_i}_{u_{j+1}}$ and the triangle inequality. The final equality is from the definition of $C$ as the lift of $\hat{C}$.
\end{proof}
We now show that it suffices to maintain the ``best quality'' fundamental chain cycle, i.e. $\max_{e^G \in E(G) \setminus E(T)} |\l \bg^G, \ba^G(e^G) \r|/\wlen_{e^G}$. To show this, we first explain how to express a cycle $\bc$ as the combination of fundamental chain cycles.
\begin{lemma}
\label{lemma:chaindecomp}
Given a circulation $\bc$ in graph $G$, recursively define $\bc^{G_i}$ for all $i = 0, \dots, d$ via \cref{def:passcore,def:passsparsecore}. Then
\[ \bc = \sum_{i=0}^d \sum_{e^G : \lvl_{e^G} = i} \bc^{G_i}_{e^i} \ba^G(e^G). \]
\end{lemma}

\begin{proof}
Define $\by \defeq \sum_{i=0}^d \sum_{e^G : \lvl_{e^G} = i} \bc^{G_i}_{e^i} \ba^G(e^G).$
We will show that $\bc_{e^G} = \by_{e^G}$ for any edge $e^G \in G \setminus T.$

First, define for any $i = 0, \dots, d$, $\Pi_i$ as the embedding $\Pi_{\cC(G_i, F_i) \to \cS(G_i, F_i)}$ and $e^i$ as the image in $G_i$ for any edge $e^G \in G.$
Clearly, $e^i$ is well-defined if $i \le \lvl_{e^G}.$
We also denote the image of $e^i$ in the core graph $\cC(G_i, F_i)$ as $\hat{e}^i.$

At any level $i$, observe that if $e^{i+1} \in G_{i+1}$, we have $\hat{e}^i = e^{i+1}$, $\Pi_i(\hat{e}^i) = \{e^{i+1}\}$ and therefore $\bc^{G_{i+1}}_{e^{i+1}} = \bc^{G_i}_{e^i}.$
Otherwise, $\bc^{G_i}_{e^i}$ is added to $\bc^{G_{i+1}}_f, f \in \Pi_i(\hat{e}^i)$.
Let $e^G$ be any edge in $G \setminus T$ at level $i.$
Following from \cref{def:passsparsecore}, we can express $\bc^{G_i}_{e^i}$ as
\begin{align}
\label{eq:expandCGi}
\bc^{G_i}_{e^i} = \bc_{e^G} + \sum_{j=0}^{i-1} \sum_{f^G: \lvl_{f^G} = j} \bc^{G_j}_{f^j} \cdot \bPi_j(\hat{f}^j)_{\hat{e}^j}.
\end{align}

On the other hand, we know that $e^G$ does not appear in any of the sparsifier cycle of $f^G$ at level $j \ge i.$
Thus, $[\ba^G(f^G)]_{e^G} = 0.$
If $\lvl_{f^G} = j < i$, $[\ba^G(f^G)]_{e^G} = -\bPi_j(\hat{f}^j)_{\hat{e}^j}$ where the $-1$ term comes from that the sparsifier cycle takes $\hat{f}^j$ and the reverse of the path $\Pi_j(\hat{f}^j).$
This yields that
\begin{align*}
\by_{e^G} 
&= \bc^{G_i}_{e^i} - \sum_{j=0}^{i-1} \sum_{f^G: \lvl_{f^G} = j} \bc^{G_j}_{f^j} \bPi_j(\hat{f}^j)_{\hat{e}^j}
  \overset{(i)}{=} \bc_{e^G},
\end{align*}
where $(i)$ follows by rearranging \eqref{eq:expandCGi}.

The lemma follows via the fact that a circulation is uniquely determined by the amount of flows on non-tree edges.
\end{proof}
\begin{lemma}
\label{lemma:goodenough}
Let $\bc, \bw$ be a valid pair. Let $T = T^{G_0,\dots,G_d}$ for a tree-chain $G_0,\dots,G_d$. Then
\[ \max_{e^G \in E(G) \setminus E(T)} \frac{|\l \bg, \ba^G(e^G) \r|}{\wlen_{e^G}} \ge \frac{1}{\O(k)} \frac{|\l \bg, \bc \r|}{\sum_{i=0}^d \|\bw^{G_i}\|_1}. \]
\end{lemma}
\begin{proof}%
Recall that for an edge $\hat{e} \in \mathcal{C}(G_i,F_i)$ with preimage $e$ in $G_i$, its length is $\bell^{\mathcal{C}(G_i,F_i)}_{\hat{e}} \defeq \wstr^{i}_e \bell^{G_i}_e$ defined in \cref{def:coregraph}. Thus by the definition of $\wlen_{e^G}$ in \cref{def:fundChainCycleOverest},
\begin{align*}
    \sum_{i=0}^d \sum_{e^G:\lvl_{e^G}=i} |\bc^{G_i}_e| \wlen_{e^G} &= \sum_{i=0}^d \sum_{e^G:\lvl_{e^G}=i} \left(\wstr^{i} \bell^{G_i}_e |\bc^{G_i}_e| + \sum_{e' \in \Pi_{\mathcal{C}(G_i,F_i)\to\SS(G_i,F_i)}(\hat{e})} \bell^{G_{i+1}}_{\hat{e}} |\bc^{G_i}_e| \right) \\
    &\overset{(i)}{\le} \sum_{i=0}^d \sum_{e^G:\lvl_{e^G}=i} \left(\O(k) \bw^{G_i}_e + \bw^{G_{i+1}}_{\hat{e}} \right) \le \O(k) \sum_{i=0}^d \|\bw^{G_i}_e\|_1,
\end{align*}
where $(i)$ follows $\wstr^{i} \le \O(k)$ from \cref{lemma:globalstretch} item \ref{item:stretchbound}, and the fact that $\bc^{G_i}, \bw^{G_i}$ are all valid pairs (see \cref{lemma:passcore,lemma:passsparsecore}) so $\bell^{G_i}_e |\bc^{G_i}_e| \le |\bw^{G_i}_e|$, and the definition of $\bw^{G_{i+1}}_e$ in \cref{def:passsparsecore}.
Additionally, note by the triangle inequality and \cref{lemma:chaindecomp} that
\[ |\l \bg, \bc\r| \le \sum_{i=0}^d \sum_{e^G:\lvl_{e^G}=i} |\bc^{G_i}_e| |\l \bg, \ba^G(e^G)\r|. \]
Hence, we get by the fact that \[ \max_{i\in[n]} \frac{\bx_i}{\by_i} \ge \frac{\sum_{i\in[n]} \bx_i}{\sum_{i\in[n]} \by_i} \] for $\bx,\by \in \R^n_{\ge0}$ that
\[ \max_{e^G \in E(G) \setminus E(T)} \frac{|\l \bg, \ba^G(e^G) \r|}{\wlen_{e^G}} \ge \frac{\sum_{i=0}^d\sum_{e^G:\lvl_{e^G}=i}|\bc^{G_i}_e||\l \bg, \ba^G(e^G) \r|}{\sum_{i=0}^d\sum_{e^G:\lvl_{e^G}=i}|\bc^{G_i}_e|\wlen_{e^G}} \ge \frac{1}{\O(k)} \frac{|\l \bg, \bc \r|}{\sum_{i=0}^d \|\bw^{G_i}\|_1}. \]
\end{proof}
\begin{rem*}
We can adapt the statement and proof of \cref{lemma:goodenough} to remove the $\O(k)$ by being more careful. However, this further complicates the statements of \cref{lemma:goodenough} and its interaction with \cref{sec:jtree}, and the extra $\O(k)$ does not meaningfully affect our runtimes.
\end{rem*}
We can now complete the proof of \cref{thm:norebuild}.
\begin{proof}[Proof of \cref{thm:norebuild}]
  The first part is to dynamically maintain an explicit $O(\log n)$ branching tree chain with path embeddings using \cref{lemma:dynamicchain}. With $O(1)$ overhead, the algorithm can also maintain the values of $\l \bg, \ba^G(e^G) \r$, $\wlen_{e^G}$ for all fundamental chain cycles of the $O(\log n)^d$ trees in the branching tree chain, because the branching tree chain maintains all edge gradients/lengths explicitly, and \cref{def:fundChainCycleOverest,lemma:gradcorrect}. Hence in $\O(1)$ overhead it can maintain the maximizer $\argmax_{e^G \in E(G) \setminus E(T)} \frac{|\l \bg, \ba^G(e^G) \r|}{\wlen_{e^G}}$ as desired in \cref{lemma:goodenough} for each tree. To show that the best out of these works, note that by \cref{lemma:dynamicchain} with high probability there is a tree-chain with
  \[ \frac{1}{\O(k)} \frac{|\l \bg, \bc \r|}{\sum_{i=0}^d \|\bw^{G_i}\|_1} \ge \frac{1}{\O(k)\O(\gamma_l)^{O(d)}} \frac{|\l \bg, \bc^{(t)} \r|}{\sum_{i=0}^d \|\bw^{(\prev^{(t)}_i)}\|_1} \ge \kappa \frac{|\l \bg, \bc^{(t)} \r|}{\sum_{i=0}^d \|\bw^{(\prev^{(t)}_i)}\|_1}, \] for $\kappa = 1/(\O(k)\O(\gamma_l)^{O(d)})$ as desired.
  
  The total runtime is $(m+Q)m^{o(1)}$ for $Q = \sum_{t\in[\tau]} \Enc(U^{(t)})$ by \cref{lemma:dynamicchain} for the choice $\gamma_s = \gamma_r = \exp(\log^{3/4}\log\log m), d = \log^{1/8}m$ and $k = m^{1/d}$.
  Thus $\kappa = \exp(-O(\log^{7/8}m\log\log m))$.
  Also, $Q$ approximates $\sum_{t\in[\tau]} |U^{(t)}|$ up to a polylog factor since $U^{(t)}$ contains only edge insertions/deletions.
  
  Finally, we can rebuild levels $i, i+1, \dots, d$ in time $m^{1+o(1)}/k^i$ time because the graphs on level $i$ have $m(\gamma_s/k)^i$ edges, there are $O(\log n)^d$ such graphs, and the initialization time is almost linear.
\end{proof}

\section{Rebuilding Data Structure Levels}
\label{sec:rebuilding}

The goal of this section is to use \cref{thm:norebuild} to get \cref{thm:MMCHiddenStableFlow} through a \emph{rebuilding game} to handle the cases where $\|\bw^{(\prev^{(t)}_i)}\|_1$ is much larger than  $\|\bw^{(t)}\|_1$ for some $0 \le i \le d$. The surprising aspect is that this is doable despite the fact that the $\bw^{(t)}$ are all hidden.
We now introduce the \emph{rebuilding game} that captures these notions.
Each round of the rebuilding game corresponds to our algorithm successfully returning a good enough cycle.
When is this does not happen, we instead have to rebuild part of our data structure.
The rebuilding game is designed to let us formally reason about strategies for rebuilding the data structure when it fails to find a good cycle.
\newcommand{\Wrange}{K} %

\paragraph{Rebuilding game parameters and definition.}
  The rebuilding game has several parameters: integers parameters size $m > 0$ and depth $d > 0$, 
  update frequency $0 < \gamma_g < 1$,
  a rebuilding cost $C_r \geq 1$, a weight range $\Wrange \geq 1$, and a recursive size reduction parameter $k \defeq m^{1/d} \geq 2$, and finally an integer round count $T > 0$. 
\begin{definition}[Rebuilding Game]
  \label{def:game}
  The rebuilding game is played between a player and an advesary and proceeds in rounds $t=1,2,\ldots,T$.
  Additionally, the steps (moves) taken by the player are indexed as $s = 1,2,\ldots$.
  Every step $s$ is associated with a tuple $\prev^{(s)} := (\prev^{(s)}_0, \dots, \prev^{(s)}_d) \in [T]^{d+1}$.
  Both the player and adversary know $\prev^{(s)}$.
  At the beginning of the game, at round $t = 1$ and step $s = 1$, we initially set $\prev^{(1)}_i = 1$ for all levels $i \in \setof{0,1,\ldots,d}$.
  
  At the beginning each round $t \ge 1$,
  \begin{enumerate}
  \item The adversary first chooses a positive real weight $W^{(t)}$ satisfying $\log W^{(t)} \in (-\Wrange, \Wrange).$
    This weight is \textbf{hidden} from the player.
    \label{enu:gameStart}
  \item 
    \label{enu:fixingTest}
    Then, \textbf{while} either of the following conditions hold,
    \begin{align}
      \label{eq:gameLossSumCond}
      &\sum_{i=0}^d W^{(\prev^{(s)}_i)} > 2(d+1) W^{(t)}
      \\
      \nonumber
      \text{or} \qquad &
      \\
      &\text{For some level $l$, at least $\gamma_g m/k^l$ rounds}
        \nonumber
      \\
      &\text{have passed since the last rebuild of level $l$.}
        \label{eq:gameProgCountCond}
    \end{align}
    the adversary can (but does not have to) force the player to perform a \emph{fixing step}.
    The player may also choose to perform a fixing step, regardless of whether the adversary forces it or not.
    In a fixing step, the player picks a level $i \in \setof{0,1,\ldots,d}$, and we then set $\prev^{(s+1)}_j \gets t$ for $j \in \{i,i+1, \dots, d\}$, and $\prev^{(s+1)}_j \gets \prev^{(s)}_j$ for $j \in \{0, \dots, i-1\}$. We call this a \emph{fix} at level $i$, and we say the levels $j \geq i$ have been \emph{rebuilt}. 
    This move costs $C_rm/k^i$ time.
  \item When the player is no longer performing fixing steps, the round finishes.
  \end{enumerate}
\end{definition}

  The goal of the player is to complete all $T$ rounds in total time cost
  $O(\frac{C_r  K d}{\gamma_g} (m + T))$.
  
\begin{rem}
  We emphasize an important point about our terminology in the rebuilding game: A \emph{fix} at level $i$ causes a \emph{rebuild} of all levels $j \geq i$.
  The adversary can force a rebuild of a level $l$ if it has participated in  $\gamma_g m/k^l$ rounds since it was last rebuilt
  -- and the latest rebuild may have been triggered by a fix at level $l$ or by a fix at some level $i < l$.
\end{rem}
To translate the rebuilding game to the setting of \cref{thm:MMCHiddenStableFlow,thm:norebuild} we can set $W^{(t)} \defeq \|\bw^{(t)}\|_1$. We give an algorithm where the player completes all $T$ rounds with total time cost $O(\frac{C_r  K d}{\gamma_g} (m + T))$. For our choice of parameters, this will be almost linear.
Note that it is trivial for the player to finish all $T$ rounds in time $O(C_rmT)$, as they could just always do a fix at level $i = 0$, as this sets all $\prev^{(s)}_j \assign t$.
\newcommand{\fixc}{\texttt{fix}}
\newcommand{\progc}{\texttt{round}}

\begin{algorithm}[!ht]
  \ForEach{$i = 0,\ldots, d$.}{
    We maintain a "fixing count", $\fixc_i$, initialized to zero.\;
    And we maintain a "round count", $\progc_i$, also initialized to zero.\;
  }
  \ForEach{round $t = 1,2, \ldots, T$ of the game}{
    
    \If{there is a level $l$ with $\progc_l \geq \gamma_g m/k^l$}{
      Find the smallest level $i$ such that $\progc_i \geq \gamma_g m/k^i$\;
      \label{lne:WINfix} Fix level $i$, thus rebuilding levels $j \geq i$.\;
      For levels $j = i,i+1,\ldots,d$
      set $\fixc_j \gets 0$ and $\progc_j \gets 0$. \;
      \tcp{ We call this a \emph{WIN at level i}.}
    }
    \While{the adversary continues to force a fixing step}{
      Let $i$ be the smallest level in $0,\ldots,d$ s.t. for all $j > i$,
      $\fixc_j = 2\Wrange$\;
      \label{lne:LOSSfix} Fix level $i$, thus rebuilding levels $j \geq i$.\;
      Set $\fixc_i \gets \fixc_i + 1$  \;
      \tcp{We call this a \emph{LOSS at level i}.}
    }
    For all levels $j = 0,1,\ldots,d$, set
    $\progc_j \gets \progc_j + 1.$\;
  }
  \caption{Strategy for the rebuilding game.}
  \label{alg:rebuildStrategy}
\end{algorithm}

The following lemma tells us that a fairly simple strategy can deterministically\footnote{Note that
  when we employ the rebuilding game strategy in our overall data structure, the data structure uses randomization. 
  However, the randomized steps succeed with high probability union-bounded across the entire algorithm.
  The rebuilding game strategy corresponds to the behavior of the data structure assuming all these data structure randomization steps are successful. 
  We address this formally in the proof of
  Theorem~\ref{thm:MMCHiddenStableFlow}.} ensure that the player always
wins the rebuilding game.

\begin{lemma}[Strategy for Rebuilding Game]
  \label{lemma:rebuildinggame}
  There is a deterministic strategy given by \Cref{alg:rebuildStrategy} for the player to finish $T$ rounds of the rebuilding game in time $O(\frac{C_r  K d}{\gamma_g} (m + T))$.
\end{lemma}

Before we state the proof, we first introduce some important
terminology for understanding \Cref{alg:rebuildStrategy} and its
analysis.

\paragraph{The rebuilding game algorithm.} Overall, our goal is the following: we (as the player) want to ensure that our vector
of weights $W^{(\prev^{(s)}_i)}$ at different levels $i \in
\setof{0,1,\ldots,d}$ is such that we must frequently succeed in completing a round without making too many fixing steps,
and we need to ensure we do not spend too much time on fixing steps.
To implement our strategy in \Cref{alg:rebuildStrategy}, we
maintain two counters $\progc_i$ and $\fixc_i$ for each level $i$.
These two counters are used to decide which level to rebuild in each
step $s$ of the game.

The first counter, $\progc_i$, is very simple. It counts the number of
rounds that have occurred since level $i$ was last rebuilt.
Ideally, we would like to complete as many rounds as each level can handle, before we reset it.
The adversary can force a fixing step if it has been more than $\gamma_g m/k^i$ rounds since level $i$ was last rebuilt. In the setting of \cref{thm:norebuild}, this corresponds to a level of the branching tree chain accumulating enough updates that it should be rebuilt.
We preempt the adversary by always rebuilding a level if it has been through this many rounds, regardless of whether the adversary forces us to or not.
When this occurs, the level can ``pay for itself'', since the cost of fixing is low when amortized across the rounds since the last rebuild.
Thus we declare a ``WIN'' at level
$i$ and rebuild levels $j \geq i$.

The second counter, $\fixc_i$, is the more interesting one.
When a fixing step occurs and we decide to fix level $i$ (thus rebuiling levels $j \geq i$)
we say a ``LOSS'' occurred at level $i$.
The $\fixc_i$ counter tracks how many times a fixing step occured and
we had a LOSS at level $i$, counted since the last
time we rebuilt level $i$ due to a WIN at some level $l \le i$.
In a fixing step, we always decide to let the LOSS occur at the
largest level index $i$ where $\fixc_i < 2K$. 

\subsection{Analyzing the rebuilding game algorithm}

Before we start our formal proof of \Cref{lemma:rebuildinggame}, we will outline the main
the main elements of the analyses of the time cost of using \Cref{alg:rebuildStrategy} to play the rebuilding game.

Ideally, we would like to say that ``if a LOSS occurs
at level $i$, then the weight  $W^{(\prev^{(s)}_i)}$ must be
large compared to the current round $t$ weight $W^{(t)}$'', because
then rebuilding would reduce $W^{(\prev^{(s+1)}_i)}$.
However, this is not true, as it may be that some other
level's weight $W^{(\prev^{(s)}_j)}$ for $j < i$ is large enough to make
Equation~\eqref{eq:gameLossSumCond} hold, allowing the adversary to
force a fixing step.
Instead, the invariant we maintain is this: We ensure that when a
fixing step occurs, if we choose to rebuild level $i$, it must be that
either.
\begin{description}
\item[Case (A)] either some level $l < i$ has an even larger weight than
all levels $j \geq i$
and thus can be ``blamed'' for the fixing step, or
\item[Case (B)]
no level $j > i$ has a weight large enough to force a
fixing step. 
\end{description}

\paragraph{Handling Case (B).}
In case (B), we significantly reduce the
weight $W^{(\prev^{(s+1)}_i)}$ compared to $W^{(\prev^{(s)}_i)}$ at level $i$ for the next step $s+1$.
Once we have $\fixc_j = 2K$ for all $j \geq i$, there must be level $l < i$ with weight
larger than level $i$, as repeated occurrence of (B) ensures this.
\paragraph{Handling Case (A).}
Now, the remaining key point is to make sure that in case (A), we do
not waste too much time rebuilding at level $i$ before moving to
rebuilding at level $i-1$, so that we eventually start rebuilding the
most problematic level $l < i$ with larger weight.
Fortunately, our threshold of $2K$
fixes at level $i$ is low enough to ensure this.
Finally, to help us formalize that we make progress
on reducing the weight at level $i$ specifically in Case (A), we introduce a notion of ``prefix maximizing''
levels. At any step $s$, we say a level $i$ is ``prefix maximizing'' if its weight $W^{(\prev^{(s)}_i)}$ is strictly
larger than the weight $W^{(\prev^{(s)}_j)}$ at all levels $j < i$
This leads to the following definition.
\newcommand{\badLevelIndices}{\mathcal{I}}
\begin{definition}
  In the Rebuilding Game, at the start of each step $s$, we define a
  set $\badLevelIndices^{(s)}$ of ``prefix maximizing'' levels,
 given by
  \[
    \badLevelIndices^{(s)} = %
    \Set{ i \in \setof{0,1,\ldots,d} }{ W^{(\prev^{(s)}_i)} >
      W^{(\prev^{(s)}_j)} \forall j < i }.
  \]
\end{definition}
Note that $0 \in \badLevelIndices^{(s)}$ for all $s,$ and the player does not know which levels are prefix maximizing. 

This next lemma shows formally that strategy of
\Cref{alg:rebuildStrategy} successfully implements the kind of weight
tracking we described above.
Concretely, the lemma tells us that when a level $i$ is prefix maximizing, the weight of the level must be pushed down as
$\fixc_i$ increases.
\begin{lemma}[Bound on fixing step count]
  In the rebuilding game, suppose the player uses the strategy of
  \Cref{alg:rebuildStrategy}, then we always have $\fixc_0 < 2\Wrange$.
    \label{lem:fixingCountBound}
\end{lemma}
This lemma tells use that we never have $2K$ LOSSes at level $0$ before the next WIN at level $0$.
Since \Cref{alg:rebuildStrategy} trivially ensures $\fixc_j \le 2\Wrange$ for all levels $j > 0$,
this gives us tight control over the number of fixing steps that can be forced by the adversary.

\begin{proof}[Proof of \Cref{lem:fixingCountBound}]
  We consider an instance of the rebuilding game and suppose the player uses the strategy given by \Cref{alg:rebuildStrategy}.
  To prove our lemma, we first introduce a condition which must be satisfied in each step by each level $i$ which is prefix maximizing.
  This condition essentially states that the $\fixc_i$ counter is correctly tracking an upper bound on the level weight $W^{(\prev^{(s)}_i)}$.
  For convenience of our analysis, we define the condition for levels regardless of whether they are prefix maximizing, although we only need to show that it holds for such levels.
\begin{definition*}%
  At the start of step $s$,
    if for some level $i$ we have,
  \begin{align}
    \log_2\left(W^{(\prev^{(s)}_i)} \right)
    < \Wrange -\fixc_i
    &&
       (\text{$\fixc_i$ correctness condition})
    \label{eq:LevelFailBoundsWeight}
  \end{align}
  we say that level $i$ satisfies \emph{the $\fixc_i$ correctness condition at step} $s$.
\end{definition*}
Given this notion, we can now state the induction hypothesis.
  \begin{inducthyp*}%
    At the start of step $s$, Condition~\eqref{eq:LevelFailBoundsWeight} holds for each $i \in \badLevelIndices^{(s)}$.
\end{inducthyp*}

  We will prove this claim by induction on the step count $s$.
  First, we establish that proving this claim is sufficient to prove the lemma.
  By assumption, we have $\log_2\left(W^{(\prev^{(s)}_i)} \right) > -K,$ and hence the above claim would imply $\fixc_i < 2K$ for all $i \in \badLevelIndices^{(s)},$ for all $s.$ Since $0 \in \badLevelIndices^{(s)}$ for all $s,$ we get that $\fixc_0 < 2K$ always.

  We first establish the base case $s = 1.$  Note that trivially for all $i \in
  \setof{0,1,2,\ldots,d}$, we have $\fixc_i = 0$, and by
  assumption we have
  $\log_2(W^{(\prev^{(s)}_i)}) < \Wrange =  \Wrange-\fixc_i$.
  This implies Condition~\eqref{eq:LevelFailBoundsWeight} holds for every level
  $i$, and hence it holds for each level $i \in \badLevelIndices^{(s)}$.
  This establishes the base case.

  Now, we now assume the induction hypothesis at the \emph{start} of step $s$ and prove it for the \emph{start} of step $s+1$, i.e. we want to show
  Condition~\eqref{eq:LevelFailBoundsWeight} holds for each $i \in
  \badLevelIndices^{(s+1)}$ at the start of step $s+1$.
  We break the analysis into two main cases, depending on what
  happens in step $s$.
  The first case (1) is when a WIN occurs.
  The second case (2) is when a 
  LOSS occurs at some level $i$.
  We further break the second case into two sub-cases, separately handling 
  when (2A) $i$ is not a prefix maximizing level at step $s$
  and when (2B) $i$ is a prefix maximizing level at step $s$.
  Conceptually, the key case is (2B), when we have a LOSS and $i \in
  \badLevelIndices^{(s)}$, which means we have to
  ensure that we make progress by reducing $W^{(\prev^{(s+1)}_i)}$
  compared to the earlier value $W^{(\prev^{(s)}_i)}$.

  \paragraph{Case 1: a WIN occurs.}
  In this case, a WIN must occur at some level $i \in 
  \setof{0,1,2,\ldots,d}$ (since $\gamma_g m/k^d\le 1$).
  Let $i$ denote the level at which the WIN occurs.
  In this case, for all levels $j \geq i$, we rebuild and set
  $\fixc_i = 0$, and by the
  definition of $\Wrange$, we thus have (for the updated value of $\fixc_i$) that
  $\log_2(W^{(\prev^{(s+1)}_i)}) < \Wrange =  \Wrange-\fixc_i$.
  Thus, for each $j \geq i$, we have that
  Condition~\eqref{eq:LevelFailBoundsWeight} holds,
  and thus, in particular, it must hold for each $j \in \badLevelIndices^{(s+1)}$.

  For each level $l < i$,
  $\fixc_l$ does not change and $W^{(\prev^{(s+1)}_l)} = W^{(\prev^{(s)}_l)}$.
  The latter implies that for each $l < i$,
  $l \in \badLevelIndices^{(s+1)}$ if and only if $l \in \badLevelIndices^{(s)}$.
  We also conclude that Condition~\eqref{eq:LevelFailBoundsWeight} holds at the beginning of step $s+1$ if it held for level $l$ at the beginning of step $s$.
  Hence, by the induction hypothesis at step $s$,
  Condition~\eqref{eq:LevelFailBoundsWeight} holds for all  $l < i$
  with $l \in \badLevelIndices^{(s+1)}$.

  This proves the induction hypothesis for step $s+1$ in the case where a WIN occurs.
  
  \paragraph{Case 2: a LOSS occurs.}
  We next consider the case when a \emph{LOSS} occurs in step $s$, and we let the current round be denoted by $t$.
  The \emph{LOSS} occurs at some level $i$.
  In order to analyze this case, we are going to split it further into two subcases 2A and 2B,
  depending on whether the LOSS occurs at level $i$ which is in the prefix maximizing set or not (2B and 2A respectively).
  However, first we make some observations that are common to both cases 2A and 2B. 

  To start, we deal with levels $l < i$. As in the case of a WIN, we again have that, for each $l < i$,
  $\fixc_l$ does not change and $W^{(\prev^{(s+1)}_l)} = W^{(\prev^{(s)}_l)}$.
  The latter implies that for each $l < i$,
  $l \in \badLevelIndices^{(s+1)}$ if and only if $l \in \badLevelIndices^{(s)}$.
  Hence, by the induction hypothesis,   Condition~\eqref{eq:LevelFailBoundsWeight} holds
 for each $l < i$ with
  $l \in \badLevelIndices^{(s+1)}.$
  
  Next we need to deal with levels  $j \geq i$.
  As a LOSS occurs in step $s$, we must have that
    $\sum_{j = 0}^d  W^{(\prev^{(s)}_j)} \geq 2(d+1) W^{(t)}$.
  This implies
  \begin{equation}
    \max_{j = 0}^d  W^{(\prev^{(s)}_j)} \geq \frac{2(d+1)}{d+1}
    W^{(t)} \geq 2 W^{(t)}
    .\label{eq:12}
  \end{equation}

  We claim that in this case, we must have
  \begin{equation}
  \badLevelIndices^{(s)} \cap
  \setof{i+1, \ldots, d} = \emptyset
  .
  \label{eq:noBadLevelsStrictlyBelowLOSSLevel}
\end{equation}

  Suppose for a contradiction that for some $j > i$ we have $j \in \badLevelIndices^{(s)}$.
  As $\fixc_j = 2\Wrange$, we conclude by the induction hypothesis, that at the start of 
  step $s$, we have $\log_2 (W^{(\prev^{(s)}_j)})
  < \Wrange-\fixc_j = -\Wrange$, and hence $\log_2(W^{(\prev^{(s)}_j)}) < -\Wrange$.
  But, this is impossible, as $\log_2(W^{(\prev^{(s)}_j)}) > -\Wrange$
  by the game definitions.

  \paragraph{Subcase 2A: LOSS at level $i \not\in \badLevelIndices^{(s)}$.}
  We now further restrict to the case when at the start of step $s$, we have $i
  \not\in \badLevelIndices^{(s)}$.
  We thus have at the start of step $s$, by the condition observed in Equation~\eqref{eq:noBadLevelsStrictlyBelowLOSSLevel}, that
  $\badLevelIndices^{(s)} \cap \setof{i,i+1, \ldots, d} = \emptyset$.

  This allows us to conclude that
  \begin{equation}
        \text{ for all }  j \in \setof{i,i+1,
          \ldots, d} \text{ there exists } l < j \text{ with } W^{(\prev^{(s)}_l)} >
        W^{(\prev^{(s)}_j)}.
        \label{eq:levelsAboveFixAreGood}
      \end{equation}
      By Equation~\eqref{eq:levelsAboveFixAreGood}, we conclude that there exists $l < i$
      with $W^{(\prev^{(s)}_l)} > \max_{j \in \setof{i,i+1,\ldots,d}}
      W^{(\prev^{(s)}_j)}$.
      Furthermore, we can conclude that $\max_{h \in \setof{0,\ldots,i-1} }
      W^{(\prev^{(s)}_h)} \geq 2 W^{(t)}$, since the maximum in
      Equation~\eqref{eq:12} is not achieved by an index $\geq i$.
      Consequently, when the LOSS at level $i$ occurs and we rebuild,
      for all levels $j \in \setof{i,i+1,\ldots,d}$,
      we set $\prev^{(s+1)}_j \gets t$, and hence at the start of step $s+1$, we have for all $j \geq i$ that
      \[
        W^{(\prev^{(s+1)}_j)} = W^{(t)}
        \le \frac{1}{2} \max_{h \in \setof{0,\ldots,i-1} } W^{(\prev^{(s)}_h)}
        = \frac{1}{2} \max_{h \in \setof{0,\ldots,i-1 } } W^{(\prev^{(s+1)}_h)}
      \]
      and hence for all $j \geq i$ we conclude that $j \not\in \badLevelIndices^{(s+1)}$.
      Altogether, this proves the induction hypothesis for step $s+1$ in the case
      where a LOSS occurs at level $i$ and  $i \not\in \badLevelIndices^{(s)}$.

\paragraph{Subcase 2B: LOSS at level $i \in \badLevelIndices^{(s)}$.}
  We now consider the case when at the start of step $s$, we have
    $i \in \badLevelIndices^{(s)}$. This is the most important case,
    where we ensure a reduction in the weight $W^{(\prev^{(s+1)}_i)}$
    compared to $W^{(\prev^{(s)}_i)}$.
    By Equation~\eqref{eq:noBadLevelsStrictlyBelowLOSSLevel}, we have $\badLevelIndices^{(s)} \cap
    \setof{i+1, \ldots, d} = \emptyset$, and hence we conclude that $W^{(\prev^{(s)}_i)}
    = \max_{j = 0}^d  W^{(\prev^{(s)}_j)}\geq 2W^{(t)}$.
    By the induction hypothesis, we have (labelling $\fixc_i$
    explicitly by step for clarity) 
    \[
      \log_2(W^{(\prev^{(s)}_i)}) < K - \fixc_i^{(s)}
    \]
    hence, 
    \begin{equation}
     \log_2(W^{(\prev^{(s+1)}_i)}) =  \log_2(W^{(t)}) \le \log_2(W^{(\prev^{(s)}_i)}) -1 <
      K - \fixc_i^{(s)} -1 = K - \fixc_i^{(s+1)}.\label{eq:NewLevelFailBoundsWeight}
    \end{equation}
    Thus, Condition~\eqref{eq:LevelFailBoundsWeight} holds for $i$ at the
    end of step $s+1$, regardless of whether $i \in \badLevelIndices^{(s+1)}$.

    Furthermore, for all $j > i$, we set $W^{(\prev^{(s+1)}_j)} = W^{(t)} =
    W^{(\prev^{(s+1)}_i)}$, and hence $j \not\in
    \badLevelIndices^{(s+1)}$.
    
     Again, recall that we already dealt with established Condition~\eqref{eq:LevelFailBoundsWeight} for  $l < i$ above in Equation~\eqref{eq:levelsAboveFixAreGood}.
      Thus the induction hypothesis holds
    for step $s+1$ when $i \in \badLevelIndices^{(s)}$,

    \medskip
    \noindent
    This completes
    our case analysis, establishing the inductive hypothesis, and hence the lemma.
\end{proof} %

At this point, armed with the conclusion of
\Cref{lem:fixingCountBound}, we are ready to analyze the
running time of the rebuilding game strategy given by
\Cref{alg:rebuildStrategy}, to prove the main lemma of this section,
\Cref{lemma:rebuildinggame}.

Before starting the proof, we will briefly outline its main elements.
The costs of the rebuilding game occur during fixes as part of either a WIN or a LOSS
in \Cref{lne:WINfix} and \Cref{lne:LOSSfix} respectively.

We use a standard amortization argument to account for the cost of fixes that occur during WINs.
We can count the cost occurred during WINs separately  at each level, and finally add it up across these.
In each level, the cost per round can be bounded by $C_r/\gamma_g$.

Next, we have to account for the cost of fixes carried out during LOSSes.
These fixes are all accounted for by increases in some $\fixc_j$ counter.
We then bound the total cost of fixes that later have their $\fixc_j$ counter
reset by amortizing the cost toward the rounds that cause the reset
of the $\fixc_j$ counter through a WIN at some level $i \le j$.
Because \Cref{lem:fixingCountBound} guarantees the $\fixc_j$ counters are bounded by $2K$,
we can bound the additional cost amortized toward the rounds during the WIN at level $i$ by 
$\frac{4K C_r}{\gamma_g}$ per step.
Finally the bound on the $\fixc_j$ counters from \Cref{lem:fixingCountBound} also tells us that
the leftover cost unaccounted for by amortization through resets is also bounded, this time by $4KC_r m$ in total.

\begin{proof}[Proof of \Cref{lemma:rebuildinggame}.]
  To bound the running time, we use a simple amortized analysis across the steps of the rebuilding game.

  The round counter at level $i$, i.e., $\progc_i$ increases by $1$ in each round and hence the sum of the increases is $T$.
  If a WIN occurs at level $i$, we incur a time cost through a fix of level $i$, with a cost of $C_r m/k^i$ (\Cref{lne:WINfix}).
  At the same time, we reduce the round counter at $i$ and all deeper levels $j > i$, and in particular, we reduce $\progc_i$ by $\gamma_g m/k^i$.
  We will amortize the cost of this fix toward the rounds that increased $\progc_i$ from zero to the threshold  $\gamma_g m/k^i$, and thus the amortized cost from fixes during WINs at level $i$ per round is at most $C_r/\gamma_g$.
  When we add this up across $T$ rounds, the total cost from fixes during  WINs at level $i$ is $T\cdot C_r/\gamma_g$.
  Thus the total cost added across our $d+1$ levels from fixes during WINs is
  \begin{equation}
    \label{eq:costWINfix}
    \text{cost from fixes during WINs } \le T(d+1)C_r/\gamma_g
    .
  \end{equation}

  All the cost incurred during a LOSS at some level $i$ (\Cref{lne:LOSSfix}) leads to an increase of the fixing step counter $\fixc_i$ by 1, and has an associated time cost of $C_r m/k^i$.
  We will break the cost form LOSSes into two parts:
  \begin{enumerate}
  \item Cost accounted for by a $\fixc_i$ counter increase where the fix counter is later reset to 0.
    \label{enu:costLOSSReset}
  \item Cost accounted for by a $\fixc_i$ counter increase where the fix counter is not reset.
    \label{enu:costLOSSNoReset}
  \end{enumerate}
  We can bound the cost arising from Part~\ref{enu:costLOSSNoReset} very easily: By \Cref{lem:fixingCountBound}, $\fixc_i \le 2K$,
  and so the cost from fixing of level $i$ without a reset of the counter following is bounded by $\fixc_i \cdot C_r m/k^i \le 2K \cdot C_r m/k^i$.
  Adding this cost across all levels we get that the total cost from Part~\ref{enu:costLOSSNoReset}
  is upper bounded by
  \begin{equation}
    \label{eq:costLOSSNoResetFix}
    \text{cost from fixes during LOSSes with no $\fixc$ counter reset } \le \sum_{i = 0}^d 2K \cdot C_r m/k^i \le 4KC_r m
    .
  \end{equation}

  Finally, we bound the cost from Part~\ref{enu:costLOSSReset}.
  Consider the resetting of some counter $\fixc_j$ associated with a level $j$.
  Any such counter is reset during a WIN at some level $i < j$.
  We will bound the cost part by amortizing it toward the rounds that caused this WIN at level $i$.
  In particular, note that the level $i$ experienced least $\gamma_g m/k^i$ rounds since it was last rebuilt
  and at this point $\fixc_j$ was reset (though it may also have been reset again since).
  This means we can count the cost associated with the increases in $\fixc_j$ toward the WIN at level $i$.
  The total cost we need to account for in this way toward the WIN at level $i$ is then
  \[
    \sum_{j = i}^d \fixc_j \cdot C_r m/k^j \le 2K \sum_{j = i}^d C_r m/k^j \le 4K C_r m/k^i
    .
  \]
  Thus, the amortized cost per round associated with WINs at level $i$ through these $\fixc_j$ resets is at most
  \[
    \frac{4K C_r m/k^i}{\gamma_g m/k^i} = \frac{4K C_r}{\gamma_g}.
  \]
  As we have $T$ rounds, the total cost associated with WINs at level $i$ through $\fixc_j$ resets is then  $\frac{4K C_r}{\gamma_g} T$.
  Since a round can contribute toward a WIN at each of our $d+1$ levels,
  this means the cost amortized to toward a round across all levels is .
    \begin{equation}
    \label{eq:costLOSSResetFix}
    \text{cost from fixes during LOSSes with $\fixc$ counter reset }
    \le \sum_{i = 0}^d T\cdot \frac{4K C_r}{\gamma_g}
    = \frac{4K C_r(d+1)}{\gamma_g} T
    .
  \end{equation}

  Finally, adding together the costs accounted for in Equations~\eqref{eq:costWINfix}, \eqref{eq:costLOSSNoResetFix}, and \eqref{eq:costLOSSResetFix}, and the cost of executing rounds,
  we get a bound on the total cost of $O(\frac{KC_rd}{\gamma_g} (T+m))$ as desired.
\end{proof}

\subsection{Dynamic Min-Ratio Cycle Using the Rebuilding Game}

In this section we combine \cref{thm:norebuild} and \cref{lemma:rebuildinggame} to show \cref{thm:MMCHiddenStableFlow} which gives a data structure for returning min-ratio cycles in dynamic graphs with hidden stable-flow chasing updates.
\begin{proof}[Proof of \cref{thm:MMCHiddenStableFlow}]
Let $W^{(t)} \defeq \|\bw^{(t)}\|_1$. The adversary plays the following strategy. They feed the inputs $\bg^{(t)},\bell^{(t)},U^{(t)}$ to \cref{thm:norebuild} and get a cycle $\bDelta$. Let $\kappa^{(\ref{thm:norebuild})}$ be the approximate parameter from \cref{thm:norebuild}. Let $\kappa \defeq \kappa^{(\ref{thm:norebuild})}/(2d+2)$. The adversary checks whether \begin{align} \l \bg^{(t)}, \bDelta\r/\|\bell^{(t)} \circ \bDelta\|_1 \le -\kappa\alpha. \label{eq:checkit} \end{align} Because $\bDelta$ is represented using $m^{o(1)}$ edges on a tree $T$, this can be performed in amortized $m^{o(1)}$ time by using a dynamic tree (\cref{algo:LCT}). If \eqref{eq:checkit} holds, the adversary allow a progress step, and this completes the \textsc{Query}() operation of \cref{thm:MMCHiddenStableFlow}. Otherwise, they force the player to perform a fixing step. This is valid because by \cref{thm:norebuild} we must have
\begin{align*}
    \frac{\kappa^{(\ref{thm:norebuild})}}{2d} \frac{\l \bg^{(t)}, \bc^{(t)} \r}{\|\bw^{(t)}\|_1} \le -\kappa\alpha \le \frac{\l \bg^{(t)}, \bDelta \r}{\|\bell^{(t)} \circ \bDelta\|_1} \le \kappa^{(\ref{thm:norebuild})} \cdot \frac{\l \bg^{(t)}, \bc^{(t)} \r}{\sum_{i=0}^d \|\bw^{(\prev^{(t)}_i)}\|_1},
\end{align*}
so $\sum_{i=0}^d \|\bw^{(\prev^{(t)}_i)}\|_1 \ge 2(d+1) \|\bw^{(t)}\|_1$. Our algorithm for \cref{thm:MMCHiddenStableFlow} is then to implement the player's strategy in \cref{lemma:rebuildinggame} on top of \cref{thm:norebuild}.

Correctness of \textsc{Query}() follows by definition. To bound the runtime, by \cref{thm:norebuild} we can take the constants $d = \log^{1/8}m$, $k = \exp(O(\log^{7/8}m))$, $\gamma_g = \exp(-O(\log^{7/8}m\log\log m))$, $C_r = \exp(O(\log^{7/8}m\log\log m))$, $T = (m+Q) \exp(O(\log^{7/8}m\log\log m))$. Thus by \cref{lemma:rebuildinggame} the total runtime to execute the player's algorithm is $(m+Q) \exp(O(\log^{7/8}m\log\log m))$ as desired.
\end{proof}

%% file: combine.tex
\section{Computing the Min-Cost Flow via Min-Ratio Cycles}
\label{sec:combine}
In this section we given the full pseudocode for proving \cref{thm:main}, modulo getting an initial point and final point, which are explained in \cref{lemma:initialpoint,lemma:finalpoint}.

\begin{algorithm}
\caption{$\textsc{MinCostFlow}(G, \bd, \bc, \bu^+, \bu^-, \bf^{(0)}, F^*)$. Takes graph $G$, demands $\bd$, costs $\bc$, upper/lower capacities $\bu^+,\bu^-$, initial feasible flow $\bf^{(0)}$ (\cref{lemma:initialpoint}), and guess of the optimal flow $F^*$ \label{algo:mincost}}
\SetKwProg{Globals}{global variables}{}{}
\SetKwProg{Proc}{procedure}{}{}
\Globals{}{
    $\alpha \gets 1 / (1000 \log m U)$ \\
    $\kappa \gets \exp(-O(\log^{7/8} m \log \log m))$ \tcp*{Approximation quality in \cref{thm:MMCHiddenStableFlow}}
    $d \assign O(\log^{1/8} n)$ \tcp*{Data structure depth}
    $\mathcal{D}^{(HSFC)}$ \tcp*{Hidden Stable-Flow Chasing (HSFC) data structure in \cref{thm:MMCHiddenStableFlow}}
    $T_1, T_2, \dots, T_s$ for $s \gets O(\log n)^d$ \tcp*{Trees maintained by data structure $\mathcal{D}^{(HSFC)}$}
    $\eps \gets \kappa\alpha/(1000s)$. \tcp*{Error tolerated within each tree.}
    $\mathcal{D}^{(T_i)}$ \tcp*{Dynamic tree data structure for trees $T_i$}
    $\bf_1, \dots, \bf_s \assign \vec{0} \in \R^E$ and $\bf \defeq \bf^{(0)} + \sum_{i\in[s]} \bf_i$ \tcp*{Flows on trees $T_i$}
    $\wt{\bf}^{(t)}$ \tcp*{Approximate flow at stage $t$, remembers which edges have been updated.}
    $\bf^{(t)} \assign \bf^{(0)}$ \tcp*{Total flow at stage $t$, implicitly stored}
    $r \assign \infty$ \tcp*{Estimate of cost difference from optimal.}
}
\Proc{$\textsc{MinCostFlow}(G, \bd, \bc, \bu^+, \bu^-, \bf^{(0)}, F^*)$}{
    \While{$\bc^\top\bf^{(t)} - F^* \ge (mU)^{-10}$}{
        \If{$t$ is a multiple of $\lfloor\eps m\rfloor$\label{line:reinit}}{
            Explicitly compute $\bf^{(t)} \assign \bf^{(0)} + \sum_{i\in[s]} \bf_i$, $\wt{\bf}^{(t)} \assign \bf^{(t)}$. \label{line:initf} \\
            $r \assign \bc^\top\wt{\bf}^{(t)} - F^*$.\label{line:initr} \tcp*{Cost difference from optimal.}
            $\bg^{(t)} \assign \bg(\wt{\bf}^{(t)}), \bell^{(t)} \assign \bell(\wt{\bf}^{(t)})$ \label{line:initlg}\tcp*{\cref{def:lg}}
            Rebuild $\mathcal{D}^{(HSFC)}$ and update the $T_i$. \label{line:initdhsfc}
            \tcp*{Because $r$ may have changed by a $1+\eps$ factor.}
        }
        $U^{(t)} \assign \bigcup_{i\in[s]} \mathcal{D}^{(T_i)}.\textsc{Detect}()$ \label{line:detect} \tcp*{\cref{algo:LCT}}
        \ForEach{  $e \in U^{(t)}$}{
        Set $\wt{\bf}^{(t)}_e \assign \bf^{(t)}_e = \bf^{(0)}_e + \sum_{i\in[s]}(\bf_i)_e$, $\bell^{(t)}_e \assign \bell(\wt{\bf}^{(t)})_e$ \label{line:deffl}\tcp*{\cref{def:lg}}
        $\bg^{(t)}_e \assign 20m\bc_e/r + \alpha(\bu^+_e - \wt{\bf}^{(t)}_e)^{-1-\alpha} - \alpha(\wt{\bf}^{(t)}_e - \bu^-_e)^{-1-\alpha}$ \label{line:defg} 
        }
        \tcc{No change to $e \notin U^{(t)}$}
        \lForEach{ $e \notin U^{(t)}$}{     $\bg^{(t)}_e \assign \bg^{(t-1)}_e, \bell^{(t)}_e \assign \bell^{(t-1)}_e$, $\wt{\bf}^{(t)}_e \assign \wt{\bf}^{(t-1)}_e$.}
        $\mathcal{D}^{(HSFC)}.\textsc{Update}(U^{(t)}, \bg^{(t)}, \bell^{(t)})$, and update the $T_i$ for $i \in [s]$ \tcp*{\cref{thm:MMCHiddenStableFlow}}
        $(i, \bDelta) \assign \mathcal{D}^{(HSFC)}.\textsc{Query}()$, where $i \in [s]$ and $\bDelta = (u_1,v_1) \oplus T_i[v_1, u_2] \oplus (u_2,v_2) \oplus \dots \oplus (u_l,v_l) \oplus T_i[v_l, u_1]$ for edges $(u_i,v_i)$ and $l \le m^{o(1)}$. \label{line:query} \tcp*{$\bDelta$ represented via $m^{o(1)}$ off-tree edges and paths on $T_i$}
        $\bDelta \assign \eta\bDelta$ for $\eta \gets -\kappa^2\alpha^2/(800\langle \bg^{(t)},\bDelta\rangle)$ \label{line:scale} \tcp*{Scale $\bDelta$ so $\langle \bg^{(t)},\bDelta\rangle = -\kappa^2\alpha^2/800$}
        $\bf_i \assign \bf_i + \bDelta$ using $\mathcal{D}^{(T_i)}$, \cref{algo:LCT} item \ref{item:nonpositiveflow}
        \tcp*{Implicitly set $\bf^{(t)} \defeq \bf^{(t-1)} + \bDelta$} \label{line:defft}
        $t \assign t+1$.
    }
    \Return{$\bf^{(t)}$}
}
\end{algorithm}

We explain the implementation in procedure \textsc{MinCostFlow} given in \Cref{algo:mincost}.
The algorithm maintains approximate lengths $\bell^{(t)}$ and gradients $\bg^{(t)}$, updating them when the dynamic tree data structures $\mathcal{D}^{(T_i)}$ report that some edge has accumulated many changes.
It updates these lengths and gradients, and passes the result to a data structure $\mathcal{D}^{(HSFC)}$ which dynamically maintains the trees $T_1, \dots, T_s$ and a min-ratio cycle on them under hidden stable-flow chasing updates.
We check that the updates to $\bell^{(t)},\bg^{(t)}$ are indeed hidden stable-flow chasing in \cref{lemma:setuphidden}.
Finally, the $\mathcal{D}^{(HSFC)}$ rebuilds itself every $\eps m$ iterations, after which the residual cost $\bc^\top\bf - F^*$ might have changed by a $1+\eps$ factor. It terminates when $\bc^\top\bf - F^* \le (mU)^{-10}$, which happens within $m^{1+o(1)}$ iterations.

To analyze the progress of the algorithm, we will show that \textsc{MinCostFlow} (\cref{algo:mincost}) satisfies the hypotheses of our main IPM result \cref{thm:IPM}. Thus, applying \cref{thm:IPM} shows that \textsc{MinCostFlow} (\cref{algo:mincost}) computes a mincost flow to high accuracy in $\O(m\kappa^{-2})$ iterations for some $\kappa = m^{-o(1)}$.

We first note that $\bell^{(t)}$ and $\bg^{(t)}$ are approximately correct lengths and gradients at all times.
\begin{lemma}[Stability in \textsc{MinCostFlow}]
\label{lemma:gl}
During a call to \textsc{MinCostFlow} (\cref{algo:mincost}), for $\bf^{(t)} \gets \bf^{(0)} + \sum_{i\in[s]}\bf_i$, we have $\bell^{(t)} \approx_{1.1} \bell(\bf^{(t)})$, for $r$ defined in line \ref{line:initr}, $r \approx_{1+\eps} \bc^\top\bf^{(t)} - F^*$, and
\[ \left\|\mL(\bf^{(t)})^{-1}(\bg^{(t)} - (\bc^\top\bf^{(t)}-F^*)/r \cdot \bg(\bf^{(t)})) \right\|_\infty \le 10s\eps = \alpha\kappa/100. \]
\end{lemma}
\begin{proof}
To show $\bell^{(t)} \approx_{1.1} \bell(\bf^{(t)})$ it suffices to check that $\bf \assign \wt{\bf}^{(t)}$ and $\bar{\bf} \assign \bf^{(t)}$ satisfy the hypotheses of \cref{lemma:stablebell},
precisely $\|\mL(\bf^{(t)})(\bf^{(t)} - \wt{\bf}^{(t)})\|_\infty \le s\eps$. Indeed, this follows directly by the guarantees of \textsc{Detect} in \cref{algo:LCT} and the fact there are $s$ trees, because if no tree returned $e$, then the total error is at most $s\eps$.

To show the bound on the gradient, we use \cref{lemma:stablebg}. Because we have argue above that $\|\mL(\bf^{(t)})^{-1}(\bf^{(t)} - \wt{\bf}^{(t)})\|_\infty \le s\eps$, it suffices to check that $r \approx_{1+\eps} \bc^\top\bf^{(t)} - F^*$. Recall that $r$ is reset every $\lfloor \eps m\rfloor$ iterations in line \ref{line:initr} of \cref{algo:mincost}. For the scaled circulation $\bDelta$ in line \ref{line:scale}, \cref{lemma:stablebc}, for $\wt{\bg} = \bg^{(t)}$ and $\wt{\bell} = \bell^{(t)}$ in \cref{algo:mincost}, tells us
\[ \frac{|\bc^\top\bDelta|}{\bc^\top\bf^{(t)} - F^*} \le |\bg^{(t)}\bDelta|/(\kappa m) \le \alpha^2\kappa/(800m) \le 1/(800m) ,\]
where the hypotheses of \cref{lemma:stablebc} are satisfied because of the guarantee of $\mathcal{D}^{(HSFC)}.\textsc{Query}()$ (\cref{thm:MMCHiddenStableFlow}), and we used the bound on $|\bg^{(t)}\bDelta|$ from line \ref{line:scale} of \cref{algo:mincost}. Hence over $\eps m$ iterations, $\bc^\top\bf^{(t)} - F^*$ can change by at most a $(1+1/(800m))^{\eps m} \le 1+\eps$ factor, as desired.
\end{proof}

Our next goal is to define circulations $\bc^{(t)}$ and upper bounds $\bw^{(t)}$ to make $\bg^{(t)},\bell^{(t)},\bc^{(t)},\bw^{(t)}$ as defined in \textsc{MinCostFlow} (\cref{algo:mincost}) satisfy the hidden stable-flow chasing property. This shows that the solutions $\bDelta$ returned by the data structure have a good ratio.
\begin{lemma}
\label{lemma:setuphidden}
Let $\bg^{(t)},\bell^{(t)},U^{(t)}$ be defined as in an execution of \textsc{MinCostFlow} (\cref{algo:mincost}). For $\bf^* \defeq \argmin_{\mB^\top\bf=\bd} \bc^\top\bf$, let $\bc^{(t)} \defeq \bf^* - \bf$ and $\bw^{(t)} = 50 + |\bell^{(t)} \circ \bc^{(t)}|$. Then $\bg^{(t)},\bell^{(t)},U^{(t)}$ satisfy the hidden stable-flow chasing (\cref{def:hiddenStableFlowChasing}) with circulations $\bc^{(t)}$ and upper bounds $\bw^{(t)}$.
\end{lemma}
\begin{proof}
We check each item of \cref{def:hiddenStableFlowChasing} carefully. For the circulation condition in item \ref{item:circulation}, note that $\mB^\top\bc^{(t)} = \mB^\top\bf^* - \mB^\top\bf^{(t)} = \bd - \bd = 0$ because $\bf^*$ and $\bf^{(t)}$ both route the demand $\bd$. For the width condition in item \ref{item:width}, by the definition of $\bw^{(t)} = 50 + |\bell^{(t)} \circ \bc^{(t)}|$ we trivially have $|\bell^{(t)} \circ \bc^{(t)}| \le \bw^{(t)}$ coordinate-wise. 

To check that the upper bounds $\bw^{(t)}$ are stable (item \ref{item:widthstable}), for an edge $e$ let $t' \in [\last^{(t)}_e,t]$ be so that $e$ was not updated by any $U^{(t)}$ since stage $t'$.
By the guarantees of \textsc{Detect} we know that \[ \bell^{(t)}_e|\bf^{(t)}_e - \bf^{(t')}_e| \le \eps. \]
Hence we get that
\begin{align*}
    |\bw^{(t)}_e - \bw^{(t')}_e| \overset{(i)}{=} \bell^{(t)}_e|\bf^{(t)}_e - \bf^{(t')}_e| \le \eps \overset{(ii)}{\le} 1/100|\bw^{(t')}_e|,
\end{align*}
Here, $(i)$ follows because $\bell^{(t)}_e = \bell^{(\last^{(t)}_e)}_e$ because $e$ was not updated in any $U^{(t)}$ since stage $\last^{(t)}_e$ and $(ii)$ is because $\bw^{(t')}_e \ge 50$ for all $t'$. Hence $|\bw_e^{(t)}| \le 1.1|\bw^{(t')}_e|$ as desired.

To check that the lengths and widths are quasipolynomially bounded for item \ref{item:quasipoly}, note that
\[ \min_{e \in E}\{\bu^+_e - \bf_e, \bf_e - \bu^-_e\} \ge \Phi(\bf^{(t)})^{-1/\alpha} \ge \exp(-O(\log^2 mU)), \]
by our assumption that $\Phi(\bf^{(t)}) \le \O(m)$ always.
Also, $|\bc^{(t)}_e| \le O(U)$ for all $e \in E$. This shows that $\log \bell^{(t)}_e, \log \bw^{(t)}_e \le O(\log^2 mU)$ for all $e \in E$. The lower bound $\bw^{(t)}_e \ge 50$ is by definition. Also, the lower bound $\bell^{(t)}_e \ge 1/U^{1+\alpha} \ge 1/U^2$ is trivial by the definition of $\bell^{(t)}_e$.
\end{proof}
As a result, we deduce that $\mathcal{D}^{(HSFC)}$ succeeds whp.
This allows us to prove that \textsc{MinCostFlow} satisfies the hypotheses of \cref{algo:mincost}, and allows us to bound the total number of iterations.
\begin{lemma}
\label{cor:iteration}
An execution of \textsc{MinCostFlow} (\cref{algo:mincost}) runs for $\O(m\kappa^{-2}\alpha^{-2})$ iterations.
\end{lemma}
\begin{proof}
We will define
$\wt{\bg}, \wt{\bell}, \bDelta, \eta$ and flows $\bf^{(t)}$ to show that an execution of \textsc{MinCostFlow} (\cref{algo:mincost}) satisfies the hypotheses of \cref{thm:IPM}, which implies that \textsc{MinCostFlow} (\cref{algo:mincost}) terminates in $\O(m\kappa^{-2}\alpha^{-2})$ iterations.

For $\bf^{(t)}$ as defined in \textsc{MinCostFlow} (\cref{algo:mincost}), note that $\Phi(\bf^{(t)}) \le 200m \log(mU)$ at all times. This is because it holds at the initial point $\bf^{(0)}$ (\cref{lemma:initialpoint}) and the potential is decreasing.

Next we define $\wt{\bg}, \wt{\bell}$, the approximate gradients and lengths.
Let $\wt{\bg} = r/(\bc^\top\bf^{(t)}-F^*) \cdot \bg^{(t)}$ and $\wt{\bell} = \bell^{(t)}$ for $\bg^{(t)}, \bell^{(t)}$ as defined in \textsc{MinCostFlow} (\cref{algo:mincost}). By \cref{lemma:gl} we know that $\wt{\bell} \approx_{1.1} \bell(\bf^{(t)})$ and
\begin{align*} 
\|\mL(\bf^{(t)})^{-1}(\wt{\bg} - \bg(\bf^{(t)}))\|_\infty =~ & r/(\bc^\top\bf^{(t)}-F^*) \left\|\mL(\bf^{(t)})^{-1}(\bg^{(t)} - (\bc^\top\bf^{(t)}-F^*)/r \cdot \bg(\bf^{(t)})) \right\|_\infty \\ \overset{(i)}{\le}~&(1+\eps)\alpha\kappa/100 \le \alpha\kappa/50,
\end{align*}
where $(i)$ uses the bounds on $r$ and $\bg^{(t)}$ in \cref{lemma:gl}. Thus for $\bc^{(t)} = \bf^* - \bf^{(t)}$ as in \cref{lemma:setuphidden},
\begin{align*}
    \frac{\bg^{(t)\top}\bc^{(t)}}{\|\bw^{(t)}\|_1} = (\bc^\top\bf^{(t)}-F^*)/r \cdot \frac{\wt{\bg}^\top(\bf^* - \bf^{(t)})}{50m + \|\mL^{(t)}\bc^{(t)}\|_1} \overset{(i)}{\le} -(1-\eps)\alpha/4 \le -\alpha/8,
\end{align*}
where $(i)$ follows from the first item of \cref{thm:IPM} (\cref{lemma:optCirculationQuality}) and $r \approx_{1+\eps} \bc^\top\bf^{(t)}-F^*$ from \cref{lemma:gl}.
Hence, by the guarantees of $\mathcal{D}^{(HSFC)}.\textsc{Query}()$ (\cref{thm:MMCHiddenStableFlow}) as called in line \ref{line:query} of \textsc{MinCostFlow} (\cref{algo:mincost}), we know that whp.
\begin{align} \frac{\bg^{(t)\top}\bDelta}{\|\mL^{(t)}\bDelta\|_1} \le -\kappa\alpha/8. \label{eq:qualitybound} \end{align}
Thus, for the scaling $\eta$ as in \textsc{MinCostFlow} (\cref{algo:mincost}), \cref{thm:IPM} (where we change $\kappa$ to $\alpha\kappa/8$ for this setting) shows that the algorithm computes a high-accuracy flow in $\O(m\kappa^{-2}\alpha^{-2})$ iterations.
\end{proof}
The final piece is to analyze the runtime of \textsc{MinCostFlow} (\cref{algo:mincost}) by bounding the total size of the update batches $U^{(t)}$ as defined in line \ref{line:detect}.
\begin{lemma}
\label{lemma:sizedt}
Consider a call to \textsc{MinCostFlow} (\cref{algo:mincost}) and let $U^{(t)}$ be as in line \ref{line:detect}. Then $\sum_t |U^{(t)}| \le \O(m\kappa^{-2}\alpha^{-2}\eps^{-1}) \le m^{1+o(1)}$.
\end{lemma}
\begin{proof}
Because of the guarantee in line \ref{line:scale} of \cref{algo:mincost}, we know that $|\bg^{(t)\top}\bDelta| = \kappa^2\alpha^2/800$. Additionally by the guarantees of $\mathcal{D}^{(HSFC)}.\textsc{Query}()$ (\cref{thm:MMCHiddenStableFlow}) as called in line \ref{line:query} of \textsc{MinCostFlow} (\cref{algo:mincost}), we get that (see \eqref{eq:qualitybound}) \[ \|\mL^{(t)}\bDelta\|_1 \le 8/(\kappa\alpha)|\bg^{(t)\top}\bDelta| \le 1. \]
Hence the sum of $\|\mL^{(t)}\bDelta\|_1$ over all iterations is $\O(m\kappa^{-2}\alpha^{-2})$ by our bound on the number of iterations in \cref{cor:iteration}. Each time an update on edge $e$ in $U^{(t)}$ it contributes $\Omega(\eps)$ to this sum by the guarantees of \textsc{Detect} in \cref{algo:LCT}. Hence $\sum_t |U^{(t)}| \le \O(m\kappa^{-2}\alpha^{-2}\eps^{-1})$.
\end{proof}
Combining these pieces shows our main result \cref{thm:main} on computing min-cost flows.
\begin{proof}[Proof of \cref{thm:main}]
Given a min-cost flow instance, we will first use \cref{lemma:initialpoint,lemma:finalpoint} to compute the initial flow $\bf^{(0)}$. 
We then run $\textsc{MinCostFlow}$ with initial flow $\bf^{(0)}$, and use \cref{lemma:finalpoint} to round to an exact min-cost flow.

The only remaining piece to analyze is the runtime.
The main component of the runtime is the data structure $\mathcal{D}^{(HSFC)}$ (\cref{thm:MMCHiddenStableFlow}).
The inputs $\bg^{(t)},\bell^{(t)},U^{(t)}$ to $\mathcal{D}^{(HSFC)}$ satisfy the hidden stable-flow chasing property by \cref{lemma:setuphidden}.
Hence the data structure $\mathcal{D}^{(HSFC)}$ runs in total time $\eps^{-1}(m+Q)m^{o(1)} = m^{1+o(1)}$ time by \cref{thm:MMCHiddenStableFlow}, because the data structure reinitializes $\O(\eps^{-1})$ times (in line \ref{line:reinit}), and $Q = \sum_t |U^{(t)}| \le m^{1+o(1)}$ by \cref{lemma:sizedt}.

The remaining runtime components can be handled in $sm^{o(1)} = m^{o(1)}$ time per operation by using dynamic trees (\cref{algo:LCT}), as there are $s$ trees, and the fact that the cycle in line \ref{line:query} of \textsc{MinCostFlow} (\cref{algo:mincost}) is represented by $\exp(O(\log^{7/8}m\log\log m)) \le m^{o(1)}$ paths, so the total runtime is $m^{1+o(1)}$ as desired.
\end{proof}

\label{lastPageOfMaxFlow}

%% file: generalconvex.tex
\section{General Convex Objectives}
\label{sec:general}
The goal of this section is to extend our algorithms to the setting of optimizing single commodity flows for general decomposable convex objectives.
\subsection{General Setup for General Convex Objectives}
Formally, for a graph $G = (V, E)$ let $h_e: \R \to \R \cup \{+\infty\}$ be convex functions. For a flow $\bf$ let $h(\bf) \defeq \sum_{e \in E} h_e(\bf_e)$. Our goal is to minimize $h(\bf)$ over all flows $\bf$ routing a demand $\bd$, i.e. $\mB^\top\bf = \bd$.

We cast this in the setting of empirical risk minimization (see~\cite{LSZ19}) by introducing new variables $\by \in \R^E$ and convex sets $\mathcal{X}_e \defeq \{(f, y) : y < h_e(f)\}$:
\begin{align}
    \min_{\mB^\top\bf=\bd} h(\bf) = \min_{\substack{\mB^\top\bf=\bd \\ \by \in \R^E : \by_e \le h(\bf_e) \forall e \in E}} {\vone}^\top\by = \min_{\substack{\mB^\top\bf=\bd \\ (\bf_e,\by_e) \in \mathcal{X}_e \forall e \in E}} {\vone}^\top\by. \label{eq:erm}
\end{align}
Let $F^* \defeq \min_{\mB^\top\bf=\bd} h(\bf)$. We will assume that we have access to gradients and Hessians of $\nu$-self-concordant barriers for $\mathcal{X}_e,$ $\psi_e: \mathcal{X}_e \to \R$.
Explicit self-concordant barriers are known for several natural objectives $h_e$ (see e.g. Chapter 9.6. of~\cite{BV04:book}, or Section 4 of~\cite{N98:notes}),
and it is known that every subset $\mathcal{X} \subseteq \R^n$ admits an $n$-self-concordant barrier~\cite{N98:notes,N04:book,Chewi21:arxiv,LY21}.

We now formally introduce the definition of self-concordance.
\begin{definition}[$\nu$-self-concordance {\cite[Definition 4.2.2]{N04:book}}]
\label{def:selfcon}
We say that a function $\psi: \mathcal{X} \to \R$ on an open set $\mathcal{X} \subseteq \R^n$ is a self-concordant barrier if $\psi$ is convex, $\psi(\bx) \to \infty$ as $\bx$ approaches the boundary of $\mathcal{X}$, and for all $\bx \in \mathcal{X}$ and $\bv \in \R^n$
\[ \left|\g^3 \psi(\bx)\left[\bv, \bv, \bv\right]\right|
\le
2\left(\bv^\top\g^2\psi\left(\bx\right)\bv\right)^{3/2}.
\]
We say that $f$ is $\nu$-self-concordant for some $\nu > 0$ if $f$ is self-concordant and for all $\bx \in \mathcal{X}$ and $\bv \in \R^n$ we have
\[
\l \g \psi\left(\bx \right), \bv\r^2
\le
\nu \bv^\top\g^2\psi\left(\bx\right)\bv.
\]
\end{definition}

Analyzing the runtime of our algorithm requires assuming that various quantities are quasipolynomially bounded such as the starting flow, demands, and convex objectives, and the underlying self-concordant barriers.
\begin{assump}
\label{assump}
We make the following assumptions for our method, for a parameter $K = \O(1)$.
\begin{enumerate}
    \item We have access in $\O(1)$ time to gradients/Hessians of the self-concordant barriers $\psi_e(\bf_e,\by_e)$.
    \item All capacities, demands, and costs are polynomially bounded, i.e. $|\bf_e| \le m^K$ for all $e$, $\|\bd\|_\infty \le m^K$, and $|h_e(x)| \le O(m^K+|x|^K)$ for all $x \in \R$.
    \item We shift the barriers $\Psi(\bf_e, \by_e)$ such that $\inf_{|\bf_e|,|\by_e|\le m^K} \Psi(\bf_e, \by_e) = 0.$ We can shift the barriers because that does not affect self-concordance. This implies that $\bzeta_e(\bf_e) \ge 1$ on the whole domain.
    \item There is a feasible flow $\bf^{(0)}$ and variables $|\by^{(0)}_e| \le m^K$ such that $m^{-K}\mI \pe \g^2 \psi_e(\bf^{(0)}_e, \by{(0)}_e) \pe m^K \mI$ for all $e$, and $\psi_e(\bf^{(0)}_e, \by^{(0)}_e) \le K$.
    \item The parameters $\alpha, \eps, \kappa$ used throughout are all less than $1/(1000\nu)$.
    \item \label{item:polybounded} The Hessian is quasipolynomially bounded as long as the function value is $\O(1)$ bounded, i.e. for all points $|\bf_e|, |\by_e| \le m^K$ with $\psi_e(\bf_e,\by_e) \le \O(1)$, we have $\g^2\psi_e(\bf_e,\by_e) \pe \exp(\log^{O(1)}m)\mI$.
\end{enumerate}
\end{assump}
We assume everything stated above for the remainder of the section. The final assumption in item~\ref{item:polybounded} is to ensure that all lengths/gradients encountered in the algorithm are bounded by $\exp(\log^{O(1)}m)$. This holds for all explicit $O(1)$-self-concordant barriers we have encountered, such as those for entropy-regularized optimal transport, matrix scaling, and normed flows.

We make direct use of the following lemmas from~\cite{N04:book}.
\begin{lemma}[{\cite[Theorem 4.2.4]{N04:book}}, first part]
\label{lem:GradDotBounded}
For a self concordant function $f$,
and any $\bx$ and $\by$ in its domain, we have
\[
\l \nabla f\left( \bx \right), \by - \bx\r
<
\nu.
\]
\end{lemma}

\begin{lemma}[{\cite[Theorem 4.1.7]{N04:book}}, first part]
\label{lem:GradStable}
For a self concordant function $f$,
and any $\bx$ and $\by$ in its domain, we have
\[
\l \nabla f\left(\by \right) - \nabla f\left(\bx \right),
\by - \bx\r
\geq
\frac{\norm{\bx - \by}_{\nabla^2 f\left( \bx \right)}^2}
{1 + \norm{\bx - \by}_{\nabla^2 f\left( \bx \right)}}.
\]
\end{lemma}

\begin{lemma}[{\cite[Theorem 4.1.6]{N04:book}}]
\label{lem:HessianStable}
For a self concordant function $f$,
and any $\bx$ and $\by$ in its domain such that
\[
\norm{\bx - \by}_{\nabla^2 f\left( \bx \right)} < 1
\]
we have
\[
\left( 1 - \norm{\bx - \by}_{\nabla^2 f\left( \bx \right)} \right)^2
\nabla^{2} f\left( \bx\right)
\preceq
\nabla^{2} f\left( \by\right)
\preceq
\frac{1}{\left( 1 - \norm{\bx - \by}_{\nabla^2 f\left( \bx \right)} \right)^2}
\nabla^{2} f\left( \bx\right)
\]
\end{lemma}

Fix some $\alpha \in (0, \nicefrac{1}{10}),$ set a path parameter $t$ and minimize the following objective
\[ \Psi_t(\bf, \by) \defeq t \cdot \vone^\top\by + \sum_{e \in E} \exp(\alpha\psi_e(\bf_e, \by_e)) = \sum_{e \in E} \left(t\by_e + \exp(\alpha\psi_e(\bf_e, \by_e))\right), \] over $\mB^\top\bf = \bd$.
This is analogous to our $\alpha$-power potential in 
Equation~\ref{eq:karmarkar} at the start of Section~\ref{sec:ipm}.

Note that for a fixed flow $\bf$, we can eliminate the variables $\by$ in the following way. We should set $\by_e = y_e(\bf_e)$ for $y_e(\bf_e) \defeq \argmin_y ty + \exp(\alpha\psi_e(\bf_e, y))$.
Thus we can write
\begin{align*} \min_{\mB^\top\bf=\bd, \by} \Psi_t(\bf, \by) = \min_{\mB^\top\bf=\bd} \sum_{e \in E} \left(ty_e(\bf_e) + \exp(\alpha\psi_e(\bf_e,y_e(\bf_e)))\right), \end{align*}
Let $\bzeta_e(\bf_e) \defeq \exp(\alpha\psi_e(\bf_e, y(\bf_e)))$, and $\zeta_e(\bf_e) = ty_e(\bf_e) + \bzeta_e(\bf_e)$ and define the potential
\begin{align}
    Z_t(\bf) \defeq \sum_{e \in E} \zeta_e(\bf_e). \label{eq:pot}
\end{align}
Our first main lemma (\cref{lemma:zetasc}) will be that up to scaling, the function $\zeta_e$ is self-concordant.
To show this, we start by studying the derivatives of the function $ty_e(\bf_e) + \zeta_e(\bf_e)$.
\begin{definition}
\label{def:partial}
For a function $\psi: \R^n \to \R$, and a sequence $(i_1,i_2,\dots,i_k) \in [n]^k$, define the mixed partials
\[ \psi_{x_{i_1},\dots,x_{i_k}} = \frac{\partial}{\partial x_{i_1}}\dots\frac{\partial}{\partial x_{i_k}}\psi. \]
\end{definition}

\begin{lemma}
\label{lemma:deriv}
Let $f: \mathcal{X} \to \R$ be a convex function on an open set $\mathcal{X} \subseteq \R^2$. For $x \in \R$ let $y(x) \defeq \mathrm{argmin}_y \psi(x, y)$. Let $\zeta(x) \defeq \psi(x, y(x))$. Then for $\bv \defeq \begin{bmatrix} 1 \\ y'(x) \end{bmatrix}$,
\begin{align}
    \zeta'(x) &= \l \g \psi(x, y(x)), \bv \r = f_x(x,y(x)) \label{eq:firstder}, \\
    \zeta''(x) &= \bv^\top \g^2 \psi(x, y(x)) \bv \label{eq:secondder}, \\
    \zeta'''(x) &= \g^3 \psi(x, y(x))[\bv, \bv, \bv] \label{eq:thirdder}.
    \\
    y'(x) &= -\psi_{xy}(x,y(x))/\psi_{yy}(x,y(x)). \label{eq:yder}
\end{align}
\end{lemma}
\begin{proof}
Note that $\psi_y(x, y(x)) = 0$ by the optimality of $y(x)$.
By the chain rule for total derivatives
\[ \zeta'(x) = f_x(x, y(x)) + \psi_y(x, y(x))y'(x) = \l \g \psi(x, y(x)), \bv \r \] which shows the first equality \eqref{eq:firstder}.

Taking the derivative of the first equality of \eqref{eq:firstder} gives us
\begin{align*}
\zeta''(x) &= \psi_{xx}(x, y(x)) + 2\psi_{xy}(x, y(x))y'(x) + \psi_{yy}(x, y(x))y'(x)^2 + f_y(x, y(x))y''(x) \\
&= \psi_{xx}(x, y(x)) + 2\psi_{xy}(x, y(x))y'(x) + \psi_{yy}(x, y(x))y'(x)^2 = \bv^\top \g^2 \psi(x, y(x)) \bv,
\end{align*}
where we have used that $\psi_y(x, y(x)) = 0$. This shows \eqref{eq:secondder}.

Taking the derivative of \eqref{eq:secondder} gives
\begin{align*} \zeta'''(x) =~&\psi_{xxx}(x, y(x)) + 3\psi_{xxy}(x, y(x))y'(x) + 3\psi_{xyy}(x, y(x))y'(x)^2 + \psi_{yyy}(x, y(x))y'(x)^3 \\ +~&2\left(\psi_{xy}(x, y(x)) + \psi_{yy}(x, y(x))y'(x)\right)y''(x).
\end{align*}
However, note that taking the derivative of the identity $\psi_y(x, y(x)) = 0$
gives us
\[
\psi_{xy}\left(x, y\left(x\right)\right)
+
\psi_{yy}\left(x, y\left(x\right)\right)y'\left(x\right)
=
0.
\]
Plugging this into the above gives us
\begin{align*} 
\zeta'''(x) & =\psi_{xxx}(x, y(x)) + 3\psi_{xxy}(x, y(x))y'(x) + 3\psi_{xyy}(x, y(x))y'(x)^2 + \psi_{yyy}(x, y(x))y'(x)^3 \\ 
&= \g^3 \psi(x, y(x))[\bv, \bv, \bv]
\end{align*}
as desired.

To show \eqref{eq:yder}, recall that $f_y(x,y(x)) = 0$. Taking a derivative of this in $x$ gives
\[ \psi_{xy}(x,y(x)) + \psi_{yy}(x,y(x))y'(x) = 0, \]
which rearranges to \eqref{eq:yder} as desired.
\end{proof}

Now we show that the $\zeta_e$ functions are self-concordant. Note that we do not claim that $\zeta_e$ is $\nu$-self-concordant, just self-concordant.
\begin{lemma}
\label{lemma:zetasc}
For all $e \in E$, $\alpha^{-1}\zeta_e/4$ is a self-concordant function.
\end{lemma}
\begin{proof}
We calculate
\begin{align*}
    \alpha^{-1}\zeta_e'''(\bf_e) &= \alpha^{-1}\g^3(\exp(\alpha\psi_e(\bf_e,y_e(\bf_e))))[\bv,\bv,\bv] \\ &= \Big( \g^3\psi_e(\bf_e,y_e(\bf_e))[\bv,\bv,\bv] \\ ~&+ 3\alpha \g^2\psi_e(\bf_e,y_e(\bf_e))[\bv,\bv] + \alpha^2 \l \g\psi_e(\bf_e,y_e(\bf_e)), \bv\r^3 \Big)\bzeta_e(\bf_e),
\end{align*}
where $\bv = \begin{bmatrix} 1 \\ y_e'(\bf_e) \end{bmatrix}$. By $\nu$-self-concordance of $\psi_e$, we can bound the previous expression by
\[ \alpha^{-1}\zeta_e'''(\bf_e) \le 4(\g^2\psi_e(\bf_e,y_e(\bf_e))[\bv,\bv]\bzeta_e(\bf_e))^{3/2} \le 4(\alpha^{-1}\zeta_e''(\bf_e))^{3/2}, \]
where the last inequality follows by the formula for $\bell(\bf_e)$. Scaling by a factor of $4$ completes the proof.
\end{proof}

Our algorithm for solving \eqref{eq:erm} will fix a value of $t$ and reduce the value of the potential \eqref{eq:pot} until $1^\top \by \le F^* + 50\nu m/t$. Once this holds, we will double the value of $t$ and start a new phase. Each phase will requires approximately $m^{1+o(1)}\alpha^{-2}$ iterations.

We now formally define the gradients and lengths. The gradient is
\begin{align}
    \bg(\bf)_e \defeq [\g Z_t(\bf)]_e = \zeta_e'(\bf_e), \label{eq:generalg}
\end{align}
and the lengths we define as
\begin{align}
    \bell(\bf)_e &\defeq \sqrt{\alpha^{-1}\zeta_e''(\bf_e)} \nonumber \\
    &= \sqrt{(\bv^\top \g^2\psi_e(\bf_e,y_e(\bf_e)) \bv + \alpha \l \g\psi_e(\bf_e,y_e(\bf_e)), \bv\r^2)\bzeta_e(\bf_e)} \enspace \text{ for } \enspace \bv \defeq \begin{bmatrix} 1 \\ y_e'(\bf_e) \end{bmatrix}. \label{eq:generall}
\end{align}
Here, the equality starting line 2 follows from \cref{lemma:deriv}, \eqref{eq:secondder} applied to the function $\zeta_e(\bf_e)$.

Define $\bf_t^* \defeq \argmin_{\mB^\top\bf=\bd} Z_t(\bf)$. We now bound the optimality gap of $\bf_t^*$. This will ultimately show that if $Z_t(\bf_t^*) - F^*$ is much larger than $4m$, then we can reduce the potential by $m^{-o(1)}$ in a single step.
\begin{lemma}[Optimality gap]
\label{lemma:dualgap}
For sufficiently small $\alpha = \wt{\Omega}(1/\log\max(t,2))$, we have that \[ Z_t(\bf_t^*) - F^* \le 4m. \]
\end{lemma}
\begin{proof}
Recall that $\bf^{(0)}, \by^{(0)}$ are the initially feasible points. Let $\bf^*$ be the optimal flow and $\by^*_e = h_e(\bf^*_e)$.
We will upper bound $Z_t(\bf)$ for a flow $\bf = \beta \bf^{(0)} + (1-\beta)\bf^*$ and $\by = \beta \by^{(0)} + (1-\beta)\by^*$ for a parameter $\beta \in [0,1]$ chosen later. 
Define $Q = h(\bf^{(0)}) - F^*$, the optimality gap of the original flow. By our assumptions, we know that $\log Q = \O(1)$. We set $\beta = \min(1, m/(tQ))$. For $s \in [0, 1]$ let $\bf^{(s)} \defeq \bf^{(0)} + s(\bf^* - \bf^{(0)})$ and $\by^{(s)} \defeq \by^{(0)} + s(\by^* - \by^{(0)})$.

Define the function $\bar{\psi}_e(s) \defeq \psi_e(\bf^{(s)}_e, \by^{(s)}_e)$, which is $\nu$-self-concordant as it is the restriction of $\psi_e$ onto a line. By self-concordance, $\bar{\psi}_e''(s) \ge \bar{\psi}_e'''(s)/(2\sqrt{\bar{\psi}_e''(s)})$, so integrating both sides gives
\[ \bar{\psi}_e'(1-\beta) - \bar{\psi}_e'(0) \ge \sqrt{\bar{\psi}_e''(1-\beta)} - \sqrt{\bar{\psi}_e''(0)}. \]
By $\nu$-self-concordance and \cite[Theorem 4.2.4]{N04:book}
(Lemma~\ref{lem:GradDotBounded}), we know $\beta\bar{\psi}'_e(1-\beta) \le \nu$, and $ \bar{\psi}_e'(0) \ge -\sqrt{\nu \bar{\psi}_e''(0)}$.
Rearranging this gives us
\begin{align} \sqrt{\bar{\psi}_e''(1-\beta)} \le 2\sqrt{\nu \bar{\psi}_e''(0)} + \nu/\beta. \label{eq:firstpsibound} \end{align}
We use this to bound $\bar{\psi}_e(\bf_e, \by_e) = \bar{\psi}_e(1-\beta)$.
We can assume that $\bar{\psi}_e'(0) \ge 0$ because $\bar{\psi}_e'(s)$ is an increasing function, so we might as well start our integration at the minimizer on the line.

Now, by rearranging the $\nu$-self-concordance condition we get
\[ \bar{\psi}_e'(s) \le \frac{\nu \bar{\psi}''_e(s)}{\bar{\psi}_e'(s) + 1} + 1. \] Integrating both sides gives us
\[ \bar{\psi}_e(1-\beta) - \bar{\psi}_e(0) \le \nu\log\left(\frac{\bar{\psi}_e'(1-\beta) + 1}{\bar{\psi}_e'(0) + 1}\right) + 1 \le \nu \log\left(\sqrt{\nu\bar{\psi}_e''(1-\beta)}+1\right) + 1. \] 
Recall by our assumption that $\bar{\psi}_e(0) = \psi_e(\bf^{(0)}_e,\by^{(0)}_e) \le K = \O(1)$. Additionally by \eqref{eq:firstpsibound} the RHS of the above expression is also bounded by $\O(1) + \log (\nu/\beta) \le \O(1) + \max(0, O(\log t))$ because $\log Q = \O(1)$. Thus, we get that $\bar{\psi}_e(1-\beta) = \O(1) + \max(0, O(\log t))$,
and in turn for $\alpha = \wt{\Omega}(1/\log\max(t,2))$, 
\[ Z_t(\bf) \le t \cdot 1^\top \by + \sum_{e \in E} \exp(\alpha\psi_e(\bf_e,\by_e)) \le t\beta Q + 2m \le 4m. \]
\end{proof}

Using this, we will bound the quality of the solution $\bf_t^* - \bf$, i.e. how negative its gradient is compared to its total length.
\begin{lemma}
\label{lemma:qualggeneral}
Let $\alpha$ be set as in \cref{lemma:dualgap}.
If $1^\top \by(\bf) - F^* \ge 10m/t$ and $\|\mL(\bf)^{-1}(\wt{\bg} - \bg(\bf)\|_\infty \le \eps$ for $\eps \le \alpha/100$ then
\[ \wt{\bg}^\top(\bf_t^* - \bf) \le -\alpha/4 \cdot \|\mL(\bf)(\bf_t^* - \bf)\|_1 - m/4. \]
\end{lemma}
\begin{proof}
We first handle the case where $\alpha\|\mL(\bf)(\bf_t^* - \bf)\|_1 \le 10m$. In this case, by the convexity of $Z_t(\bf)$, we get
\[ \bg(\bf)^\top(\bf_t^* - \bf)
\le Z_t(\bf_t^*) - Z_t(\bf)
=
\left(Z_t(\bf_t^*)  - F^{*} \right) - 
\left( Z_t(\bf) - F^{*}\right)
\]
which upon applying Lemma~\ref{lemma:dualgap} to the first term,
and the assumption of $1^\top \by(\bf) - F^* \ge 10m/t$ to the second gives
\[
\le 4m - t \cdot 10m/t \le -6m.
\]
Thus, we get
\begin{align*}
\bg(\bf)^\top(\bf_t^* - \bf)
&\le \bg(\bf)^\top(\bf_t^* - \bf) + \|\mL(\bf)^{-1}(\wt{\bg} - \bg(\bf))\|_\infty\|\mL(\bf)(\bf_t^* - \bf)\|_1 \\ &\le -6m + \eps \cdot 10m/\alpha \le -5m \le -\alpha/4 \cdot \|\mL(\bf)(\bf_t^* - \bf)\|_1 - m/4.
\end{align*}
Now, we handle the case where $\alpha\|\mL(\bf)(\bf_t^* - \bf)\|_1 \ge 10m$.
Consider the function
\[
\hat{\zeta}_e(\bf_e) = \alpha^{-1}\zeta_e(\bf_e)/4,
\]
which by \cref{lemma:zetasc} is self-concordant.
Invoking \cite[Theorem 4.1.7]{N04:book} (Lemma~\ref{lem:GradStable})
on this function, we get for all $e$,
\begin{align*} (\hat{\zeta}_e'([\bf_t^*]_e) - \hat{\zeta}_e'(\bf_e))([\bf_t^*]_e - \bf_e) &\ge \frac{\hat{\zeta}_e(\bf_e)''|[\bf_t^*]_e - \bf_e|^2}{1 + \sqrt{\hat{\zeta}_e(\bf_e)''}|[\bf_t^*]_e - \bf_e|}
\\ &\ge \sqrt{\hat{\zeta}_e(\bf_e)''}|[\bf_t^*]_e - \bf_e| - 1.
\end{align*}
Rearranging the above equation gives us
\begin{align*}
    \zeta_e'(\bf_e)([\bf_t^*]_e - \bf_e) &\le 4\alpha\left(\hat{\zeta}_e'([\bf_t^*]_e)([\bf_t^*]_e - \bf_e)-\sqrt{\hat{\zeta}_e''(\bf_e)}|[\bf_t^*]_e - \bf_e| + 1\right)\\
    &= 4\alpha\hat{\zeta}_e'([\bf_t^*]_e)([\bf_t^*]_e - \bf_e) - \alpha/2 \cdot \bell(\bf_e)|[\bf_t^*]_e - \bf_e| + 4\alpha.
\end{align*}
By the optimality of $\bf_t^*$ we know that $\bg(\bf_t^*) = \mB \bz$ for some $\bz$,
so the first term is $0$ because the difference
$\bf_t^* - \bf$ is a circulation.
Hence
\[
\bg(\bf)^\top(\bf_t^* - \bf_e)
\le
-\alpha/2 \cdot \|\mL(\bf)(\bf_t^* - \bf)\|_1 + 4\alpha m.
\]
which upon incorporating errors from the approximate gradient gives
\begin{align*} \wt{\bg}^\top(\bf_t^* - \bf) &\le -\alpha/2 \cdot \|\mL(\bf)(\bf_t^* - \bf)\|_1 + 4\alpha m + \|\mL(\bf)^{-1}(\bf_t^* - \bf)\|_1 \|\mL(\bf)(\wt{\bg} - \bg(\bf))\|_\infty \\ &\le (-\alpha/2+\eps) \cdot \|\mL(\bf)(\bf_t^* - \bf)\|_1 + 4\alpha m \le -\alpha/4 \cdot \|\mL(\bf)(\bf_t^* - \bf)\|_1 - m/4
\end{align*}
as long as $\alpha\|\mL(\bf)(\bf_t^* - \bf)\|_1 \ge 10m$, for the choice of $\alpha$.
\end{proof}

We move towards analyzing how a step $\bDelta$ decreases the potential $Z_t$. We start by showing that the gradients and lengths are stable in a Hessian ball.
\begin{lemma}[Stability bounds]
\label{lemma:generalstable}
For a flow $\bf$ and vector $\bf \in \R^E$ satisfying $\|\mL(\bf)(\bf - \bar{\bf})\|_\infty \le \eps$ for $\eps < 1/1000$, then $\bell(\bf) \approx_{1+5\eps} \bell(\bar{\bf})$, $\|\mL(\bf)^{-1}(\bg(\bf) - \bg(\bar{\bf}))\|_\infty \le \eps$.
\end{lemma}
\begin{proof}
The stability of lengths follows from self-concordance of $\alpha^{-1}\zeta_e/4$ shown in \cref{lemma:zetasc},
plus the Hessian stability of such functions shown
in \cite[Theorem 4.1.6]{N98:notes} (Lemma~\ref{lem:HessianStable}).
To analyze gradient stability, let $\phi(s) \defeq \zeta_e'(\bf^{(s)})$. Now
\begin{align*}
    |\phi(s)'| = |\bar{\bf}_e - \bf_e|\zeta_e''(\bf^{(s)}) = \alpha|\bar{\bf}_e - \bf_e|\bell(\bf^{(s)})^2 \le 2\alpha\bell(\bf_e)^2|\bar{\bf}_e - \bf_e| \le 2\alpha\eps\bell_e(\bf).
\end{align*}
Hence $|\bell_e(\bf)^{-1}(\phi(1) - \phi(0))| \le 2\alpha\eps \le \eps$ as desired.
\end{proof}
We can use \cref{lemma:generalstable} to show that a good quality circlation $\bDelta$ decreases the potential $Z_t$.
\begin{lemma}
\label{lemma:zdecrease}
Let $\wt{\mL} \approx_2 \mL(\bf)$ and $\|\mL(\bf)^{-1}(\wt{\bg} - \bg(\bf))\|_\infty \le \eps$ for $\eps < \kappa/100$. If circulation $\bDelta$ satisfies $\wt{\bg}^\top\bDelta/\|\wt{\mL}\bDelta\|_1 \le -\kappa$, then for $\eta > 0$ chosen so that $\eta\wt{\bg}^\top\bDelta = -\kappa^2/50$ satisfies
\[ Z_t(\bf + \eta\bDelta) \le Z_t(\bf) - \kappa^2/100. \]
\end{lemma}
\begin{proof}
Let $\bar{\bDelta} = \eta\bDelta$. Define $\bf^{(s)} = \bf + s\bar{\bDelta}$, and $\phi(s) = Z_t(\bf^{(s)})$. By Taylor's theorem, we know
\begin{align*} Z_t(\bf + \eta\bDelta) - Z_t(\bf) &= \phi(1) - \phi(0) \le \phi'(0) + \max_{s \in [0,1]} \phi''(s)^2/2 \\ 
&\overset{(i)}{\le} \eta \bg(\bf)^\top \bDelta + \bar{\bDelta}^\top \g^2 Z_t(\bf) \bar{\bDelta} \\
&\le \eta\bar{\bg}^\top \bDelta + (\bg(\bf)-\bar{\bg})^\top\bar{\bDelta} + \alpha\|\mL(\bf)\bar{\bDelta}\|_2^2 \\
&\le -\kappa^2/50 + \|\mL(\bf)^{-1}(\bg(\bf)-\bar{\bg})\|_\infty\|\mL(\bf)\bar{\bDelta}\|_1 + \|\mL(\bf)\bar{\bDelta}\|_1^2 \\ &\le -\kappa^2/50 + \eps\kappa/50 + (\kappa/50)^2 \le -\kappa^2/100,
\end{align*}
where $(i)$ follows from length stability in \cref{lemma:generalstable}.
\end{proof}

We can now state and show our main result on optimizing flows under general convex objectives.
\begin{theorem}[General convex flows]
\label{thm:general}
Let $G$ be a graph with $m$ edges, and let $\bd$ be a demand. Given convex functions $h_e: \R \to \R \cup \{+\infty\}$ and $\nu$-self-concordant barriers $\psi_e(f,y)$ on the domain $\{(f,y) : y \le h_e(f)\}$ satisfying the guarantees of \cref{assump}, there is an algorithm that runs in $m^{1+o(1)}$ time and outputs a flow $\bf$ with $\mB^\top\bf=\bd$ and for any fixed constant $C > 0$, \[ h(\bf) \le \min_{\mB^\top\bf^*=\bd} h(\bf^*) + \exp(-\log^Cm). \]
\end{theorem}
\begin{proof}
Initialize $t = m^{-\O(1)}$, and set $\alpha = \wt{\Omega}(1)$ as in \cref{lemma:dualgap}. For this fixed value of $t$ run the analogue of \cref{algo:mincost}, and we repeat the same analysis as in \cref{sec:combine}. We will store the approximate values of $\bf, \by(\bf)$. Every $\wt{\Omega}(m)$ iterations, we recompute $\bf, \by(\bf)$ exactly and check whether $1^\top\by(\bf) - F^* \le 20m/t$. If so, we double $t$ and proceed to the next phase. We stop when $t = m^{\O(1)}$, so there are at most $\O(1)$ phases.

By \cref{lemma:qualggeneral,lemma:zdecrease}, the value of $Z_t$ decreases by $\kappa^{-2}\alpha^{-2} = m^{-o(1)}$ per iteration. When $t$ doubles, because we know that $1^\top\by(\bf) - F^* \le 20m/t$ by the stopping condition, $Z_{2t}(\bf) \le 20m + Z_t(\bf)$, i.e. the potential increases by at most $20m$. Hence over all $\O(1)$ phases, the total potential increase is $\O(m)$. So the algorithm runs in at most $m^{1+o(1)}$ iterations. The number of gradient/length changes is bounded by $m^{1+o(1)}$ if they are updated lazily by \cref{lemma:generalstable}.

Because $Z_t(\bf) \le \O(m)$ always, by the choice of $\alpha$ we know that $\psi_e(\bf_e,\by_e) \le \O(1)$ at all times. Thus, by item \ref{item:polybounded} of \cref{assump}, all lengths are quasipolynomially bounded during the algorithm. Additionally, \cref{lemma:qualggeneral} and an identical analysis to  \cref{lemma:setuphidden} for $\bc \defeq \bf_t^* - \bf$ and $\bw \defeq 50 + \|\bell(\bf) \circ \bc\|_1$ shows that the updates to $\bg,\bell$ satisfy the hidden stable-flow chasing property (\cref{def:hiddenStableFlowChasing}). Hence our min-ratio cycle data structure \cref{thm:MMCHiddenStableFlow} succeeds whp. in total time $m^{1+o(1)}$ as desired.
\end{proof}

\subsection{Applications: \texorpdfstring{$p$}{pnorm}-Norms, Entropy-Regularized Optimal Transport, and Matrix Scaling}
\label{sec:applicationsgeneral}
Using our main result \cref{thm:general} we can give algorithms for the problems of normed flows, isotonic regression, entropy-regularized optimal transport, and matrix scaling. We start by discussing $p$-norm flows. In this case, we allow the convex functions on our edges to be the sum of arbitrarily weighted power functions where the power is at most $\O(1)$.
\begin{theorem}
\label{thm:pnorm}
Consider a graph $G = (V, E)$ and demand $\bd$ whose entries are bounded by $\exp(\log^{O(1)}m)$, and convex functions $h_e$ which are the sum of $\O(1)$ $p$-norm terms, i.e. $h_e(x) = \sum_{i=1}^{c_e} w_i|x|^{p_i}$ for $c_e \le \O(1)$ and $p_i \le \O(1)$, and $w_i \in [0,\exp(\log^{O(1)}m)]$ for all $i \in [c_e]$. Let $h(\bf) \defeq \sum_{e \in E} h_e(\bf_e)$. Then in $m^{1+o(1)}$ time we can compute a flow $\bf$ satisfying $\mB^\top\bf=\bd$ and for any constant $C>0$
\[ h(\bf) \le \min_{\mB^\top\bf^*=\bd} h(\bf^*) + \exp(-\log^Cm). \]
\end{theorem}
\begin{proof}
By splitting up an edge $e$ into $c_e$ edges in a path we can assume that each $h_e(x) = w|x|^p$ for some $p \le \O(1)$. It is known that the function $\psi(x,y) \defeq -2\log y-\log(y^{2/p}-x^2)$ is $4$-self-concordant for the region $\{(x,y) : y \ge |x|^p\}$ \cite[Example 9.2.1]{Nem04}.
For this barrier, all items in \cref{assump} hold by observation, except those involving $\g^2\psi(x,y)$ which we now calculate.
\[ \g^2\psi(x,y) = \begin{bmatrix} \frac{4x^2}{(y^{2/p}-x^2)^2} + \frac{2}{y^{2/p} - x^2} & -\frac{4xy^{2/p-1}}{p(y^{2/p}-x^2)^2} \\ -\frac{4xy^{2/p-1}}{p(y^{2/p}-x^2)^2} & \frac{2}{y^2} + \frac{4y^{4/p-2}}{p^2(y^{2/p}-x^2)^2} + \frac{(2p-4)y^{2/p-2}}{p^2(y^{2/p}-x^2)} \\ \end{bmatrix}. \]
Clearly, if $-\log y, -\log(y^{2/p}-x^2) \le \O(1)$, and $|y|, |x| \le m^{\O(1)}$, then all terms of $\g^2\psi(x,y)$ are bounded by $\exp(\log^{O(1)}m)$ as desired, which verifies item \ref{item:polybounded} of \cref{assump}. Thus all assumptions are satisfied, so the result follows by \cref{thm:general}.
\end{proof}
The same barriers allow us to solve the problem of $\ell_p$ isotonic regression \cite{KRS15}. Given a directed acyclic graph $G = (V, E)$ and a vector $\by \in \R^V$, the $\ell_p$ isotonic regression problems asks to return a vector $\bx$ that satisfies $\bx_u \le \bx_v$ for all directed arcs $(u, v) \in E$, and minimizes $\|\mW(\bx - \by)\|_p$ for a weight vector $\mW \ge 0$. Linear algebraically, this is $\min_{\bx \in \R^V, \mB\bx \ge \vec{0}} \|\mW(\bx - \by)\|_p$. We show that the dual of this problem is a flow problem. Let $q$ be the dual norm of $p$. Then
\begin{align*}
    \min_{\bx \in \R^V, \mB\bx \ge \vec{0}} \|\mW(\bx - \by)\|_p &= \min_{\bx \in \R^V} \max_{\bf \ge \vec{0}, \|\bz\|_q \le 1} \bz^\top\mW(\bx - \by) - \bf^\top\mB\bx \\
    &= \max_{\bf \ge \vec{0}, \|\bz\|_q \le 1} \min_{\bx \in \R^V} \bz^\top\mW(\bx - \by) - \bf^\top\mB\bx \\ &= \max_{\bf \ge \vec{0}, \|\mW^{-1}\mB^\top\bf\|_q \le 1} -\bf^\top(\mB\by).
\end{align*}
Let $\bc = \mB\by$. By rescaling, the objective becomes computing \[ \min_{\bf \ge \vec{0}} \bc^\top\bf + \|\mW^{-1}\mB^\top\bf\|_q^q. \]
Given a high-accuracy solution to this objective, we can extract the desired original potentials $\bx$ by taking a gradient of the objective. To turn this objective into the $q$-norm of a flow, add a few vertex $v^*$ to the graph $G$, and an undirected edge between $(v, v^*)$ for all $v \in V$. Assign this edge the convex function $\bw_v|x|^q$, and for every other original edge $e \in E$, assign it the convex function $\bc_e\bf_e$, and restrict $\bf \ge 0$ (eg. using a logarithmic barrier). Finally, force $\bf$ to to have $0$ demand on the graph $G$ with the extra vertex $v^*$. This is now a clearly equivalent flow problem. As we have already described the self-concordant barriers for linear objectives, $\bf_e \ge 0$, and $q$-norms in the proof of \cref{thm:pnorm}, we get:
\begin{theorem}
\label{thm:isotonic}
Given a directed acyclic graph $G$, vector $\by \in \R^V$, and $p \in [1,\infty]$, we can compute in $m^{1+o(1)}$ time vertex potentials $\bx$ with $\mB\bx \ge 0$ and for any constant $C>0$
\[ \|\bx - \by\|_p \le \min_{\mB\bx^*\ge0} \|\bx^* - \by\|_p + \exp(-\log^C m). \]
\end{theorem}

Next we discuss the pair of problems of entropy-regularized optimal transport/min-cost flow and matrix scaling, which are duals. The former problem is
\[ \min_{\mB^\top\bf=\bd} \sum_{e \in E} \bc_e\bf_e + \bf_e \log\bf_e. \] The matrix scaling problem asks to given a matrix $\mA \in \R^{n \times n}_{\ge0}$ with non-negative entries to compute positive diagonal matrices $\mX, \mY$ so that $\mX\mA\mY$ is doubly stochastic, i.e. all row and column sums are exactly $1$. By \cite[Theorem 4.6]{CMTV17}, it suffices to minimize the objective
\[ \sum_{(i,j): \mA_{ij} \neq 0} \mA_{ij}\exp(\bx_i - \by_j) - \sum_{i=1}^n \bx_i - \sum_{i=1}^n \by_i \]
to high accuracy.
Consider turning the pair $(i,j)$ with $\mA_{ij} \neq 0$ into an edge $e = (i,j+n)$ in a graph $G$ with $2n$ vertices with weight $\bw_e \defeq \mA_{ij}$. Let $\bz = \begin{bmatrix} \bx \\ \by \end{bmatrix}$ be the concatenation of the $x, y$ vectors. Then the above problem becomes $\sum_{e \in E(G)} \bw_e\exp(\bz_i - \bz_j) - \sum_{i=1}^{2n} \bz_i.$ The optimality conditions for this objective are $\mB^\top\mW\exp(\mB \bz) = \bd$, where $\bd$ is $+1$ on the vertices $\{1, \dots, n\}$ and $-1$ on $\{n+1, \dots, 2n\}$.
Rearranging this gives us $\mB^\top\bf=\bd$ for $\bf = \mW\exp(\mB\bz)$, or $\mB\bz = \log(\mW^{-1}\bf)$. This is the exact optimality condition for the flow problem
\[ \min_{\mB^\top\bf=\bd} \sum_{e \in E} (-1-\log \bw_e)\bf_e + \bf_e \log(\bf_e), \] which is exactly entropy-regularized optimal transport for $\bc_e \defeq (-1-\log \bw_e)$.
If the entries of $\mA$ are polynomially lower and upper bounded, then given an (almost) optimal flow $\bf$, we know by KKT conditions that $\bf = \mW\exp(\mB\bz)$ for some dual variable $\bz$ which we can then efficiently recover. So it suffices to give high-accuracy algorithms for the entropy regularized OT problem.
\begin{theorem}
\label{thm:entropyot}
Given a graph $G = (V,E)$, demands $\bd$, costs $\bc \in \R^E$, and weights $\bw_e \in \R_{\ge0}$, all bounded by $\exp(\log^{O(1)}m)$, let $h(\bf) \defeq \sum_{e \in E} \bc_e\bf_e + \bw_e\bf_e\log\bf_e$. Then in $m^{1+o(1)}$ time we can find a flow $\bf$ with $\mB^\top\bf=\bd$ and for any constant $C>0$ \[ h(\bf) \le \min_{\mB^\top\bf^*=\bd} h(\bf^*) + \exp(-\log^Cm). \]
\end{theorem}
\begin{proof}
By splitting the edge $e$ into two edges, we can handle the terms $\bc_e\bf_e$ and $\bw_e\bf\log\bf_e$ separately. For the $\bc_e\bf_e$ term, we can handle it using the self-concordant barrier $\psi(x,y) \defeq -\log(y-\bc_e x)$. For the term $\bw_e\bf\log\bf_e$, we use the $2$-self-concordant barrier $\psi(x,y) \defeq -\log x - \log(y - x \log x)$ \cite[Example 9.2.4]{Nem04}. As in the proof of \cref{thm:pnorm}, all items of \cref{assump} hold directly, except that we must compute $\g^2\psi(x,y)$.
\[ \g^2\psi(x,y) = \begin{bmatrix}\frac{1}{x^2} + \frac{(1 + \log x)^2}{(y - x \log x)^2} + \frac{1}{x(y-x\log x)} & -\frac{1+\log x}{(y-x\log x)^2} \\ -\frac{1+\log x}{(y-x\log x)^2} & \frac{1}{(y-x \log x)^2} \end{bmatrix}. \]
Indeed, if $-\log x, -\log(y - x \log x) \le \O(1)$, and $|x|, |y| \le \exp(\log^{O(1)}m)$, then all terms in the Hessian are quasipolynomially bounded. This verifies item \ref{item:polybounded} of \cref{assump}, so the result follows by \cref{thm:general}.
\end{proof}

\begin{corollary}[Matrix scaling]
\label{cor:matrixscaling}
Given a matrix $\mA \in \R^{n\times n}_{\ge0}$ whose nonzero entries are in $[n^{-O(1)}, n^{O(1)}]$, we can find in time $\mathsf{nnz}(\mA)^{1+o(1)}$ positive diagonal matrices $\mX, \mY$ such that all row and column sums of $\mX\mA\mY$ are within $\exp(-\log^Cn)$ of $1$ for any constant $C>0$.
\end{corollary}

%% file: previous.tex
\section{Previous Works}
\label{sec:previous}
\label{sec:maxflowprevious}

We give a brief overview of the many approaches toward
the max-flow and min-cost flow problems.
A more detailed description of many of these approaches,
and more, can be found the CACM article by Goldberg and Tarjan~\cite{GT14}.
As there is a vast literature on flow algorithms,
this list is by no means complete:
we plan to update this section in subsequent works,
and would greatly appreciate any pointers.

\subsection{Maximum Flow}

The max-flow problem, and its dual, the min-cut problem
were first studied by Dantzig~\cite{D51},
who gave an $O(mn^2U)$ time algorithm.
Ford and Fulkerson introduced the notion of residual
graphs and augmenting paths, and showed the convergence
of the successive augmentation algorithm via the max-flow
min-cut theorem~\cite{FF54}.

Proving faster convergence of flow augmentations has
received much attention since the 1970s due to weighted
network flow being a special case of linear programs.
Works by Edmonds-Karp~\cite{EK73} and Karzanov~\cite{K73}
gave weakly, as well as strongly polynomial time algorithms
for finding maximum flows on capacitated graphs.

Partly due to the connection with linear programming,
the strongly polynomial case, as well as its generalizations
to min-cost flows and lossy generalized flows, subsequently received significantly more
attention.

To date, there have been three main
approaches for solving max-flow in the strongly-polynomial setting:
\begin{enumerate}
    \item Augmenting paths~\cite{EK73, K73, D70, D73,GG88,BK04}.
    \item Push-relabel~\cite{GT88,G08,GHKKTW15,OG21}.
    \item Pseudo-flows~\cite{H08,CH09,FHM10:arxiv}.
\end{enumerate}

These flow algorithms in turn motivated the study of dynamic
tree data structures~\cite{GN80}, which allows for the quick
identification of bottleneck edges in dynamically changing trees.
Suitably applying these dynamic trees gives a max-flow algorithm in the strongly-polynomial setting with runtime $\O(nm)$, which is within
polylog factors of the flow decomposition barrier.
This barrier lower bounds the combinatorial complexity
of representing the final flow as a collection of paths.

Obtaining faster algorithms hinge strongly upon handling paths
using data structures and measuring progress more numerically~\cite{EK73,G85,GR98,DS08}.
Such views date back to the Edmonds-Karp~\cite{EK73} weakly
polynomial algorithm based on finding bottleneck shortest paths
which takes $\O(m^2\log{U})$ time.
Karzanov \cite{K73} and independently Even-Tarjan \cite{ET75} further showed that in unit capacity graphs, maximum flow can be solved in time $O(m\min(\sqrt{m},n^{2/3}))$ by combining a fast bottleneck finding approach with a dual-based convergence argument. A related algorithm by Hopcroft-Karp \cite{HK73} showed that maximum bipartite matching can be solved in $O(m\sqrt{n})$ time.

Our algorithm in some sense can be viewed as implementing a data structure that identifies approximate bottlenecks
in $n^{o(1)}$ time per update, except we use a much more
complicated definition of `bottleneck' motivated by interior
point methods.
Subquadratic running times using numerical methods started
with the study of scaling algorithms for weighted matchings~\cite{G85}
and negative length shortest paths/negative cycle detection~\cite{G95}.
In these directions, Goldberg and Rao~\cite{GR98} used binary blocking flows
to obtain a runtime of $O(m\min(\sqrt{m},n^{2/3}\log{U})$ for max-flow.

More systematic studies of numerical approaches to network
flows took place via the Laplacian paradigm~\cite{ST04}.
Daitch and Spielman~\cite{DS08} made the critical observation
that when interior point methods are applied to single commodity
flow problems, the linear systems that arise are graph Laplacians, which can be solved in nearly-linear time \cite{ST04}.
This immediately implied $\O(m^{1.5}\log{U})$ time algorithms
for min-cost flow problems with integral costs/capacities,
and provided the foundations for further improvements.
Christiano, Kelner, Madry, Spielman, and Teng then gave the
first exponent beyond $1.5$ for max-flow: an algorithm that computes
$(1+\eps)$-approximate max-flows in undirected graphs
in time $\O(m n^{1/3} \eps^{-8/3})$~\cite{CKMST11}.
This motivated substantial progress on numerically driven
flow algorithms, which broadly fall into two categories:
\begin{itemize}
    \item Obtaining faster approximation algorithms for
    max-flow and its multi-commodity generalizations in undirected
    graphs through first-order methods~\cite{KMP12,S13,KLOS14,P16,S17},
    leading to a runtime of $\O(m \epsilon^{-1})$ for $(1+\eps)$-approximate max-flow.
    \item Reducing the iteration complexity of high accuracy
    methods such as interior point methods: from $m^{1/2}$
to $n^{1/2}$, or $m^{1/3+o(1)}$ for unit capacity max-flow~\cite{CKMST11, M13, LS19:arxiv, M16,LS19:arxiv,KLS20}.
\end{itemize}

Over the past two years, further progress took place via
data structured tailored to electrical flows arising in interior point methods.
These led to near-optimal runtimes for max-flow and min-cost flow on
dense graphs~\cite{BLNPSSW20, BLLSSSW21}
as well as improvements over $m^{1.5}$ in
sparse capacitated settings~\cite{GLP21:arxiv,BGJLLPS21:arxiv}.
Our approach broadly falls into this category,
except we use dynamic tree-like data structures as the starting point as opposed to electrical flow data structures,
and modify our interior point methods towards them.
Notably, we use interior point methods based on undirected min-ratio cycles instead of electrical flows. Hence, our methods use $\Omega(m)$ iterations instead of $m^{1/2}$ or $n^{1/2}$ which is common to
all algorithms subsequent to \cite{DS08}.

\subsection{Minimum-Cost Flows}

Work on the minimum cost flow problem can be traced back
to the Hungarian algorithm for the assignment problem~\cite{K12}.
This problem is a special case of minimum cost flow on
bipartite graphs with unit capacity edges.
When generalized to graphs with arbitrary integer capcities,
the algorithm runs in $\O((n + F)m)$ time where $F$
is the total units of flow sent.
Algorithms with similar running time guarantees include
many variants of network simplex~\cite{AGOT92},
and the out-of-kilter algorithm~\cite{F61}.

Strongly polynomial time algorithms for
min-cost flows have been extensively
studied~\cite{T85,GT88b,OPT93,O93,O96},
with the fastest runtime also about $\O(nm)$.
Many these algorithms are also based on augmenting
minimum mean cycles, which are closely related to
our undirected minimum-ratio cycles.
However, the admissible cycles in these algorithms
are directed, and their analysis are with obtaining
strongly polynomial time as goal.

The assignment problem has been a focal point for studying
scaling algorithms that obtain high accuracy solutions 
numerically~\cite{GT87,GT89b,DPS18}.
This is partly due to the negative-length shortest path
problem also reducing to it~\cite{G85,G95}.
These scaling algorithms obtain runtimes of the form
of $\O(m^{1.5} \log{U})$, but also extend to matching
problems on non-bipartite graphs~\cite{G85,DPS18}.
However, to date scaling arguments tend to work on only
one of capacities or costs (similar to the reductions
in Appendix~\ref{appendix:scaling}), and all previous
runtimes beyond the $\Theta(nm)$ flow decomposition barrier
for computing minimum-cost flows have been via interior point 
methods~\cite{DS08,LS19:arxiv,BLLSSSW21,AMV21:arxiv,BGJLLPS21:arxiv}.

%% file: proofs.tex
\section{Omitted Proofs}
\label{sec:proofs}

\subsection{Proof of \texorpdfstring{\cref{lemma:initialpoint}}{initialpoint}}
\label{app:initialpoint}

\begin{proof}
The graph $\wt{G}$ will have one more vertex than $G$, denoted by $v^*$. Additionally, we will add a directed edge between $v^*$ and $v$ for all $v \in V(G)$, where the direction will be decided later. Thus, $\wt{G}$ will have at most $m + n$ edges.

Initially define $\bf^\init_e = (\bu^-_e + \bu^+_e)/2 \forall e \in E(G)$. However, $\bf^\init$ will not route the demand $\bd$, and we denote the demand it routes by $\barbd \defeq \mB^\top \bf^\init$. Now, we will describe how to generate the edge between $v^*$ and $v$. If $\barbd_v = \bd_v$, then we do not add any edge between $v$ and $v^*$. If $\barbd_v > \bd_v$, then add an edge $e_v = (v \to v^*)$ with upper capacity $2(\barbd_v - \bd_v)$ (and lower capacity $0$), and set $\bf^\init_{e_v} = \barbd_v - \bd_v$. If $\barbd_v < \bd_v$ then add an edge $e_v = (v^* \to v)$ of upper capacity $2(\bd_v - \barbd_v)$ (and lower capacity $0$), and set $\bf^\init_{e_v} = \bd_v - \barbd_v$.
Finally, set $\wt{\bd}_{v^*} = 0$ and $\wt{\bd}_v = \bd_v \forall v \in V(G)$, and $\wt{\bc}_{e_v} = 4mU^2$ for the new edges $e_v$, and $\wt{\bc}_e = \bc_e \forall e \in E(G)$.

It is direct to check that all capacities of edges in $\wt{G}$ are integral and bounded by $2mU$, and costs are bounded by $4mU^2$.
It suffices then to prove that if $G$ supports a feasible flow, then the optimal flow in $G$ must put $0$ units on any of the $e_v$ edges. Indeed, note that the $e_v$ edges only can contribute nonnegative cost as $\wt{\bf}_{e_v} \ge 0$ by definition, and if any of them supports one unit of flow (i.e. $\wt{\bf}_{e_v} \ge 1$), then that contributes $4mU^2$ to the cost. The maximum possible cost of edges $e \in E$ is bounded by $mU^2$, as there are $m$ edges of capacity at most $U$, each with cost at most $U$, and the minimum is at least $-mU^2$. Hence if $G$ supports a feasible flow, $\wt{\bf}_{e_v} = 0$ for all the $e_v$ edges.

We conclude by bounding $\Phi(\bf^\init)$. Note that $\wt{\bf}_e$ are all half-integral, and $\wt{\bf}_e \le mU \forall e \in E(\wt{G})$. Also, $F^* \ge -mU^2$ by the above discussion, and because all costs are bounded by $4mU^2$, $\wt{\bc}^\top \bf^\init \le 2m \cdot 4mU^2 \cdot mU = 8m^3U^3$. Thus
\begin{align*}
    \Phi(\bf^\init) \overset{(i)}{\le} 20m \log(8m^3U^3 + mU^2) + \sum_{e \in E} \left((1/2)^{-\alpha} + (1/2)^{-\alpha} \right) \le 200m \log mU,
\end{align*}
where $(i)$ used the above bounds and that $\bu_e^+ - \bf^\init_e, \bf^\init_e - \bu^-_e$ are all half-integral.
\end{proof}

\subsection{Proof of \texorpdfstring{\Cref{thm:expanderStatement}}{expanderStatement}}
\label{app:expanderStatement}

The goal of this section is to show \Cref{thm:expanderStatement}.
To obtain our result, we require the following result from  \cite{SW19}.

\begin{theorem}[see {\cite[Theorem 1.2]{SW19}}] \label{thm:getExpander}
Given an unweighted, undirected $m$-edge graph $G$, there is an algorithm that finds a partition of $V(G)$ into sets $V_1, V_2, \ldots, V_k$ such that for each $1 \leq j \leq k$, $G[V_j]$ is a $\psi$-expander for  $\psi = \Omega(1/\log^3(m))$ and there are at most $m/4$ edges that are not contained in one of the expander graphs. The algorithm runs in time $O(m\log^7(m))$ and succeeds with probability at least $1-n^{-C}$ for any constant $C$ fixed before the start of the algorithm.
\end{theorem}

We run \Cref{alg:decompose} (given below) to obtain the graphs $G_i$ as desired in \cref{thm:expanderStatement}.

\begin{algorithm}[!ht]
$\ell \gets \lceil\log_2 \Delta_{max}(G) \rceil + 1; G_{\ell} = G$\label{lne:decompFirstLine}\;
\For{$i = \ell, \ell - 1, \ldots, 1$}{
    Let $G_i^{\rcirclearrow}$ denote the graph $G_i$ after adding $2^i$ self-loops to each vertex.\label{lne:GiLoopy}\;
    Compute an expander decomposition $V_0, V_1, \ldots, V_k$ of $G_i^{\rcirclearrow}$ as described in \Cref{thm:getExpander}.\;
    $G_{i-1} \gets \left(\bigcup_{0 \leq j \leq k} E_{G_i}(V_j, V \setminus V_j)\right)$.\label{lne:initGi}\;
    $G_{i} \gets G_i \setminus G_{i-1}$.

}
\caption{$\textsc{Decompose}(G)$}
\label{alg:decompose}
\end{algorithm}

\begin{claim}\label{clm:initializdGiIsGood}
For each $i$, the graph $G_i$ has at initialization in \Cref{lne:decompFirstLine} or \Cref{lne:initGi} at most $2^i n$ edges.
\end{claim}
\begin{proof}
We prove by induction on $i$. For the base case, $i = \ell$, observe that $2^{\ell} \geq \Delta_{max}(G)$ and since $G_{\ell}$ is a subgraph of $G$, we have $|E(G_{\ell})| \leq 2^{\ell} n$.

For $i \mapsto i-1$, we observe that $G_i$ is unchanged since its initialization until at least after $G_{i-1}$ was defined in \Cref{lne:initGi}. Thus, using the induction hypothesis and the fact above, we can conclude that $G_i^{\rcirclearrow}$  (defined in \Cref{lne:GiLoopy}) consists of at most $2^i n$ edges from $G_i$ plus $2^i n$ edges from all self-loops. But by \Cref{thm:getExpander}, this implies that $|\bigcup_{0 \leq j \leq k} E_{G_i}(V_j, V \setminus V_j)| = |\bigcup_{0 \leq j \leq k} E_{G_i^{\rcirclearrow}}(V_j, V \setminus V_j)| \leq 2^{i+1}n/4 = 2^{i-1}n$, and since this is exactly the edge set of $G_{i-1}$, the claim follows.
\end{proof}

\begin{proof}[Proof of  \Cref{thm:expanderStatement}]
Using \Cref{clm:initializdGiIsGood} and the insight that each graph $G_i$, after initialization, can only have edges deleted from it, we conclude that $|E(G_i)| \leq 2^i n$ for each $i$. 

For the minimum degree property of each $G_i$ with $i > 0$, we observe by \Cref{thm:getExpander}, that for $G_i^{\rcirclearrow}$ and vertex $v$ in expander $V_j$, $\deg_{G_i}(v) = |E_{G_i}(v , V_j \setminus \{v\})| = |E_{G_i^{\rcirclearrow}}(v , V_j \setminus \{v\})| \geq \psi 2^{i}$.
\end{proof}

\subsection{Proof of \texorpdfstring{\cref{lemma:globalstretch}}{globalstretch}}
\label{app:globalstretch}

We show \cref{lemma:globalstretch} using the following steps. First, we assume for the majority of the section that the weights $\bv = \mathbf{1}$, i.e. the all ones vector. We explain later a standard reduction to this case.
Given a low stretch tree $T$ on a graph with lengths $\bell$, and a target set of roots $R$, we explain how to find a forest $F$ (depending on $R$) that has low total stretch (\cref{def:stretchf}). This involves defining a notion of congestion on edges $e \in E(T)$. Then we explain how to handle dynamic edge insertions and deletions by adding new roots to the tree, and decrementally maintain the forest. The trickiest part is to explain how to add roots so that we can return valid stretch \emph{overestimates}. At a high level, this is done by computing a heavy-light decomposition of $T$, and using it to inform our root insertions.

It is useful to maintain the invariant that our set of roots is \emph{branch-free} at all times, i.e. that the lowest common ancestor (LCA) of any two roots $r_1, r_2 \in R$ is also in $R$. This is necessary to make it easier to construct the forest $F$ given $R$. Intuitively, forcing our set of roots to be branch-free is not a big restriction, as any set of roots can be made branch free by at most doubling its size.
\begin{definition}[Branch-free sets]
\label{def:branchfree}
For a rooted tree $T$ on vertices $V$, we say that a set $R \subseteq V$ is \emph{branch-free} if the LCA of any vertices $r_1, r_2 \in R$ is also in $R$.
\end{definition}
Given a branch-free set of roots $R$, we build a forest $F$ in the following way. We start with some total ordering/permutation $\pi$ on the tree edges $E(T)$, and for any two ``adjacent'' roots $r_1, r_2 \in R$, we delete the smallest edge with respect to $\pi$ from $T$. Here, adjacent means that no root is on the path between $r_1, r_2$. It is crucial in this construction that $R$ is branch-free, so that there are exactly $|R|-1$ adjacent pairs of roots.
\begin{definition}[Forest given roots]
\label{def:forestroots}
Given a rooted tree $T$, a branch-free set of roots $R \subseteq V$, and a total ordering $\pi$ on $E(T)$, define $F_T(R,\pi)$ as the forest obtained by removing the smallest tree edge with respect to $\pi$ from every path between adjacent roots in $T$.
\end{definition}
It is direct to verify that $F_T(R,\pi)$ has exactly $|R|$ connected components, each of which contains a unique vertex in $R$. We now explain how to construct $\pi$. $\pi$ sorts the edges by their \emph{congestions}.
\begin{definition}[Congestion]
\label{def:congestion}
Given a graph $G=(V,E)$ with lengths $\bell$, tree $T$ we define the \emph{congestions} of edges $e \in E(T)$ as
\[ \Cong^{T,\bell}_e \defeq \sum_{\substack{e' = (u, v) \in E(G) \\ \text{ s.t. } e \in T[u,v]}} 1/\bell_{e'}. \]
\end{definition}
We show that if $\pi$ is ordered by increasing congestions, then $F = F_T(R,\pi)$ has low total stretch.
\begin{lemma}[Valid $\pi$]
\label{lemma:validpi}
For a graph $G=(V,E)$ with lengths $\bell$, a rooted tree $T$, and a branch-free set of roots $R$, let $\pi$ be a total ordering on $E(T)$ sorted by increasing $\Cong^{T,\bell}_e$ (\cref{def:congestion}). Then for $F = F_T(R,\pi)$, we have $\sum_{e \in E} \str^{F,\bell}_e \le 2\sum_{e \in E} \str^{T,\bell}_e$.
\end{lemma}
\begin{proof}
Let $\hat{E}$ be the set of edges deleted from $E(T)$ to get $F_T(R,\pi)$. For an edge $e \in \hat{E}$ on a path between adjacent roots $r_1(e), r_2(e) \in R$, define $L_e = \langle \bell, |\bp(T[r_1(e),r_2(e)])|\rangle$ as the length of the path between $r_1(e), r_2(e)$ in $T$. First note by \cref{def:stretchf} that for an edge $e' = (u, v) \in E(G)$
\begin{align}
    \str^{F,\bell}_{e'} \le \str^{T,\bell}_{e'} + \sum_{e \in \hat{E} \cap T[u,v]} L_e/\bell_{e'}. \label{eq:routingbound}
\end{align}
Thus, we can bound
\begin{align*}
    \sum_{e' \in E(G)} \str^{F,\bell}_{e'} &\overset{(i)}{\le} \sum_{e' \in E(G)} \str^{T,\bell}_{e'} + \sum_{e'=(u,v) \in E(G)} \sum_{e\in \hat{E}\cap T[u,v]} L_e/\bell_{e'} \\
    &= \sum_{e' \in E(G)} \str^{T,\bell}_{e'} + \sum_{e \in \hat{E}} L_e \Cong^{T,\bell}_e = \sum_{e' \in E(G)} \str^{T,\bell}_{e'} + \sum_{e \in \hat{E}} \sum_{f \in T[r_1(e),r_2(e)]} \bell_f \Cong^{T,\bell}_e \\
    &\overset{(ii)}{\le} \sum_{e' \in E(G)} \str^{T,\bell}_{e'} + \sum_{e \in \hat{E}} \sum_{f \in T[r_1(e),r_2(e)]} \bell_f \Cong^{T,\bell}_f \overset{(iii)}{\le} \sum_{e' \in E(G)} \str^{T,\bell}_{e'} + \sum_{e \in E(T)} \bell_e \Cong^{T,\bell}_e \\
    &= 2\sum_{e' \in E(G)} \str^{T,\bell}_{e'}.
\end{align*}
Here $(i)$ follows from \eqref{eq:routingbound}, $(ii)$ follows from the fact that $\pi$ is sorted by increasing $\Cong^{T,\bell}_e$ so $\Cong^{T,\bell}_e \le \Cong^{T,\bell}_f$ for all $f \in T[r_1(e),r_2(e)]$, and $(iii)$ follows as $T[r_1(e),r_2(e)]$ are disjoint paths.
\end{proof}
To handle item \ref{item:degbound} of \cref{lemma:globalstretch}, we initialize the set of roots $R$ to have size $O(m/k)$ to already satisfy item \ref{item:degbound}. This set of roots $R$ exists by a standard decomposition result due to \cite{ST04}.
\begin{lemma}[Tree Decomposition, \cite{ST03,ST04}]
  \label{lemma:treeDecomp}
  There is a deterministic linear-time algorithm that on a graph $G=(V, E)$ with weights $\bw \in \R^E_{> 0}$, a rooted spanning tree $T$, and a reduction parameter $k$, outputs a decomposition $\cW$ of $T$ into edge-disjoint sub-trees such that:
  \begin{enumerate}
  \item $\Abs{\cW} = O(m/k).$ \label{item:sizeBound}
  \item  \label{item:branchFreeBoundary} $R \defeq \partial \cW \subseteq V$, defined as the subset of vertices appear in multiple components, is branch-free.
  \item For every component $C \subseteq V$ of $\cW$, the total weight of edges adjacent to non-boundary vertices of $C$ is at most $O(\norm{\bw}_1 \cdot k / m)$, i.e. $\sum_{e: e \cap C \not\subseteq \partial \cW} \bw_e \le 40 \cdot \norm{\bw}_1 \cdot k / m.$ \label{item:totalWgtBound}
  \end{enumerate}
\end{lemma}
Item \ref{item:degbound} of \cref{lemma:globalstretch} then follows from \cref{lemma:treeDecomp} by taking $\bw_e = 1$ for all $e \in E$.

We now explain how to add roots to $R$ under insertions and deletions of edges $e = (u, v)$. While a na\"{i}ve approach is to add both $u, v$ to the set $R$, i.e. $R \assign R \cup \{u,v\}$ (and then add more roots to make it branch free), this does not work because $\str^{T,\bell}_e$ might fluctuate significantly, and not be globally upper bounded as we want. To fix this, we introduce a more complex procedure that adds $\O(1)$ additional roots to $R$ to control the number of potential roots that an edge $e$ is assigned too.

Formally, we will construct an auxiliary tree with the same root and vertex set as $T$.
This tree is constructed via a heavy-light decomposition on $T$.
Then we replace each heavy chain with a balanced binary tree.
Thus, this auxiliary tree has height $O(\log^2 n)$.
When a vertex $u$ is added to $R$, we walk up the tree induced by the heavy-light decomposition and add all the ancestors of $u$.
Thus, at most $O(\log^2 n)$ additional vertices will be added to $R$, but each edge will only be assigned to $O(\log^2 n)$ distinct roots.

We introduce one additional piece of notation.
Given a rooted tree $T_H$ (the tree defined by the heavy-light decomposition on $T$) and vertex $u$, define $u^{\uparrow T_H}$ as the set of its ancestors in $T_H$ plus itself.
We extend the notation to any subset of vertices by defining $R^{\uparrow T_H} = \bigcup_{u \in R} u^{\uparrow T_H}.$

\begin{lemma}[Heavy-Light Decomposition of Trees, \cite{ST83}]
  \label{lemma:HLD}
  There is a linear-time algorithm that given a rooted tree $T$ with $n$ vertices outputs a collection of vertex disjoint tree paths $\{P_1, \ldots, P_t\}$ (called \emph{heavy chains}), such that the following hold for every vertex $u$:
  \begin{enumerate}
  \item There is exactly one heavy chain $P_i$ containing $u.$
  \item If $P_i$ is the heavy chain containing $u,$ at most one child of $u$ is in $P_i.$ 
  \item There are at most $O(\log n)$ heavy chains that intersect with $u^{\uparrow T}.$
  \end{enumerate}
  In addition, edges that are not covered by any heavy chain are called \emph{light edges}.
\end{lemma}

\begin{lemma}
  \label{lemma:HLDBranchFree}
  There is a linear-time algorithm that given a tree $T$ rooted at $r$ with $n$ vertices outputs a rooted tree $T_H$ supported on the same vertex set such that
  \begin{enumerate}
  \item\label{cond1:HLDBranchFree}
    The height of $T_H$ is $O(\log^2 n).$
  \item\label{cond2:HLDBranchFree} For any subset of vertices $R$ in $T$, $R^{\uparrow T_H}$ is branch-free in $T$.
  \item\label{cond3:HLDBranchFree} Given any total ordering on tree edges $\pi$ and nonempty vertex subset $R \subseteq V(T),$ $\root^F_u \in u^{\uparrow T_H}$ for every vertex $u$ where $F = F_T(R^{\uparrow T_H}, \pi).$ 
  \item\label{cond4:HLDBranchFree} Given any total ordering on tree edges $\pi$ and nonempty vertex subsets $R_1, R_2 \subseteq V(T)$, $\root^{F_1}_u = \root^{F_2}_u$ if $R_1^{\uparrow T_H} \cap u^{\uparrow T_H} = R_2^{\uparrow T_H} \cap u^{\uparrow T_H}$ where $F_i = F_T(R_i^{\uparrow T_H}, \pi), i = 1, 2.$
  That is, the root of $u$ in any rooted spanning forest of the form $F_T(R^{\uparrow T_H}, \pi)$ is determined by the intersection of $u$'s ancestors and forest roots.
  \end{enumerate}
\end{lemma}
\begin{proof}
  We first present the construction of the rooted tree $T_H.$ We root $T_H$ at the root $r$ of $T$ and compute its heavy-light decomposition in linear-time via \Cref{lemma:HLD}. Let $\{P_1, \ldots, P_t\}$ be the resulting decomposition. For every path $P_i$, we build a balanced binary search tree (BST) $T_i$ over its vertices, $V(P_i)$, with respect to their depth in $T.$ The depth of a vertex in $T$ is defined as the distance to the root $r.$ In addition, we make the vertex with minimum depth the root of the BST $T_i.$ $T_H$ is then obtained from $T$ by replacing every path $P_i$ by BST $T_i.$

  To show condition~\ref{cond1:HLDBranchFree}, observe that the path $T_H[u, r]$ consists of $O(\log n)$ node-to-root paths in some balanced BSTs and $O(\log n)$ light edges. Each node-to-root path in some balanced BST has length at most $O(\log n).$ Therefore, the $T_H[u, r]$-path has length at most $O(\log^2 n)$.

  Next, we prove condition~\ref{cond2:HLDBranchFree}. For any two vertices $u$ and $v$ in $R^{\uparrow T_H}$, let $w$ be their lowest common ancestor in $T.$ Let $P_w$ be the heavy chain containing $w.$ 
  Thus, for at least one of  $T[u, r]$ and  $T[v, r],$  $w$ must be the first vertex of $P_w$ that appears on that path or else $w$ has two distinct children that belong to $P_w.$
  Then, $w$ appears in either $T_H[u, r]$ or $T_H[v, r]$ as well and thus is included in $R^{\uparrow T_H}.$

  We prove Condition~\ref{cond3:HLDBranchFree} by induction on the depth of $u$ in $T_H$, i.e. the size of $u^{\uparrow T_H}.$
  If $|u^{\uparrow T_H}| = 1$, $u$ is the root of $T_H$ and $R^{\uparrow T_H}$ contains $u$ for any nonempty $R.$
  Thus, $\root^F_u$ is $u$ itself.
  Next, we consider the case where $|u^{\uparrow T_H}| = k+1.$
  Let $v$ be the first vertex in $R^{\uparrow T_H}$ on the path $T_H[u, r].$
  Let $P$ be the heavy chain containing $v$ and $b$ be the first vertex in $P$ on the path $T_H[u, r].$

  If $P$ does not contain $u$, let $a$ be the vertex before $b$ on the path $T_H[u, r].$
  The sequence $u, a, b, v$ shows up in the same order as in the path $T_H[u, r].$
  Since $R^{\uparrow T_H}$ does not contain $a$, $R^{\uparrow T_H}$ does not contain any vertex in the subtree of $T_H$ rooted at $a$ as well as the subtree of $T$ rooted at $a.$
  Thus, $u$, $a$, and $b$ are connected in the forest $F_T(R^{\uparrow T_H}, \pi)$ and share the same root.
  The size of $b^{\uparrow T_H}$ is less than the size of $u^{\uparrow T_H}$ and we can apply induction hypothesis to argue that $\root^F_u = \root^F_b \in b^{\uparrow T_H} \subsetneq u^{\uparrow T_H}.$

  If $P$ contains $u$, we will show that $u$ is connected to some other vertex $w \in P \cap u^{\uparrow T_H}$ in the rooted forest $F_T(R^{\uparrow T_H}, \pi).$
  Let $C$ be the set of vertices in $P$ connected to $u.$
  Observe that $C$ forms a contiguous subpath of $P$ and contains one root from $R^{\uparrow T_H}$.
  Recall that the subtree in $T_H$ corresponding to $P$ is a balanced binary search tree keyed by depth in $T$.
  Let $B$ be such binary search tree.
  It is known that given a binary search tree and a range on keys, the set of nodes in the tree within the range is closed under taking lowest common ancestor.
  Let $w$ be the lowest common ancestor of all vertices in $C$ in the BST $B.$
  $w$ must be an element of $R^{\uparrow T_H}$ and therefore $\root^F_u = w.$
  This concludes the proof of Condition~\ref{cond3:HLDBranchFree}.

  To prove Condition~\ref{cond4:HLDBranchFree}, it suffices to argue the case where $R_2 = R_1 \cup \{r\}$ and $R_1^{\uparrow T_H}$ contains every ancestor of $r.$
  Specifically, we prove that $\root^{F_2}_u = \root^{F_1}_u$ for every vertex $u$ which does not have $r$ as its ancestor.
  Let $C$ be the component of $F_1$ in which $r$ lives and $w$ be the root of $C.$
  Adding $r$ as a new root removes some edge between $r$ and $w$ and divides $C$ into two components $C_1$ and $C_2.$
  Suppose that $r \in C_1$ and $w \in C_2.$
  Condition~\ref{cond3:HLDBranchFree} says that $r$ is ancestor w.r.t. $T_H$ to every vertex in $C_1.$
  However, only vertices in $C_1$ have their root changed.
  This concludes the proof of Condition~\ref{cond4:HLDBranchFree}.
\end{proof}
We now provide an algorithm for \cref{lemma:globalstretch} and prove that it works. At a high level, the algorithm will first compute an LSST. Then it will compute global stretch overestimates based on the tree $T_H$ from \cref{lemma:HLDBranchFree}. Then, it will initialize a set of roots of size $O(m/k)$ by adding all endpoints of large stretch edges as terminals, and by calling \cref{lemma:treeDecomp} to bound the degree of each component of $F$ as needed in item \ref{item:degbound} of \cref{lemma:globalstretch}.
\begin{proof}[Proof of \cref{lemma:globalstretch}]
For the weights $\bv$, we first construct a graph $G_{\bv}$ that has $\lceil m\bv_e/\|\bv\|_1 \rceil$ unweighted copies of the edge $e$ for each $e \in E(G)$. Note that $G_{\bv}$ has at most $2m$ edges:
\[ \sum_{e \in E} \left\lceil \frac{m\bv_e}{\|\bv\|_1} \right\rceil \le \sum_{e \in E} \left(1 + \frac{m\bv_e}{\|\bv\|_1}\right) = 2m. \] Let $T$ be a LSST on $G_{\bv}$ computed using \cref{thm:an}, with an arbitrarily chosen root $r$, and let $T_H$ be the tree in \cref{lemma:HLDBranchFree}. Let $\pi$ be the permutation sorted by increasing congestions (\cref{lemma:validpi}).

Notice that $G \subseteq G_{\bv}$ and thus every spanning tree/forest of $G_{\bv}$ is also a spanning tree/forest of $G.$
Furthermore, either the tree stretch or forest stretch of edge $e \in E(G)$ is equal to the one of any of $e$'s copy in $G_{\bv}.$

We now explain how to compute stretch overestimates $\wstr_e$. For $i \ge 0$, let $B_i$ be the set of vertices within distance $i$ of the root $r$ in $T_H$. Let $D = O(\log^2 n)$ be the height of $T_H$. We define
\begin{align} \wstr_e \defeq 2\sum_{i=0}^D \str^{F_T(B_i,\pi),\bell}_e. \label{eq:defupperbounds} \end{align}
In this definition, $\wstr_e$ take identical values among copies of $e$ in $G_{\bv}.$
By \cref{lemma:validpi} we know that
\begin{align} \sum_{e \in E(G_{\bv})} \wstr_e = 2\sum_{i=0}^D \sum_{e \in E(G_{\bv})} \str^{F_T(B_i,\pi),\bell}_e \le O(D \cdot m\gamma_{LSST}) = O(m\gamma_{LSST}\log^2n). \label{eq:totalstretchbound} \end{align}
Thus the total stretch bound is fine. Shortly, we will explain the full algorithm for maintaining the set of roots and why $\wstr_e$ are valid stretch overestimates for our algorithm.

To explain how we maintain the set of roots, we first explain how to initialize a set of roots. First run \cref{lemma:treeDecomp} on $T$ with uniform weights $\bw_e = 1$ for all $e \in E(G)$ to output a set $|\cW^T| = O(m/(k\log^2 n))$, and such that each component has total adjacent weight $k$ (minus the boundaries), i.e. at most $k$ adjacent edges (as $\bw_e = 1$ for all $e$). Start by defining $R_0 \assign \partial \cW^T$. So far, $|R_0| = O(m/(k\log^2 n))$.

Also for any edge $e \in E(G)$ with $\wstr_e \ge O(k\gamma_{LSST}\log^4 n)$, add both endpoints of $e$ to $R_0$.
As
\[
\sum_{e \in E(G)} \wstr_e
\le \sum_{e \in E(G)} \left\lceil \frac{m\bv_e}{\|\bv\|_1} \right\rceil \wstr_e
= \sum_{e \in E(G_{\bv})} \wstr_e
\le O(m\gamma_{LSST}\log^2 n),
\] Markov's inequality tells us that the number of edges $e \in E(G)$ with $\wstr_e \ge O(k\gamma_{LSST}\log^4 n)$ is bounded by $O(m/(k\log^2 n))$. Thus, overall $|R_0| = O(m/(k\log^2 n))$. Now, define our initial sets of branch-free roots as $R \defeq R_0^{\uparrow T_H}$. Because the height of $T_H$ is $O(\log^2 n)$ (\cref{lemma:HLDBranchFree}), we know $|R| = O(m/k)$.

We now handle edge insertions and deletions. When an edge $e = (u, v)$ is inserted or deleted, we add $u^{\uparrow T_H} \cup v^{\uparrow T_H}$ to $R$, i.e. $R \assign R \cup (u^{\uparrow T_H} \cup v^{\uparrow T_H})$. This is branch free by \cref{lemma:HLDBranchFree}. If $e$ was inserted, assign it to have $\wstr_e = 1$, as both endpoints are roots in $R$. We also update the forest $F \defeq F_T(R,\pi)$. Because $R$ is incremental and $\pi$ is a total ordering, $F$ is decremental.

We now verify all items of \cref{lemma:globalstretch}. Item \ref{item:cccount} follows because initially $|R| = O(m/k)$, and the height of $T_H$ is $O(\log^2 n)$, so each edge insertion/deletion increases the size of $R$ by $O(\log^2 n)$. Item \ref{item:avgstretchbound} follows because
\[ \sum_{e \in E(G)} \bv_e \wstr_e \le \frac{\|\bv\|_1}{m} \sum_{e \in E(G)} \left\lceil \frac{m\bv_e}{\|\bv\|_1} \right\rceil \wstr_e \le \frac{\|\bv\|_1}{m} \sum_{e \in E(G_{\bv})} \wstr_e = O(\|\bv\|_1\gamma_{LSST}\log^2n), \]
where the last inequality follows from \eqref{eq:totalstretchbound}. 

Let $F_0 \defeq F_T(R, \pi)$ be the initial rooted spanning forest to be output.
Let $\cW$ be a refinement of $\cW^T$ induced by the connectivity in $F_0.$
We output $\cW$ as the desired edge-disjoint partition of $F_0$ into $O(m/k)$ subtrees.
$\cW$ contains at most $O(m/k)$ subtrees because $F_0$ is obtained from $T$ by removing $|R| - 1$ edges and $\cW^T$ is a edge-disjoint partition of $T.$
Item \ref{item:degbound} follows because $R \supseteq \partial \cW$, where $\cW$ was a partition of $G$ into pieces of total degree $O(k)$.

We conclude by checking item \ref{item:stretchbound}, i.e. that $\wstr_e$ upper-bounds $\str^{F, \bell}_e$ at any moment for every edge $e=(u, v)$. If $u, v$ are in the same connected component of $F$, then $\str^{F,\bell}_e = \str^{T,\bell}_e \le \wstr_e$ by noting that $T = F_T(B_0^{\uparrow T_H}, \pi)$. In the other case where $u$ and $v$ are disconnected, let $R^{\uparrow T_H}$ be the current set of roots of $F$. There must be some non-negative integer $i$ (and $j$) such that $u^{\uparrow T_H} \cap R^{\uparrow T_H} = u^{\uparrow T_H} \cap B_i$ (and $v^{\uparrow T_H} \cap R^{\uparrow T_H} = v^{\uparrow T_H} \cap B_j$ respectively). To finish, note that item \ref{cond4:HLDBranchFree} of \Cref{lemma:HLDBranchFree} ensures that
  \begin{align*}
    \root^F_u &= \root^{F_i}_u, \root^{F}_v = \root^{F_j}_v, \text{ and therefore} \\
    \str^{F, \bell}_e &\le \str^{F_i, \bell}_e + \str^{F_j, \bell}_e \le \wstr_e. \qedhere
  \end{align*}
Finally, it can be checked that the total runtime is $\O(m)$, as every operation can be implemented efficiently.
\end{proof}

\subsection{Proof of \texorpdfstring{\Cref{lemma:strMWU}}{strMWU}}
\label{app:mwu}

\begin{proof}
  Let $W = O(\gamma_{LSST}\log^2 n)$ be such that items \ref{item:stretchbound}, \ref{item:avgstretchbound} of \cref{lemma:globalstretch} imply
  \begin{align*}
    \sum_{e \in E} \bv_e \wstr_e &\le W \norm{\bv}_1, \text{ and} \\
    \max_{e \in E} \wstr_e &\le k W\log^2 n.
  \end{align*}
  Let $t = 10 k W\log^2 n = \O(k)$.
  The algorithm sequentially constructs edge weights $\bv_1, \ldots, \bv_t$ in a multiplicative weight update fashion and trees $T_1, \dots, T_t$, forests $F_1, \dots, F_t$, and stretch overestimates $\wstr^1, \dots, \wstr^t$ via \cref{lemma:globalstretch}.

  Initially, $\bv_1 \defeq \mathbf{1}$ is an the all 1's vector.
  After computing $T_i$, $\bv_{i+1}$ is defined as
  \begin{align*}
    \bv_{i+1, e} \defeq \bv_{i, e} \exp\left(\frac{\wstr^{i}_e}{t}\right) = \exp\left(\frac{1}{t} \sum_{j=1}^i \wstr^{j}_e\right) \forall e \in E.
  \end{align*}
  Finally we define the distribution $\blambda$ to be uniform over the set $\{1, \dots, t\}.$

  To show the desired bound \eqref{eq:strMWU}, we first relate it with $\norm{\bv_{t+1}}$ using the following:
  \begin{align*}
    \max_{e \in E} \frac{1}{t} \sum_{i=1}^t \wstr^{i}_e \le \log \left(\sum_e \exp\left(\frac{1}{t} \sum_{i=1}^t \wstr^{i}_e\right)\right) = \log \norm{\bv_{t+1}}_1,
  \end{align*}
  where $\bv_{t+1}$ is defined similarly even though it is never used in the algorithm.

  Next, we upper bounds $\norm{\bv_{i}}_1$ inductively for every $i=1, \dots, t+1$.
  Initially, $\bv_{1} = \mathbf{1}$ and we have $\norm{\bv_{1}}_1 = m.$
  To bound $\norm{\bv_{i+1}}$, we plug in the definition and have the following:
  \begin{align*}
    \norm{\bv_{i+1}}_1
    &= \sum_e \bv_{i, e} \exp\left(\frac{\wstr^{i}_e}{t}\right) \le \sum_e \bv_{i, e} \left(1 + 2 \cdot \frac{\wstr^{i}_e}{t}\right) \\
    &= \norm{\bv_{i}}_1 + \frac{2}{t} \sum_{e} \bv_{i, e} \wstr^{i}_e \le \norm{\bv_{i}}_1 + \frac{2}{t} W \norm{\bv_{i}}_1 = \left(1 + \frac{2W}{t}\right) \norm{\bv_{i}}_1,
  \end{align*}
  where the first inequality comes from the bound $\wstr^{i}_e \le k W = 0.1 t$ and $e^x \le 1 + 2x$ for $0 \le x \le 0.1.$
  Applying the inequality iteratively yields
  \begin{align*}
    \exp\left(\max_{e \in E} \frac1t\sum_{i=1}^t\wstr^{i}_e\right) = \bv_{t+1,e} \le \norm{\bv_{t+1}}_1 \le \left(1 + \frac{2W}{t}\right)^t \norm{\bv_{1}}_1 \le \exp(2W) m.
  \end{align*}
  The desired bound \eqref{eq:strMWU} now follows by taking the logarithm of both sides.
\end{proof}

%% file: scaling.tex
\section{Cost and Capacity Scaling for Min-Cost Flows}
\label{appendix:scaling}

In this section, we describe a cost and capacity scaling scheme \cite{G85, GT89, AGOT92} that reduces the min-cost flow problem to $O(\log mU \log C)$ instances with polynomially bounded cost and capacity.
We prove the following lemma:
\begin{lemma}
\label{lemma:polyScaling}
Suppose there is an algorithm $\cA$ that solves \eqref{eq:mincostopt} on any $m$-edge graph and $\poly(m)$-bounded integral demands, costs, and lower/upper capacities in $T_{\cA}(m)$ time.
There is an algorithm that on a graph $G=(V, E)$ and a min-cost flow instance $\cI = (G, \bd, \bc, \bu^-, \bu^+)$ with integral demands $\bd$,  integral lower/upper capacities $\bu^-, \bu^+ \in \{-U, \dots, U\}^E$, and integral costs $\bc \in \{-C, \ldots, C\}^E,$ solves $\cI$ exactly in $O(T_{\cA}(m) \log m  \log mU \log C)$-time.
\end{lemma}

Instead of \eqref{eq:mincostopt}, we consider
the equivalent \emph{min-cost circulation} problem:
\begin{align}
  \label{eq:MinCostCirculation}
  \min_{\substack{\mB^\top\bf=0 \\ 0 \le \bf_e \le \bu_e \forall e \in E}} \bc^\top \bf,
\end{align}
where cost $\bc \in \{-C, \dots, C\}^E$ and capacity $\bu \in \{1, \dots, U\}^E.$
It satisfies strong duality with dual problem:
\begin{align}
  \label{eq:MinCostCirculationDual}
  \max_{\substack{\mB \by + \bss^- - \bss^+ = \bc \\ \bss^-, \bss^+ \in \R^E_{\ge 0}}} -\bss^{+ \top} \bu.
\end{align}

Given a (directed) graph $G=(V, E)$ with costs $\bc$, capacities $\bu$, and some feasible circulation $\bf$ to \eqref{eq:MinCostCirculation}, we can define its residual graph $G(\bf) = (V, E(\bf))$ with costs $\bc(\bf)$, and capacities $\bu(\bf)$ as follows.
For any arc (directed edge) $e=(u, v) \in E$, we include $e$ with cost $\bc_e$ and capacity $\bu_e - \bf_e$ if it's not saturated, i.e. $\bf_e < \bu_e.$
We also include its reverse arc $\rev(e) = (v, u)$ with cost $-\bc_e$ and capacity $\bf_e$ if $\bf_e > 0.$
Given any directed graph $G$, we use $\mB(G) \in \{-1, 0, 1\}^{E \times V}$ to denote its edge-vertex incidence matrix that respects the edge orientation.

Given some positive integers $m, C, U$, we define $T_{MCC}(m, C, U)$ to be the time for exactly solving \eqref{eq:MinCostCirculation} on a graph of at most $m$ arcs with costs $\bc \in \{-C, \dots, C\}^E$, capacities $\bu \in \{1, \dots, U\}^E$ w.h.p..
A direct implication of \cref{thm:main} shows that
\begin{corollary}
\label{coro:polyMCF}
$T_{MCC}(m, \poly(m), \poly(m)) = m^{1+o(1)}.$
\end{corollary}

\subsection{Reduction to Polynomially Bounded Cost Instances}
\label{sec:costScaling}

In this section, we present a cost scaling scheme (\cref{algo:costscaling}) for reducing to $O(\log C)$ instances with polynomially bounded cost.
\begin{lemma}
\label{lemma:costScaling}
Suppose there is an algorithm $\cA$ that gives an integral exact minimizer to \eqref{eq:MinCostCirculation} on any $m$-edge graph and $m^{10}$-bounded integral costs, $U$-bounded integral capacities in $T_{\cA}(m, U)$ time.
\cref{algo:costscaling} takes as input a graph $G = (V, E)$ and a instance of \eqref{eq:MinCostCirculation} $\cI = (G, \bc, \bu)$ with costs $\bc \in \{-C, \dots, C\}^E$ and capacities $\bu \in \{1, \dots, U\}^E$, solves $\cI$ exactly in $O(T_{\cA}(m, U) \log C + m \log C)$-time.
In other words, $T_{MCC}(m, C, U) = O((T_{MCC}(m, m^{10}, U) + m)\log C).$\footnote{In the proof, we do not make any effort on reducing the exponent of the polynomial bound on costs.}
\end{lemma}

\begin{algorithm}[!ht]
  \caption{Cost Scaling Scheme for Solving \eqref{eq:MinCostCirculation} \label{algo:costscaling}}
  \SetKwProg{Proc}{procedure}{}{}
  \Proc{$\textsc{CostScaling}(G = (V, E), \bc \in \{-C, \dots, C\}^E, \bu \in \{1, \dots, U\}^E)$}{
    $\bf^{(0)} \gets 0.$ \\
    $T \gets O(\log C)$ \\
    \For{$t = 0, \dots, T - 1$}{
        Let $\wt{G}(\bf^{(t)}), \wt{\bc}(\bf^{(t)}), \bu(\bf^{(t)})$ be the cost rounded residual graph (\cref{def:costRoundedResG}) of $\bf^{(t)}.$ \\
        Solve \eqref{eq:MinCostCirculation} on $\wt{G}(\bf^{(t)}), \wt{\bc}(\bf^{(t)}), \bu(\bf^{(t)})$ \\
        Let $\bDelta_{\bf}$ be the primal optimal. \\
        Extract dual optimal $\bDelta_{\by}$ via \cref{lemma:computeDual}. \\
        $\bf^{(t+1)} \gets \bf^{(t)} + \bDelta_{\bf}$ \\
        $\by^{(t+1)} \gets \by^{(t)} + \bDelta_{\by}$
    }
    Output $\bf^{(T)}$
  }
\end{algorithm}

In \eqref{eq:MinCostCirculation}, the problem is equivalent under any perturbation to the cost with $\mB\by$ for any real vector $\by \in \R^V.$
To see this, given any circulation $\bf$ (not even feasible), the cost $\bc^\top \bf$ is equal to $(\bc - \mB \by)^\top \bf$ for any $\by$ because $\mB^\top \bf = 0.$
Given such $\by$, we define the \emph{reduced cost} of $\bc$ w.r.t. $\by$ as $\bc - \mB \by.$

Here we introduce the idea of \emph{$\eps$-optimality} which will be used to characterize exact minimizers to integral instance of \eqref{eq:MinCostCirculation}.
\begin{definition}
\label{def:epsOpt}
Given a parameter $\eps > 0$, and a feasible circulation $\bf$ to \eqref{eq:MinCostCirculation}, we say $\bf$ is \emph{$\eps$-optimal} if there is some vertex potential $\by \in \R^V$ such that $\min_e (\bc(\bf) - \mB(\bf)\by)_e > -\eps,$ where $G(\bf)$ is the residual graph, $\bc(\bf)$ is the residual cost w.r.t. $\bf$, and $\mB(\bf)$ is the edge-vertex incidence matrix of $G(\bf).$
\end{definition}

From \cref{def:epsOpt}, the $0$ circulation is $C$-optimal as initial cost of any edge is at least $-C.$
In the integral case, a flow $\bf$ is an exact minimizer if it is $1 / (n+1)$-optimal:
\begin{lemma}
\label{lemma:epsOpt}
In the case where the costs $\bc$ in \eqref{eq:MinCostCirculation} is from $\{-C, \dots, C\}^E$, a feasible integral circulation $\bf$ is an exact minimizer if it is $1 / (n+1)$-optimal.
\end{lemma}
\begin{proof}
Let $\wh{\bc} = \bc(\bf) - \mB(\bf) \by$ be the witness of $1/(n+1)$-optimality of $\bf.$
The cost of every circulation in $G(\bf)$ is identical between $\bc(\bf)$ and $\wh{\bc}.$
In the residual network of $\bf$, every simple cycle has cost at least $-n / (n+1)$ due to $\min_e \wh{\bc}_e > -1 / (n+1).$
However, since $\bf$ is integral, so do its residual graph, costs, and capacities.
Every negative cycle in $G(\bf)$ should have cost at most $-1.$
The fact that every simple cycle in $G(\bf)$ has cost at least $-n / (n+1) > -1$ denies the existence of negative cycles in $G(\bf)$.
Thus, $0$ is the exact minimizer to the residual problem w.r.t. $\bf$ and $\bf$ is an optimal solution to \eqref{eq:MinCostCirculation}.
\end{proof}

Intuitively, \cref{algo:costscaling} works by computing a augmenting flow $\bDelta$ given a $\eps$-optimal solution $\bf$ such that $\bf + \bDelta$ is $(\eps / 2)$-optimal.
Thus, the algorithm runs for $O(\log C + \log n)$ iterations until it reaches a $1 / (n+1)$-optimal solution.
$\bDelta$ is computed via solving \eqref{eq:MinCostCirculation} with costs rounded to polynomial size.
That is, given an integral feasible circulation $\bf$, we define a \emph{rounded residual graph} $\wt{G}_{\bf}$ as follows:
\begin{definition}
\label{def:costRoundedResG}
Given a $\eps$-optimal integral circulation $\bf$ w.r.t. a vertex potential $\by \in \R^V$, we define its \emph{cost rounded residual graph} $\wt{G}(\bf) = (V, \wt{E}(\bf))$ with costs $\wt{\bc}(\bf)$ and capacities $\bu(\bf)$ as follows:
Let $\wh{\bc}(\bf)$ be the reduced cost of $\bc(\bf)$ w.r.t. $\by$, i.e. $\wh{\bc}(\bf) = \bc(\bf) - \mB(\bf) \by$.
For any arc $e=(u, v) \in G(\bf)$, include $e$ in $\wt{G}(\bf)$ with the same residual capacity $\bu(\bf)_{e}$ and cost $\wt{\bc}(\bf)_e$ obtained by rounding $\wh{\bc}(\bf)_{e}$ to the nearest integral multiple of $\eps / m^8.$
\end{definition}

\begin{rem}
\label{rem:costRoundedResG}
When solving \eqref{eq:MinCostCirculation} on $\wt{G}(\bf)$, we can always ignore edges whose cost is more than $\eps m.$
Any simple cycle containing such edge has non-negative cost because costs are at least $-\eps.$
There will be an optimal solution that does not use any of such edge.

Thus, every edge we care about is an integral multiple of $\eps / m^8$ within the range $[-\eps, \eps m]$.
Via dividing the costs by $\eps / m^8$, all the costs are integers within $\{-m^{10}, \dots, m^{10}\}.$
\end{rem}

Given a current $\eps$-optimal integral circulation $\bf$ w.r.t. $\by$, \cref{algo:costscaling} finds a augmenting flow $\bDelta_{\bf}$ via solving \eqref{eq:MinCostCirculation} on the instance $\cI = (\wt{G}(\bf), \wt{\bc}(\bf), \bu(\bf))$ with polynomially bounded costs (\cref{rem:costRoundedResG}).
Let $\bDelta_{\by}$ be the corresponding dual \eqref{eq:MinCostCirculationDual} optimal to the instance $\cI$.
We show that $\bf + \bDelta_{\bf}$ is $\eps/2$-optimal w.r.t. $\by + \bDelta_{\by}.$
This is formulated as the following lemma:
\begin{lemma}
\label{lemma:costRoundedProgress}
Given an $\eps$-optimal integral circulation $\bf$ w.r.t. $\by$, let $\bDelta_{\bf}$ and $\bDelta_{\by}$ be the optimal primal dual solution to \eqref{eq:MinCostCirculation} on the instance $\cI = (\wt{G}(\bf), \wt{\bc}(\bf), \bu(\bf))$.
$\bf + \bDelta_{\bf}$ is $\eps/2$-optimal w.r.t. $\by + \bDelta_{\by}.$
\end{lemma}
\begin{proof}
Clearly, $\bf + \bDelta_{\bf}$ is an integral feasible circulation.
Let $\bDelta_{\by}, \bss^-, \bss^+$ be the corresponding dual solution to \eqref{eq:MinCostCirculationDual}.
We have that $\mB(\bf) \bDelta_{\by} + \bss^- - \bss^+ = \wt{\bc}(\bf).$
By complementary slackness, we know that for any arc $e = (u, v)$
\begin{enumerate}
    \item If $\Delta_{\bf, e} < \bu(\bf)_e$, $\bss^+_e = 0$, and
    \item If $\Delta_{\bf, e} > 0$, $\bss^-_e = 0$.
\end{enumerate}

For any arc $e = (u, v)$, it is included in the residual graph $G(\bf + \bDelta_{\bf})$ if $\Delta_{\bf, e} < \bu(\bf)_e$.
Therefore, we have $\bss^+_e = 0$ and $(\mB(\bf) \bDelta_{\by})_e \le \wt{\bc}(\bf)_e.$
The reduced residual cost on $e$ w.r.t. $\by + \bDelta_{\by}$ is 
\begin{align*}
    \wh{\bc}(\bf)_e - (\mB(\bf) \bDelta_{\by})_e \ge \wh{\bc}(\bf)_e - \wt{\bc}(\bf)_e \ge \frac{-\eps}{m^8.}
\end{align*}
Its reverse, $\rev(e) = (v, u)$, is included in the residual graph $G(\bf + \bDelta_{\bf})$ if $\Delta_{\bf, e} > 0$.
Therefore, we have $\bss^- = 0$ and $(\mB(\bf) \bDelta_{\by})_{\rev(e)} = -(\mB(\bf) \bDelta_{\by})_{e} \le -\wt{\bc}(\bf)_e = \wt{\bc}(\bf)_{\rev(e)}.$
The reduced residual cost on $\rev(e)$ w.r.t. $\by + \bDelta_{\by}$ is 
\begin{align*}
    \wh{\bc}(\bf)_{\rev(e)} - (\mB(\bf) \bDelta_{\by})_{\rev(e)} \ge \wh{\bc}(\bf)_{\rev(e)} - \wt{\bc}(\bf)_{\rev(e)} \ge \frac{-\eps}{m^8.}
\end{align*}
\end{proof}

However, the algorithm implementing \cref{thm:main} only gives primal optimal solution.
We need a separate routine for extracting the dual solution from the primal one.

\begin{lemma}
\label{lemma:computeDual}
There is an algorithm that given an instance $\cI = (G = (V, E), \bc, \bu)$ of \eqref{eq:MinCostCirculation} where $G$ has $m$ edges, costs $\bc \in \{-C, \dots, C\}^E$, and capacities $\bu \in \{-U, \dots, U\}$, computes an optimal primal and dual solution $\bf, \by, \bss^-, \bss^+$ to \eqref{eq:MinCostCirculation} and \eqref{eq:MinCostCirculationDual} in $O(T_{MCC}(m, C, U))$-time.
\end{lemma}
\begin{proof}
First, we compute $\bf$, the optimal primal solution to \eqref{eq:MinCostCirculation}, in $T_{MCC}(m, C, U)$-time.
Due to the optimality of $\bf$, the residual graph $G(\bf)$ has no negative cycles.
Then, we can compute a distance label on $G(\bf)$ as follows:
Add a supervertex $s$ to $G(\bf)$ with arcs toward every vertex in $G(\bf)$ of $0$ costs.
Then, we can compute a shortest path tree rooted at $s$ by solving an un-capacitated min-cost flow with demands $\bd_s = n$, $\bd_u = -1, u \in V$ in $O(T_{MCC}(m, C, U))$-time using standard reduction.

Let $\by_u, u \in V$ be the distance from $s$ to $u.$
Since $\by$ is a valid distance label on $G(\bf)$, we have $\mB(\bf)\by \le \bc(\bf)$ where $\mB(\bf)$ is the edge-vertex incidence matrix of the residual graph $G(\bf).$
Then, we will construct $\bss^-, \bss^+ \ge 0$ such that $(\by, \bss^-, \bss^+)$ is the optimal dual solution.

For any arc $e=(u, v) \in G$, if $0 < \bf_e < \bu_e$, we set both $\bss^-_e = \bss^+_e = 0.$
If $\bf_e = 0$, we set $\bss^-_e = \bc_e - (\by_v - \by_u)$ and $\bss^+_e = 0.$
If $\bf_e = \bu_e$, we set $\bss^- = 0$ and $\bss^+ = \by_v - \by_u - \bc_e.$

Next, we check that $(\by, \bss^-, \bss^+)$ is a feasible dual solution.
For any arc $e=(u, v) \in G$, if $0 < \bf_e < \bu_e$, both $e$ and $\rev(e)$ appears in $G(\bf)$.
Thus, we have both $\by_v - \by_u \le \bc_e$ and $\by_u - \by_v \le -\bc_e$ and therefore $\by_v - \by_u = \bc_e.$
Otherwise, $\by_v - \by_u + \bss^-_e - \bss^+_e = \bc_e$ holds by the definitions of $\bss^-_e$ and $\bss^+_e.$

Finally, we check that $\bc^\top \bf = -\bss^{+ \top} \bu.$
Complementary slackness and the fact that $\bf$ is a circulation yield
\begin{align*}
    \by^\top \mB^\top \bf + \bss^{-\top} \bf + \bss^{+ \top} (\bu - \bf) = 0.
\end{align*}
Rearrangement yields
\begin{align*}
    -\bss^{+ \top} \bu = (\mB \by - \bss^- + \bss^+)^\top \bf = \bc^\top \bf,
\end{align*}
where the last equality comes from dual feasibility of $(\by, \bss^-, \bss^+).$
\end{proof}

\begin{proof}[Proof of \cref{lemma:costScaling}]
Initially, we start with a $C$-optimal flow $\bf^{(0)} = 0 \in R^E$ w.r.t. potential $0 \in \R^V.$
At any iteration $t+1$, $\bf^{(t+1)}$ is an $\eps/2$-optimal flow w.r.t. $\by^{(t+1)}$ if $\bf^{(t)}$ is an $\eps$-optimal flow w.r.t. $\by^{(t)}$ (\cref{lemma:costRoundedProgress}).
Thus, after $T$ iterations, $\bf^{(T)}$ is an $C / 2^T$-optimal flow w.r.t. $\by^{(T)}.$
Taking $T = O(\log C)$, we have $C / 2^T < 1 / (n+1)$ and hence $\bf^{(T)}$ is an optimal solution to \eqref{eq:MinCostCirculation} due to \cref{lemma:epsOpt}.

Every iteration, we solve for an optimal primal dual solution to \eqref{eq:MinCostCirculation} on a cost-rounded residual graph using \cref{lemma:computeDual}.
Such instance has $m$ edges, polynomially bounded costs (\cref{rem:costRoundedResG}), $U$-bounded capacities and thus takes $O(T_{\cA}(m, U))$-time given an algorithm $\cA$ for solving \eqref{eq:MinCostCirculation}.
There is also an $O(m)$-overhead for constructing cost-rounded residual graph and updating $\bf^{(t+1)}$ and $\by^{(t+1)}$ in each iteration.
Overall, the runtime is $O(T_{\cA}(m, U) \log C + m \log C)$ since $C = \Omega(\poly(m)).$
\end{proof}

\subsection{Reduction to Polynomially Bounded Capacity Instances}
\label{sec:capScaling}

It remains to address the case of min-cost flows
with polynomially bounded costs,
but possibly large capacity.
Here we use capacity scaling.

\begin{lemma}
\label{lemma:capscaling}
Suppose there is an algorithm $\cA$ that gives an integral exact minimizer to \eqref{eq:MinCostCirculation} on any $m$-edge graph and $m^{10}$-bounded integral costs, $m^{40}$-bounded integral capacities in $T_{\cA}(m)$ time.
\cref{algo:capscaling} takes as input a graph $G = (V, E)$ and an instance of \eqref{eq:MinCostCirculation} $\cI = (G, \bc, \bu)$ with costs $\bc \in \{-m^{10}, \dots, m^{10}\}^E$ and capacities $\bu \in \{1, \dots, U\}^E$, solves $\cI$ exactly in $O(T_{\cA}(m) \log m\log U + m\log m \log U)$-time.
In other word, $T_{MCC}(m, m^{10}, U) = O(T_{MCC}(m, m^{10}, m^{40}) \log m \log mU).$
\end{lemma}

In each iteration, \cref{algo:capscaling} augments the current integral circulation $\bf$ with $\bDelta$, a constant approximate integral solution to \eqref{eq:MinCostCirculation} on the residual graph.
After $O(\log(C m U))$ iterations, the optimal objective value on the residual graph is at most $-0.1.$
This indicates that the value is $0$ and we have reached an optimal solution because the residual graph is always an integral instance and has optimal value either $0$ or at most $-1.$

To find a constant approximate solution, \cref{algo:capscaling} first finds a $\poly(m)$-approximate solution of value $-x, x > 0.$
Then, one can round the residual capacities down to integral multiples of $x / \poly(m)$ and show that the optimal solution to the rounded residual instance is a constant approximation.
Solving \eqref{eq:MinCostCirculation} on the rounded residual instance is equivalent to solving with polynomially bounded capacities, which can be done using \cref{coro:polyMCF}.

The first component is an algorithm that computes a $\poly(m)$-approximate solution to \eqref{eq:MinCostCirculation}.
In particular, we find a $Cm$-approximate solution when the costs $\bc \in \{-C, \dots, C\}^E.$
Given an instance $\cI = (G, \bc, \bu)$,
Roughly speaking, the algorithm finds a negative weight cycle in $G$ with largest bottleneck.
This is done via performing binary search over the bottleneck, which has $m$ different values, and then detecting negative cycle by solving an unit capacity version of \eqref{eq:MinCostCirculation}.

\begin{lemma}
\label{lemma:polyApproxMCC}
Suppose there is an algorithm $\cA$ that gives an integral exact minimizer to \eqref{eq:MinCostCirculation} on any $m$-edge graph and $m^{10}$-bounded integral costs, $m^{40}$-bounded integral capacities in $T_{\cA}(m)$ time.
There is an algorithm that takes as input a graph $G = (V, E)$ and a instance of \eqref{eq:MinCostCirculation} $\cI = (G, \bc, \bu)$ with costs $\bc \in \{-m^{10}, \dots, m^{10}\}^E$ and capacities $\bu \in \{1, \dots, U\}^E$, outputs an $m^{12}$-approximate solution $\bf$ such that
\begin{align*}
    \frac{1}{m^{12}} \bc^\top \bf^* \ge \bc^\top \bf,
\end{align*}
where $\bf^*$ is the optimal solution to \eqref{eq:MinCostCirculation}.
The algorithm runs in $O(T_{\cA}(m)\log m)$-time.
\end{lemma}
\begin{proof}
For any directed cycle $C$ in $G$, we define its bottleneck as $\bu(C) \defeq \min_{e \in C} \bu_e.$
The algorithm finds a negative weighted cycle $C^*$ in $G$ with maximum bottleneck.
This is done by first performing binary search over all possible bottleneck capacities, which has $m$ of them.
Let $u$ be the bottleneck capacity we want to check.
We construct graph $G_u$ from $G$ by removing all edges with capacities smaller than $u.$
Via a standard reduction to unit capacity min-cost circulation on the instance $\cI' = (G_u, \bc, 1)$, we can either find a negative cost cycle in $G_u$ or determine there's none in  $O(T_{\cA}(m))$-time.
There are $O(\log m)$ stages in binary search and each stage is done in $O(T_{\cA}(m))$-time using the given min-cost circulation algorithm $\cA$.
Overall, we can find a negative cost cycle $C^*$ with maximum bottleneck in $O(T_{\cA}(m)\log m)$-time.

Let $\bp(C^*)$ be the flow vector corresponding to $C^*$.
We show that $u(C^*) \bp(C^*)$ is a $m^{12}$-approximate solution to \eqref{eq:MinCostCirculation} on instance $\cI = (G, \bc, \bu).$
Let $\bf^*$ be the optimal solution to \eqref{eq:MinCostCirculation}.
Decompose $\bf^*$ as a non-negative linear combination of edge-disjoint directed cycles in $G$, i.e.
\begin{align*}
    \bf^* = \sum_{i=1}^m a^*_i \bp(C_i),
\end{align*}
where $\{C_1, \dots, C_m\}$ is a edge disjoint collection of cycles in $G$, $a^*_i$ is an non-negative coefficient, and $\bp(C_i)$ denotes the flow vector corresponding to cycle $C_i$ for $i = 1, \dots, m$.
Due to optimality of $\bf^*$, we can assume that $a^*_i > 0$ only if $C_i$ is a negative cost cycle.

Observe that $a^*_i$ is at most $u(C_i)$, the bottleneck of $C_i.$
And the cost of each cycle is at least $-m^{11}.$
Thus, we can bound the cost of $\bf^*$ by
\begin{align*}
    \bc^\top \bf^* = \sum_{i: \bc^\top \bp(C_i) < 0} a^*_i \bc^\top \bp(C_i) \ge -\sum_{i: \bc^\top \bp(C_i) < 0} u(C_i) m^{11} \ge -m^{12} u(C^*),
\end{align*}
where the last inequality comes from the definition of $C^*$ being the maximum bottleneck of all negative cost cycles in $G.$

On the other hand, the weight of $C^*$ is at most $-1$ because costs are integers.
Thus, the cost of $u(C^*) \bp(C^*)$ is at most $-u(C^*)$ and the proof concludes.
\end{proof}

Given any $B > 0$, we can round edge capacities down to integral multiples of $B / m^{20}$ and remove edges of capacity over $m B.$
The optimal circulation for the rounded instance is feasible in the original instance.
Furthermore, it is a constant approximation.
First, let us define the rounded instance formally.

\begin{definition}
\label{def:capRoundedG}
Given a graph $G = (V, E)$ with costs $\bc \in \{-m^{10}, \dots, m^{10}\}^E$ and capacities $\bu \in \{1, \dots, U\}^E$ and a positive value $B > 0.$
We define its \emph{capacity rounded graph} $G^B = (V, E^B)$ with costs $\bc$ and capacities $\bu^B$ as follows:
We include every arc $e \in E$ to $G^B$ and assign its capacity $\bu^B_e$ to be the nearest integral multiple of $\lceil B / m^{20} \rceil$ below $\bu_e$ or $\lceil Bm^{20} \rceil$ if $\bu_e > Bm^{20}.$
\end{definition}
\begin{rem}
\label{rem:capRoundedG}
By scaling down the rounded capacities by $\lceil B / m^{20} \rceil$, the capacity of every edge is a positive integer at most $m^{40}.$
Thus, solving \eqref{eq:MinCostCirculation} on the capacity rounded instance $\cI^B = (G^B, \bc, \bu^B)$ is equivalent to solving on an instance with $m^{40}$-bounded capacities.
In addition, one can recover the optimal solution for the capacity rounded instance by scaling up with $\lceil B / m^{20} \rceil$, which is still integral.
\end{rem}

\begin{lemma}
\label{lemma:capRoundedGConstApprox}
Given a graph $G = (V, E)$ with costs $\bc \in \{-m^{10}, \dots, m^{10}\}^E$ and capacities $\bu \in \{1, \dots, U\}^E$ and a positive value $B > 0.$
Suppose that $-B$ is $m^{12}$-approximation to the optimal value of \eqref{eq:MinCostCirculation} on the instance $\cI = (G, \bc, \bu)$.
Let $\bf^B$ be the optimal solution for \eqref{eq:MinCostCirculation} on the capacity rounded instance $\cI^B = (G^B, \bc, \bu^B)$.
$\bf^B$ is an integral $1.1$-approximate solution for the original instance $\cI = (G, \bc, \bu)$.
\end{lemma}
\begin{proof}
Clearly, $\bf^B$ is a feasible solution for $\cI$ because $\bf^B \le \bu^B \le \bu$ and $G^B$ is a subgraph of $G.$
Let $\bf^*$ be the optimal solution for \eqref{eq:MinCostCirculation} on the instance $\cI = (G, \bc, \bu)$.

One can decompose $\bf^*$ as a non-negative linear combination of edge-disjoint directed cycles in $G$, i.e.
\begin{align*}
    \bf^* = \sum_{i=1}^m a^*_i \bp(C_i),
\end{align*}
where $\{C_1, \dots, C_m\}$ is a edge disjoint collection of cycles in $G$, $a^*_i$ is an non-negative coefficient, and $\bp(C_i)$ denotes the flow vector corresponding to cycle $C_i$ for $i = 1, \dots, m$.
Due to optimality of $\bf^*$, we can assume that $a^*_i > 0$ only if $C_i$ is a negative cost cycle.
Also, $a^*_i$ is at most be bottleneck capacity of the cycle $C_i$, i.e. $a^*_i \le \min_{e \in C_i} \bu_e.$

First, we claim that $\bf^*_e \le \lceil Bm^{20} \rceil$ for any arc $e.$
Otherwise, there is a negative cost cycle $C_i$ in the decomposition with $a^*_i > \lceil Bm^{20} \rceil$.
The cost of $C_i$ is at least $-1$ due to integral costs.
In this case, we use the fact that $-m^{12}B \le \bc^\top \bf^*$ and deduce
\begin{align*}
    \bc^\top \bf^* < -\lceil Bm^{20} \rceil \le -Bm^{20} \le \frac{\bc^\top \bf^*}{m^{12}} m^{20} = m^8 \bc^\top \bf^* < 0,
\end{align*}
which leads to a contradiction.

If $B < 2 m^{20}$, we have $\bu^B = \bu$ because $\lceil B / m^{20}\rceil = 1$.
In this case, $\bf^B$ is exactly $\bf^*.$
Otherwise, round down the cycle decomposition of $\bf^*$ to integral multiples of $\lceil B / m^{20} \rceil.$
That is, we define
\begin{align*}
    \wt{\bf} = \sum_{i=1}^m \wt{a}_i \bp(C_i),
\end{align*}
where $\wt{a}_i$ is the nearest integral multiple of $\lceil B / m^{20} \rceil$ at most $a^*_i.$
We have that $\wt{a}_i \le \min_{e \in C_i} \bu^B_e$ and hence $\wt{\bf}$ is a feasible solution for the capacity rounded instance.
In addition, we have $\wt{a}_i \ge a^*_i - \lceil B / m^{20} \rceil$ for any $i.$
Using these facts, we have
\begin{align*}
    \bc^\top \wt{\bf}
    = \sum_{i=1}^m \wt{a}_i \bc^\top \bp(C_i) 
    &\overset{(i)}{\le} \sum_{i=1}^m a^*_i \bc^\top \bp(C_i) - \sum_{i=1}^m \lceil \frac{B}{m^{20}} \rceil \bc^\top \bp(C_i) \\
    &\overset{(ii)}{\le} \sum_{i=1}^m a^*_i \bc^\top \bp(C_i) + \sum_{i=1}^m \lceil \frac{B}{m^{20}} \rceil m^{11} \\
    &\overset{(iii)}{\le} \sum_{i=1}^m a^*_i \bc^\top \bp(C_i) + \sum_{i=1}^m 2 \frac{B}{m^{20}} m^{11} \\
    &= \bc^\top \bf^* + 2 \frac{B}{m^8} \\
    &\overset{(iv)}{\le} \bc^\top \bf^* + \frac{-2\bc^\top \bf^*}{m^8} \\
    &= (1 - 2m^{-8}) \bc^\top \bf^*,
\end{align*}
where (i) comes from $\wt{a}_i \ge a^*_i - \lceil B / m^{20} \rceil$, (ii) comes from that any simple cycle has cost at least $-m^{11}$, (iii) comes from $B \ge 2m^{20}$ and $\lceil B / m^{20} \rceil \le 2B / m^{20}$, and (iv) comes from $\bc^\top \bf^* \le -B$ as $-B$ is the value of a $m^{12}$-approximate solution.
We conclude the proof by observing that
\begin{align*}
    \bc^\top \bf^* \le \bc^\top \bf^B \le \bc^\top \wt{\bf} \le (1 - 2m^{-8}) \bc^\top \bf^*.
\end{align*}
\end{proof}

\begin{algorithm}[!ht]
  \caption{Capacity Scaling Scheme for Solving \eqref{eq:MinCostCirculation} \label{algo:capscaling}}
  \SetKwProg{Proc}{procedure}{}{}
  \Proc{$\textsc{CapacityScaling}(G = (V, E), \bc \in \{-m^{10}, \dots, m^{10}\}^E, \bu \in \{1, \dots, U\}^E)$}{
    $\bf^{(0)} \gets 0.$ \\
    $T \gets O(\log U)$ \\
    \For{$t = 0, \dots, T - 1$}{
        Compute $-x \le 0$ to be the value of an $m^{12}$-approximate solution to \eqref{eq:MinCostCirculation} via \cref{lemma:polyApproxMCC}. \tcp{$x \ge 0$.}
        \If{$x = 0$}{
            $\bf^{(t)}$ is an optimal solution and we end the for loop here.
        }
        Let $G^{x}(\bf^{(t)}), \bc(\bf^{(t)}), \bu^{x}(\bf^{(t)})$ be the capacity rounded graph (\cref{def:capRoundedG}) of the residual graph $G(\bf^{(t)})$ with costs $\bc(\bf^{(t)})$and capacities $\bu(\bf^{(t)})$. \\
        Solve \eqref{eq:MinCostCirculation} on $G^{x}(\bf^{(t)}), \bc(\bf^{(t)}), \bu^{x}(\bf^{(t)})$ \\
        Let $\bDelta_{\bf}$ be the primal optimal. \\
        $\bf^{(t+1)} \gets \bf^{(t)} + \bDelta_{\bf}$
    }
    Output $\bf^{(T)}$
  }
\end{algorithm}

\begin{proof}[Proof of \cref{lemma:capscaling}]

Let $\bf^*$ be the optimal solution.
For any $t$, $\bf^{(t)}$ is integral since the augmenting circulation $\bDelta_{\bf}$ is always integral (\cref{rem:capRoundedG}).
Therefore, the optimal solution $\bf^* - \bf^{(t)}$ to the residual instance w.r.t. $\bf^{(t)}$ is integral and have cost at most $-1$ or $0.$

\cref{lemma:capRoundedGConstApprox} states that the augmenting circulation $\bDelta_{\bf}$ is always a $2$-approximation to the residual instance.
Thus, at any iteration $t$, we have
\begin{align*}
    \bc^\top (\bf^* - \bf^{(t)}) \le \bc^\top \bDelta_{\bf} \le \frac{1}{2}\bc^\top (\bf^* - \bf^{(t)}) \le 0.
\end{align*}
Using the definition of $\bf^{(t+1)}$ and induction yield
\begin{align*}
    0 \ge \bc^\top (\bf^* - \bf^{(t + 1)}) = \bc^\top (\bf^* - \bf^{(t)}) - \bc^\top \bDelta_{\bf} \ge \frac{1}{2}\bc^\top (\bf^* - \bf^{(t)}) \ge \frac{1}{2^{t+1}} \bc^\top \bf^*.
\end{align*}
Since the costs and capacities are bounded by $m^{10}$ and $U$, $\bc^{\top}\bf^*$ is at least $-m^{11} U.$
After $T = O(\log m + \log U) = O(\log mU)$ iterations, we have
\begin{align*}
    0 \ge \bc^\top (\bf^* - \bf^{(T)}) \ge \frac{1}{2^{T}} \bc^\top \bf^* \ge \frac{-m^{11}U}{2^T} \ge \frac{-1}{2}.
\end{align*}
Combining with the previous observation that $\bc^\top (\bf^* - \bf^{(T)})$ is either $0$ or at most $-1$, we have that $\bc^\top (\bf^* - \bf^{(T)}) = 0$ and hence $\bf^{(T)}$ is an optimal solution.

Each iteration spends $O(T_{\cA}(m)\log m)$-time for computing an $m^{12}$-approximate solution, plus $O(T_{\cA}(m))$-time for computing $\bDelta_{\bf}$ on a capacity rounded instance, and $O(m)$ for constructing instances and computing $\bf^{(t+1)}$.
Overall, the runtime is $O(T_{\cA}(m)\log m \log mU + m \log m \log mU).$
\end{proof}

Now, we can prove \cref{lemma:polyScaling} by combining \cref{lemma:costScaling} and \cref{lemma:capscaling}.
\begin{proof}[Proof of \cref{lemma:polyScaling}]
As any instance of \eqref{eq:mincostopt} can be reduced to \eqref{eq:MinCostCirculation} with linear overhead in the number of edges and $\poly(m)$ scaling on costs and capacities, the lemma follows directly from \cref{lemma:costScaling} and \cref{lemma:capscaling}.
\end{proof}

%% file: applications.tex
\section{Applications}
\label{sec:related}\label{sec:applications}

Our results directly imply faster running times for algorithms that invoke network flow primitives.

\paragraph{Extensions of \Cref{thm:main}.} Our main result can be generalized to take vertex capacities and costs by standard transformations (for lower capacities on vertices being zero, one can simple split each vertex $v$ into $v_{in}$ and $v_{out}$ such that all in-going edges to $v$ are incident to $v_{in}$ and all out-going edges to $v_{out}$ after the split, and then insert and edge $(v_{in}, v_{out})$ with the desired capacity and cost). Further, we can generalize our algorithm to handle the \emph{flow diffusion} problems \cite{WFHMR17, FWY20, CPW21:arxiv} where $\bd \in \mathbb{R}^V, \bd^{\top} \mathbf{1} \geq \mathbf{0}$, is considered a vertex capacity vector instead of a demand and one wants to find a flow $\bf$ that satisfies $\mB^{\top} \bf \leq \bd$ while minimizing over a cost function on $\bf$. This can be realized by adding special vertices $s,t$ and an edge $(s,v)$ (resp. $(v,t)$) for each vertex $v \in V$ where $\bb_v \leq 0$ (resp. $\bb_v > 0$), with lower capacity $0$, upper capacity $|\bb_v|$ and cost $0$.

Previously, considerable effort 
\cite{CK19,C21, BGS21} was directed towards obtaining approximate max-flow algorithms that can handle vertex-capacities in undirected graphs where the above mentioned transformations do not translate. Diffusion has been considered for the cost function taken to be the $\ell_2$-norm ~\cite{HRW20,CPW21:arxiv}. We recover using simple reductions a simple almost linear time algorithm that can handle a wide range of cost function.

We can also obtain an algorithm that runs in near-linear time to compute $p$-norm flows, i.e. flow problems where one is given a weight matrix $\mW$ and solves the problem $\min_{\mB^{\top} \bf = \bd } \|\mW \bf \|_p^p$ up to a polynomially small error. An even more general problem is considered in \Cref{thm:pnorm}. Previous work, either achieved super-linear run-time \cite{adil2019iterative,adil2021almost} or was only able to solve the problem when $\mW$ was taken to be the identity matrix \cite{KPSW19, AS20}.

\paragraph{Bipartite Matching \& Optimal Transport.} Many popular variations of matching problems are well-known to be reducible to min-cost flow in bipartite graphs, i.e. graphs $G = ( V, E)$ where there is a partition $V_1, V_2$ of $V$ such that each edge has exactly one endpoint in $V_1$ and one in $V_2$.

In the standard matching problem, one is given the task of maximizing the number of edges without common vertex in an undirected graph. In the \emph{perfect} matching problem, the algorithm has to output a matching of size $|V|/2$ or conclude that such a matching does not exist. A substantial generalization of perfect matching problems is the \emph{worker assignment} problem: given upper capacities $\bu^+ \in \mathbb{R}^E_{\geq 0}$ and costs $\bc \in \mathbb{R}^E$ over the edges and has $\bb \in \mathbb{N}_{\geq 0}^V$, the goal is to either compute a weight $\bw \in \mathbb{N}^{E}$ such that each vertex $v \in V$ has edges of total weight $\bb_v$ incident and where $\bc^{\top} \bw$ is minimized over all such choices, or decide that no such weight $\bw$ exists. Our result implies that the the worker assignment problem can be solved in time $m^{1+o(1)} \log^2 U$ in bipartite graphs. We refer the reader to \cite{GT89} for an in-depth description of the reduction to min-cost flow.

Our result can further also be used to solve the optimal transportation problem, even with entropic regularization (see \cite{dvurechensky2018computational, guo2020fast}), which is crucial for applications in machine learning. In this problem, one is given a bipartite graph $G = (V_1 \cup V_2, E)$, demand $\bd$, where $\bd$ is non-negative on $V_1$ and non-positive on $V_2$, costs $\bc$, and the goal is to find a flow $\bf$ that satisfies $\mB^{\top}\bf = \bd$ and minimizes $\bc^{\top} \bf + H(\bf)$ where $H(\bf) = \sum_{e \in E} \bf_e \log(\bf_e)$. We can use our result in \cref{thm:entropyot} to obtain the first almost-linear time algorithm to obtain an optimal flow $\bf$ to high accuracy (also called transportation plan). This improves even over the run-time of $\tilde{O}(n^2)$ taken by current state-of-the-art low accuracy solvers \cite{cuturi2013sinkhorn,benamou2015iterative, altschuler2017near, dvurechensky2018computational}. Without the entropic regularization the problem is reducible directly to the worker assignment problem.

The matrix scaling problem \cite{ALOW17,CMTV17} asks: given a matrix $\mA \in \R^{n\times n}_{\ge0}$ with non-negative polynomially bounded entries, to compute positive diagonal matrices $\mX, \mY$ such that all row and column sums of $\mX\mA\mY$ are $1$. As shown in \cref{sec:applicationsgeneral}, the dual of the matrix scaling problem is optimal transport with entropic regularization. Hence we achieve an algorithm for solving matrix scaling to high accuracy in almost-linear time even when the entries of the matrices $\mX, \mY$ may be exponentially large.
    
\paragraph{Negative Shortest-Paths and Cycle Detection.} We obtain a almost linear time algorithm to compute the Single-Source Shortest Paths from a dedicated source vertex $s$ in a directed, possibly negatively weighted graph by invoking \Cref{cor:main} with costs set to edge weights, $\bu^- = \mathbf{0}, \bu^+ = n \cdot \mathbf{1}$ and $\bd_s = n$ and $\bd_v = -1$ for all $v \in V$. For a graph with weights bounded by $W$ in absolute value, this gives an algorithm with running time $m^{1+o(1)}\log W.$ Further, we can find a negative directed cycle in a graph by choosing $\bu^- = \bd = \mathbf{0}, \bu^+ = \mathbf{1}$, letting the cost vector equal the weights and check whether the computed flow $\bf$ is non-zero. If it is not then $\bf$ is a negative cost circulation and using Cut-Link Trees \cite{ST83} on can recover a negative cycle. For both problems, we give the first almost linear time algorithm. 

\paragraph{Connectivity \& Gomory-Hu Trees.} Another family of classic combinatorial problems are connectivity problems where many reductions to maximum flow have been found during the last years. It is well-known that from a $(s,t)$ maximum flow, i.e. the maximum amount of flow that can be sent in a unit-weighted graph from a vertex $s$ to vertex $t$, one can find an $(s,t)$ min-cut in almost linear time, that is a bipartition $V_1, V_2$ of the vertex set $V$ of the graph with $s \in V_1, t \in V_2$ such that the number of edges with tail in $V_1$ and head in $V_2$ is minimized. 

Our algorithm implies an algorithm that finds the \emph{global} min-cut obtained by miminizing over $(s,t)$ cuts for all pairs $s,t \in V$, in time $mn^{1/2+o(1)}$ time in directed graphs~\cite{CLNPQS21:arxiv}.
For undirected graphs, using a reduction from \cite{li2020deterministic}, we obtain the first almost linear algorithm to compute a Steiner min-cut which is the minimum $(s,t)$-cut for $s,t \in S$ for a fixed input set $S \subseteq V$. 

Our result also implies the first $m^{1+o(1)}$ time algorithm to compute a  global vertex min-cut in undirected graphs via \cite{LNPSY21}, i.e. a tripartition $A,B,S$ of $V$ such that there is no edge from $A$ to $B$ where the size of $S$ is minimized. It further gives $m^{1+o(1)} \poly(k)$ time algorithm to construct a $k$-vertex connectivity oracle (see \cite{pettie2022optimal}).

Finally, we consider algorithms to compute Gomory-Hu trees that is a weighted tree $T$ over the vertex set of $G$ such that for any two vertices $s,t \in V$, the $(s,t)$ min-cut in $G$ has the same value as in $T$. Our result gives the first $m^{1+o(1)}$ time algorithm to compute Gomory-Hu trees in unweighted graphs (via \cite{AKLPST21:arxiv,Z21:arxiv}), or to a $(1+\epsilon)$-approximation in weighted graphs (via \cite{LP21}) for arbitrarily small constant $\epsilon$.

We point out that we improve for all cited problems the run-time by polynomial factors (in $m$).

\paragraph{Directed Expanders.} We say a cut $(S, V \setminus S)$ in a digraph $G$ is $\phi$-out-sparse if $\frac{|E(S, V \setminus S)|}{\min\{\vol(S), \vol(V \setminus S)\}} < \phi$ where $E(S, V \setminus S)$ is the set of edges with tail in $S$ and head in $V \setminus S$ and $\vol(X)$ is the sum of degrees of vertices in $X$ in $G$. A graph $G$ is called a $\phi$-expander if $G$ allows no $\phi$-out-sparse cut.

Applying our max-flow algorithm to a straightforward extension of the cut-matching game \cite{khandekar2009graph, louis2010cut} gives a $m^{1+o(1)}$ time algorithm that given any graph $G$ and parameter $\phi \in (0,1/O(\log^2m)]$, either outputs a $O(\phi \log^2m)$-out-sparse cut or certifies that $G$ is a $\phi$-expander. The algorithm also works when a $\phi$-out-sparse cut is redefined to be a cut $(S,V \setminus S)$ with $\frac{|E(S, V \setminus S)|}{\min\{|S|, |V \setminus S|\}} < \phi$. This improves over the previously best run-time of $\tilde{O}(m/\phi)$ for sparse graphs for a wide range of values for $\phi$.

As a concrete application, we obtain a $mn^{0.5+o(1)}$ total time algorithm for the problem of maintaining strongly-connected graphs in a graph undergoing edge deletions that works against an adaptive adversary (via \cite{bernstein2020deterministic}), improving on the previously best time of $mn^{2/3+o(1)}$.

\paragraph{Isotonic Regression.}
Isotonic regression is a classic shape-constrained nonparametric regression method.
The problem is formulated as follows: we are given a DAG (Directed Acyclic Graph) $G=(V,E)$ and a vector $\by \in \R^V.$ The goal is to find project $\by$ on to the space of vectors that are \emph{isotonic} with respect to $G.$
A vector $\bx \in \R^V$ is said to be isotonic with respect to $G$ if the embedding of $V$ into $\R$ given by $\bx$ is weakly order-preserving with respect to the partial order described by $G.$ The projection is usually computed using a weighted $\ell_p$ norm.
This can be captured as the following convex program, $\min_{\bx} \norm{\mW(\bx-\by)}_p$ subject to the constraints $\bx_i \le \bx_j$ for all $(i,j) \in E.$

We give an almost linear time algorithm for computing a $1/\poly(n)$ additive approximate solution to Isotonic regression for all $p \in [1, \infty).$
 The previous best time bounds were $\O(m^{1.5})$ for $p \in [1,\infty)$~\cite{KRS15}, $O(nm\log \frac{n^2}{m})$ for
$p \in (1,\infty)$~\cite{HQ03},  and $O(nm + n^{2} \log n)$ for
$p=1$~\cite{Stout13}. We stress that this running time is almost-linear in the number of edges in the underlying DAG, which could be significantly smaller than the number of edges in the transitive closure, which determines the running time of some algorithms~\cite{Stout21}.

%% file: refs.bib
@article{DPS18,
  author    = {Ran Duan and
               Seth Pettie and
               Hsin{-}Hao Su},
  title     = {Scaling Algorithms for Weighted Matching in General Graphs},
  journal   = {{ACM} Trans. Algorithms},
  volume    = {14},
  number    = {1},
  pages     = {8:1--8:35},
  year      = {2018},
  note      = {Available at:~\url{https://arxiv.org/abs/1411.1919}}
}

@article{GT88b,
author = {Galil, Zvi and Tardos, \'{E}va},
title = {An $O(n^2(m + n\log n)\log n)$ Min-Cost Flow Algorithm},
year = {1988},
issue_date = {April 1988},
publisher = {Association for Computing Machinery},
address = {New York, NY, USA},
volume = {35},
number = {2},
issn = {0004-5411},
journal = {J. ACM},
pages = {374–386},
numpages = {13}
}

@article{T85,
author = {Tardos, \'{E}va},
title = {A Strongly Polynomial Minimum Cost Circulation Algorithm},
year = {1985},
issue_date = {1985},
publisher = {Springer-Verlag},
address = {Berlin, Heidelberg},
volume = {5},
number = {3},
issn = {0209-9683},
journal = {Combinatorica},
pages = {247–255},
numpages = {9}
}

@inproceedings{O96,
author = {Orlin, James B.},
title = {A Polynomial Time Primal Network Simplex Algorithm for Minimum Cost Flows},
year = {1996},
isbn = {0898713668},
publisher = {Society for Industrial and Applied Mathematics},
address = {USA},
booktitle = {Proceedings of the Seventh Annual ACM-SIAM Symposium on Discrete Algorithms},
pages = {474–481},
numpages = {8},
location = {Atlanta, Georgia, USA},
series = {SODA '96}
}

@article{OPT93,
author = {Orlin, James B. and Plotkin, Serge A. and Tardos, \'{E}va},
title = {Polynomial Dual Network Simplex Algorithms},
year = {1993},
issue_date = {June 1993},
publisher = {Springer-Verlag},
address = {Berlin, Heidelberg},
volume = {60},
number = {1–3},
journal = {Math. Program.},
pages = {255–276},
numpages = {22}
}

@article{F61,
  title={An out-of-kilter method for minimal-cost flow problems},
  author={Fulkerson, Delbert R},
  journal={Journal of the Society for Industrial and Applied Mathematics},
  volume={9},
  number={1},
  pages={18--27},
  year={1961},
  publisher={SIAM}
}

@article{GT89b,
author = {Goldberg, Andrew V. and Tarjan, Robert E.},
title = {Finding Minimum-Cost Circulations by Canceling Negative Cycles},
year = {1989},
issue_date = {Oct. 1989},
publisher = {Association for Computing Machinery},
address = {New York, NY, USA},
volume = {36},
number = {4},
journal = {J. ACM},
pages = {873–886},
numpages = {14}
}

@article{O93,
  author    = {James B. Orlin},
  title     = {A Faster Strongly Polynomial Minimum Cost Flow Algorithm},
  journal   = {Oper. Res.},
  volume    = {41},
  number    = {2},
  pages     = {338--350},
  year      = {1993}
}

@article{K12,
  author    = {Harold W. Kuhn},
  title     = {A tale of three eras: The discovery and rediscovery of the Hungarian Method},
  journal   = {Eur. J. Oper. Res.},
  volume    = {219},
  number    = {3},
  pages     = {641--651},
  year      = {2012}
}

@article{LY21,
  author    = {Yin Tat Lee and
               Man{-}Chung Yue},
  title     = {Universal Barrier Is \emph{n}-Self-Concordant},
  journal   = {Math. Oper. Res.},
  volume    = {46},
  number    = {3},
  pages     = {1129--1148},
  year      = {2021},
  note      = {Available at:~\url{https://arxiv.org/abs/1809.03011}}
}

@article{alon1995graph,
  title={A graph-theoretic game and its application to the k-server problem},
  author={Alon, Noga and Karp, Richard M and Peleg, David and West, Douglas},
  journal={SIAM Journal on Computing},
  volume={24},
  number={1},
  pages={78--100},
  year={1995},
  publisher={SIAM}
}

@misc{Chewi21:arxiv,
      title={The entropic barrier is $n$-self-concordant}, 
      author={Sinho Chewi},
      year={2021},
      eprint={2112.10947},
      archivePrefix={arXiv},
      primaryClass={math.MG},
      note={Available at:~\url{https://arxiv.org/abs/2112.10947}}
}

@book{N04:book,
  author    = {Yurii E. Nesterov},
  title     = {Introductory Lectures on Convex Optimization - {A} Basic Course},
  series    = {Applied Optimization},
  volume    = {87},
  publisher = {Springer},
  year      = {2004},
  note={Available at:~\url{https://wwwfr.uni.lu/content/download/92121/1121193/file/NesB.pdf}}
}

@misc{N98:notes,
  title={Introductory lectures on convex programming},
  author={Nesterov, Yu},
  year={1998}
}

@article{BV04:book,
  title={Convex Optimization},
  author={Boyd, Stephen and Vandenberghe, Lieven},
  year={2004},
  publisher={Cambridge University Press}
}

@inproceedings{adil2019iterative,
  title={Iterative refinement for ℓp-norm regression},
  author={Adil, Deeksha and Kyng, Rasmus and Peng, Richard and Sachdeva, Sushant},
  booktitle={Proceedings of the Thirtieth Annual ACM-SIAM Symposium on Discrete Algorithms},
  pages={1405--1424},
  year={2019},
  organization={SIAM}
}

@article{altschuler2017near,
  title={Near-linear time approximation algorithms for optimal transport via Sinkhorn iteration},
  author={Altschuler, Jason and Niles-Weed, Jonathan and Rigollet, Philippe},
  journal={Advances in neural information processing systems},
  volume={30},
  year={2017}
}

@article{benamou2015iterative,
  title={Iterative Bregman projections for regularized transportation problems},
  author={Benamou, Jean-David and Carlier, Guillaume and Cuturi, Marco and Nenna, Luca and Peyr{\'e}, Gabriel},
  journal={SIAM Journal on Scientific Computing},
  volume={37},
  number={2},
  pages={A1111--A1138},
  year={2015},
  publisher={SIAM}
}

@article{cuturi2013sinkhorn,
  title={Sinkhorn distances: Lightspeed computation of optimal transport},
  author={Cuturi, Marco},
  journal={Advances in neural information processing systems},
  volume={26},
  year={2013}
}

@inproceedings{dvurechensky2018computational,
  title={Computational optimal transport: Complexity by accelerated gradient descent is better than by Sinkhorn’s algorithm},
  author={Dvurechensky, Pavel and Gasnikov, Alexander and Kroshnin, Alexey},
  booktitle={International conference on machine learning},
  pages={1367--1376},
  year={2018},
  organization={PMLR}
}

@inproceedings{adil2021almost,
  title={Almost-Linear-Time Weighted $\ell_p$-Norm Solvers in Slightly Dense Graphs via Sparsification},
  author={Adil, Deeksha and Bullins, Brian and Kyng, Rasmus and Sachdeva, Sushant},
  booktitle={48th International Colloquium on Automata, Languages, and Programming (ICALP 2021)},
  year={2021},
  organization={Schloss Dagstuhl-Leibniz-Zentrum f{\"u}r Informatik}
}

@inproceedings{KPSW19,
  title = {Flows in Almost Linear Time via Adaptive Preconditioning},
  booktitle = {Proceedings of the 51st {{Annual ACM SIGACT Symposium}} on {{Theory}} of {{Computing}}},
  author = {Kyng, Rasmus and Peng, Richard and Sachdeva, Sushant and Wang, Di},
  year = {2019},
  pages = {902--913}
}

@inproceedings{SW19,
  author    = {Thatchaphol Saranurak and
               Di Wang},
  editor    = {Timothy M. Chan},
  title     = {Expander Decomposition and Pruning: Faster, Stronger, and Simpler},
  booktitle = {Proceedings of the Thirtieth Annual {ACM-SIAM} Symposium on Discrete
               Algorithms, {SODA} 2019, San Diego, California, USA, January 6-9,
               2019},
  pages     = {2616--2635},
  publisher = {{SIAM}},
  year      = {2019},
  note      = {Available at:~\url{https://arxiv.org/abs/1812.08958}}
}

@inproceedings{NSW17,
  author    = {Danupon Nanongkai and
               Thatchaphol Saranurak and
               Christian Wulff{-}Nilsen},
  editor    = {Chris Umans},
  title     = {Dynamic Minimum Spanning Forest with Subpolynomial Worst-Case Update Time},
  booktitle = {58th {IEEE} Annual Symposium on Foundations of Computer Science, {FOCS}
               2017, Berkeley, CA, USA, October 15-17, 2017},
  pages     = {950--961},
  publisher = {{IEEE} Computer Society},
  year      = {2017},
  note      = {Available at:~\url{https://arxiv.org/abs/1708.03962}}
}

@inproceedings{guo2020fast,
  title={Fast algorithms for computational optimal transport and wasserstein barycenter},
  author={Guo, Wenshuo and Ho, Nhat and Jordan, Michael},
  booktitle={International Conference on Artificial Intelligence and Statistics},
  pages={2088--2097},
  year={2020},
  organization={PMLR}
}

@inproceedings{G08,
  title={The partial augment--relabel algorithm for the maximum flow problem},
  author={Goldberg, Andrew V},
  booktitle={European Symposium on Algorithms},
  pages={466--477},
  year={2008},
  organization={Springer}
}

@article{GG88,
  title={A computational comparison of the Dinic and network simplex methods for maximum flow},
  author={Goldfarb, Donald and Grigoriadis, Michael D},
  journal={Annals of Operations Research},
  volume={13},
  number={1},
  pages={81--123},
  year={1988},
  publisher={Springer}
}

@article{CH09,
  title={A computational study of the pseudoflow and push-relabel algorithms for the maximum flow problem},
  author={Chandran, Bala G and Hochbaum, Dorit S},
  journal={Operations research},
  volume={57},
  number={2},
  pages={358--376},
  year={2009},
  publisher={INFORMS}
}

@article{FHM10:arxiv,
  author    = {Barak Fishbain and
               Dorit S. Hochbaum and
               Stefan Muller},
  title     = {Competitive Analysis of Minimum-Cut Maximum Flow Algorithms in Vision
               Problems},
  journal   = {CoRR},
  volume    = {abs/1007.4531},
  year      = {2010},
  note      = {Available at:~\url{http://arxiv.org/abs/1007.4531}},
  eprinttype = {arXiv},
  eprint    = {1007.4531}
}

@inproceedings{GHKKTW15,
  author    = {Andrew V. Goldberg and
               Sagi Hed and
               Haim Kaplan and
               Pushmeet Kohli and
               Robert Endre Tarjan and
               Renato F. Werneck},
  editor    = {Nikhil Bansal and
               Irene Finocchi},
  title     = {Faster and More Dynamic Maximum Flow by Incremental Breadth-First
               Search},
  booktitle = {Algorithms - {ESA} 2015 - 23rd Annual European Symposium, Patras,
               Greece, September 14-16, 2015, Proceedings},
  series    = {Lecture Notes in Computer Science},
  volume    = {9294},
  pages     = {619--630},
  publisher = {Springer},
  year      = {2015},
  note      = {Available at:~\url{https://www.microsoft.com/en-us/research/wp-content/uploads/2016/11/ghkktw_ESA2015.pdf}}
}

@article{BK04,
  author    = {Yuri Boykov and
               Vladimir Kolmogorov},
  title     = {An Experimental Comparison of Min-Cut/Max-Flow Algorithms for Energy Minimization in Vision},
  journal   = {{IEEE} Trans. Pattern Anal. Mach. Intell.},
  volume    = {26},
  number    = {9},
  pages     = {1124--1137},
  year      = {2004},
  note      = {Available at:~\url{https://arxiv.org/abs/1202.3367}}
}

@article{G95,
  author    = {Andrew V. Goldberg},
  title     = {Scaling Algorithms for the Shortest Paths Problem},
  journal   = {{SIAM} J. Comput.},
  volume    = {24},
  number    = {3},
  pages     = {494--504},
  year      = {1995}
}

@inproceedings{KMP12,
  author    = {Jonathan A. Kelner and
               Gary L. Miller and
               Richard Peng},
  editor    = {Howard J. Karloff and
               Toniann Pitassi},
  title     = {Faster approximate multicommodity flow using quadratically coupled
               flows},
  booktitle = {Proceedings of the 44th Symposium on Theory of Computing Conference,
               {STOC} 2012, New York, NY, USA, May 19 - 22, 2012},
  pages     = {1--18},
  publisher = {{ACM}},
  year      = {2012},
  note = {Available at:~\url{https://arxiv.org/abs/1202.3367}}
}

@inbook{LP21,
author = {Li, Jason and Panigrahi, Debmalya},
title = {Approximate Gomory–Hu Tree is Faster than n – 1 Max-Flows},
year = {2021},
isbn = {9781450380539},
publisher = {Association for Computing Machinery},
address = {New York, NY, USA},
pages = {1738–1748},
numpages = {11},
note = {Available at:~\url{https://arxiv.org/abs/2111.02022}}
}

@article{AKLPST21:arxiv,
  author    = {Amir Abboud and
               Robert Krauthgamer and
               Jason Li and
               Debmalya Panigrahi and
               Thatchaphol Saranurak and
               Ohad Trabelsi},
  title     = {Gomory-Hu Tree in Subcubic Time},
  journal   = {CoRR},
  volume    = {abs/2111.04958},
  year      = {2021},
  note      = {Available at:~\url{https://arxiv.org/abs/2111.04958}},
  eprinttype = {arXiv},
  eprint    = {2111.04958}
}

@article{Z21:arxiv,
  author    = {Tianyi Zhang},
  title     = {Gomory-Hu Trees in Quadratic Time},
  journal   = {CoRR},
  volume    = {abs/2112.01042},
  year      = {2021},
  note      = {Available at:~\url{https://arxiv.org/abs/2112.01042}},
  eprinttype = {arXiv},
  eprint    = {2112.01042}
}

@article{pettie2022optimal,
  title={Optimal Vertex Connectivity Oracles},
  author={Pettie, Seth and Saranurak, Thatchaphol and Yin, Longhui},
  journal={Accepted to STOC'22},
  year={2022}
}

@inproceedings{LNPSY21,
author = {Li, Jason and Nanongkai, Danupon and Panigrahi, Debmalya and Saranurak, Thatchaphol and Yingchareonthawornchai, Sorrachai},
title = {Vertex Connectivity in Poly-Logarithmic Max-Flows},
year = {2021},
isbn = {9781450380539},
publisher = {Association for Computing Machinery},
address = {New York, NY, USA},
booktitle = {Proceedings of the 53rd Annual ACM SIGACT Symposium on Theory of Computing},
pages = {317–329},
numpages = {13},
location = {Virtual, Italy},
series = {STOC 2021},
note = {Available at~\url{https://arxiv.org/abs/2104.00104}}
}

@inproceedings{li2020deterministic,
  title={Deterministic Min-cut in Poly-logarithmic Max-flows},
  author={Li, Jason and Panigrahi, Debmalya},
  booktitle={2020 IEEE 61st Annual Symposium on Foundations of Computer Science (FOCS)},
  pages={85--92},
  year={2020},
  organization={IEEE}
}

@article{CLNPQS21:arxiv,
  author    = {Ruoxu Cen and
               Jason Li and
               Danupon Nanongkai and
               Debmalya Panigrahi and
               Kent Quanrud and
               Thatchaphol Saranurak},
  title     = {Minimum Cuts in Directed Graphs via Partial Sparsification},
  journal   = {CoRR},
  volume    = {abs/2111.08959},
  year      = {2021},
  url       = {https://arxiv.org/abs/2111.08959},
  eprinttype = {arXiv},
  eprint    = {2111.08959},
  bibsource = {dblp computer science bibliography, https://dblp.org}
}

@article{HRW20,
  author    = {Monika Henzinger and
               Satish Rao and
               Di Wang},
  title     = {Local Flow Partitioning for Faster Edge Connectivity},
  journal   = {{SIAM} J. Comput.},
  volume    = {49},
  number    = {1},
  pages     = {1--36},
  year      = {2020},
  note = {Available at:~\url{https://arxiv.org/abs/1704.01254}}
}

@inbook{C21,
author = {Chuzhoy, Julia},
title = {Decremental All-Pairs Shortest Paths in Deterministic near-Linear Time},
year = {2021},
isbn = {9781450380539},
publisher = {Association for Computing Machinery},
address = {New York, NY, USA},
booktitle = {Proceedings of the 53rd Annual ACM SIGACT Symposium on Theory of Computing},
pages = {626–639},
numpages = {14},
note = {Available at:~\url{https://arxiv.org/abs/2109.05621}}
}

@inproceedings{bernstein2020deterministic,
  title={Deterministic decremental reachability, SCC, and shortest paths via directed expanders and congestion balancing},
  author={Bernstein, Aaron and Gutenberg, Maximilian Probst and Saranurak, Thatchaphol},
  booktitle={2020 IEEE 61st Annual Symposium on Foundations of Computer Science (FOCS)},
  pages={1123--1134},
  year={2020},
  organization={IEEE}
}

@article{louis2010cut,
  title={Cut-matching games on directed graphs},
  author={Louis, Anand},
  journal={arXiv preprint arXiv:1010.1047},
  year={2010}
}

@article{khandekar2009graph,
  title={Graph partitioning using single commodity flows},
  author={Khandekar, Rohit and Rao, Satish and Vazirani, Umesh},
  journal={Journal of the ACM (JACM)},
  volume={56},
  number={4},
  pages={1--15},
  year={2009},
  publisher={ACM New York, NY, USA}
}

@inproceedings{CK19,
author = {Chuzhoy, Julia and Khanna, Sanjeev},
title = {A New Algorithm for Decremental Single-Source Shortest Paths with Applications to Vertex-Capacitated Flow and Cut Problems},
year = {2019},
isbn = {9781450367059},
publisher = {Association for Computing Machinery},
address = {New York, NY, USA},
booktitle = {Proceedings of the 51st Annual ACM SIGACT Symposium on Theory of Computing},
pages = {389–400},
numpages = {12},
location = {Phoenix, AZ, USA},
series = {STOC 2019},
note = {Available at~\url{https://arxiv.org/abs/1905.11512}}
}

@article{GN80,
  author    = {Zvi Galil and
               Amnon Naamad},
  title     = {An $O(EV\log^2V)$ Algorithm for the Maximal Flow Problem},
  journal   = {J. Comput. Syst. Sci.},
  volume    = {21},
  number    = {2},
  pages     = {203--217},
  year      = {1980},
  url       = {https://doi.org/10.1016/0022-0000(80)90035-5},
  doi       = {10.1016/0022-0000(80)90035-5},
  bibsource = {dblp computer science bibliography, https://dblp.org}
}

@inproceedings{M16,
  author    = {Aleksander M{\k{a}}dry},
  title     = {Computing Maximum Flow with Augmenting Electrical Flows},
  booktitle = {57th {IEEE} Annual Symposium on Foundations of Computer Science, {FOCS}
               2016, 9-11 October 2016, Hyatt Regency, New Brunswick, New Jersey,
               {USA}},
  pages     = {593--602},
  publisher = {{IEEE} Computer Society},
  year      = {2016},
  note      = {Available at~\url{https://arxiv.org/abs/1608.06016}}
}

@article{LS19:arxiv,
  author    = {Yin Tat Lee and
               Aaron Sidford},
  title     = {Solving Linear Programs with Sqrt(rank) Linear System Solves},
  journal   = {CoRR},
  volume    = {abs/1910.08033},
  year      = {2019},
  url       = {http://arxiv.org/abs/1910.08033},
  eprinttype = {arXiv},
  eprint    = {1910.08033},
  bibsource = {dblp computer science bibliography, https://dblp.org}
}

@inproceedings{KLS20,
  author    = {Tarun Kathuria and
               Yang P. Liu and
               Aaron Sidford},
  title     = {Unit Capacity Maxflow in Almost {$O(m^{4/3})$}
               Time},
  booktitle = {61st {IEEE} Annual Symposium on Foundations of Computer Science, {FOCS}
               2020, Durham, NC, USA, November 16-19, 2020},
  pages     = {119--130},
  publisher = {{IEEE}},
  year      = {2020},
  url       = {https://doi.org/10.1109/FOCS46700.2020.00020},
  doi       = {10.1109/FOCS46700.2020.00020},
  bibsource = {dblp computer science bibliography, https://dblp.org}
}

@inproceedings{CKMST11,
  author    = {Paul Christiano and
               Jonathan A. Kelner and
               Aleksander M{\k{a}}dry and
               Daniel A. Spielman and
               Shang{-}Hua Teng},
  title     = {Electrical flows, {L}aplacian systems, and faster approximation of maximum flow in undirected graphs},
  booktitle = {Proceedings of the 43rd {ACM} Symposium on Theory of Computing, {STOC}
               2011, San Jose, CA, USA, June 6-8 2011},
  pages     = {273--282},
  publisher = {{ACM}},
  year      = {2011},
  note = {Available at \url{https://arxiv.org/abs/1010.2921}}
}

@inproceedings{ST04,
  author    = {Daniel A. Spielman and
               Shang{-}Hua Teng},
  title     = {Nearly-linear time algorithms for graph partitioning, graph sparsification, and solving linear systems},
  booktitle = {Proceedings of the 36th Annual {ACM} Symposium on Theory of Computing, {STOC} 2004,
               Chicago, IL, USA, June 13-16, 2004},
  pages     = {81--90},
  year      = {2004},
  note      = {Available at \url{https://arxiv.org/abs/0809.3232}, 
  \url{https://arxiv.org/abs/0808.4134},
  \url{https://arxiv.org/abs/cs/0607105}}
}

@article{GR98,
	author    = {Andrew V. Goldberg and
	Satish Rao},
	title     = {Beyond the Flow Decomposition Barrier},
	journal   = {Journal of the ACM},
	volume    = {45},
	number    = {5},
	pages     = {783--797},
	year      = {1998},
	url       = {http://doi.acm.org/10.1145/290179.290181},
	doi       = {10.1145/290179.290181},
	bibsource = {dblp computer science bibliography, https://dblp.org},
	note = {Announced at FOCS'97}
}

@article{H08,
author = {Hochbaum, Dorit S.},
title = {The Pseudoflow Algorithm: A New Algorithm for the Maximum-Flow Problem},
journal = {Operations Research},
volume = {56},
number = {4},
pages = {992-1009},
year = {2008},
doi = {10.1287/opre.1080.0524}
}

@article{GT88,
  author    = {Andrew V. Goldberg and
               Robert Endre Tarjan},
  title     = {A new approach to the maximum-flow problem},
  journal   = {J. {ACM}},
  volume    = {35},
  number    = {4},
  pages     = {921--940},
  year      = {1988},
  url       = {https://doi.org/10.1145/48014.61051},
  doi       = {10.1145/48014.61051},
  bibsource = {dblp computer science bibliography, https://dblp.org}
}

@article{K73,
	title={On finding maximum flows in networks with special structure and some applications},
	author={Karzanov, Alexander V},
	journal={Matematicheskie Voprosy Upravleniya Proizvodstvom},
	volume={5},
	pages={81--94},
	year={1973}
}

@article{D70,
    author = {Dinic, E.A.},
    year = {1970},
    title = {Algorithm for solution of a problem of maximum flow in networks with power
estimation},
    journal = {Soviet Mathematics Doklady},
    volume = {11},
    pages = {1277-1280}
}

@article{D73,
    author = {Dinic, E.A.},
    year = {1973},
    title = {Metod porazryadnogo sokrashcheniya nevyazok i transportnye zadachi},
    journal = {Issledovaniya po Diskretno\v{i} Matematike},
    address = {Nauka, Moskva, CIS},
    note = {In Russian. Title translation: Excess scaling and transportation problems.}
}

@article{EK73,
    author = {Edmonds, Jack and Karp, Richard M.},
    title = {Theoretical Improvements in Algorithmic Efficiency for Network Flow Problems},
    year = {1972},
    issue_date = {April 1972},
    publisher = {Association for Computing Machinery},
    address = {New York, NY, USA},
    volume = {19},
    number = {2},
    issn = {0004-5411},
    url = {https://doi.org/10.1145/321694.321699},
    doi = {10.1145/321694.321699},
    journal = {Journal of the ACM},
    pages = {248-264},
    numpages = {17}
}

@book{FF54,
author="Ford, L. R. and D. R. Fulkerson",
title="Maximal Flow through a Network.",
address="Santa Monica, CA",
year="1954",
doi="",
publisher="RAND Corporation"
}

@article{GT14,
  title={Efficient maximum flow algorithms},
  author={Goldberg, Andrew V and Tarjan, Robert E},
  journal={Communications of the ACM},
  volume={57},
  number={8},
  pages={82--89},
  year={2014},
  publisher={ACM New York, NY, USA},
  note = {Available at~\url{https://cacm.acm.org/magazines/2014/8/177011-efficient-maximum-flow-algorithms}}
}

@inproceedings{K84,
  author    = {Narendra Karmarkar},
  title     = {A New Polynomial-Time Algorithm for Linear Programming},
  booktitle = {{STOC}},
  pages     = {302--311},
  publisher = {{ACM}},
  year      = {1984}
}

@inproceedings{BLNPSSW20,
  title={Bipartite matching in nearly-linear time on moderately dense graphs},
  author={Brand, Jan van den and Lee, Yin-Tat and Nanongkai, Danupon and Peng, Richard and Saranurak, Thatchaphol and Sidford, Aaron and Song, Zhao and Wang, Di},
  booktitle={2020 IEEE 61st Annual Symposium on Foundations of Computer Science (FOCS)},
  pages={919--930},
  year={2020},
  organization={IEEE}
}

@inproceedings{CS21,
  title={Deterministic algorithms for decremental shortest paths via layered core decomposition},
  author={Chuzhoy, Julia and Saranurak, Thatchaphol},
  booktitle={Proceedings of the 2021 ACM-SIAM Symposium on Discrete Algorithms (SODA)},
  pages={2478--2496},
  year={2021},
  organization={SIAM}
}

@article{AN19:journal,
  title={Using petal-decompositions to build a low stretch spanning tree},
  author={Abraham, Ittai and Neiman, Ofer},
  journal={SIAM Journal on Computing},
  volume={48},
  number={2},
  pages={227--248},
  year={2019},
  publisher={SIAM}
}

@article{AGOT92,
  author    = {Ravindra K. Ahuja and
               Andrew V. Goldberg and
               James B. Orlin and
               Robert Endre Tarjan},
  title     = {Finding minimum-cost flows by double scaling},
  journal   = {Math. Program.},
  volume    = {53},
  pages     = {243--266},
  year      = {1992},
  url       = {https://doi.org/10.1007/BF01585705},
  doi       = {10.1007/BF01585705},
  bibsource = {dblp computer science bibliography, https://dblp.org}
}

@article{G85,
  author    = {Harold N. Gabow},
  title     = {Scaling Algorithms for Network Problems},
  journal   = {J. Comput. Syst. Sci.},
  volume    = {31},
  number    = {2},
  pages     = {148--168},
  year      = {1985},
  url       = {https://doi.org/10.1016/0022-0000(85)90039-X},
  doi       = {10.1016/0022-0000(85)90039-X},
  bibsource = {dblp computer science bibliography, https://dblp.org}
}

@article{GT89,
  author    = {Harold N. Gabow and
               Robert Endre Tarjan},
  title     = {Faster Scaling Algorithms for Network Problems},
  journal   = {{SIAM} J. Comput.},
  volume    = {18},
  number    = {5},
  pages     = {1013--1036},
  year      = {1989},
  url       = {https://doi.org/10.1137/0218069},
  doi       = {10.1137/0218069},
  bibsource = {dblp computer science bibliography, https://dblp.org}
}

@article{DS08,
  author    = {Samuel I. Daitch and
               Daniel A. Spielman},
  title     = {Faster Approximate Lossy Generalized Flow via Interior Point Algorithms},
  journal   = {CoRR},
  volume    = {abs/0803.0988},
  year      = {2008},
  url       = {http://arxiv.org/abs/0803.0988},
  eprinttype = {arXiv},
  eprint    = {0803.0988},
  bibsource = {dblp computer science bibliography, https://dblp.org}
}

@article{GLP21:arxiv,
  title={Fully dynamic electrical flows: sparse maxflow faster than Goldberg-Rao},
  author={Gao, Yu and Liu, Yang P and Peng, Richard},
  journal={FOCS 2021},
  year={2021},
  note={Available at \url{https://arxiv.org/abs/2101.07233}}
}

@inproceedings{DLY21,
  author    = {Sally Dong and
               Yin Tat Lee and
               Guanghao Ye},
  title     = {A nearly-linear time algorithm for linear programs with small treewidth:
               a multiscale representation of robust central path},
  booktitle = {{STOC}},
  pages     = {1784--1797},
  publisher = {{ACM}},
  year      = {2021}
}

@article{BGJLLPS21:arxiv,
  author    = {Brand, Jan van den and
               Yu Gao and
               Arun Jambulapati and
               Yin Tat Lee and
               Yang P. Liu and
               Richard Peng and
               Aaron Sidford},
  title     = {Faster Maxflow via Improved Dynamic Spectral Vertex Sparsifiers},
  journal   = {CoRR},
  volume    = {abs/2112.00722},
  year      = {2021},
  url       = {https://arxiv.org/abs/2112.00722},
  eprinttype = {arXiv},
  eprint    = {2112.00722},
  bibsource = {dblp computer science bibliography, https://dblp.org}
}

@inproceedings{BLSS20,
  author    = {Brand, Jan van den and
               Yin Tat Lee and
               Aaron Sidford and
               Zhao Song},
  title     = {Solving tall dense linear programs in nearly linear time},
  booktitle = {Proccedings of the 52nd Annual {ACM} {SIGACT} Symposium on Theory
               of Computing, {STOC} 2020, Chicago, IL, USA, June 22-26, 2020},
  pages     = {775--788},
  publisher = {{ACM}},
  year      = {2020},
  note      = {Available at~\url{https://arxiv.org/abs/2002.02304}}
}

@inproceedings{CLS19,
  author    = {Michael B. Cohen and
               Yin Tat Lee and
               Zhao Song},
  title     = {Solving linear programs in the current matrix multiplication time},
  booktitle = {Proceedings of the 51st Annual {ACM} {SIGACT} Symposium on Theory
               of Computing, {STOC} 2019, Phoenix, AZ, USA, June 23-26, 2019},
  pages     = {938--942},
  publisher = {{ACM}},
  year      = {2019},
  note = {Available at~\url{https://arxiv.org/abs/1810.07896}}
}

@article{AMV21:arxiv,
  author    = {Kyriakos Axiotis and
               Aleksander M{\k{a}}dry and
               Adrian Vladu},
  title     = {Faster Sparse Minimum Cost Flow by Electrical Flow Localization},
  journal   = {CoRR},
  volume    = {abs/2111.10368},
  year      = {2021},
  url       = {https://arxiv.org/abs/2111.10368},
  eprinttype = {arXiv},
  eprint    = {2111.10368},
  bibsource = {dblp computer science bibliography, https://dblp.org}
}

@inproceedings{BLLSSSW21,
  author    = { Brand, Jan van den and
               Yin Tat Lee and
               Yang P. Liu and
               Thatchaphol Saranurak and
               Aaron Sidford and
               Zhao Song and
               Di Wang},
  title     = {Minimum cost flows, MDPs, and $\ell_1$-regression
               in nearly linear time for dense instances},
  booktitle = {{STOC}},
  pages     = {859--869},
  publisher = {{ACM}},
  year      = {2021}
}

@inproceedings{LSZ19,
  author    = {Yin Tat Lee and
               Zhao Song and
               Qiuyi Zhang},
  title     = {Solving Empirical Risk Minimization in the Current Matrix Multiplication Time},
  booktitle = {Conference on Learning Theory, {COLT} 2019, Phoenix,
               AZ, {USA}, June 25-28, 2019},
  series    = {Proceedings of Machine Learning Research},
  volume    = {99},
  pages     = {2140--2157},
  publisher = {{PMLR}},
  year      = {2019},
  note      = {Available at~\url{https://arxiv.org/abs/1905.04447}}
}

@article{CPW21:arxiv,
  title={$\ell_2$-norm Flow Diffusion in Near-Linear Time},
  author={Chen, Li and Peng, Richard and Wang, Di},
  journal={arXiv preprint arXiv:2105.14629},
  year={2021},
  note = {Available at~\url{https://arxiv.org/abs/2105.14629}}
}

@inproceedings{CGHPS20,
  title={Fast dynamic cuts, distances and effective resistances via vertex sparsifiers},
  author={Chen, Li and Goranci, Gramoz and Henzinger, Monika and Peng, Richard and Saranurak, Thatchaphol},
  booktitle={2020 IEEE 61st Annual Symposium on Foundations of Computer Science (FOCS)},
  pages={1135--1146},
  year={2020},
  organization={IEEE}
}

@inproceedings{ST03,
  title={Solving sparse, symmetric, diagonally-dominant linear systems in time $O(m^{1.31})$},
  author={Spielman, Daniel A and Teng, Shang-Hua},
  booktitle={44th Annual IEEE Symposium on Foundations of Computer Science, 2003. Proceedings.},
  pages={416--427},
  year={2003},
  organization={IEEE}
}

@article{ST83,
  title={A data structure for dynamic trees},
  author={Sleator, Daniel D and Tarjan, Robert Endre},
  journal={Journal of computer and system sciences},
  volume={26},
  number={3},
  pages={362--391},
  year={1983},
  publisher={Elsevier}
}

@inproceedings{KLOS14,
  author    = {Jonathan A. Kelner and
               Yin Tat Lee and
               Lorenzo Orecchia and
               Aaron Sidford},
  editor    = {Chandra Chekuri},
  title     = {An Almost-Linear-Time Algorithm for Approximate Max Flow in Undirected
               Graphs, and its Multicommodity Generalizations},
  booktitle = {Proceedings of the Twenty-Fifth Annual {ACM-SIAM} Symposium on Discrete
               Algorithms, {SODA} 2014, Portland, Oregon, USA, January 5-7, 2014},
  pages     = {217--226},
  publisher = {{SIAM}},
  year      = {2014},
  url       = {https://doi.org/10.1137/1.9781611973402.16},
  doi       = {10.1137/1.9781611973402.16},
  bibsource = {dblp computer science bibliography, https://dblp.org}
}

@inproceedings{R08,
  author    = {Harald R{\"{a}}cke},
  editor    = {Cynthia Dwork},
  title     = {Optimal hierarchical decompositions for congestion minimization in
               networks},
  booktitle = {Proceedings of the 40th Annual {ACM} Symposium on Theory of Computing,
               Victoria, British Columbia, Canada, May 17-20, 2008},
  pages     = {255--264},
  publisher = {{ACM}},
  year      = {2008},
  url       = {https://doi.org/10.1145/1374376.1374415},
  doi       = {10.1145/1374376.1374415},
  bibsource = {dblp computer science bibliography, https://dblp.org}
}

@inproceedings{M10,
  author    = {Aleksander M{\k{a}}dry},
  title     = {Fast Approximation Algorithms for Cut-Based Problems in Undirected
               Graphs},
  booktitle = {51th Annual {IEEE} Symposium on Foundations of Computer Science, {FOCS}
               2010, October 23-26, 2010, Las Vegas, Nevada, {USA}},
  pages     = {245--254},
  publisher = {{IEEE} Computer Society},
  year      = {2010},
  url       = {https://doi.org/10.1109/FOCS.2010.30},
  doi       = {10.1109/FOCS.2010.30},
  bibsource = {dblp computer science bibliography, https://dblp.org}
}

@inproceedings{S13,
  author    = {Jonah Sherman},
  title     = {Nearly Maximum Flows in Nearly Linear Time},
  booktitle = {54th Annual {IEEE} Symposium on Foundations of Computer Science, {FOCS}
               2013, 26-29 October, 2013, Berkeley, CA, {USA}},
  pages     = {263--269},
  publisher = {{IEEE} Computer Society},
  year      = {2013},
  url       = {https://doi.org/10.1109/FOCS.2013.36},
  doi       = {10.1109/FOCS.2013.36},
  bibsource = {dblp computer science bibliography, https://dblp.org}
}

@article{D51,
  title={Application of the simplex method to a transportation problem},
  author={Dantzig, George B},
  journal={Activity analysis and production and allocation},
  year={1951},
  publisher={Wiley}
}

@inproceedings{GT87,
  title={Solving minimum-cost flow problems by successive approximation},
  author={Goldberg, Andrew and Tarjan, Robert},
  booktitle={Proceedings of the nineteenth annual ACM symposium on Theory of computing},
  pages={7--18},
  year={1987}
}

@inproceedings{M13,
  title={Navigating central path with electrical flows: From flows to matchings, and back},
  author={M{\k{a}}dry, Aleksander},
  booktitle={2013 IEEE 54th Annual Symposium on Foundations of Computer Science},
  pages={253--262},
  year={2013},
  organization={IEEE}
}

@inproceedings{P16,
  title={Approximate undirected maximum flows in $O(m \mathrm{polylog}(n))$ time},
  author={Peng, Richard},
  booktitle={Proceedings of the twenty-seventh annual ACM-SIAM symposium on Discrete algorithms},
  pages={1862--1867},
  year={2016},
  organization={SIAM}
}

@inproceedings{S17,
  title={Area-convexity, $\ell_{\infty}$ regularization, and undirected multicommodity flow},
  author={Sherman, Jonah},
  booktitle={Proceedings of the 49th Annual ACM SIGACT Symposium on Theory of Computing},
  pages={452--460},
  year={2017}
}

@inproceedings{LS20,
  title={Faster energy maximization for faster maximum flow},
  author={Liu, Yang P and Sidford, Aaron},
  booktitle={Proceedings of the 52nd Annual ACM SIGACT Symposium on Theory of Computing},
  pages={803--814},
  year={2020}
}

@article{bernstein2020fully,
  title={Fully-dynamic graph sparsifiers against an adaptive adversary},
  author={Bernstein, Aaron and Brand, Jan van den and Gutenberg, Maximilian Probst and Nanongkai, Danupon and Saranurak, Thatchaphol and Sidford, Aaron and Sun, He},
  journal={arXiv preprint arXiv:2004.08432},
  year={2020}
}

@article{leighton1999multicommodity,
  title={Multicommodity max-flow min-cut theorems and their use in designing approximation algorithms},
  author={Leighton, Tom and Rao, Satish},
  journal={Journal of the ACM (JACM)},
  volume={46},
  number={6},
  pages={787--832},
  year={1999},
  publisher={ACM New York, NY, USA}
}

@article{BGS21,
  title={Deterministic decremental sssp and approximate min-cost flow in almost-linear time},
  author={Bernstein, Aaron and Gutenberg, Maximilian Probst and Saranurak, Thatchaphol},
  journal={arXiv preprint arXiv:2101.07149},
  year={2021}
}

@inproceedings{DGGLPSY22,
author = {Sally Dong and Yu Gao and Gramoz Goranci and Yin Tat Lee and Richard Peng and Sushant Sachdeva and Guanghao Ye},
title = {Nested Dissection Meets IPMs: Planar Min-Cost Flow in Nearly-Linear Time},
booktitle = {Proceedings of the 2022 Annual ACM-SIAM Symposium on Discrete Algorithms (SODA)},
chapter = {},
pages = {124-153},
doi = {10.1137/1.9781611977073.7},
URL = {https://epubs.siam.org/doi/abs/10.1137/1.9781611977073.7},
eprint = {https://epubs.siam.org/doi/pdf/10.1137/1.9781611977073.7},
year = {2022}
}

@article{LV21,
  title={Tutorial on the Robust Interior Point Method},
  author={Lee, Yin Tat and Vempala, Santosh S},
  journal={arXiv preprint arXiv:2108.04734},
  year={2021}
}

@article{R88,
  title={A polynomial-time algorithm, based on Newton's method, for linear programming},
  author={Renegar, James},
  journal={Mathematical programming},
  volume={40},
  number={1},
  pages={59--93},
  year={1988},
  publisher={Springer}
}

@article{V90,
  author    = {Pravin M. Vaidya},
  title     = {An Algorithm for Linear Programming which Requires $O(((m+n)n^2 + (m+n)^{1.5}n)L)$ Arithmetic Operations},
  journal   = {Math. Program.},
  volume    = {47},
  pages     = {175--201},
  year      = {1990}
}

@inproceedings{ALOW17,
  author    = {Zeyuan Allen{-}Zhu and
               Yuanzhi Li and
               Rafael Mendes de Oliveira and
               Avi Wigderson},
  title     = {Much Faster Algorithms for Matrix Scaling},
  booktitle = {{FOCS}},
  pages     = {890--901},
  publisher = {{IEEE} Computer Society},
  year      = {2017}
}

@inproceedings{CMTV17,
  author    = {Michael B. Cohen and
               Aleksander M{\k{a}}dry and
               Dimitris Tsipras and
               Adrian Vladu},
  title     = {Matrix Scaling and Balancing via Box Constrained Newton's Method and
               Interior Point Methods},
  booktitle = {{FOCS}},
  pages     = {902--913},
  publisher = {{IEEE} Computer Society},
  year      = {2017}
}

@article{Nem04,
  title={Interior point polynomial time methods in convex programming},
  author={Nemirovski, Arkadi},
  journal={Lecture notes},
  volume={42},
  number={16},
  pages={3215--3224},
  year={2004},
  publisher={Citeseer},
  note={Available at \url{https://www2.isye.gatech.edu/~nemirovs/Lect_IPM.pdf}}
}

@article{KRS15,
  title={Fast, provable algorithms for isotonic regression in all l\_p-norms},
  author={Kyng, Rasmus and Rao, Anup and Sachdeva, Sushant},
  journal={Advances in neural information processing systems},
  volume={28},
  year={2015}
}

@article{HQ03,
author = {Hochbaum, D.S. and Queyranne, M.},
title = {Minimizing a Convex Cost Closure Set},
journal = {SIAM Journal on Discrete Mathematics},
volume = {16},
number = {2},
pages = {192-207},
year = {2003}
}

@article{Stout13,
year={2013},
issn={0178-4617},
journal={Algorithmica},
volume={66},
number={1},
doi={10.1007/s00453-012-9628-4},
title={{Isotonic} {Regression} via Partitioning},
url={http://dx.doi.org/10.1007/s00453-012-9628-4},
publisher={Springer-Verlag},
author={Stout, Q. F.},
pages={93-112},
language={English}
}

@article{Stout21,
  title={$\ell_{p}$ Isotonic Regression Algorithms Using an $\ell_0$ Approach},
  author={Quentin F. Stout},
  journal={ArXiv},
  year={2021},
  volume={abs/2107.00251}
}

@article{OG21,
  title={A fast maximum flow algorithm},
  author={Orlin, James B and Gong, Xiao-yue},
  journal={Networks},
  volume={77},
  number={2},
  pages={287--321},
  year={2021},
  publisher={Wiley Online Library}
}

@inproceedings{AS20,
  title={Faster p-norm minimizing flows, via smoothed q-norm problems},
  author={Adil, Deeksha and Sachdeva, Sushant},
  booktitle={Proceedings of the Fourteenth Annual ACM-SIAM Symposium on Discrete Algorithms},
  pages={892--910},
  year={2020},
  organization={SIAM}
}

@article{BNW22,
  author    = {Aaron Bernstein and
               Danupon Nanongkai and
               Christian Wulff{-}Nilsen},
  title     = {Negative-Weight Single-Source Shortest Paths in Near-linear Time},
  journal   = {CoRR},
  volume    = {abs/2203.03456},
  year      = {2022},
  url       = {https://doi.org/10.48550/arXiv.2203.03456},
  doi       = {10.48550/arXiv.2203.03456},
  eprinttype = {arXiv},
  eprint    = {2203.03456},
  bibsource = {dblp computer science bibliography, https://dblp.org}
}

@article{WZ92,
  title = {A Combinatorial Interior Point Method for Network Flow Problems},
  author = {Wallacher, C. and Zimmermann, U.},
  year = {1992},
  month = aug,
  journal = {Mathematical Programming},
  volume = {56},
  number = {1-3},
  pages = {321--335},
  issn = {0025-5610, 1436-4646},
  doi = {10.1007/BF01580905},
  langid = {english},
}

@article{ET75,
  title = {Network Flow and Testing Graph Connectivity},
  author = {Even, Shimon and Tarjan, R. Endre},
  year = {1975},
  journal = {SIAM journal on computing},
  volume = {4},
  number = {4},
  pages = {507--518},
  publisher = {{SIAM}}
}

@article{HK73,
  title = {An \$n\^\{5/2\} \$ {{Algorithm}} for {{Maximum Matchings}} in {{Bipartite Graphs}}},
  author = {Hopcroft, John E. and Karp, Richard M.},
  year = {1973},
  month = dec,
  journal = {SIAM Journal on Computing},
  volume = {2},
  number = {4},
  pages = {225--231},
  publisher = {{Society for Industrial and Applied Mathematics}},
  issn = {0097-5397},
  doi = {10.1137/0202019}
}

@inproceedings{WFHMR17,
  title={Capacity releasing diffusion for speed and locality},
  author={Wang, Di and Fountoulakis, Kimon and Henzinger, Monika and Mahoney, Michael W and Rao, Satish},
  booktitle={International Conference on Machine Learning},
  pages={3598--3607},
  year={2017},
  organization={PMLR}
}

@inproceedings{FWY20,
  author    = {Kimon Fountoulakis and
               Di Wang and
               Shenghao Yang},
  title     = {p-Norm Flow Diffusion for Local Graph Clustering},
  booktitle = {Proceedings of the 37th International Conference on Machine Learning,
               {ICML} 2020, 13-18 July 2020, Virtual Event},
  series    = {Proceedings of Machine Learning Research},
  volume    = {119},
  pages     = {3222--3232},
  publisher = {{PMLR}},
  year      = {2020}
  }

@inproceedings{AMV20,
  title={Circulation control for faster minimum cost flow in unit-capacity graphs},
  author={Axiotis, Kyriakos and M{\k{a}}dry, Aleksander and Vladu, Adrian},
  booktitle={2020 IEEE 61st Annual Symposium on Foundations of Computer Science (FOCS)},
  pages={93--104},
  year={2020},
  organization={IEEE}
}
